\documentclass[12pt, reqno, a4paper]{amsart}

\pdfoutput=1

\usepackage[margin=1in,footskip=0.35in]{geometry}

\usepackage{amsmath}
\usepackage{etex}

\usepackage[UKenglish]{babel}

\usepackage{amssymb,amsthm}
\usepackage[shortlabels]{enumitem}
\usepackage{array}
\usepackage{tikz-cd}
\usepackage{mathtools}
\usepackage{graphicx}
\usepackage{mathrsfs}
\usepackage{stmaryrd}
\usepackage[scr=euler,cal=dutchcal]{mathalfa}
\usepackage{stackengine}
\usepackage[hyphens]{url}

\usepackage{bussproofs}
\usepackage{todonotes}

\usepackage{bm}
\usepackage{bbm}
\usepackage{thmtools}
\usepackage{thm-restate}

\DeclareSymbolFont{lettersA}{U}{txmia}{m}{it}
\DeclareMathSymbol{\sigmaup}{\mathord}{lettersA}{27}
\DeclareMathSymbol{\tauup}{\mathord}{lettersA}{28}

\usepackage[%
  default,
  full,
  fullproof,	% enables display of the 'fullproof' proof environment
]{optional}

\usepackage[
    pagebackref,
    breaklinks=true,
    colorlinks = true,
    urlcolor = black,
    linkcolor= black,
    citecolor= black,
    filecolor= black,
]{hyperref}

\renewcommand*{\backref}[1]{}
\renewcommand*{\backrefalt}[4]{{
    \ifcase #1 Not cited.%
          \or Cited on page~#2.%
          \else Cited on pages #2.%
    \fi%
    }}

\makeatletter

\renewcommand{\tocsection}[3]{%
  \indentlabel{\@ifnotempty{#2}{\bfseries\ignorespaces#1 #2\quad}}\bfseries#3}
\renewcommand{\tocsubsection}[3]{%
  \indentlabel{\@ifnotempty{#2}{\ignorespaces#1 #2\quad}}#3}
\def\@dotsep{4.5}
\def\@tocline#1#2#3#4#5#6#7{\relax
  \ifnum #1>\c@tocdepth % then omit
  \else
    \par \addpenalty\@secpenalty\addvspace{#2}%
    \begingroup \hyphenpenalty\@M
    \@ifempty{#4}{%
      \@tempdima\csname r@tocindent\number#1\endcsname\relax
    }{%
      \@tempdima#4\relax
    }%
    \parindent\z@ \leftskip#3\relax \advance\leftskip\@tempdima\relax
    \rightskip\@pnumwidth plus1em \parfillskip-\@pnumwidth
    #5\leavevmode\hskip-\@tempdima{#6}\nobreak
    \leaders\hbox{$\m@th\mkern \@dotsep mu\hbox{.}\mkern \@dotsep mu$}\hfill
    \nobreak
    \hbox to\@pnumwidth{\@tocpagenum{\ifnum#1=1\fi#7}}\par
    \nobreak
    \endgroup
  \fi}
\AtBeginDocument{%
\expandafter\renewcommand\csname r@tocindent0\endcsname{0pt}
}
\def\l@subsection{\@tocline{2}{0pt}{2.5pc}{5pc}{}}
\makeatother

\usepackage{macros}

\theoremstyle{plain}
\newtheorem{theorem}{Theorem}[section]
\newtheorem{lemma}[theorem]{Lemma}
\newtheorem{proposition}[theorem]{Proposition}
\newtheorem{corollary}[theorem]{Corollary}
\newtheorem{claim}[theorem]{Claim}
\newtheorem*{claim*}{Claim}
\newtheorem{fact}[theorem]{Fact}
\newtheorem*{fact*}{Fact} 

\theoremstyle{definition}
\newtheorem{definition}[theorem]{Definition}
\newtheorem{remark}[theorem]{Remark}
\newtheorem{notation}[theorem]{Notation}
\newtheorem{example}[theorem]{Example}

\theoremstyle{definition}

\newcounter{SplitEnum}
\newcounter{StartEnum}

\usepackage{comment}
\opt{fullproof}{\newcommand\FLAGPROOF{}}
\ifdefined\FLAGPROOF 
  \newenvironment{fullproof}{\begin{proof}}{\end{proof}}
\else
  \excludecomment{fullproof}
\fi

\begin{document}

\title{Finitely accessible arboreal adjunctions and Hintikka formulae}

\author{Luca Reggio}
\address{Dipartimento di Matematica ``Federigo Enriques'',
  Universit\`a degli Studi di Milano, Italy}
\email{luca.reggio@unimi.it}

\author{Colin Riba}
\address{LIP\\ ENS de Lyon, France}
\email{colin.riba@ens-lyon.fr}

% make the title area
\maketitle
\thispagestyle{plain}
\pagestyle{plain}

%%%%%%%%%%%%%%%%%%%%%%%%%%%%%%%%%%%%%%%%%%%%%%%%%%%%%%%%%%%%%%%%%%%%%%%%%%%
\begin{abstract}
%%%%%%%%%%%%%%%%%%%%%%%%%%%%%%%%%%%%%%%%%%%%%%%%%%%%%%%%%%%%%%%%%%%%%%%%%%%
Arboreal categories provide an axiomatic framework 
in which abstract notions of bisimilarity and back-and-forth games can be defined.
They act on extensional categories, typically consisting of relational structures, via
arboreal adjunctions.
In many cases,
equivalence of structures in fragments of infinitary first-order logic
can be captured by transferring the bisimilarity relation along the adjunction.

In most applications, the categories involved are locally finitely presentable
and the adjunctions are finitely accessible.
Our main result identifies the expressive power of this class of adjunctions.
We show that the ranks of back-and-forth games in the arboreal category
are definable by formulae \emph{\`a la} Hintikka,
and thus the relation between extensional objects induced by bisimilarity
is always coarser than equivalence in
infinitary first-order logic. 

Our approach leverages Gabriel-Ulmer duality for
locally finitely presentable categories,
and Hodges' \mbox{word-constructions}.
\end{abstract}

\setcounter{tocdepth}{1}
\tableofcontents
\setcounter{tocdepth}{3}

%%%%%%%%%%%%%%%%%%%%%%%%%%%%%%%%%%%%%%%%%%%%%%%%%%%%%%%%%%%%%%%%%%%%%%%%%%%
\section{Introduction}
\label{sec:intro}
%%%%%%%%%%%%%%%%%%%%%%%%%%%%%%%%%%%%%%%%%%%%%%%%%%%%%%%%%%%%%%%%%%%%%%%%%%%

Model comparison games, such as Ehrenfeucht-Fraïssé and pebble games, are an
important tool in model theory, as they provide semantic characterisations of
equivalence of models in given logic fragments, see e.g.~\cite{hodges93book,marker02book}.
They are one of the few model-theoretic methods
that restrict to finite structures and are thus central to finite model theory;
see e.g.~\cite{ef99book,libkin04book}.
The present work is concerned with a class of categories
(which extends the class of \emph{arboreal categories} of~\cite{ar21icalp,ar21arboreal})
in which abstract games can be defined,
generalising the games mentioned above and many others.

We start with a slight reformulation of a well-known example.
Let $\FG{\vec x \mid \varphi}$ be a finite structure
consisting of generators $\vec x = x_1,\dots,x_n$
subject to a condition $\varphi = \varphi(\vec x)$
expressed as a conjunction
of equality constraints $x_i \Eq x_j$ and order constraints $x_i \Lt x_j$.
Consider 
ordered structures $M$ and $N$, and homomorphisms $m,n$ as shown below.
\[
\begin{tikzcd}[column sep=huge]
& \FG{\vec x \mid \varphi}
  \arrow[rightarrowtail, bend right=20]{dl}[above, xshift=-5pt]{m}
  \arrow[rightarrowtail, bend left=20]{dr}{n}
\\
  M
& 
& N
\end{tikzcd}
\]

\noindent
The arrows $\emb$ denote order embeddings.
This situation can be regarded as a position in the Ehrenfeucht-Fraïssé game
played by \emph{Spoiler} and \emph{Duplicator} on~$M$ and~$N$.
In the next round, Spoiler moves by choosing one of the embeddings $m$ and $n$ (say $m$),
and by extending it as shown by the solid arrows below.
\[
\begin{tikzcd}[column sep=huge]
& \FG{\vec x \mid \varphi}
  \arrow[rightarrowtail, bend right=20]{dl}[above, xshift=-5pt]{m}
  \arrow[rightarrowtail, bend left=20]{dr}{n}
  \arrow[rightarrowtail]{d}
\\
  M
& \FG{\vec x, x' \mid \varphi \land \varphi'}
  \arrow[rightarrowtail]{l}[below]{\text{(Spoiler)}}
  \arrow[rightarrowtail, dotted]{r}[below]{\text{(Duplicator)}}
& N
\end{tikzcd}
\]

\noindent
Duplicator has to respond by extending the other embedding (in this case $n$)
as indicated by the dotted arrow above.
The game then continues from the new position, and Duplicator
wins if they can respond to all of Spoiler's moves.
For instance, Duplicator wins 
if $M$ and $N$ are dense linear orders without endpoints
(such as $\QQ$ and $\RR$), see e.g.~\cite[\S 14.2]{rosenstein82book}.
A well-known result of Karp~\cite{karp65stm} then implies that $M$ and $N$
satisfy the same formulae of infinitary first-order logic
$\Lang_{\infty} = \Lang_{\infty,\omega}$
where infinite conjunctions and disjunctions (of any arity) are allowed.

An advantage of the game perspective is that parameters such as the number of rounds in
Ehrenfeucht-Fraïssé games,
or the number of pebbles in pebble games, naturally correspond to complexity measures
of formulae such as the quantifier rank, or the number of variables.
Stratifying
formulae in terms of the number of logical resources they use
is pervasive in finite model theory and is at the heart of 
descriptive complexity, cf e.g.\ \cite{Immerman1999}.

An insight of the \emph{comonadic} approach to model comparison games,
put forward in \cite{adw17lics} and developed in~\cite{as18csl,as21jlc} by
Abramsky, Dawar and their collaborators,
is that the set of all plays of a game in a given relational
structure forms itself a structure of the same type.
Furthermore, this assignment yields an endofunctor on the category of structures 
and their homomorphisms, which is in fact
a comonad, called a \emph{game comonad}.
This applies to various games,
such as Ehrenfeucht-Fraïssé, pebble, and modal bisimulation games.
Game comonads, and their axiomatic extension in the form of \emph{arboreal categories}
and \emph{arboreal adjunctions}~\cite{ar21icalp,ar21arboreal},
provide a purely categorical view on many model comparison games,
as well as on
the corresponding logical resources and combinatorial parameters of structures.
In particular, arboreal categories allow abstract notions of games to be defined,
and equivalence of objects in these games can be captured by a notion of
\emph{bisimilarity} closely related to Joyal, Nielsen and Winskel's open-map
bisimilarity~\cite{jnw93lics,jnw96ic}.
For a recent survey on game comonads and arboreal categories, see~\cite{AR2024}.

An arboreal category $\A$ comes equipped with a factorisation system
(consisting of \emph{quotients} and \emph{embeddings})
which induces a notion of back-and-forth game, and thus of back-and-forth equivalence, between objects of $\A$.
The arboreal category $\A$ typically acts on a category $\E$,
called the \emph{extensional category} (e.g.\ a category of relational structures),
via an \emph{arboreal adjunction}, i.e.\ 
an adjunction of the form
\begin{equation*}
\begin{tikzcd}
  \A
  \arrow[bend left]{r}[above]{\Ladj}
  \arrow[phantom]{r}[description]{\textnormal{\footnotesize{$\bot$}}}
& \E.
  \arrow[bend left]{l}[below]{R}
\end{tikzcd}
\end{equation*}

\noindent
In particular we can transfer the back-and-forth equivalence relation along the
adjunction to the category~$\E$:
two objects $M,N\in \E$ are equivalent in the transferred relation if, and only if,
$RM$ and $RN$ are back-and-forth equivalent in $\A$.
In this way, we recover equivalence of structures in a number of logic fragments.

In many concrete examples, the logics involved are fragments of $\Lang_\infty$,
and the arboreal adjunctions satisfy the following additional
properties:
\begin{enumerate}[(i)]
\item
the categories $\E$ and $\A$ are locally finitely presentable (in the sense of, e.g., \cite{ar94book}),

\item
the right adjoint $R$ is finitary, i.e.\ it preserves filtered colimits.
\end{enumerate}

\noindent
An adjunction between locally finitely presentable categories where the right adjoint preserves filtered colimits is customarily called a \emph{finitely accessible} adjunction, hence we refer to arboreal adjunctions $\Ladj \colon \A \inadj \E \cocolon R$ satisfying the two conditions above as \emph{finitely accessible arboreal adjunctions}.
For example, the arboreal adjunction associated with the Ehrenfeucht-Fraïssé comonad
is finitely accessible, and the induced back-and-forth equivalence relation captures equivalence of structures in the full logic~$\Lang_{\infty}$.

Our main result yields a converse to this for finitely accessible arboreal adjunctions.
Under mild additional assumptions, we show that
if $M,N \in \E$ are $\Lang_\infty$-equivalent,
then $R M$ and $R N$ are back-and-forth equivalent in $\A$.
Hence, the arboreal back-and-forth equivalence relations
cannot distinguish between $\Lang_{\infty}$-equivalent models.
This identifies a tight upper bound for the expressive power of finitely accessible arboreal adjunctions.

Our results apply to what we will call \emph{wooded categories}, obtained by weakening the notion of arboreal category.
The axioms for an arboreal category imply that back-and-forth equivalence in an arboreal category $\A$ coincides with bisimilarity via (a suitable notion of) open maps.
Moreover, there is a functor $\Path \colon \A \to \Tree$, into the category $\Tree$ of trees and height-preserving monotone maps, which allows us to regard many constructions in $\A$ as ``concrete over trees''.
While these properties support the arboreal notion of back-and-forth equivalence,
they are not needed for our main results.
Thus, among the axioms for arboreal categories, we adopt only those
that are needed to define back-and-forth equivalence.
This essentially amounts to requiring a good a factorisation system;
we call the resulting categories \emph{wooded}
as they may be much less densely populated with trees than arboreal categories.

%%%%%%%%%%%%%%%%%%%%%%%%%%%%%%%%%%%%%%%%%%%%%%%%%%%%%%%%%%%%%%%%%%%%%%%%%%%
\subsubsection*{Technical overview of our main result}
%%%%%%%%%%%%%%%%%%%%%%%%%%%%%%%%%%%%%%%%%%%%%%%%%%%%%%%%%%%%%%%%%%%%%%%%%%%
We give a high-level overview of the proof structure of our main results
and the tools used.

An important assumption in our main result is based on the observation that,
in most examples, the class of embeddings in the arboreal (or, more generally, wooded) category~$\A$
is determined by the class of substructure embeddings in~$\E$.
To make this precise we rely on Gabriel-Ulmer duality~\cite{gu71lnm} (cf.\ also~\cite{ar94book,ap98ja}), which can be viewed as a syntax-semantics duality between locally finitely presentable categories and \emph{cartesian theories}. The latter, as introduced by Coste in~\cite{coste76benabou}, consist of implications between formulae which are constructed using only (finite) conjunctions
and provably unique existential quantifiers
(cf.\ also~\cite[\S\S D1--2]{johnstone02book}).
Locally finitely presentable categories can then be identified, up to equivalence,
with the categories of models of cartesian theories
(where morphisms are all homomorphisms).

Consider now a \emph{finitely accessible wooded adjunction}
$\Ladj \colon \C \inadj \E \cocolon R$, i.e.\ a finitely accessible adjunction between locally finitely presentable categories such that $\C$ is wooded.
Assume that $\Sig$ is a (possibly many-sorted and infinite) signature 
such that $\E$ is (equivalent to) the category of models of a cartesian
theory in $\Sig$.
We say that $\Ladj \colon \C \inadj \E \cocolon R$
\emph{detects embeddings in $\Sig$} when the following holds:%
\footnote{Our main result actually assumes
a weaker condition of ``detection of \emph{path} embeddings''.}
\begin{itemize}
\item
a morphism of $\C$ is an embedding
if, and only if, its image under the left adjoint $\Ladj$ is
an embedding of $\Sig$-structures in $\E$.
\end{itemize}

\noindent
Our main result 
states that
\begin{equation*}
\begin{array}{l l l}
  \text{$M,N \in \E$ equivalent in $\Lang_{\infty}(\Sig)$}
& \longimp
& 
  \text{$R M, R N$ back-and-forth equivalent in $\C$.}
\end{array}
\end{equation*}

The proof strategy is as follows.
The back-and-forth games in $\C$ are closed in the usual sense
and are equipped with a customary notion of ordinal rank of positions;
cf.\ e.g.~\cite[\S 3.4]{hodges93book} or~\cite[\S 20B]{kechris95book}.
Assuming some definability of embeddings in $\C$,
and working with Coste's syntactic categories,
for each $a \in \C$ we devise ``Hintikka formulae'' that define
the ordinal ranks of the games played in $a$. 
Under Gabriel-Ulmer duality, the right adjoint $R$ induces an interpretation of theories in the reverse direction, so the Hintikka formulae for $R M$ in $\C$ (which are written in a signature $\Sigbis$ such that $\C$ is the category of models of a cartesian theory in $\Sigbis$) can be translated to formulae
for $M$ in $\E$.

To obtain the desired result, we show that our
assumption about the definability of embeddings in $\C$ is satisfied
if $\Ladj \colon \C \inadj \E \cocolon R$ detects embeddings in $\Sig$.
For this purpose, we use some definability of substructure embeddings in $\E$, so that
Hintikka formulae for $R M$ in~$\C$
depend on formulae over $\Ladj R M$ in $\E$.
This time there is no interpretation induced by Gabriel-Ulmer duality (since we consider the left adjoint, instead of the right one), so to translate the infinitary theory of $\Ladj R M$
into the infinitary theory of $R M$
we use a version of 
Hodges' \emph{word-constructions}~\cite{hodges74,hodges75la}
(cf.\ also \cite[\S 9.3]{hodges93book}).

Note, however, that the notion of substructure embedding in a locally finitely
presentable category $\E$ is somewhat problematic.
As mentioned above, any such $\E$ is equivalent to the category 
of models of a cartesian theory in some signature $\Sig$.
But there are several possible choices of $\Sig$ for a given $\E$,
and the notion of substructure embedding in $\E$ depends on such a choice.
Even when $\Sig$ can be chosen to be mono-sorted, finite
and purely relational, Hintikka formulae relative to $\E$
can be much more complicated than Hintikka formulae relative to the
category of $\Sig$-structures. The situation is considerably simplified if
the finitely accessible adjunction
$\Ladj \colon \C \inadj \E \cocolon R$ has the \emph{factorisation property},
i.e.\ it factors through the usual reflection of $\E$ into the category
of $\Sig$-structures.
We show that the factorisation property holds in some natural examples, 
but that it fails in general (even if we assume that the wooded category $\C$ is arboreal).

%%%%%%%%%%%%%%%%%%%%%%%%%%%%%%%%%%%%%%%%%%%%%%%%%%%%%%%%%%%%%%%%%%%%%%%%%%%
\subsubsection*{Organization of the paper.}
%%%%%%%%%%%%%%%%%%%%%%%%%%%%%%%%%%%%%%%%%%%%%%%%%%%%%%%%%%%%%%%%%%%%%%%%%%%
After some preliminaries in~\S\ref{sec:prelim},
in~\S\ref{sec:path} we present back-and-forth equivalence
for wooded catategories and for finitely accessible wooded adjunctions.
This leads to the statement of our main result, namely Theorem~\ref{thm:path:main},
in~\S\ref{sec:path:results}.
We provide some background on arboreal categories in~\S\ref{sec:arboreal},
as many of our examples of wooded categories fall into this class.
Much of the remainder of the paper is then devoted to the proof
of Theorem~\ref{thm:path:main}.
In~\S\S\ref{sec:lfp}--\ref{sec:coste} we recall some material on
locally finitely presentable
categories and functorial semantics.
In~\S\ref{sec:hintikka}, 
assuming some definability of embeddings in a locally finitely presentable
wooded category $\C$,
we devise Hintikka formulae for ordinal ranks,
and obtain a weaker version of Theorem~\ref{thm:path:main}.
In \S\ref{sec:emb}, we discuss the definability of substructure
embeddings in a locally finitely presentable category $\E$.
We then complete the proof of Theorem~\ref{thm:path:main} in~\S\ref{sec:wc}:
using word-constructions, we transfer
the formulae defining substructure embeddings in $\E$ to the wooded category $\C$.
Finally, we discuss the factorisation property in~\S\ref{sec:fact}.

%%%%%%%%%%%%%%%%%%%%%%%%%%%%%%%%%%%%%%%%%%%%%%%%%%%%%%%%%%%%%%%%%%%%%%%%%%%
\section{Preliminaries}
\label{sec:prelim}
%%%%%%%%%%%%%%%%%%%%%%%%%%%%%%%%%%%%%%%%%%%%%%%%%%%%%%%%%%%%%%%%%%%%%%%%%%%

We assume familiarity with the basic notions of category theory;
standard references include \cite{maclane98book,ahs06book}.
Throughout, $\Set$ denotes the category of (small) sets and functions.
If $\cat C$ and $\cat D$ are categories, then
$\funct{\cat C, \cat D}$ is the category of functors from $\cat C$ to $\cat D$,
with morphisms the natural transformations.
The functor category $\funct{\cat C^\op,\Set}$ is denoted by $\presh{\cat C}$
and referred to as the category of \emph{presheaves over $\cat C$}.
A \emph{lex-morphism} is a functor between finitely complete
categories that preserves finite limits.
If $\cat C$ and $\cat D$ are finitely complete, then $\lex\funct{\cat C,\cat D}$
is the full subcategory of $\funct{\cat C,\cat D}$
consisting of the lex-morphisms. 
A functor is \emph{finitary} if it preserves filtered colimits.

%%%%%%%%%%%%%%%%%%%%%%%%%%%%%%%%%%%%%%%%%%%%%%%%%%%%%%%%%%%%%%%%%%%%%%%%%%%
\subsection{Structures and infinitary first-order logic}
\label{sec:prelim:struct}
%%%%%%%%%%%%%%%%%%%%%%%%%%%%%%%%%%%%%%%%%%%%%%%%%%%%%%%%%%%%%%%%%%%%%%%%%%%

We consider \emph{many-sorted}, possibly infinite, first-order signatures.
Such a signature $\Sig$ consists of a set $\Sort(\Sig)$ of \emph{sorts},
together with collections $\Fun(\Sig)$ and $\Rel(\Sig)$ of 
\emph{function} and \emph{relation symbols}, respectively.
We reserve the symbol $\sig$ for mono-sorted purely relational signatures.

A \emph{$\Sig$-structure} $M$ 
is a family of sets $(M(\sort) \mid \sort \in \Sort(\Sig))$
together with,
for each function symbol $f \in \Fun(\Sig)$,
a function $f^M$ interpreting $f$,
and for each relation symbol $R \in \Rel(\Sig)$,
a relation $R^M$ interpreting $R$.
A \emph{homomorphism} of $\Sig$-structures
$h \colon M \to N$
is
a family
of functions
\(
(
h^\sort \colon M(\sort) \to N(\sort)
\mid
\sort \in \Sort(\Sig)
)
\)
preserving the interpretations of function and relation symbols.
We denote by $\Struct(\Sig)$ the category of $\Sig$-structures and their homomorphisms.

The \emph{atomic formulae} in $\Sig$
are the (sorted) equalities $(t \Eq_\sort u)$ 
and the formulae of the form $R(t_1,\dots,t_n)$,
where $R \in \Rel(\Sig)$
and
$t,u,t_1,\dots,t_n$ are terms in $\Sig$.
We write $\Lang_\infty(\Sig)$,
or simply $\Lang_\infty$ when the signature is clear from the context,
for the (large) set
of 
formulae 
built
from atomic formulae in~$\Sig$
using arbitrary (set-indexed) conjunctions and disjunctions,
along with $\neg, \exists, \forall$.
For a cardinal $\kappa$,
we write
$\Lang_\kappa$
for the set of formulae whose conjunctions and disjunctions
have cardinality $<\kappa$.
Hence, $\Lang_\omega$ consists of the ordinary finite first-order formulae.
In any case,
we only consider formulae whose subformulae have 
finitely many free variables
(thus, with the notation of~\cite{hodges93book,ef99book},
$\Lang_\infty=\Lang_{\infty,\omega}$).
There is a rich literature on the (non-trivial) expressiveness of $\Lang_\infty$.
See e.g.~\cite{hodges93book}.

A \emph{formula-in-context} is a pair $(\varphi, \Env)$,
always denoted by $\Env \sorting \varphi$,
consisting of a formula~$\varphi$
and a context
\[\Env = x_1:\sort_1,\dots,x_n:\sort_n,\]
where the $x_i$'s are pairwise distinct.
We require $\Env$ to declare all the free variables of $\varphi$
with the appropriate sort.
Given a structure $M \in \Struct(\Sig)$,
a formula-in-context $\Env \sorting \varphi$ induces a set
\[
\begin{array}{l l l}
  \I{\Env \mid \varphi}_M
& \coloneqq
& \{ (a_1,\dots,a_n) \in M(\sort_1) \times \dots \times M(\sort_n)
  \mid
  M \models \varphi(a_1,\dots,a_n) \}.
\end{array}
\]

\noindent
We say that $\Env \sorting \varphi$ is \emph{valid} in $M$ if 
\[
\begin{array}{l l l}
  \I{\Env \mid \varphi}_M
& =
& M(\sort_1) \times \cdots \times M(\sort_n).
\end{array}
\]

\noindent
When sorts are notationally irrelevant,
we simply write $\vec x$ for $\Env = x_1:\sort_1,\dots,x_n:\sort_n$,
and similarly for homomorphisms and carriers of structures.

We write $\emptyset$ for the empty context.
Note that if $\varphi$ is a sentence,
then $\I{\emptyset \mid \varphi}_M$ is a subset of the singleton set $\one$.
In this case, we have
$M \models \varphi$ precisely when $\I{\emptyset \mid \varphi}_M = \one$.

%%%%%%%%%%%%%%%%%%%%%%%%%%%%%%%%%%%%%%%%%%%%%%%%%%%%%%%%%%%%%%%%%%%%%%%%%%%
\subsection{Cartesian theories}
\label{sec:prelim:coste}
%%%%%%%%%%%%%%%%%%%%%%%%%%%%%%%%%%%%%%%%%%%%%%%%%%%%%%%%%%%%%%%%%%%%%%%%%%%
Cartesian theories consist of implications between $\land\exists$-formulae
in which only provably unique existential quantifications are allowed.
They have been introduced by Coste~\cite{coste76benabou}
as a purely logical characterisation of locally finitely presentable categories.
We briefly review the latter in~\S\ref{sec:prelim:lfp};
further explanations and details are provided in~\S\S\ref{sec:lfp}--\ref{sec:coste}.
We mostly follow the presentation in~\cite[\S D1]{johnstone02book}.

Let $\Sig$ be a signature.
A \emph{sequent in $\Sig$} is a triple of the form
\[
\begin{array}{l l l}
  \psi
& \thesis_\Env
& \varphi
\end{array}
\]

\noindent
where $\Env \sorting \psi$ and $\Env \sorting \varphi$ are formulae-in-context of $\Sig$.
A \emph{theory $\theory$ in $\Sig$} is a set of sequents in $\Sig$.

%%%%%%%%%%%%%%%%%%%%%%%%%%%%%%%%%%%%%%%%%%%%%%%%%%%%%%%%%%%%%%%%%%%%%%%%%%%
\begin{definition}[Cartesian theory]
\label{def:prelim:coste:th}
%%%%%%%%%%%%%%%%%%%%%%%%%%%%%%%%%%%%%%%%%%%%%%%%%%%%%%%%%%%%%%%%%%%%%%%%%%%
%Let $\Sig$ be a signature,
Let $\theory$ be a theory in a signature $\Sig$.
\begin{enumerate}[(1)]
\item
The set of formulae-in-context
that are \emph{cartesian over $\theory$}
is inductively defined as follows:
\begin{itemize}
\item atomic formulae-in-context are cartesian over $\theory$;
\item cartesian formulae are closed under finite conjunctions;

\item given a cartesian formula-in-context
$\Env,x:\sort \sorting \varphi$,
the formula
$\Env \sorting (\exists x:\sort) \varphi$
is cartesian over $\theory$ if $\theory$ proves the sequent
\(
  \varphi
  \land
  \varphi[y/x]
  ~\thesis_{\Env, x: \sort, y: \sort}~
  x \Eq_\sort y
\).
\end{itemize}

\item
A sequent $\psi \thesis_\Env \varphi$ is \emph{cartesian over $\theory$}
if both $\psi, \varphi$ are.

\item
We say that $\theory$ is a \emph{cartesian theory}
when its axioms can be well-ordered so that each axiom
is cartesian over the set of its predecessors.
\end{enumerate}
\end{definition}

Let $M \in \Struct(\Sig)$.
A sequent $\psi \thesis_\Env \varphi$ is \emph{valid} in
$M$
if $\I{\Env \mid \psi}_M \sle \I{\Env \mid \varphi}_M$, and
$M$ is a \emph{model} of a theory $\theory$
if all sequents in $\theory$
are valid in $M$.
We denote by 
\[
\Mod(\theory)
\] 
the full subcategory
of $\Struct(\Sig)$ defined by the models of $\theory$.

%%%%%%%%%%%%%%%%%%%%%%%%%%%%%%%%%%%%%%%%%%%%%%%%%%%%%%%%%%%%%%%%%%%%%%%%%%%
\begin{remark}
\label{rem:prelim:coste:compl}
%%%%%%%%%%%%%%%%%%%%%%%%%%%%%%%%%%%%%%%%%%%%%%%%%%%%%%%%%%%%%%%%%%%%%%%%%%%
Cartesian theories consist of (sequents of)
$\land\exists$-formulae, hence they are special cases of \emph{regular} theories.
Further, they admit a completeness theorem stating that
any cartesian theory has enough models in $\Set$ to determine intuitionistic
provability (see e.g.~\cite[Proposition D1.5.1]{johnstone02book}).
In particular, classical and intuitionistic provability coincide for cartesian theories.%
\footnote{This is a particular case of \emph{Barr's theorem}.
See e.g.~\cite[\S 5]{negri16jlc} (and also~\cite[\S X]{mm92book}).}
%%%%%%%%%%%%%%%%%%%%%%%%%%%%%%%%%%%%%%%%%%%%%%%%%%%%%%%%%%%%%%%%%%%%%%%%%%%
\end{remark}
%%%%%%%%%%%%%%%%%%%%%%%%%%%%%%%%%%%%%%%%%%%%%%%%%%%%%%%%%%%%%%%%%%%%%%%%%%%

%%%%%%%%%%%%%%%%%%%%%%%%%%%%%%%%%%%%%%%%%%%%%%%%%%%%%%%%%%%%%%%%%%%%%%%%%%%
\begin{example}[Cartesian theories of functors]
\label{ex:prelim:coste:funct}
%%%%%%%%%%%%%%%%%%%%%%%%%%%%%%%%%%%%%%%%%%%%%%%%%%%%%%%%%%%%%%%%%%%%%%%%%%%
Given a small category $\cat C$, consider the signature $\Sig(\cat C)$
with sorts the objects of $\cat C$, function symbols the morphisms of $\cat C$,
and no predicate symbols.
Let $\theory(\cat C)$ be the  
equational theory with axioms
\[
\begin{array}{l !{\qquad} l}
  \True
  \thesis_{x:a}
  \id_a(x) \Eq_a x

& \True
  \thesis_{x:a}
  (g \comp f)(x) \Eq_c g(f(x))
\end{array}
\]

\noindent
where $f \in \cat C\funct{a,b}$ and $g \in \cat C\funct{b,c}$.
Then $\Mod(\theory(\cat C))$
is equivalent to the functor category $\funct{\cat C,\Set}$.
Further, if $\cat C$ has finite limits, then
$\theory(\cat C)$ can be extended to a cartesian theory~$\theory$ such that $\Mod(\theory) \cong \lex\funct{\cat C,\Set}$.
See e.g.\ \cite[Example D1.4.8]{johnstone02book}.
%%%%%%%%%%%%%%%%%%%%%%%%%%%%%%%%%%%%%%%%%%%%%%%%%%%%%%%%%%%%%%%%%%%%%%%%%%%
\end{example}
%%%%%%%%%%%%%%%%%%%%%%%%%%%%%%%%%%%%%%%%%%%%%%%%%%%%%%%%%%%%%%%%%%%%%%%%%%%

%%%%%%%%%%%%%%%%%%%%%%%%%%%%%%%%%%%%%%%%%%%%%%%%%%%%%%%%%%%%%%%%%%%%%%%%%%%
\begin{example}
\label{ex:prelim:coste:pos}
%%%%%%%%%%%%%%%%%%%%%%%%%%%%%%%%%%%%%%%%%%%%%%%%%%%%%%%%%%%%%%%%%%%%%%%%%%%
\emph{Horn theories} form an important class of cartesian theories.
For instance, the category $\Pos$
of posets and monotone maps
is the category of models
of the following Horn theory in $\sig = \{\Leq\}$:
\[
\begin{array}{l !{\qquad} l l l !{\qquad} l l l}
  \True
  \thesis_{x}
  x \Leq x

& x \Leq y
  ~\land~
  y \Leq z
& \thesis_{x,y,z}
& x \Leq z

& x \Leq y
  ~\land~
  y \Leq x
& \thesis_{x,y}
& x \Eq y
\end{array}
\]
%%%%%%%%%%%%%%%%%%%%%%%%%%%%%%%%%%%%%%%%%%%%%%%%%%%%%%%%%%%%%%%%%%%%%%%%%%%
\end{example}
%%%%%%%%%%%%%%%%%%%%%%%%%%%%%%%%%%%%%%%%%%%%%%%%%%%%%%%%%%%%%%%%%%%%%%%%%%%

%%%%%%%%%%%%%%%%%%%%%%%%%%%%%%%%%%%%%%%%%%%%%%%%%%%%%%%%%%%%%%%%%%%%%%%%%%%
\begin{remark}
\label{rem:prelim:coste:horn}
%%%%%%%%%%%%%%%%%%%%%%%%%%%%%%%%%%%%%%%%%%%%%%%%%%%%%%%%%%%%%%%%%%%%%%%%%%%
The class of cartesian theories strictly extends that of Horn theories.
For instance, the category $\Cat$ of small categories
is the category of models of a cartesian theory,
see~\cite[Example D1.1.7(e)]{johnstone02book},
but there is no signature $\Sig$ such that $\Cat$
is the category of models of a Horn theory in $\Sig$,
see~\cite[Example D2.4.7]{johnstone02book}.
%%%%%%%%%%%%%%%%%%%%%%%%%%%%%%%%%%%%%%%%%%%%%%%%%%%%%%%%%%%%%%%%%%%%%%%%%%%
\end{remark}
%%%%%%%%%%%%%%%%%%%%%%%%%%%%%%%%%%%%%%%%%%%%%%%%%%%%%%%%%%%%%%%%%%%%%%%%%%%

%%%%%%%%%%%%%%%%%%%%%%%%%%%%%%%%%%%%%%%%%%%%%%%%%%%%%%%%%%%%%%%%%%%%%%%%%%%
\subsection{Locally finitely presentable categories}
\label{sec:prelim:lfp}
%%%%%%%%%%%%%%%%%%%%%%%%%%%%%%%%%%%%%%%%%%%%%%%%%%%%%%%%%%%%%%%%%%%%%%%%%%%
We briefly discuss some aspects of locally finitely presentable (lfp)
categories needed to present our main examples and results.

Recall that an object $X$ of a locally small category $\E$ is
\emph{finitely presentable} if the functor
$\E\funct{X,-} \colon \E \to \Set$ preserves all filtered colimits that exist in $\E$.

%%%%%%%%%%%%%%%%%%%%%%%%%%%%%%%%%%%%%%%%%%%%%%%%%%%%%%%%%%%%%%%%%%%%%%%%%%%
\begin{definition}
\label{def:lfp}
%%%%%%%%%%%%%%%%%%%%%%%%%%%%%%%%%%%%%%%%%%%%%%%%%%%%%%%%%%%%%%%%%%%%%%%%%%%
A locally small category $\E$ is \emph{locally finitely presentable} (\emph{lfp})
provided it satisfies the following conditions:
\begin{enumerate}[(i)]
\item
\label{item:lfp-small-dense}
$\E$
admits a (small) set $\cat C$ of finitely presentable objects
such that every object of~$\E$ is a filtered colimit
of objects from $\cat C$;

\item
\label{item:lfp-cocomplete}
$\E$ is cocomplete
(and thus complete, see~\cite[Remark 1.56]{ar94book}).
\end{enumerate}

\noindent
A \emph{morphism of lfp categories} is a limit-preserving finitary functor
between lfp categories.
%%%%%%%%%%%%%%%%%%%%%%%%%%%%%%%%%%%%%%%%%%%%%%%%%%%%%%%%%%%%%%%%%%%%%%%%%%%
\end{definition}
%%%%%%%%%%%%%%%%%%%%%%%%%%%%%%%%%%%%%%%%%%%%%%%%%%%%%%%%%%%%%%%%%%%%%%%%%%%

%%%%%%%%%%%%%%%%%%%%%%%%%%%%%%%%%%%%%%%%%%%%%%%%%%%%%%%%%%%%%%%%%%%%%%%%%%%
\begin{example}
\label{ex:lfp:struct}
%%%%%%%%%%%%%%%%%%%%%%%%%%%%%%%%%%%%%%%%%%%%%%%%%%%%%%%%%%%%%%%%%%%%%%%%%%%
Each category $\Struct(\Sig)$ is lfp.
The finitely presentable objects of $\Struct(\Sig)$,
denoted by $\FG{\vec x \mid \varphi}$,
are given by
finitely many generators $\vec x = x_1,\dots,x_n$
subject to a condition $\varphi = \varphi(\vec x)$
expressed by
a finite conjunction of atomic formulae from $\Sig$
(so that $\vec x \sorting \varphi$ is a formula-in-context).
%%%%%%%%%%%%%%%%%%%%%%%%%%%%%%%%%%%%%%%%%%%%%%%%%%%%%%%%%%%%%%%%%%%%%%%%%%%
\end{example}
%%%%%%%%%%%%%%%%%%%%%%%%%%%%%%%%%%%%%%%%%%%%%%%%%%%%%%%%%%%%%%%%%%%%%%%%%%%

%%%%%%%%%%%%%%%%%%%%%%%%%%%%%%%%%%%%%%%%%%%%%%%%%%%%%%%%%%%%%%%%%%%%%%%%%%%
\begin{fact}
\label{fact:lfp:struct}
%%%%%%%%%%%%%%%%%%%%%%%%%%%%%%%%%%%%%%%%%%%%%%%%%%%%%%%%%%%%%%%%%%%%%%%%%%%
Consider $M \in \Struct(\Sig)$
and
$\FG{\vec x \mid \varphi}$ as in Example~\ref{ex:lfp:struct}.
Given $\vec a \in \I{\vec x \mid \varphi}_M$,
there is a unique homomorphism
$h \colon \FG{\vec x \mid \varphi} \to M$
that takes each generator $x_i$ to $a_i$. 
This results in a bijection
\[
\begin{array}{l l l}
  \Struct(\Sig)\funct{\FG{\vec x \mid \varphi}, M}
& \cong
& \I{\vec x \mid \varphi}_M.
\end{array}
\]
%%%%%%%%%%%%%%%%%%%%%%%%%%%%%%%%%%%%%%%%%%%%%%%%%%%%%%%%%%%%%%%%%%%%%%%%%%%
\end{fact}
%%%%%%%%%%%%%%%%%%%%%%%%%%%%%%%%%%%%%%%%%%%%%%%%%%%%%%%%%%%%%%%%%%%%%%%%%%%

References on Example~\ref{ex:lfp:struct} and Fact~\ref{fact:lfp:struct}
include \cite[Remarks 5.1 and 5.5]{ar94book}
and \cite[Lemma D2.4.1 and Lemma D1.4.4(ii)]{johnstone02book}.
We give some details in~\S\ref{sec:coste} below.

%%%%%%%%%%%%%%%%%%%%%%%%%%%%%%%%%%%%%%%%%%%%%%%%%%%%%%%%%%%%%%%%%%%%%%%%%%%
\begin{remark}
\label{rem:lfp:struct}
%%%%%%%%%%%%%%%%%%%%%%%%%%%%%%%%%%%%%%%%%%%%%%%%%%%%%%%%%%%%%%%%%%%%%%%%%%%
Generalising Example~\ref{ex:lfp:struct},
if $\theory$ is a cartesian theory in~$\Sig$, then
$\Mod(\theory)$ is lfp, and moreover it is
reflective and closed under filtered colimits in $\Struct(\Sig)$.
In fact, a category $\E$ is lfp if, and only if,
there is a cartesian theory $\theory$ in a
(possibly infinite, many-sorted) signature $\Sig$ such that
$\E$ is equivalent to $\Mod(\theory)$.
See Corollary~\ref{cor:lfp-iff-models-of-T} and Lemma~\ref{lem:coste:mod:filtcolim} below and also \cite[Theorem 5.9]{ar94book}.
%%%%%%%%%%%%%%%%%%%%%%%%%%%%%%%%%%%%%%%%%%%%%%%%%%%%%%%%%%%%%%%%%%%%%%%%%%%
\end{remark}
%%%%%%%%%%%%%%%%%%%%%%%%%%%%%%%%%%%%%%%%%%%%%%%%%%%%%%%%%%%%%%%%%%%%%%%%%%%

Finally, we state a well-known sufficient criterion to establish that a category is lfp
(see e.g.~\cite[Theorem~1.11 and Remark~1.23]{ar94book}).
To this end, recall that a subcategory $I \colon \cat K \into \E$
is \emph{dense} if
each $X \in \E$
is the canonical colimit of the 
diagram
$(I/X) \to \cat K \into \E$,
i.e.\ the cocone $\gamma$
with $\gamma_{(k,g \colon I k \to X)} \deq g$ is colimiting.

%%%%%%%%%%%%%%%%%%%%%%%%%%%%%%%%%%%%%%%%%%%%%%%%%%%%%%%%%%%%%%%%%%%%%%%%%%%
\begin{lemma}
\label{l:lfp-criterion}
%%%%%%%%%%%%%%%%%%%%%%%%%%%%%%%%%%%%%%%%%%%%%%%%%%%%%%%%%%%%%%%%%%%%%%%%%%%
Let $\E$ be a cocomplete locally small category.
If $\E$ admits an (essentially) small full dense subcategory $\cat K$
consisting of finitely presentable objects,
then $\E$ is lfp. 
%%%%%%%%%%%%%%%%%%%%%%%%%%%%%%%%%%%%%%%%%%%%%%%%%%%%%%%%%%%%%%%%%%%%%%%%%%%
\end{lemma}
%%%%%%%%%%%%%%%%%%%%%%%%%%%%%%%%%%%%%%%%%%%%%%%%%%%%%%%%%%%%%%%%%%%%%%%%%%%

%%%%%%%%%%%%%%%%%%%%%%%%%%%%%%%%%%%%%%%%%%%%%%%%%%%%%%%%%%%%%%%%%%%%%%%%%%%
\begin{example}
\label{ex:presheaf-cats-lfp}
%%%%%%%%%%%%%%%%%%%%%%%%%%%%%%%%%%%%%%%%%%%%%%%%%%%%%%%%%%%%%%%%%%%%%%%%%%%
For any small category $\cat C$, the presheaf category $\presh{\cat C}$ is lfp.
This follows from Lemma~\ref{l:lfp-criterion},
since the representable functors form a full dense subcategory of 
$\presh{\cat C}$ consisting of finitely presentable objects.
\end{example}

%%%%%%%%%%%%%%%%%%%%%%%%%%%%%%%%%%%%%%%%%%%%%%%%%%%%%%%%%%%%%%%%%%%%%%%%%%%
\begin{example}
\label{ex:prelim:forests}
%%%%%%%%%%%%%%%%%%%%%%%%%%%%%%%%%%%%%%%%%%%%%%%%%%%%%%%%%%%%%%%%%%%%%%%%%%%
Recall that a poset $(\forest,\leq)$ is a \emph{forest} if,
for all $x\in \forest$,
the set 
\[
\begin{array}{l l l}
  \down x
& \deq
& \{y\in \forest \mid y\leq x\}
\end{array}
\] 

\noindent
is a finite chain.
The \emph{roots} of a forest $(\forest,\leq)$ are the minimal elements.
The \emph{covering relation} $\cvr$ associated with the partial order $\leq$ is
defined by $x\cvr y$ if, and only if, $x < y$ and there is no $z$ such that $x < z < y$.
A \emph{tree} is a forest with at most one root.

A \emph{forest morphism} is a function between forests that preserves roots
and the covering relation.
The category of forests and forest morphisms is denoted by $\Forest$,
and the full subcategory of \emph{non-empty} trees by~$\Tree$. 

The categories $\Forest$ and $\Tree$ are equivalent
(via the functor $\Forest\to\Tree$ that adds a least element,
whose quasi-inverse $\Tree\to\Forest$ removes the root),
and they are equivalent to the presheaf category $\presh{\NN}$. 
In particular, $\Forest$ and $\Tree$ are lfp categories.
\end{example}

%%%%%%%%%%%%%%%%%%%%%%%%%%%%%%%%%%%%%%%%%%%%%%%%%%%%%%%%%%%%%%%%%%%%%%%%%%%
\subsection{Factorisation systems}
\label{sec:prelim:facto}
%%%%%%%%%%%%%%%%%%%%%%%%%%%%%%%%%%%%%%%%%%%%%%%%%%%%%%%%%%%%%%%%%%%%%%%%%%%

Given arrows $e$ and $m$ in a category $\C$,
we say that $e$ has the \emph{left lifting property} with respect to $m$,
or that $m$ has the \emph{right lifting property} with respect to $e$,
and write $e\pitchfork m$, if every commutative square as on the left-hand side below
\[
\begin{array}{l !{\qquad} l}

\begin{tikzcd}
  \unit
  \arrow{d}[swap]{e}
  \arrow{r}
& \unit
  \arrow{d}{m}
\\
  \unit
  \arrow{r}
& \unit
\end{tikzcd}

&

\begin{tikzcd}
  \unit
  \arrow{d}[swap]{e}
  \arrow{r}
& \unit
  \arrow{d}{m}
\\
  \unit \arrow{ur}[description]{d}
  \arrow{r}
& \unit
\end{tikzcd}

\end{array}
\]

\noindent
admits a \emph{diagonal filler},
i.e.\ an arrow $d$ making the rightmost diagram commute.

%%%%%%%%%%%%%%%%%%%%%%%%%%%%%%%%%%%%%%%%%%%%%%%%%%%%%%%%%%%%%%%%%%%%%%%%%%%
\begin{definition}
\label{def:weak-f-s}
%%%%%%%%%%%%%%%%%%%%%%%%%%%%%%%%%%%%%%%%%%%%%%%%%%%%%%%%%%%%%%%%%%%%%%%%%%%
A pair of classes of morphisms $(\Q,\M)$ in a category $\C$ is a
\emph{weak factorisation system} provided it satisfies the following conditions:
\begin{enumerate}[(i)]
\item
Every morphism $f$ in $\C$ can be decomposed as $f = m \comp e$
with $e\in \Q$ and $m\in \M$.

\item
$\Q=\{e\mid \forall m\in \M, e\pitchfork m\}$
and
$\M=\{m\mid \forall e\in \Q, e\pitchfork m\}$.
\end{enumerate}

\noindent
If, in addition,
all arrows in $\Q$ are epimorphisms and all arrows in $\M$ are monomorphisms,
then $(\Q,\M)$ is a \emph{proper factorisation system}. 
\end{definition} 

%%%%%%%%%%%%%%%%%%%%%%%%%%%%%%%%%%%%%%%%%%%%%%%%%%%%%%%%%%%%%%%%%%%%%%%%%%%
\begin{remark}
%%%%%%%%%%%%%%%%%%%%%%%%%%%%%%%%%%%%%%%%%%%%%%%%%%%%%%%%%%%%%%%%%%%%%%%%%%%
Note that a weak factorisation system $(\Q,\M)$ is completely determined
by either of the two classes $\Q$ or $\M$.
%%%%%%%%%%%%%%%%%%%%%%%%%%%%%%%%%%%%%%%%%%%%%%%%%%%%%%%%%%%%%%%%%%%%%%%%%%%
\end{remark}
%%%%%%%%%%%%%%%%%%%%%%%%%%%%%%%%%%%%%%%%%%%%%%%%%%%%%%%%%%%%%%%%%%%%%%%%%%%

Recall that a monomorphism is \emph{strong}
if it has the right lifting property with respect to all epimorphisms;
\emph{strong epimorphisms} are defined dually.
If $(\Q,\M)$ is a proper factorisation system,
it follows at once that every strong monomorphism belongs to
$\M$ and every strong epimorphism belongs to $\Q$.

%%%%%%%%%%%%%%%%%%%%%%%%%%%%%%%%%%%%%%%%%%%%%%%%%%%%%%%%%%%%%%%%%%%%%%%%%%%
\begin{remark}
\label{rem:prelim:fact:strong}
%%%%%%%%%%%%%%%%%%%%%%%%%%%%%%%%%%%%%%%%%%%%%%%%%%%%%%%%%%%%%%%%%%%%%%%%%%%
In any category, a necessary and sufficient condition for
(epimorphisms, strong monomorphisms) to form a proper factorisation system
is that every morphism can be decomposed into an epimorphism followed by a
strong monomorphism. Dually, for the (strong epimorphisms, monomorphisms)
factorisation system.
%%%%%%%%%%%%%%%%%%%%%%%%%%%%%%%%%%%%%%%%%%%%%%%%%%%%%%%%%%%%%%%%%%%%%%%%%%%
\end{remark}
%%%%%%%%%%%%%%%%%%%%%%%%%%%%%%%%%%%%%%%%%%%%%%%%%%%%%%%%%%%%%%%%%%%%%%%%%%%

%%%%%%%%%%%%%%%%%%%%%%%%%%%%%%%%%%%%%%%%%%%%%%%%%%%%%%%%%%%%%%%%%%%%%%%%%%%
\begin{fullproof}
%%%%%%%%%%%%%%%%%%%%%%%%%%%%%%%%%%%%%%%%%%%%%%%%%%%%%%%%%%%%%%%%%%%%%%%%%%%
Let $\Q$ consist of the epimorphisms and
let $\M$ consist of the strong monos of $\cat C$.
It is clear that $\M=\{m\mid \forall e\in \Q, e\pitchfork m\}$.
We show that $\Q=\{e\mid \forall m\in \M, e\pitchfork m\}$.
We of course have $\Q \sle \{e\mid \forall m\in \M, e\pitchfork m\}$.

Let $f$ such that $f \pitchfork m$ for every strong mono $m$.
Factor $f$ as $f = m \comp e$ with $e \in \Q$ and $m \in \M$.
We thus have a diagonal filler $d$ as in
\[
\begin{tikzcd}
  \unit
  \arrow{d}[left]{f}
  \arrow[twoheadrightarrow]{r}{e}
& \unit
  \arrow[rightarrowtail]{d}{m}
\\
  \unit
  \arrow{r}[below]{\id}
  \arrow{ur}[description]{d}
& \unit
\end{tikzcd}
\]

But then $m$ is an iso since it is a both a mono and split epi.
Hence $d$ is an iso and $f$ is an epi.
%%%%%%%%%%%%%%%%%%%%%%%%%%%%%%%%%%%%%%%%%%%%%%%%%%%%%%%%%%%%%%%%%%%%%%%%%%%
\end{fullproof}
%%%%%%%%%%%%%%%%%%%%%%%%%%%%%%%%%%%%%%%%%%%%%%%%%%%%%%%%%%%%%%%%%%%%%%%%%%%

%%%%%%%%%%%%%%%%%%%%%%%%%%%%%%%%%%%%%%%%%%%%%%%%%%%%%%%%%%%%%%%%%%%%%%%%%%%
\begin{example}
\label{ex:prelim:fact:lfp}
%%%%%%%%%%%%%%%%%%%%%%%%%%%%%%%%%%%%%%%%%%%%%%%%%%%%%%%%%%%%%%%%%%%%%%%%%%%
Each lfp category has two
proper factorisation
systems, one given by taking as $\M$ the class of strong monomorphisms,
the other one given by taking as $\Q$ the class of strong epimorphisms
(combine~\cite[Proposition 1.61, \S 1D]{ar94book}
with Remark~\ref{rem:prelim:fact:strong}).
%%%%%%%%%%%%%%%%%%%%%%%%%%%%%%%%%%%%%%%%%%%%%%%%%%%%%%%%%%%%%%%%%%%%%%%%%%%
\end{example}
%%%%%%%%%%%%%%%%%%%%%%%%%%%%%%%%%%%%%%%%%%%%%%%%%%%%%%%%%%%%%%%%%%%%%%%%%%%

%%%%%%%%%%%%%%%%%%%%%%%%%%%%%%%%%%%%%%%%%%%%%%%%%%%%%%%%%%%%%%%%%%%%%%%%%%%
\begin{example}
\label{ex:prelim:fact:forest}
%%%%%%%%%%%%%%%%%%%%%%%%%%%%%%%%%%%%%%%%%%%%%%%%%%%%%%%%%%%%%%%%%%%%%%%%%%%
In the case of a presheaf category,
the two factorisation systems in Example~\ref{ex:prelim:fact:lfp}
coincide with the (epimorphism, monomorphism) factorisation system;
cf.\ e.g.~\cite[4.3.10.g]{borceux94vol1}.%
\footnote{\label{footnote:unique-fs}In fact, any elementary topos has exactly one proper factorisation system,
namely the (epimorphism, monomorphism) factorisation system.}
In particular, for the remaining of this paper, we consider the
categories $\Tree$ and $\Forest$ with the
(epimorphism, monomorphism) factorisation systems.
\end{example}

%%%%%%%%%%%%%%%%%%%%%%%%%%%%%%%%%%%%%%%%%%%%%%%%%%%%%%%%%%%%%%%%%%%%%%%%%%%
\begin{example}
\label{ex:prelim:fact:struct}
%%%%%%%%%%%%%%%%%%%%%%%%%%%%%%%%%%%%%%%%%%%%%%%%%%%%%%%%%%%%%%%%%%%%%%%%%%%
In $\Struct(\Sig)$, the proper factorisation systems 
in Example~\ref{ex:prelim:fact:lfp}
can be described concretely as follows.
The monomorphisms (respectively, epimorphisms)
are the morphisms that are sortwise injective (respectively, surjective).
The strong monomorphisms are the embeddings of structures,
i.e.\ the monomorphisms that \emph{reflect} (as well as preserve)
the interpretations of the relation symbols. 
Just recall that embeddings coincide with regular monos in $\Struct(\Sig)$,
while regular monos coincide with strong monos in the presence of
(epi, regular mono) factorisations; cf.\ e.g.\ \cite[0.5 and Remark~5.1]{ar94book}.%

In general, the  
factorisation systems (strong epi, mono) and (epi, strong mono)
on $\Struct(\Sig)$ are distinct.
For example, if $\sig$ is a mono-sorted purely relational signature and
$h\colon M\to N$ is a morphism in $\Struct(\sig)$,
consider the usual factorisation of $h$ through its set-theoretic image $h[M]$:
\[
M \to h[M] \to N.
\]

\noindent
Then the (strong epi, mono) factorisation of $h$ is obtained by defining
the interpretation in $h[M]$ of an arbitrary relation symbol $R\in\Rel(\sig)$
as $R^{h[M]}\coloneqq h[R^{M}]$,
whereas to obtain the (epi, strong mono) factorisation of $h$ we let
$R^{h[M]}$ be the restriction of $R^N$ to (the appropriate power of) $h[M]$.
%%%%%%%%%%%%%%%%%%%%%%%%%%%%%%%%%%%%%%%%%%%%%%%%%%%%%%%%%%%%%%%%%%%%%%%%%%%
\end{example}
%%%%%%%%%%%%%%%%%%%%%%%%%%%%%%%%%%%%%%%%%%%%%%%%%%%%%%%%%%%%%%%%%%%%%%%%%%%

The following are well-known properties of weak factorisation systems,
cf.~\cite{freyd1972categories} or~\cite{riehl08notes}.

%%%%%%%%%%%%%%%%%%%%%%%%%%%%%%%%%%%%%%%%%%%%%%%%%%%%%%%%%%%%%%%%%%%%%%%%%%%
\begin{lemma}
\label{lem:prelim:fact:base}
%%%%%%%%%%%%%%%%%%%%%%%%%%%%%%%%%%%%%%%%%%%%%%%%%%%%%%%%%%%%%%%%%%%%%%%%%%%
Let $(\Q,\M)$ be a weak factorisation system on $\cat C$.
\begin{enumerate}[(1)]
\item
\label{item:prelim:fact:base:comp}
$\Q$ and $\M$ are closed under composition.

\item
\label{item:prelim:fact:base:iso}
$\Q \cap \M = \{\text{isomorphisms} \}$.

\item
\label{item:prelim:fact:base:pullback}
The pullback of an $\M$-morphism along any morphism,
if it exists, is again in $\M$.
\setcounter{SplitEnum}{\value{enumi}}
\end{enumerate}

\noindent
Moreover, if $(\Q,\M)$ is proper, the following hold for all composable
morphisms $f, g$ of $\cat C$.
\begin{enumerate}[(1)]
\setcounter{enumi}{\value{SplitEnum}}

\item
\label{item:prelim:fact:base:quotient}
$g \comp f \in \Q$ implies $g \in \Q$.

\item
\label{item:prelim:fact:base:emb}
$g \comp f \in \M$ implies $f \in \M$.
\end{enumerate}
%%%%%%%%%%%%%%%%%%%%%%%%%%%%%%%%%%%%%%%%%%%%%%%%%%%%%%%%%%%%%%%%%%%%%%%%%%%
\end{lemma}
%%%%%%%%%%%%%%%%%%%%%%%%%%%%%%%%%%%%%%%%%%%%%%%%%%%%%%%%%%%%%%%%%%%%%%%%%%%

Let $\C$ be a category equipped with a proper factorisation system $(\Q,\M)$.
We shall refer to $\M$-morphisms as \emph{embeddings} and denote them by $\emb$.
$\Q$-morphisms will be referred to as \emph{quotients} and denoted by~$\epi$. 

In the same way that one usually defines the poset of subobjects of a given object $X\in\C$, we can define the poset $\Emb{X}$ of $\M$-subobjects of $X$. Given embeddings $m\colon S\emb X$ and $n\colon T\emb X$, let us say that $m\trianglelefteq n$ provided there is a morphism $i\colon S\to T$ such that $m=n\circ i$, as displayed below.
\[\begin{tikzcd}
S \arrow[rightarrowtail]{r}{m} \arrow[dashed, swap]{d}{i} & X \\
T \arrow[rightarrowtail]{ur}[swap]{n} & {}
\end{tikzcd}\] 
(Note that, if it exists, $i$ is necessarily an embedding by Lemma~\ref{lem:prelim:fact:base}.\ref{item:prelim:fact:base:emb}.)
This yields a preorder on the class of all embeddings with codomain $X$. The symmetrization $\sim$ of $\trianglelefteq$ can be characterised as follows: $m\sim n$ if, and only if, there exists an isomorphism $i\colon S\to T$ such that $m=n\circ i$. Let $\Emb{X}$ be the class of $\sim$-equivalence classes of embeddings with codomain $X$, equipped with the natural partial order $\leq$ induced by~$\trianglelefteq$, and observe that the poset $\Emb{X}$ is small whenever $\C$ is well-powered. We systematically represent a $\sim$-equivalence class of embeddings by any of its representatives. 

%%%%%%%%%%%%%%%%%%%%%%%%%%%%%%%%%%%%%%%%%%%%%%%%%%%%%%%%%%%%%%%%%%%%%%%%%%%
\begin{example}
%%%%%%%%%%%%%%%%%%%%%%%%%%%%%%%%%%%%%%%%%%%%%%%%%%%%%%%%%%%%%%%%%%%%%%%%%%%
Consider $\Struct(\Sig)$ with the (epis, strong monos) factorisation system.
For all $M\in \Struct(\Sig)$,
$\Emb{M}$ is isomorphic to the poset of \emph{substructures} of $M$
(i.e., subsets of $M$ equipped with the induced relations)
ordered by set-theoretic inclusion.
%%%%%%%%%%%%%%%%%%%%%%%%%%%%%%%%%%%%%%%%%%%%%%%%%%%%%%%%%%%%%%%%%%%%%%%%%%%
\end{example}
%%%%%%%%%%%%%%%%%%%%%%%%%%%%%%%%%%%%%%%%%%%%%%%%%%%%%%%%%%%%%%%%%%%%%%%%%%%

For any morphism $f\colon X\to Y$ in $\C$ and embedding $m\colon S\emb X$, we can consider the $(\Q,\M)$ factorisation $S\epi \exists_f S \emb Y$ of $f\circ m$. This yields a monotone map ${\exists_f\colon \Emb{X}\to\Emb{Y}}$ sending $m$ to the embedding $\exists_f S\emb Y$. (Note that the map $\exists_f$ is well-defined because factorisations are unique up to isomorphism.) In particular, if $f$ is an embedding,
we have $\exists_f(m) = f \comp m$ for all $m\in \Emb{X}$.

For the following observation, see e.g.\ \cite[Lemma 2.7(a)]{ar21arboreal}.

%%%%%%%%%%%%%%%%%%%%%%%%%%%%%%%%%%%%%%%%%%%%%%%%%%%%%%%%%%%%%%%%%%%%%%%%%%%
\begin{lemma}
\label{lem:prelim:fact:emb-order-embedding}
%%%%%%%%%%%%%%%%%%%%%%%%%%%%%%%%%%%%%%%%%%%%%%%%%%%%%%%%%%%%%%%%%%%%%%%%%%%
Let $(\Q,\M)$ be a proper factorisation system on $\cat C$.
Given an embedding $f \colon X \emb Y$ in $\cat C$,
the function $\exists_f \colon \Emb{X} \to \Emb{Y}$
is an order-embedding.
%%%%%%%%%%%%%%%%%%%%%%%%%%%%%%%%%%%%%%%%%%%%%%%%%%%%%%%%%%%%%%%%%%%%%%%%%%%
\end{lemma}
%%%%%%%%%%%%%%%%%%%%%%%%%%%%%%%%%%%%%%%%%%%%%%%%%%%%%%%%%%%%%%%%%%%%%%%%%%%

In this paper, we mostly deal with factorisation systems
satisfying an additional stability condition:

%%%%%%%%%%%%%%%%%%%%%%%%%%%%%%%%%%%%%%%%%%%%%%%%%%%%%%%%%%%%%%%%%%%%%%%%%%%
\begin{definition}[Stable factorisation system]
%%%%%%%%%%%%%%%%%%%%%%%%%%%%%%%%%%%%%%%%%%%%%%%%%%%%%%%%%%%%%%%%%%%%%%%%%%%
A proper factorisation system $(\Q,\M)$ is \emph{stable}
if for every $e\in \Q$ and $m\in\M$ with common codomain,
the pullback of~$e$ along~$m$ exists and is a quotient.%
\footnote{Our terminology follows~\cite{ar21arboreal}, but differs
from the usual one, which requires the pullback of $e \in \Q$
along \emph{any} morphism to exist and to be in $\Q$.}
%%%%%%%%%%%%%%%%%%%%%%%%%%%%%%%%%%%%%%%%%%%%%%%%%%%%%%%%%%%%%%%%%%%%%%%%%%%
\end{definition}
%%%%%%%%%%%%%%%%%%%%%%%%%%%%%%%%%%%%%%%%%%%%%%%%%%%%%%%%%%%%%%%%%%%%%%%%%%%

%%%%%%%%%%%%%%%%%%%%%%%%%%%%%%%%%%%%%%%%%%%%%%%%%%%%%%%%%%%%%%%%%%%%%%%%%%%
\begin{example}
\label{ex:presheaf-cats-stable-fact-sys}
%%%%%%%%%%%%%%%%%%%%%%%%%%%%%%%%%%%%%%%%%%%%%%%%%%%%%%%%%%%%%%%%%%%%%%%%%%%
In a presheaf category,
(epimorphisms, monomorphisms) form a stable factorisation system.
Just observe that the corresponding property holds in $\Set$,
and limits and colimits in presheaf categories
are computed pointwise.
In particular, the (epimorphism, monomorphism) factorisation
systems on $\Tree$ and $\Forest$ are stable.
%%%%%%%%%%%%%%%%%%%%%%%%%%%%%%%%%%%%%%%%%%%%%%%%%%%%%%%%%%%%%%%%%%%%%%%%%%%
\end{example}
%%%%%%%%%%%%%%%%%%%%%%%%%%%%%%%%%%%%%%%%%%%%%%%%%%%%%%%%%%%%%%%%%%%%%%%%%%%

%%%%%%%%%%%%%%%%%%%%%%%%%%%%%%%%%%%%%%%%%%%%%%%%%%%%%%%%%%%%%%%%%%%%%%%%%%%
\begin{example}
%%%%%%%%%%%%%%%%%%%%%%%%%%%%%%%%%%%%%%%%%%%%%%%%%%%%%%%%%%%%%%%%%%%%%%%%%%%
The factorisation system (epimorphisms, strong monomorphisms) on $\Struct(\Sig)$
is stable.
Akin to Example~\ref{ex:presheaf-cats-stable-fact-sys},
this follows from the fact that limits and epimorphisms in $\Struct(\Sig)$
are computed in $\Set$, cf.\ e.g.~\cite[Remark 5.1]{ar94book}.
%%%%%%%%%%%%%%%%%%%%%%%%%%%%%%%%%%%%%%%%%%%%%%%%%%%%%%%%%%%%%%%%%%%%%%%%%%%
\end{example}
%%%%%%%%%%%%%%%%%%%%%%%%%%%%%%%%%%%%%%%%%%%%%%%%%%%%%%%%%%%%%%%%%%%%%%%%%%%

If $(\Q,\M)$ is 
a stable factorisation system,
then each quotient $f \colon X \epi Y$ induces a monotone function
$\ladj f \colon \Emb{Y} \to \Emb{X}$
that takes $n \colon S \emb Y$ to its pullback along $f$.
The map $\ladj f$ is right adjoint to $\exists_f \colon \Emb{X} \to \Emb{Y}$ (see e.g.\
\cite[Lemma 2.6]{ar21arboreal} for a proof).
Moreover, we have the following observation (see e.g.\ \cite[Lemma 2.7(b)]{ar21arboreal}):

%%%%%%%%%%%%%%%%%%%%%%%%%%%%%%%%%%%%%%%%%%%%%%%%%%%%%%%%%%%%%%%%%%%%%%%%%%%
\begin{lemma}
\label{lem:prelim:fact:epi-order-embedding}
%%%%%%%%%%%%%%%%%%%%%%%%%%%%%%%%%%%%%%%%%%%%%%%%%%%%%%%%%%%%%%%%%%%%%%%%%%%
Let $(\Q,\M)$ be
a stable
factorisation system on~$\cat C$.
For any quotient ${f \colon X \epi Y}$ in $\cat C$,
the map $\ladj f \colon \Emb{Y} \to \Emb{X}$
is an order-embedding.
%%%%%%%%%%%%%%%%%%%%%%%%%%%%%%%%%%%%%%%%%%%%%%%%%%%%%%%%%%%%%%%%%%%%%%%%%%%
\end{lemma}
%%%%%%%%%%%%%%%%%%%%%%%%%%%%%%%%%%%%%%%%%%%%%%%%%%%%%%%%%%%%%%%%%%%%%%%%%%%

We note in passing that Lemma~\ref{lem:prelim:fact:epi-order-embedding} depends
on both classes of the factorisation system $(\Q,\M)$,
whereas Lemma~\ref{lem:prelim:fact:emb-order-embedding}
only requires a class of monomorphisms $\M$ that is closed
under compositions and satisfies the condition in
Lemma~\ref{lem:prelim:fact:base}.\ref{item:prelim:fact:base:emb}.

%%%%%%%%%%%%%%%%%%%%%%%%%%%%%%%%%%%%%%%%%%%%%%%%%%%%%%%%%%%%%%%%%%%%%%%%%%%
\section{Back-and-forth equivalence}
\label{sec:path}
%%%%%%%%%%%%%%%%%%%%%%%%%%%%%%%%%%%%%%%%%%%%%%%%%%%%%%%%%%%%%%%%%%%%%%%%%%%

In this section, 
we recall from~\cite{ar21icalp,ar21arboreal}
the general notion of back-and-forth game
which is the main subject of this paper.
These games abstract usual notions of model comparison games,
such as the classical Ehrenfeucht-Fraïssé game.
Their definition is possible in all categories equipped
with a stable factorisation system and satisfying some mild
additional assumptions, embodied in the concept of \emph{wooded categories}.

We begin in~\S\ref{sec:path:ef}
by revisiting 
the classical Ehrenfeucht-Fraïssé game from a point of view 
well suited to our axiomatic categorical treatment.

Then, in~\S\ref{sec:path:wooded-and-equivalence} we introduce
wooded categories and discuss the notion of back-and-forth equivalence in a
wooded category $\C$.
In many situations of interest, equivalence of structures in
various logic fragments can be captured by transferring this equivalence relation
on $\C$ to another category $\E$, typically of the form $\Struct(\sig)$,
via a right adjoint functor $R\colon \E\to \C$
(cf.\ e.g.~\cite{adw17lics,as21jlc}).
The ensuing relation on $\E$,
referred to as \emph{$R$(-back-and-forth)-equivalence},
is introduced in \S\ref{sec:path:adjunction} and will play a central role throughout.

Finally, 
in \S\ref{sec:path:results} we state our main result (Theorem~\ref{thm:path:main})
which identifies the expressive power of $R$-equivalence in the case of
\emph{finitely accessible wooded adjunctions}.

%%%%%%%%%%%%%%%%%%%%%%%%%%%%%%%%%%%%%%%%%%%%%%%%%%%%%%%%%%%%%%%%%%%%%%%%%%%
\subsection{From concrete to abstract model comparison games}
\label{sec:path:ef}
%%%%%%%%%%%%%%%%%%%%%%%%%%%%%%%%%%%%%%%%%%%%%%%%%%%%%%%%%%%%%%%%%%%%%%%%%%%
Let $\sig$ be a (mono-sorted) purely relational signature.
We recall the classical Ehrenfeucht-Fraïssé game $\EF_\omega(M,N)$
played on $\sig$-structures~$M$ and~$N$~\cite{Ehr1960}.
The game is played by two players, \emph{Spoiler} and \emph{Duplicator}.
Positions in the game have the form $(\vec a;\vec b)$
where $\vec a = a_1,\dots,a_n$
and $\vec b = b_1,\dots,b_n$
are finite sequences of elements from $M$ and $N$, respectively.
Such a position is winning (for Duplicator)
if the assignment
\begin{equation}
\label{eq:partial-iso}
\begin{array}{l !{\quad} l}
  a_i \mapsto b_i,
& i = 1,\dots,n
\end{array}
\end{equation}

\noindent
is a \emph{partial isomorphism}, that is,
if this assignment
is an isomorphism between
the induced substructures of $M$ and $N$ with universes
$\{a_{1},\ldots,a_{n}\}$ and $\{b_{1},\ldots,b_{n}\}$, respectively.
Note
that~\eqref{eq:partial-iso} defines a partial isomorphism
precisely when
for every atomic formula
$x_1,\dots,x_n \sorting \atom$ in $\sig$,
we have
\[
  \vec a \in \I{\vec x \mid \atom}_M
  ~~\longiff~~
  \vec b \in \I{\vec x \mid \atom}_N.
\]

From a winning position $(\vec a;\vec b)$,
a round of the game proceeds as follows:
\begin{itemize}

\item
Spoiler chooses an element from one of the two structures, say $a' \in M$;

\item
Duplicator responds by choosing an element from the other structure, say $b' \in N$.
\end{itemize}

\noindent
If the new position $(\vec a, a'; \vec b, b')$ is not winning, 
then the game ends and is won by Spoiler.
Otherwise, the round is won by Duplicator and
the game continues from the new position $(\vec a, a'; \vec b, b')$.
Duplicator wins the game $\EF_\omega(M,N)$
with initial position $(\vec a; \vec b)$
if Duplicator has a strategy from $(\vec a;\vec b)$
that is winning after $t$ rounds for all integers $t \geq 0$.

For any integer $k \geq 0$, the \emph{$k$-round} Ehrenfeucht-Fraïssé game $\EF_k(M,N)$
is the restriction of $\EF_\omega(M,N)$ to plays of length $\leq k$.%
\footnote{An equivalent way to describe (winning strategies in) Ehrenfeucht-Fraïssé games
is in terms of back-and-forth systems of partial isomorphisms,
see e.g.\ \cite{Fra1954}, \cite[\S 2.3]{ef99book} or \cite[pp.~98--99]{hodges93book}.}

%%%%%%%%%%%%%%%%%%%%%%%%%%%%%%%%%%%%%%%%%%%%%%%%%%%%%%%%%%%%%%%%%%%%%%%%%%%
\begin{example}
\label{ex:path:games}
%%%%%%%%%%%%%%%%%%%%%%%%%%%%%%%%%%%%%%%%%%%%%%%%%%%%%%%%%%%%%%%%%%%%%%%%%%%
Assume that $\sig$ consists only of the binary relation $\Lt$.
Let $M = (M,<_M)$ and $N = (N,<_N)$ be dense linear orders without endpoints
(e.g.\ $(\QQ,<)$ and $(\RR,<)$).
Then from any winning position $(\vec a;\vec b)$ in $\EF_\omega(M,N)$,
Duplicator can respond to Spoiler's move so that the new position
$(\vec a,a' ; \vec b,b')$
is again winning.
Hence, Duplicator wins the game $\EF_\omega(M,N)$
with initial position given by the empty sequences
(see e.g.~\cite[\S 14.2]{rosenstein82book}).
%%%%%%%%%%%%%%%%%%%%%%%%%%%%%%%%%%%%%%%%%%%%%%%%%%%%%%%%%%%%%%%%%%%%%%%%%%%
\end{example}
%%%%%%%%%%%%%%%%%%%%%%%%%%%%%%%%%%%%%%%%%%%%%%%%%%%%%%%%%%%%%%%%%%%%%%%%%%%

Our axiomatic categorical approach to back-and-forth equivalence
relies on the following two observations.
First, 
provided $\sig$ is finite,
partial isomorphisms can be seen as pairs of embeddings
whose common domain is a finitely presentable $\sig$-structure.
A further insight, put forward in \cite{adw17lics} and expounded
in~\cite{as18csl,as21jlc},
is that collections of plays in a game form themselves $\sig$-structures.
We now elaborate on these two ideas.

%%%%%%%%%%%%%%%%%%%%%%%%%%%%%%%%%%%%%%%%%%%%%%%%%%%%%%%%%%%%%%%%%%%%%%%%%%%
\subsubsection{Positions as pairs of embeddings}
%%%%%%%%%%%%%%%%%%%%%%%%%%%%%%%%%%%%%%%%%%%%%%%%%%%%%%%%%%%%%%%%%%%%%%%%%%%
%First,
Assume that $\sig$ is finite and
let $\vec x \sorting \varphi$ be a (finite) conjunction of atomic formulae. 
By Fact~\ref{fact:lfp:struct},
the tuples $\vec a \in \I{\vec x \mid \varphi}_M$
correspond bijectively to
homomorphisms
\[
\begin{array}{*{5}{l}}
  h_{\vec a}
& \colon
& \FG{\vec x \mid \varphi}
& \longto
& M
\end{array}
\]

\noindent
Further, $h_{\vec a}$ is an embedding of $\sig$-structures exactly
when $\vec x \sorting \varphi$
is the conjunction of all atomic formulae
$\vec x \sorting \atom$ such that $\vec a \in \I{\vec x \mid \atom}_M$.
Hence, since $\sig$ is finite, tuples $\vec a \in M$
are in bijection with
substructures $\FG{\vec x \mid \varphi} \in \Emb{M}$
(i.e.\ isomorphism classes of embeddings 
$m \colon \FG{\vec x \mid \varphi} \emb M$,
in the sense of~\S\ref{sec:prelim:facto}),
where $\varphi$ ranges over (finite) conjunctions of atomic formulae.

%%%%%%%%%%%%%%%%%%%%%%%%%%%%%%%%%%%%%%%%%%%%%%%%%%%%%%%%%%%%%%%%%%%%%%%%%%%
\begin{remark}
\label{rem:path:infinite-sig-representability}
%%%%%%%%%%%%%%%%%%%%%%%%%%%%%%%%%%%%%%%%%%%%%%%%%%%%%%%%%%%%%%%%%%%%%%%%%%%
When $\sig$ is infinite,
given $\vec a \in M$ there may be no embedding $\FG{\vec x \mid \varphi} \emb M$,
with $\varphi$ finite, taking $\vec x$ to $\vec a$. For instance, suppose that $\sig$ contains a unary relation symbol $R_i$ for each positive integer $i$, and let $M$ be any $\sig$-structure containing an element $a$ such that $a\in R_i^M$ for all $i$. Then there is no embedding $\FG{x \mid \varphi} \emb M$ taking $x$ to $a$, as such a $\varphi$ can only account for finitely many of the relations satisfied by $a$.
%%%%%%%%%%%%%%%%%%%%%%%%%%%%%%%%%%%%%%%%%%%%%%%%%%%%%%%%%%%%%%%%%%%%%%%%%%%
\end{remark}
%%%%%%%%%%%%%%%%%%%%%%%%%%%%%%%%%%%%%%%%%%%%%%%%%%%%%%%%%%%%%%%%%%%%%%%%%%%

We shall represent a position $(\vec a, \vec b)$ in the game $\EF_\omega(M,N)$ by
the corresponding pair of embeddings $(m_{\vec a}, n_{\vec b})$, as displayed below.
\[
\begin{tikzcd}
& \FG{\vec x \mid \varphi}
  \arrow[rightarrowtail,
    start anchor={[xshift=+2pt, yshift=+3pt]south west}
  ]{dl}[above, xshift=-5pt]{m_{\vec a}}
& \FG{\vec x \mid \psi}
  \arrow[rightarrowtail,
    start anchor={[xshift=-3pt, yshift=+3pt]south east}
  ]{dr}{n_{\vec b}}
\\
  M
&
&
& N
\end{tikzcd}
\]

\noindent
Note that $(\vec a, \vec b)$ is winning 
exactly when
the domains of the embeddings $m_{\vec a}$ and $n_{\vec a}$ are isomorphic
as objects of $\Struct(\sig)$.

This results in a reformulation of the game $\EF_\omega(M,N)$
in which positions are pairs of embeddings $(m,n)$,
as depicted on the left-hand side below.
\[
\begin{array}{l !{\quad}|!{\quad} l}

\begin{tikzcd}[column sep=small]
& \FG{\vec x \mid \varphi}
  \arrow[rightarrowtail,
    start anchor={[xshift=+4pt, yshift=+3pt]south west}
]{dl}[above, xshift=-5pt]{m}
& \FG{\vec x \mid \psi}
  \arrow[rightarrowtail,
    start anchor={[xshift=-4pt, yshift=+3pt]south east}
]{dr}{n}
\\
  M
&
&
& N
\end{tikzcd}

&

\begin{tikzcd}[column sep=scriptsize]
& \FG{\vec x \mid \varphi}
  \arrow[rightarrowtail,
    start anchor={[xshift=+4pt, yshift=+3pt]south west}
  ]{dl}[above, xshift=-5pt]{m}
  \arrow[phantom]{r}[description]{\cong}
  \arrow[rightarrowtail]{d}
& \FG{\vec x \mid \psi}
  \arrow[rightarrowtail,
    start anchor={[xshift=-4pt, yshift=+3pt]south east}
  ]{dr}{n}
  \arrow[rightarrowtail, dashed]{d}
\\
  M
& \FG{\vec x,x \mid \varphi'}
  \arrow[rightarrowtail]{l}{m'}
& \FG{\vec x,x \mid \psi'}
  \arrow[rightarrowtail, dashed]{r}[below]{n'}
& N
\end{tikzcd}

\end{array}
\]

\noindent
Such a position $(m,n)$ is winning when $\dom(m) \cong \dom(n)$.
If $(m,n)$ is not winning, then the game stops and is won by Spoiler.
Otherwise,
either Spoiler
chooses some ${m' \colon \FG{\vec x,x' \mid \varphi'} \emb M}$
through which $m$ factors,
and Duplicator responds 
by choosing some ${n' \colon \FG{\vec x,x' \mid \psi'} \emb N}$
through which $n$ factors
(see the right-hand diagram above),
or the other way around. 
The game then proceeds from the new position $(m',n')$.

%%%%%%%%%%%%%%%%%%%%%%%%%%%%%%%%%%%%%%%%%%%%%%%%%%%%%%%%%%%%%%%%%%%%%%%%%%%
\subsubsection{Plays as structures}
%%%%%%%%%%%%%%%%%%%%%%%%%%%%%%%%%%%%%%%%%%%%%%%%%%%%%%%%%%%%%%%%%%%%%%%%%%%
Consider now the first projection of 
a finite play in the game $\EF_\omega(M,N)$,
say
\begin{equation*}
\label{eq:path:ef:projplay}
\begin{tikzcd}
   \FG{x_1 \mid \varphi_1}
   \arrow[rightarrowtail]{r}
   \arrow[rightarrowtail, bend right=5,
      start anchor={[xshift=-1pt, yshift=-2pt]south east},
      end anchor={[xshift=0pt, yshift=4pt]south west}
   ]{drrrr}%[below]{m_1}
& \FG{x_1,x_2 \mid \varphi_2}
   \arrow[rightarrowtail]{r}
   \arrow[rightarrowtail, bend right=5,
      start anchor={[xshift=-1pt, yshift=1.5pt]south east}
      %end anchor={[xshift=0pt, yshift=-3pt]north west}
   ]{drrr}%[above]{m_2}
& \cdots
  \arrow[rightarrowtail]{r}
  %\arrow[rightarrowtail]{dr}
& \FG{x_1,x_2,\dots,x_k \mid \varphi_k}
  \arrow[rightarrowtail,
      start anchor={[xshift=-3pt, yshift=+2pt]south east}
  ]{dr}%[above, xshift=5pt]{m_k}
\\
&
&
&
& M
\end{tikzcd}
\end{equation*}

The set $M^+$ of non-empty finite sequences over $M$
can be seen as a $\sig$-structure (see Example~\ref{ex:path:bisim-FOk-equivalence} for more details)
and is equipped with a natural
forest order, namely the prefix order.
The sequence of embeddings
in the top row above can be identified with a substructure of $M^+$,
whose underlying set consists of the images in $M$ of the finite sequences
$x_1,\, x_1 x_2,\, \ldots,\, x_1 x_2 \cdots x_k$;
this substructure forms a path, i.e.\ a linear order, in
the forest order of $M^+$.

The idea of viewing collections of plays in a game as $\sig$-structures is
 a key insight from~\cite{adw17lics},
which led, via~\cite{as18csl,as21jlc},
to the arboreal setting of~\cite{ar21icalp,ar21arboreal}.
Arboreal categories provide a
description of back-and-forth games based
on an axiomatic notion of \emph{path}, which abstracts the idea 
of linearly ordered finite $\sig$-structures.
As we shall recall in~\S\ref{sec:path:path-wooded},
paths can 
be defined in any (well-powered)
category equipped with a proper factorisation system.

Before proceeding, let us mention that
the arboreal notion of path
allows one to capture games 
with more complex or constrained plays
than in Ehrenfeucht-Fraïssé games,
such as pebble games and modal bisimulation games.%
\footnote{Actually, the above reformulation of Ehrenfeucht-Fraïssé games
hides some subtleties on equality
that will be addressed in
Example~\ref{ex:path:bisim-FOk-equivalence} below.}
See~\cite{AR2024} for further examples.

%%%%%%%%%%%%%%%%%%%%%%%%%%%%%%%%%%%%%%%%%%%%%%%%%%%%%%%%%%%%%%%%%%%%%%%%%%%
\subsection{Wooded categories and back-and-forth equivalence}
\label{sec:path:wooded-and-equivalence}
%%%%%%%%%%%%%%%%%%%%%%%%%%%%%%%%%%%%%%%%%%%%%%%%%%%%%%%%%%%%%%%%%%%%%%%%%%%

%%%%%%%%%%%%%%%%%%%%%%%%%%%%%%%%%%%%%%%%%%%%%%%%%%%%%%%%%%%%%%%%%%%%%%%%%%%
\subsubsection{Paths and wooded categories}
\label{sec:path:path-wooded}
%%%%%%%%%%%%%%%%%%%%%%%%%%%%%%%%%%%%%%%%%%%%%%%%%%%%%%%%%%%%%%%%%%%%%%%%%%%
Let us fix an arbitrary well-powered category $\C$
endowed with a proper factorisation system $(\Q,\M)$.

%%%%%%%%%%%%%%%%%%%%%%%%%%%%%%%%%%%%%%%%%%%%%%%%%%%%%%%%%%%%%%%%%%%%%%%%%%%
\begin{definition}[Path and path embedding] 
\label{def:path:path}
%%%%%%%%%%%%%%%%%%%%%%%%%%%%%%%%%%%%%%%%%%%%%%%%%%%%%%%%%%%%%%%%%%%%%%%%%%%
An object $P$ of $\C$ is a \emph{path} if $\Emb{P}$ is a finite chain.
Paths will be denoted by $P,Q$ and variations thereof.

A \emph{path embedding}
in $\C$
is an embedding $P\emb X$ whose domain is a path.
We let $\Path{X}$ denote the sub-poset of $\Emb{X}$ consisting of the path embeddings.
\end{definition}

%%%%%%%%%%%%%%%%%%%%%%%%%%%%%%%%%%%%%%%%%%%%%%%%%%%%%%%%%%%%%%%%%%%%%%%%%%%
\begin{example}
\label{ex:path:tree}
%%%%%%%%%%%%%%%%%%%%%%%%%%%%%%%%%%%%%%%%%%%%%%%%%%%%%%%%%%%%%%%%%%%%%%%%%%%
The paths in the category $\Tree$ 
are the (non-empty) finite chains.
Hence, for each $X \in \Tree$, the poset $\Path{X}$ is a tree
isomorphic to $X$.
%%%%%%%%%%%%%%%%%%%%%%%%%%%%%%%%%%%%%%%%%%%%%%%%%%%%%%%%%%%%%%%%%%%%%%%%%%%
\end{example}
%%%%%%%%%%%%%%%%%%%%%%%%%%%%%%%%%%%%%%%%%%%%%%%%%%%%%%%%%%%%%%%%%%%%%%%%%%%

Generalising Example~\ref{ex:path:tree},
the next lemma shows that $\Path{X}$ is always a forest,
and in fact a non-empty tree under mild additional assumptions.

%%%%%%%%%%%%%%%%%%%%%%%%%%%%%%%%%%%%%%%%%%%%%%%%%%%%%%%%%%%%%%%%%%%%%%%%%%%
\begin{lemma}[{\cite[Lemmas~3.5 and~3.12]{ar21arboreal}}]
\label{lem:path:base}
%%%%%%%%%%%%%%%%%%%%%%%%%%%%%%%%%%%%%%%%%%%%%%%%%%%%%%%%%%%%%%%%%%%%%%%%%%%
The following statements hold for all objects $X$ and paths $P$ of $\C$.
\begin{enumerate}[(1)]
\item
\label{item:path:base:emb}
If $f \colon X \emb P$ is an embedding then $X$ is a path.

\item
\label{item:path:base:forest}
$\Path{X}$ is a forest.

\setcounter{SplitEnum}{\value{enumi}}
\end{enumerate}

\noindent
Assume that $(\Q,\M)$ is stable.
\begin{enumerate}[(1)]
\setcounter{enumi}{\value{SplitEnum}}
\item
\label{item:path:base:epi}
If $f \colon P \epi X$ is a quotient then $X$ is a path.

\setcounter{SplitEnum}{\value{enumi}}
\end{enumerate}

\noindent
Assume further that $\C$ has an initial object $\zero$.
\begin{enumerate}[(1)]
\setcounter{enumi}{\value{SplitEnum}}
\item
\label{item:path:base:zero}
$\zero$ is a path.

\item
\label{item:path:base:tree}
$\Path{X}$ is a non-empty tree,
whose root $\zero_X \emb X$ is obtained by taking a
$(\Q,\M)$ factorisation $\zero \epi \zero_X \emb X$
of the unique morphism $\zero \to X$.
\end{enumerate}
\end{lemma}

%%%%%%%%%%%%%%%%%%%%%%%%%%%%%%%%%%%%%%%%%%%%%%%%%%%%%%%%%%%%%%%%%%%%%%%%%%%
\begin{remark}
\label{rem:path:base}
%%%%%%%%%%%%%%%%%%%%%%%%%%%%%%%%%%%%%%%%%%%%%%%%%%%%%%%%%%%%%%%%%%%%%%%%%%%
Items~\ref{item:path:base:emb} and~\ref{item:path:base:epi} of
Lemma~\ref{lem:path:base} are direct consequences
of Lemma~\ref{lem:prelim:fact:emb-order-embedding} and
Lemma~\ref{lem:prelim:fact:epi-order-embedding}, respectively.
Items~\ref{item:path:base:forest},
\ref{item:path:base:zero} and~\ref{item:path:base:tree}
are the content of~\cite[Lemma 3.12]{ar21arboreal}.

Note that if the factorisation system $(\Q,\M)$ is not stable,
then the poset $\Path{X}$ is still a tree, but possibly an empty one.
\end{remark}

Motivated by Lemma~\ref{lem:path:base},
we shall now introduce the notion of wooded category,
which provides a weakening of the concept
of arboreal category from~\cite{ar21icalp,ar21arboreal}.

%%%%%%%%%%%%%%%%%%%%%%%%%%%%%%%%%%%%%%%%%%%%%%%%%%%%%%%%%%%%%%%%%%%%%%%%%%%
\begin{definition}
%%%%%%%%%%%%%%%%%%%%%%%%%%%%%%%%%%%%%%%%%%%%%%%%%%%%%%%%%%%%%%%%%%%%%%%%%%%
A category $\C$ equipped with a factorisation system $(\Q,\M)$
is a \emph{wooded category}
when the following conditions hold:
\begin{enumerate}[(i)]
\item $\C$ is locally small and well powered.
\item $\C$ has an initial object $\zero$.
\item The factorisation system $(\Q,\M)$ is stable (and thus proper).
\end{enumerate}
%%%%%%%%%%%%%%%%%%%%%%%%%%%%%%%%%%%%%%%%%%%%%%%%%%%%%%%%%%%%%%%%%%%%%%%%%%%
\end{definition}
%%%%%%%%%%%%%%%%%%%%%%%%%%%%%%%%%%%%%%%%%%%%%%%%%%%%%%%%%%%%%%%%%%%%%%%%%%%

The following lemma follows from a straightforward adaptation of~\cite[Proposition~3.13]{ar21arboreal} 
combined with~\cite[Lemma~3.15(c)]{ar21arboreal}, whose proof applies unchanged (in fact, stability of the proper factorisation system is not required).

\begin{lemma}\label{l:forest-of-paths-wooded}
The following statements hold in any wooded category $\C$:
\begin{enumerate}
\item\label{i:at-most-one-emb} For any two paths $P,Q$ in $\C$, there is at most one embedding $P\emb Q$.
\item\label{i:C-p-forest} The subcategory $\pth\C$ of $\C$ defined by paths and embeddings between them is equivalent to a small forest.
\end{enumerate}
\end{lemma}

\begin{fullproof}
We claim that for any pathwise embedding $f\colon X\to Y$ in $\C$, the map 
\[
\Path{f} \colon \Path{X} \to \Path{Y}, \ \ [m]\mapsto [f\circ m]
\] 
is a forest morphism. 
Since $\Path{f}$ is monotone, it suffices to prove that it preserves the height of elements. In turn, this is equivalent to saying that, for any path embedding $m\colon P\emb X$, the induced map $\Path{f}\colon \down m\to \down \Path{f}(m)$ is a bijection.

We start by establishing surjectivity, i.e.\ $\down\Path{f}(m)\subseteq \Path{f}(\down m)$. 
If $n\colon Q\emb Y$ is a path embedding such that $n\leq \Path{f}(m)=f\circ m$ in $\Path{Y}$, there exists an embedding $k\colon Q\emb P$ such that the following diagram commutes:
\[\begin{tikzcd}
Q \arrow[rightarrowtail]{r}{n} \arrow[rightarrowtail]{d}[swap]{k} & Y \\
P \arrow[rightarrowtail]{ur}[swap]{f\circ m} &
\end{tikzcd}\]
Then $m\circ k \colon Q\emb X$ is a path embedding that is below $m$ in $\Path{X}$, and $\Path{f}(m\circ k)= n$.

For injectivity, let $m_1\colon P_1\emb X$ and $m_2\colon P_2\emb X$ be path embeddings in $\down m$. Since $P$ is a path, $m_1$ and $m_2$ are comparable in the order of $\Path{X}$. Assume without loss of generality that $m_1\leq m_2$, i.e.\ there is an embedding $k\colon P_1\emb P_2$ such that $m_1=m_2\circ k$. 
If $\Path{f}(m_1)=\Path{f}(m_2)$, there exists an isomorphism $k'\colon P_1\to P_2$ such that $f\circ m_{1} = f\circ m_{2}\circ k'$.
We get
\[
f\circ m_2\circ k = f\circ m_{1} = f\circ m_{2}\circ k'
\]
and so, since $f\circ m_{2}$ is an embedding (hence a monomorphism), $k=k'$. This shows that $k$ is an isomorphism and thus $m_1=m_2\circ k$ entails that $m_{1}=m_{2}$ as elements of $\Path{X}$.

This yields a functor $\Path\colon \pth\C\to \Tree$.
Now, suppose that $m,n\colon P\emb Q$ are embeddings between paths. Since there is at most one forest morphism between any two finite chains, we get $\Path{m} = \Path{n}$. In particular, $[m]=\Path{m}(\id_{P}) = \Path{n}(\id_{P})=[n]$ and so there is an isomorphism $i\colon P\to P$ such that $m = n\circ i$. Since $i$ is an epimorphism, $m=n$. 

This proves item~\ref{i:at-most-one-emb}. Item~\ref{i:C-p-forest} follows from item~\ref{i:at-most-one-emb} combined with the definition of path and the assumption that $\C$ is well-powered.
\end{fullproof}

%%%%%%%%%%%%%%%%%%%%%%%%%%%%%%%%%%%%%%%%%%%%%%%%%%%%%%%%%%%%%%%%%%%%%%%%%%%
\begin{example}
\label{ex:path:wooded:tree}
%%%%%%%%%%%%%%%%%%%%%%%%%%%%%%%%%%%%%%%%%%%%%%%%%%%%%%%%%%%%%%%%%%%%%%%%%%%
$\Tree$ is a wooded category.
The initial objects are the trees having exactly one node (their root).
%%%%%%%%%%%%%%%%%%%%%%%%%%%%%%%%%%%%%%%%%%%%%%%%%%%%%%%%%%%%%%%%%%%%%%%%%%%
\end{example}
%%%%%%%%%%%%%%%%%%%%%%%%%%%%%%%%%%%%%%%%%%%%%%%%%%%%%%%%%%%%%%%%%%%%%%%%%%%

%%%%%%%%%%%%%%%%%%%%%%%%%%%%%%%%%%%%%%%%%%%%%%%%%%%%%%%%%%%%%%%%%%%%%%%%%%%
\begin{example}
\label{ex:path:wooded:pos}
%%%%%%%%%%%%%%%%%%%%%%%%%%%%%%%%%%%%%%%%%%%%%%%%%%%%%%%%%%%%%%%%%%%%%%%%%%%
The category $\Set$ with the (epi, mono) factorisation system is a wooded category.
Another example is provided by the category $\Pos$ 
of Example~\ref{ex:prelim:coste:pos}
with the (epi, strong mono) factorisation system.
Recall from Remark~\ref{rem:lfp:struct} that $\Pos$ is lfp,
hence it is (co)complete and the aforementioned factorisation system is proper
(cf.\ Example~\ref{ex:prelim:fact:lfp}).
We show that this factorisation system is stable.

Note that epimorphisms in $\Pos$ coincide with
the surjective morphisms.
Just observe that, if $f \colon X \to Y$ is a morphism in $\Pos$ and $b \in Y$,
the characteristic functions of
the sets
$\up b = \{y \mid b \leq_Y y \}$
and
$\up b \setminus \{b\}$
yield distinct morphisms $Y\rightrightarrows 2 = \{0 \leq 1\}$.
But if $b \notin \Img(f)$,
then the two compositions with $f$
yield equal morphisms $X \to Y \rightrightarrows 2$.
Hence, epimorphisms are surjective; the converse is trivial.
Now, the forgetful functor
$U \colon \Pos \to \Set$ has a left adjoint sending a set~$S$ to the poset $(S,=)$;
in particular, $U$ preserves limits.
Using the fact that $U$ preserves pullbacks, and it preserves and reflects epimorphisms,
the stability of epimorphisms under pullbacks in $\Set$ entails
at once the same property for $\Pos$.

Further, note that the strong monos in $\Pos$ are exactly the order embeddings
(an easy diagrammatic reasoning shows that order embeddings are strong monos;
the converse is provided by Remark~\ref{rem:emb:strong} below).

Both in $\Set$ and $\Pos$, the paths are the
initial objects along with the terminal ones.
%%%%%%%%%%%%%%%%%%%%%%%%%%%%%%%%%%%%%%%%%%%%%%%%%%%%%%%%%%%%%%%%%%%%%%%%%%%
\end{example}
%%%%%%%%%%%%%%%%%%%%%%%%%%%%%%%%%%%%%%%%%%%%%%%%%%%%%%%%%%%%%%%%%%%%%%%%%%%

%%%%%%%%%%%%%%%%%%%%%%%%%%%%%%%%%%%%%%%%%%%%%%%%%%%%%%%%%%%%%%%%%%%%%%%%%%%
\begin{remark}
\label{rem:path:wooded:projsep}
%%%%%%%%%%%%%%%%%%%%%%%%%%%%%%%%%%%%%%%%%%%%%%%%%%%%%%%%%%%%%%%%%%%%%%%%%%%
Generalising the argument in Example~\ref{ex:path:wooded:pos},
the stable factorisation system on $\Set$ can be ``lifted''
to other categories under appropriate assumptions.
Consider for example an lfp category $\C$,
an object $X\in\C$, and the functor $\C(X,-)\colon \C\to\Set$.
If $X$ is a projective separator in $\C$,
then the (epi, strong mono) factorisation system on $\C$ is stable,
and so $\C$ is a wooded category.
In the case of $\Pos$, we can take as $X$ the terminal object.
Similarly, $\Struct(\Sig)$ is also a wooded category with respect to the
(epi, strong mono) factorisation system,
as can be seen by taking as $X$ the sub-terminal object having exactly
one element for each sort, with empty relations and the unique possible
interpretations of the function symbols.
(Note, however, that the most important examples of wooded categories
in this paper arise from categories of coalgebras for game comonads,
and $\Struct(\Sig)$ is not one of these.)
%%%%%%%%%%%%%%%%%%%%%%%%%%%%%%%%%%%%%%%%%%%%%%%%%%%%%%%%%%%%%%%%%%%%%%%%%%%
\end{remark}
%%%%%%%%%%%%%%%%%%%%%%%%%%%%%%%%%%%%%%%%%%%%%%%%%%%%%%%%%%%%%%%%%%%%%%%%%%%

We now turn to examples of wooded categories arising from game comonads.
Most of them were introduced in~\cite{as18csl,as21jlc}
and shown to be arboreal, hence wooded, in~\cite{ar21icalp,ar21arboreal}.

%%%%%%%%%%%%%%%%%%%%%%%%%%%%%%%%%%%%%%%%%%%%%%%%%%%%%%%%%%%%%%%%%%%%%%%%%%%
\begin{example}[The categories $\cat R_k(\sig)$]
\label{ex:path:R}
%%%%%%%%%%%%%%%%%%%%%%%%%%%%%%%%%%%%%%%%%%%%%%%%%%%%%%%%%%%%%%%%%%%%%%%%%%%
Let $\sig$ be a (mono-sorted) relational signature.
A \emph{forest-ordered $\sig$-structure} is a pair
$(M, {\leq})$ where $M$ is a $\sig$-structure and $\leq$ is a forest order on~$M$.
A morphism of forest-ordered $\sig$-structures
is a homomorphism of $\sig$-structures
that is also a forest morphism. This yields a category $\cat R_\omega(\sig)$.
Given $0 < k < \omega$, 
$\cat R_k(\sig)$ is the full subcategory of $\cat R_\omega(\sig)$
on those forest-ordered $\sig$-structures $M$ such that
for each $a \in M$, the finite chain $\down a$
has cardinality at most $k$ (cf.\ Example~\ref{ex:prelim:forests}).

Let $0 < k \leq \omega$.
We equip $\cat R_k(\sig)$ with the factorisation system $(\Q,\M)$
where $\Q$ consists of the surjective morphisms and
$\M$ of the morphisms that are embeddings of $\sig$-structures
(in the sense of Example~\ref{ex:prelim:fact:struct}).
Note that $(\Q,\M)$ is stable;
this follows from Example~\ref{ex:presheaf-cats-stable-fact-sys} and the fact that
pullbacks in $\Struct(\sig)$
are computed in $\Set$
\cite[Remark 5.1(2)]{ar94book}.
%%%%%%%%%%%%%%%%%%%%%%%%%%%%%%%%%%%%%%%%%%%%%%%%%%%%%%%%%%%%%%%%%%%%%%%%%%%
\end{example}
%%%%%%%%%%%%%%%%%%%%%%%%%%%%%%%%%%%%%%%%%%%%%%%%%%%%%%%%%%%%%%%%%%%%%%%%%%%

As illustrated in the following examples, the categories of coalgebras for various game comonads
can be represented, up to isomorphism,
as subcategories of some $\cat R_k(\sig)$, or minor variants thereof.
Recall that the \emph{Gaifman graph} of a $\sig$-structure $M$
has the elements of $M$ as vertices,
and as edges those pairs of distinct vertices that both
occur in some tuple $\vec a \in R^M$ for some relation symbol $R \in \sig$.

%%%%%%%%%%%%%%%%%%%%%%%%%%%%%%%%%%%%%%%%%%%%%%%%%%%%%%%%%%%%%%%%%%%%%%%%%%%
\begin{example}
\label{ex:path:R^E}
%%%%%%%%%%%%%%%%%%%%%%%%%%%%%%%%%%%%%%%%%%%%%%%%%%%%%%%%%%%%%%%%%%%%%%%%%%%
Given $0 < k \leq \omega$,
we let $\cat R^E_k(\sig)$ be the full subcategory of $\cat R_k(\sig)$
determined by those objects $(M, {\leq})$ satisfying the following condition:
\begin{enumerate}[label=\textnormal{(E)}]
\item\label{E}
If $a,b \in M$ are adjacent in the Gaifman graph of the $\sig$-structure $M$,
then they are comparable in the forest order.
\end{enumerate}

\noindent
We equip $\cat R^E_k(\sig)$ with the restriction of the factorisation system on
$\cat R_k(\sig)$.
The paths of $\cat R^E_k(\sig)$ are
those objects in which the order is a finite chain.

Assume $k < \omega$.
The paths of $\cat R^E_k(\sig)$ are
the finite chains of cardinality $\leq k$.
When~$\sig$ is finite,
it follows from \cite[Theorem 9.1]{as21jlc}
that $\cat R^E_k(\sig)$ is isomorphic
to the category of coalgebras for the \emph{Ehrenfeucht-Fraïssé} comonad
$\gc{E}_k$ on $\Struct(\sig)$.
The objects $(M, {\leq})$ of~$\cat R^E_k(\sig)$ are forest covers of $M$ witnessing
that its \emph{tree-depth} is at most $k$ \cite{NOdM2006}.
The categories $\cat R^E_k(\sig)$ were introduced in~\cite{as18csl,as21jlc} and shown to be arboreal, hence wooded, in~\cite{ar21icalp,ar21arboreal}.
%%%%%%%%%%%%%%%%%%%%%%%%%%%%%%%%%%%%%%%%%%%%%%%%%%%%%%%%%%%%%%%%%%%%%%%%%%%
\end{example}
%%%%%%%%%%%%%%%%%%%%%%%%%%%%%%%%%%%%%%%%%%%%%%%%%%%%%%%%%%%%%%%%%%%%%%%%%%%

%%%%%%%%%%%%%%%%%%%%%%%%%%%%%%%%%%%%%%%%%%%%%%%%%%%%%%%%%%%%%%%%%%%%%%%%%%%
\begin{example}
\label{ex:path:misc}
%%%%%%%%%%%%%%%%%%%%%%%%%%%%%%%%%%%%%%%%%%%%%%%%%%%%%%%%%%%%%%%%%%%%%%%%%%%
In the same vein as Example~\ref{ex:path:R^E},
one can consider the categories of coalgebras for other game comonads
such as the
the \emph{pebble} and \emph{modal} comonads~\cite{adw17lics,as18csl}
(encoding, respectively, pebble and bisimulation games)
and the \emph{hybrid} comonad~\cite{am22mfcs}
(encoding bisimulation games for hybrid logic).
\begin{enumerate}[(1)]
\item
\label{item:path:misc:pebble}
For the pebbling comonad,
for each integer $k > 0$ this is the category $\cat R^P_k(\sig)$
whose objects have the form $(M, {\leq}, p)$, where 
$(M, {\leq}) \in \cat R_\omega(\sig)$
and $p\colon M \to \{1,\ldots,k\}$ is a ``pebbling'' function.
In addition to condition (E), these structures have to satisfy the condition (P):
if $a$ is adjacent to $b$ in the Gaifman graph of $M$, and $a < b$ in the forest order,
then for all $x$ such that $a < x \leq b$, $p(a) \neq p(x)$.

When $\sig$ is finite,
it is shown in~\cite{as21jlc} that these structures are equivalent
to the more familiar form of tree decomposition used to define
tree-width~\cite{kloks94lncs}.
Morphisms have to preserve the pebbling function.

\item
\label{item:path:misc:modal}
Suppose $\sig$ is a modal vocabulary,
i.e.\ it consists only of unary and binary relation symbols.
For the modal comonad, the category $\cat R^M_k(\sig)$
has as objects the non-empty tree-ordered Kripke structures
$M \in \cat R_k(\sig)$
satisfying the condition (M):
for $x, y \in M$, $x \cvr y$ if and only if for some unique
binary relation $\atom$ in $\sig$ (the ``transition relation''),
we have $\atom^{M}(x, y)$.

\item
For the hybrid comonad, 
one considers again a modal vocabulary $\sig$,
and the category $\cat R^H_k(\sig)$ consisting of 
non-empty tree-ordered Kripke structures
$M \in \cat R_k(\sig)$,
but with condition (M) weakened to condition~\ref{E}
of the Ehrenfeucht-Fraïssé case.
\end{enumerate}

These categories were shown to be arboreal in~\cite{ar21icalp,ar21arboreal}
and~\cite{am22mfcs}.
The paths in each of them 
are those structures 
in which the order is a finite chain.
Note that in the modal case,
ignoring the interpretations of the propositional variables
(i.e., the unary relations),
these correspond to synchronization trees consisting of a single branch, i.e.\ traces.
%%%%%%%%%%%%%%%%%%%%%%%%%%%%%%%%%%%%%%%%%%%%%%%%%%%%%%%%%%%%%%%%%%%%%%%%%%%
\end{example}
%%%%%%%%%%%%%%%%%%%%%%%%%%%%%%%%%%%%%%%%%%%%%%%%%%%%%%%%%%%%%%%%%%%%%%%%%%%

For further examples based on game comonads,
see e.g. the surveys~\cite{abramsky22fi,AR2024}.

%%%%%%%%%%%%%%%%%%%%%%%%%%%%%%%%%%%%%%%%%%%%%%%%%%%%%%%%%%%%%%%%%%%%%%%%%%%
\subsubsection{Back-and-forth games in wooded categories} 
\label{sec:path:games}
%%%%%%%%%%%%%%%%%%%%%%%%%%%%%%%%%%%%%%%%%%%%%%%%%%%%%%%%%%%%%%%%%%%%%%%%%%%
Let $\C$ be a category equipped with a proper factorisation system $(\Q,\M)$,
and let $a, b$ be objects of $\C$.
We define the following back-and-forth game $\G(a,b)$ played on $a$ and $b$.

%%%%%%%%%%%%%%%%%%%%%%%%%%%%%%%%%%%%%%%%%%%%%%%%%%%%%%%%%%%%%%%%%%%%%%%%%%%
\begin{definition}[Game $\G(a,b)$]
\label{def:games}
%%%%%%%%%%%%%%%%%%%%%%%%%%%%%%%%%%%%%%%%%%%%%%%%%%%%%%%%%%%%%%%%%%%%%%%%%%%
Positions in the game are pairs $(m,n)\in\Path{a}\times\Path{b}$;
the initial position is denoted by $(m_0,n_0)$.
The winning relation $\W(a,b)$ consists of the pairs $(m,n)$ such that
$\dom(m)\cong\dom(n)$.
\begin{itemize}
\item
If $(m_0,n_0)\notin \W(a,b)$, Duplicator loses the game. Otherwise,
the game proceeds as follows:

\item
At the start of each round,
the position is specified by a pair $(m,n)\in\Path{a}\times\Path{b}$
and the round proceeds as follows.
Either Spoiler chooses some $m'\succ m$
and Duplicator must respond with some $n'\succ n$,
or Spoiler chooses some $n''\succ n$ and Duplicator must respond with $m''\succ m$. 

\item
Duplicator wins the round if they are able to respond
and the new position is in $\W(a,b)$.
Duplicator wins the game with initial position $(m_0,n_0)$
if they have a strategy that is winning after $t$ rounds,
for all integer $t\geq 0$.
\end{itemize}
\end{definition}

%%%%%%%%%%%%%%%%%%%%%%%%%%%%%%%%%%%%%%%%%%%%%%%%%%%%%%%%%%%%%%%%%%%%%%%%%%%
\begin{remark}
\label{rem:path:games:length}
%%%%%%%%%%%%%%%%%%%%%%%%%%%%%%%%%%%%%%%%%%%%%%%%%%%%%%%%%%%%%%%%%%%%%%%%%%%
The game $\G(a,b)$ need not have finite length;
however, by the definition of path, its length is at most~$\omega$.
Also, $\G(a,b)$ has a \emph{set} of positions (as opposed to a proper class)
when $\C$ is well-powered (which is in particular the case if $\C$ is wooded).
\end{remark}

When $\C$ is a wooded category,
we write $\bot_a \colon \zero_a \emb a$ for the root
of the tree $\Path{a}$, and similarly for $b$.

%%%%%%%%%%%%%%%%%%%%%%%%%%%%%%%%%%%%%%%%%%%%%%%%%%%%%%%%%%%%%%%%%%%%%%%%%%%
\begin{definition}[Back-and-forth equivalence] 
\label{def:path:bfe}
%%%%%%%%%%%%%%%%%%%%%%%%%%%%%%%%%%%%%%%%%%%%%%%%%%%%%%%%%%%%%%%%%%%%%%%%%%%
Let $a$ and $b$ be objects of a wooded category.
We say that $a$ and $b$ are \emph{back-and-forth equivalent},
and denote it by $a \bisim b$, when Duplicator wins the game
$\G(a,b)$ with initial position $(\bot_a,\bot_b)$.
%%%%%%%%%%%%%%%%%%%%%%%%%%%%%%%%%%%%%%%%%%%%%%%%%%%%%%%%%%%%%%%%%%%%%%%%%%%
\end{definition}
%%%%%%%%%%%%%%%%%%%%%%%%%%%%%%%%%%%%%%%%%%%%%%%%%%%%%%%%%%%%%%%%%%%%%%%%%%%

For instance, as we shall see in~\S\ref{sec:arboreal},
two objects $a,b \in \Tree$ are back-and-forth equivalent if, and only if,
they are bisimilar (when regarded as Kripke frames).
More generally, we will see that if $\C$ is an arboreal category,
back-and-forth equivalence induced by games in $\C$ coincides with
\emph{open map bisimilarity} in the sense of \cite{jnw93lics}.
Moreover, in this case, if $a,b\in \C$ are back-and-forth equivalent then
the trees $\Path a$ and $\Path b$ are bisimilar.

%%%%%%%%%%%%%%%%%%%%%%%%%%%%%%%%%%%%%%%%%%%%%%%%%%%%%%%%%%%%%%%%%%%%%%%%%%%
\subsection{Wooded adjunctions and $R$-equivalence}
\label{sec:path:adjunction}
%%%%%%%%%%%%%%%%%%%%%%%%%%%%%%%%%%%%%%%%%%%%%%%%%%%%%%%%%%%%%%%%%%%%%%%%%%%
We now consider the relation obtained by transferring back-and-forth equivalence in a wooded category $\C$ along a right adjoint functor $R$.
%%%%%%%%%%%%%%%%%%%%%%%%%%%%%%%%%%%%%%%%%%%%%%%%%%%%%%%%%%%%%%%%%%%%%%%%%%%
\begin{definition}
\label{def:wooded-adjunction}
%%%%%%%%%%%%%%%%%%%%%%%%%%%%%%%%%%%%%%%%%%%%%%%%%%%%%%%%%%%%%%%%%%%%%%%%%%%
A \emph{wooded adjunction} is an adjoint pair
\[
\begin{tikzcd}
  \C
  \arrow[r, bend left=25, ""{name=U, below}, "\Ladj"{above}]
  \arrow[r, leftarrow, bend right=25, ""{name=D}, "R"{below}]
& \E
  \arrow[phantom, "\textnormal{\footnotesize{$\bot$}}", from=U, to=D] 
\end{tikzcd}
\]

\noindent
where $\C$ is wooded category. 
Since the left adjoint is unique up to natural isomorphism,
we denote a wooded adjunction simply by $R\colon \E\to\C$.
%%%%%%%%%%%%%%%%%%%%%%%%%%%%%%%%%%%%%%%%%%%%%%%%%%%%%%%%%%%%%%%%%%%%%%%%%%%
\end{definition}
%%%%%%%%%%%%%%%%%%%%%%%%%%%%%%%%%%%%%%%%%%%%%%%%%%%%%%%%%%%%%%%%%%%%%%%%%%%

We think of $\E$ as an ``extensional'' category;
this is typically a category of relational structures 
but need not be in general,
see Examples~\ref{ex:path:pos} and~\ref{ex:Diaconescu-cover} below.
A wooded adjunction $R\colon \E\to \C$ allows us to transfer the
back-and-forth equivalence relation $\bisim$ on $\C$ to $\E$.
This yields the relation of
\emph{$R$(-back-and-forth)-equivalence} on $\E$,
denoted by $\bisim_R$ and defined, for all $d,e\in\E$, as
\[
d \bisim_R e
  \ \longiff \
  Rd \bisim Re.
\]

Choosing the wooded adjunction $R\colon \E\to \C$ in a suitable way,
we can capture equivalence of structures in various fragments of first-order
logic by means of $R$-equivalence, as illustrated in the following examples.

%%%%%%%%%%%%%%%%%%%%%%%%%%%%%%%%%%%%%%%%%%%%%%%%%%%%%%%%%%%%%%%%%%%%%%%%%%%
\begin{example}
\label{ex:path:bisim-FOk-equivalence}
%%%%%%%%%%%%%%%%%%%%%%%%%%%%%%%%%%%%%%%%%%%%%%%%%%%%%%%%%%%%%%%%%%%%%%%%%%%
Let $\sig$ be a (mono-sorted) relational signature.
For each $0 < k \leq \omega$,
consider the wooded category $\cat R^E_k(\sig)$ from
Example~\ref{ex:path:R^E}.
The forgetful functor
$\Ladj_k \colon \cat R^E_k(\sig) \to \Struct(\sig)$
has a right adjoint $R_k$, hence it yields a wooded adjunction
\begin{equation}
\label{eq:wooded-adj-EF-comonadic}
\begin{tikzcd}
  \cat R^E_k(\sig)
  \arrow[r, bend left=25, ""{name=U, below}, "\Ladj_k"{above}]
  \arrow[r, leftarrow, bend right=25, ""{name=D}, "R_k"{below}]
& \Struct(\sig).
  \arrow[phantom, "\textnormal{\footnotesize{$\bot$}}", from=U, to=D] 
\end{tikzcd}
\end{equation}

\noindent
(In fact, this adjunction is comonadic.)
The functor $R_k$ sends $M\in \Struct(\sig)$ to the set of all
finite non-empty sequences 
$a_0,\dots,a_i$ ($i < k$) of elements of $M$,
equipped with the prefix order and an appropriate lifting of the relations of $M$
(cf.\ \cite[\S3.1]{as21jlc}). E.g., if $S$ is a binary relation symbol in $\sig$, its interpretation in $R_{k}M$ consists of the pairs of sequences 
\[
(s,t) = ((a_0,\dots,a_i), (b_0,\dots,b_j))\in R_{k}M \times R_{k}M
\] 
such that (i) either $s$ is a prefix of $t$, or $t$ is a prefix of $s$, and (ii) $(a_{i},b_{j})\in S^{M}$, i.e.\ the last elements of $s$ and $t$ are $S$-related in $M$.

The induced relations of $R_k$-equivalence on $\Struct(\sig)$
only capture equivalence with respect to \emph{equality-free} fragments of logic,
see~\cite[Remark 10.6]{as21jlc} and~\cite{cdj96ndjfl}.
To capture the equality symbol, we adopt the approach from~\cite{djr21lics,ar21arboreal} and
consider a fresh binary relation symbol $I$, and the expanded signature $\sig^I\deq \sig \cup\{I\}$.
The fully faithful functor
$J\colon \Struct(\sig) \to \Struct(\sig^I)$
that interprets $I$ as the identity relation has a left adjoint $H$,
which sends $M\in \Struct(\sig^I)$ to the quotient of (the $\sig$-reduct of) $M$
by the equivalence relation generated by $I^M$ \cite[Lemma~25]{djr21lics}.
The composite adjunction
\begin{equation}
\label{eq:EF-wooded-adjunction}
\begin{tikzcd}
{\cat R^E_k(\sig^I)} \arrow[r, bend left=25, ""{name=U, below}, "\Ladj_k"{above}]
\arrow[r, leftarrow, bend right=25, ""{name=D}, "R_k"{below}]
& {\Struct(\sig^I)} \arrow[r, bend left=25, ""{name=U', below}, "H"{above}]
\arrow[r, leftarrow, bend right=25, ""{name=D'}, "J"{below}] & {\Struct(\sig)},
\arrow[phantom, "\textnormal{\footnotesize{$\bot$}}", from=U, to=D] 
\arrow[phantom, "\textnormal{\footnotesize{$\bot$}}", from=U', to=D'] 
\end{tikzcd}
\end{equation}

\noindent
is wooded (although not comonadic).
In eq.~\eqref{eq:EF-wooded-adjunction} above, 
$\Ladj_k \adj R_k$ stands for the adjunction defined in
eq.~\eqref{eq:wooded-adj-EF-comonadic},
but this time for the signature $\sig^I$.

Let $R^I_k\deq R_k J$ and let $\sig$ be finite.
If $k < \omega$, it follows from~\cite[Theorem~10.4]{as21jlc} that for any two $\sig$-structures $M,N$,
the back-and-forth game $\G(R^I_k (M), R^I_k (N))$ 
corresponds precisely to
the $k$-round Ehrenfeucht-Fraïssé game $\EF_k(M,N)$.
Thus, by the classical Ehrenfeucht-Fraïssé theorem \cite{Ehr1960,Fra1954},
the induced relation of $R^I_k$-equivalence
coincides with equivalence of $\sig$-structures
in the fragment of $\Lang_{\omega}$ with quantifier rank at most $k$.

When $k = \omega$, $R_\omega(M)$ consists of \emph{all} finite sequences of elements from $M$.
The game $\G(R^I_\omega(M), R^I_\omega(N))$ then corresponds to
the infinite Ehrenfeucht-Fraïssé game $\EF_\omega(M,N)$.
By Karp's Theorem (see e.g.~\cite[Corollary~3.5.3]{hodges93book}),
the relation of $R^I_\omega$-equivalence
coincides with equivalence of $\sig$-structures in $\Lang_{\infty}$.
%%%%%%%%%%%%%%%%%%%%%%%%%%%%%%%%%%%%%%%%%%%%%%%%%%%%%%%%%%%%%%%%%%%%%%%%%%%
\end{example}
%%%%%%%%%%%%%%%%%%%%%%%%%%%%%%%%%%%%%%%%%%%%%%%%%%%%%%%%%%%%%%%%%%%%%%%%%%%

Combining Examples~\ref{ex:path:games} and~\ref{ex:path:bisim-FOk-equivalence}
yields the following
well-known consequence of Karp's Theorem.

%%%%%%%%%%%%%%%%%%%%%%%%%%%%%%%%%%%%%%%%%%%%%%%%%%%%%%%%%%%%%%%%%%%%%%%%%%%
\begin{example}
\label{ex:path:karp}
%%%%%%%%%%%%%%%%%%%%%%%%%%%%%%%%%%%%%%%%%%%%%%%%%%%%%%%%%%%%%%%%%%%%%%%%%%%
In the signature $\sig$ consisting only of the binary relation $\Lt$,
any two dense linear orderings without endpoints are
$R^I_\omega$-equivalent and thus $\Lang_\infty$-equivalent.
This applies for instance to $(\QQ,<)$ and $(\RR,<)$, see e.g.~\cite[Theorem~14.20]{rosenstein82book}.
\end{example}

%%%%%%%%%%%%%%%%%%%%%%%%%%%%%%%%%%%%%%%%%%%%%%%%%%%%%%%%%%%%%%%%%%%%%%%%%%%
\begin{example}
\label{ex:path:misc:games}
%%%%%%%%%%%%%%%%%%%%%%%%%%%%%%%%%%%%%%%%%%%%%%%%%%%%%%%%%%%%%%%%%%%%%%%%%%%
The wooded categories of Example~\ref{ex:path:misc}
lead to wooded adjunctions that capture
equivalence in
finite-variable fragments of $\Lang_{\infty}$ \cite{adw17lics},
modal logic with bounded modal-depth \cite{as18csl}
and hybrid logic with bounded hybrid modal-depth \cite{am22mfcs}, respectively.
(In fact, these are \emph{arboreal} adjunctions, cf.\ \S\ref{sec:arboreal:def}.)
\begin{enumerate}[(1)]
\item
\label{item:path:misc:games:pebble}
For the $k$-pebble comonad,
there is a wooded adjunction $R^P_k \colon \Struct(\sig) \to \cat R^P_k(\sig)$
that sends a $\sig$-structure $M$ to a structure whose universe is the set of all finite sequences of pairs of the form $(p,a)$, where $p\in\{1,\ldots,k\}$ is the \emph{pebble index} and $a\in M$.
See~\cite[\S 3.2]{as21jlc} for more details.
To recover equivalence in the $k$-variable fragment of $\Lang_\infty$,
equality is handled as in
Example~\ref{ex:path:bisim-FOk-equivalence}.
In fact, when $\sig$ is finite, \cite[Theorem 10.9]{as21jlc} shows 
that the usual $k$-pebble game played on
$\sig$-structures $M$ and $N$
corresponds to the back-and-forth game
played on $R^P_k J M$ and $R^P_k J N$.

\item
\label{item:path:misc:games:modal}
For the modal comonad,
equivalence in modal logic with modal-depth at most~$k$ is
captured by $R^M_k$-equivalence for a suitable right adjoint
$R^M_k\colon \Struct_{\bullet}(\sig) \to \cat R^M_k(\sig)$, where $\sig$ is a modal vocabulary and
$\Struct_{\bullet}(\sig)$ is the category of pointed $\sig$-structures and homomorphisms
preserving the distinguished elements;
see~\cite[\S 3.3]{as21jlc} for a detailed account.
(Note that there is no equality involved here.)

\item
\label{item:path:misc:games:hybrid}
To capture logical equivalence in hybrid logic with bounded hybrid modal-depth,
via the wooded adjunction induced by the hybrid comonad,
requires handling the equality symbol as in the case of the Ehrenfeucht-Fraïssé
and pebble comonads; cf.~\cite{am22mfcs} for more details.
\end{enumerate}
%%%%%%%%%%%%%%%%%%%%%%%%%%%%%%%%%%%%%%%%%%%%%%%%%%%%%%%%%%%%%%%%%%%%%%%%%%%
\end{example}
%%%%%%%%%%%%%%%%%%%%%%%%%%%%%%%%%%%%%%%%%%%%%%%%%%%%%%%%%%%%%%%%%%%%%%%%%%%

The wooded adjunctions 
corresponding to finite Ehrenfeucht-Fraïssé games, and to finite games for modal and hybrid logics, can be restricted by taking finite structures on both sides.
However, this is not always possible:
for instance, in the case of pebble games
one is forced to consider infinite structures as well~\cite[Corollary~8]{adw17lics}.

Let us also note that, in addition to back-and-forth equivalence, wooded adjunctions can be used to transfer other
relations to the extensional category,
such as the isomorphism relation or the homomorphism preorder.
These also result in characterisations of equivalence in important logic fragments;
cf.~\cite[\S7]{ar21arboreal} for the case of \emph{arboreal adjunctions}.

%%%%%%%%%%%%%%%%%%%%%%%%%%%%%%%%%%%%%%%%%%%%%%%%%%%%%%%%%%%%%%%%%%%%%%%%%%%
\subsubsection{Finitely accessible wooded adjunctions}
\label{sec:path:finaccadj}
%%%%%%%%%%%%%%%%%%%%%%%%%%%%%%%%%%%%%%%%%%%%%%%%%%%%%%%%%%%%%%%%%%%%%%%%%%%
Infinite Ehrenfeucht-Fraïssé games
provide an example of wooded adjunction whose back-and-forth relation
corresponds exactly to equivalence in infinitary first-order logic $\Lang_{\infty}$
(cf.\ Example~\ref{ex:path:bisim-FOk-equivalence}).
Our main result yields a partial converse to this:
we show that, for a wide class of wooded adjunctions
\[
\begin{tikzcd}
  \C
  \arrow[r, bend left=25, ""{name=U, below}, "\Ladj"{above}]
  \arrow[r, leftarrow, bend right=25, ""{name=D}, "R"{below}]
& \E,
  \arrow[phantom, "\textnormal{\footnotesize{$\bot$}}", from=U, to=D] 
\end{tikzcd}
\]

\noindent
the induced relation of $R$-equivalence is always coarser than equivalence in infinitary
first-order logic (see \S\ref{sec:path:results} for precise statement).
The class of wooded adjunctions that we shall be interested in consists
of the finitely accessible ones:

%%%%%%%%%%%%%%%%%%%%%%%%%%%%%%%%%%%%%%%%%%%%%%%%%%%%%%%%%%%%%%%%%%%%%%%%%%%
\begin{definition} 
\label{def:path:finaccadj}
%%%%%%%%%%%%%%%%%%%%%%%%%%%%%%%%%%%%%%%%%%%%%%%%%%%%%%%%%%%%%%%%%%%%%%%%%%%
A wooded adjunction $R \colon \E \to \C$ is \emph{finitely accessible} when
\begin{enumerate}[(i)]
\item
the categories $\E$ and $\C$ are lfp, and
\item
$R \colon \E \to \C$ is a morphism of lfp categories.
\end{enumerate}
%%%%%%%%%%%%%%%%%%%%%%%%%%%%%%%%%%%%%%%%%%%%%%%%%%%%%%%%%%%%%%%%%%%%%%%%%%%
\end{definition}
%%%%%%%%%%%%%%%%%%%%%%%%%%%%%%%%%%%%%%%%%%%%%%%%%%%%%%%%%%%%%%%%%%%%%%%%%%%

\noindent
Recall from Definition~\ref{def:lfp} that 
morphisms of lfp categories are limit-preserving finitary functors.
In fact, a functor between lfp categories is a morphism of lfp
categories exactly when it is a finitary right adjoint
(see Remark~\ref{rem:finitely-accessible-adj} below).

We already noted in Example~\ref{ex:lfp:struct}
that categories of the form $\Struct(\Sig)$ (and, in particular, of the form $\Struct(\sig)$)
are always lfp.
Hence the ``extensional'' category $\E$ is lfp
in all the examples of~\S\ref{sec:path:adjunction}.
To establish that the wooded categories of interest are also lfp, we shall rely on the following criterion, which is a direct consequence of Lemma~\ref{l:lfp-criterion}.

%%%%%%%%%%%%%%%%%%%%%%%%%%%%%%%%%%%%%%%%%%%%%%%%%%%%%%%%%%%%%%%%%%%%%%%%%%%
\begin{lemma}
\label{lem:path:finaccadj:lfp}
%%%%%%%%%%%%%%%%%%%%%%%%%%%%%%%%%%%%%%%%%%%%%%%%%%%%%%%%%%%%%%%%%%%%%%%%%%%
A wooded category $\C$ is lfp provided the following conditions are satisfied:
\begin{enumerate}[(i)]
\item
\label{item:path:finaccadj:lfp:cocompl}
$\C$ is cocomplete,

\item
\label{item:path:finaccadj:lfp:pathlp}
every path in $\C$ is finitely presentable, and

\item
\label{item:path:finaccadj:lfp:pathdense}
$\C$ admits a small dense full subcategory consisting of paths. 
\end{enumerate}
\end{lemma}

%%%%%%%%%%%%%%%%%%%%%%%%%%%%%%%%%%%%%%%%%%%%%%%%%%%%%%%%%%%%%%%%%%%%%%%%%%%
\begin{example}
\label{ex:path:EFk-fin-acc-wooded-adj}
%%%%%%%%%%%%%%%%%%%%%%%%%%%%%%%%%%%%%%%%%%%%%%%%%%%%%%%%%%%%%%%%%%%%%%%%%%%
Let $\sig$ be a finite (mono-sorted, relational) signature.
For all $0 < k \leq \omega$, the (comonadic) wooded adjunction
$\Ladj_k \colon {\cat R^E_k(\sig)} \inadj {\Struct(\sig)} \cocolon R_k$
in Example~\ref{ex:path:bisim-FOk-equivalence}
is finitely accessible.
To see this, note that the comonad $\Ladj_k R_k$ is finitary
(this can be deduced e.g.\ using the criterion in \cite[Theorem~3.4]{amsw2019}).
It follows that $R_k$ is finitary and $\Ladj_k$
preserves and reflects finitely presentable objects,
cf.\ e.g.\ \cite[p.~45]{diers1986} or \cite[Lemma~B.2]{reggioAiM2022}.
Since the underlying $\sig$-structure of any path in $\cat R^E_k(\sig)$ is finite,
hence finitely presentable in $\Struct(\sig)$,
we see that paths in $\cat R^E_k(\sig)$ are finitely presentable.
By Lemma~\ref{lem:path:finaccadj:lfp}, the latter category is lfp;
just observe that $\cat R^E_k(\sig)$ is cocomplete because
$\Struct(\sig)$ has this property and comonadic functors create colimits.

The adjunction in eq.~\eqref{eq:EF-wooded-adjunction} obtained by composing
with the adjunction $H\adj J$ between $\sig$-structures and
$\sig^I$-structures is also a finitely accessible wooded adjunction,
since $J$ (which interprets $I$ as the identity relation) is easily seen to be finitary.
\end{example}

%%%%%%%%%%%%%%%%%%%%%%%%%%%%%%%%%%%%%%%%%%%%%%%%%%%%%%%%%%%%%%%%%%%%%%%%%%%
\begin{remark}
\label{rem:path:EFk-fin-acc-wooded-adj}
%%%%%%%%%%%%%%%%%%%%%%%%%%%%%%%%%%%%%%%%%%%%%%%%%%%%%%%%%%%%%%%%%%%%%%%%%%%
In the previous example, the assumption that~$\sig$ is finite is necessary
for Lemma~\ref{lem:path:finaccadj:lfp} to apply; e.g., if $\sig$ is the infinite signature in Remark~\ref{rem:path:infinite-sig-representability},
there are paths $P$ that are not finitely presentable in $\cat R^E_k(\sig)$
since $\Ladj_k P$ is not finitely presentable in $\Struct(\sig)$.

On the other hand, if the signature $\sig$ is infinite,
we can still prove that the category $\cat R^E_k(\sig)$
is lfp by invoking a result of Diers which provides sufficient conditions
for the category of coalgebras for a comonad to be lfp
(see \cite[Proposition~1.12.1]{diers1986} and also \cite[Theorem~B.5]{reggioAiM2022}).
%%%%%%%%%%%%%%%%%%%%%%%%%%%%%%%%%%%%%%%%%%%%%%%%%%%%%%%%%%%%%%%%%%%%%%%%%%%
\end{remark}
%%%%%%%%%%%%%%%%%%%%%%%%%%%%%%%%%%%%%%%%%%%%%%%%%%%%%%%%%%%%%%%%%%%%%%%%%%%

By a reasoning completely analogous to that in
Example~\ref{ex:path:EFk-fin-acc-wooded-adj},
one can see that the right adjoints
in Example~\ref{ex:path:misc:games} (for finite $\sig$)
are morphisms of lfp categories
and yield finitely accessible wooded adjunctions.
%

%%%%%%%%%%%%%%%%%%%%%%%%%%%%%%%%%%%%%%%%%%%%%%%%%%%%%%%%%%%%%%%%%%%%%%%%%%%
\subsubsection{Detection of path embeddings}
\label{sec:path:detect}
%%%%%%%%%%%%%%%%%%%%%%%%%%%%%%%%%%%%%%%%%%%%%%%%%%%%%%%%%%%%%%%%%%%%%%%%%%%
Consider a finitely accessible wooded adjunction
\[
\begin{tikzcd}
  \C
  \arrow[r, bend left=25, ""{name=U, below}, "\Ladj"{above}]
  \arrow[r, leftarrow, bend right=25, ""{name=D}, "R"{below}]
& \E.
  \arrow[phantom, "\textnormal{\footnotesize{$\bot$}}", from=U, to=D] 
\end{tikzcd}
\]

\noindent
In the examples of~\S\ref{sec:path:finaccadj},
Lemma~\ref{lem:path:finaccadj:lfp} was instrumental
in proving that the wooded category $\C$ is lfp.
In particular, in Example~\ref{ex:path:EFk-fin-acc-wooded-adj},
the paths in the wooded categories $\cat R^E_k(\sig)$ are finitely presentable objects
(this is also the case with Example~\ref{ex:path:misc:games}, for finite~$\sig$).
Moreover, the adjunctions in the aforementioned examples are such 
that a morphism $f$ of the wooded category $\C$ is an embedding
(i.e.\ $f \in \M$) exactly when $\Ladj f$,
the image of $f$
under the left adjoint $\Ladj \colon \C \to \E$,
is an embedding of structures in $\E$.
This motivates the following definition:

%%%%%%%%%%%%%%%%%%%%%%%%%%%%%%%%%%%%%%%%%%%%%%%%%%%%%%%%%%%%%%%%%%%%%%%%%%%
\begin{definition}[Detection of path embeddings]
\label{def:path:detection-path-emb}
%%%%%%%%%%%%%%%%%%%%%%%%%%%%%%%%%%%%%%%%%%%%%%%%%%%%%%%%%%%%%%%%%%%%%%%%%%%
Let $\Sig$ be a signature.
A finitely accessible wooded adjunction
$\Ladj \colon \C \inadj \E \cocolon R$ 
\emph{detects path embeddings in $\Sig$}
when the following conditions hold:
\begin{enumerate}[(i)]

\item
\label{item:path:detection-path-emb:coste}
there is a cartesian theory $\theory$ in $\Sig$ such that $\E = \Mod(\theory)$,

\item
all paths in $\C$ are finitely presentable, and

\item
a morphism $f \colon P \to a$ in $\C$, with $P$ a path,
is an embedding exactly when
$\Ladj f$ is an embedding of $\Sig$-structures in $\E$.
\end{enumerate}
%%%%%%%%%%%%%%%%%%%%%%%%%%%%%%%%%%%%%%%%%%%%%%%%%%%%%%%%%%%%%%%%%%%%%%%%%%%
\end{definition}
%%%%%%%%%%%%%%%%%%%%%%%%%%%%%%%%%%%%%%%%%%%%%%%%%%%%%%%%%%%%%%%%%%%%%%%%%%%

%%%%%%%%%%%%%%%%%%%%%%%%%%%%%%%%%%%%%%%%%%%%%%%%%%%%%%%%%%%%%%%%%%%%%%%%%%%
\begin{remark}
\label{rem:path:structemb}
%%%%%%%%%%%%%%%%%%%%%%%%%%%%%%%%%%%%%%%%%%%%%%%%%%%%%%%%%%%%%%%%%%%%%%%%%%%
Definition~\ref{def:path:detection-path-emb} requires a few comments.
\begin{enumerate}[(1)]
\item
\label{item:path:structemb:rabin}
Note that
condition~\ref{item:path:detection-path-emb:coste}
of Definition~\ref{def:path:detection-path-emb}
ensures that $\E$ is lfp, 
and also a reflective subcategory of $\Struct(\Sig)$
closed under filtered colimits
(cf.\ Remark~\ref{rem:lfp:struct}).

Conversely, a full subcategory $\E$ of $\Struct(\Sig)$
that is reflective and closed under filtered colimits is
lfp \cite[Corollary 1.47]{ar94book},
and so there is a cartesian theory~$\theory$ such that $\E$ is equivalent
to $\Mod(\theory)$ (cf.\ again Remark~\ref{rem:lfp:struct}).
However, we cannot assume that $\theory$ is a theory in the signature $\Sig$; 
this follows from a slight adaptation of an example given in~\cite{rabin62ascfm}
and used in~\cite{volger79mz}.%
\footnote{See also \cite{har01ctgdc},
which makes explicit the failure of~\cite[Theorem 1.39]{ar94book}
in the case of locally \emph{finitely} presentable categories.
\opt{fullproof}{Appendix~\ref{s:lfp-cats-of-substructures} provides some details.}}
This is the reason why in Definition~\ref{def:path:detection-path-emb} we had
to explicitly require the existence of a cartesian theory $\theory$ in $\Sig$ 
such that $\E = \Mod(\theory)$.

\item
Given an lfp category $\E$,
there are several signatures $\Sig$
such that $\E$ is equivalent to 
$\Mod(\theory)$ for a cartesian theory $\theory$ in $\Sig$;
see \cite[Lemma D1.4.9]{johnstone02book}.

For a given $\Sig$,
it follows from 
completeness (in the sense of Remark~\ref{rem:prelim:coste:compl})
that the choice of a cartesian theory $\theory$
such that $\E \cong \Mod(\theory)$ is irrelevant.

On the other hand, the signature $\Sig$ is an important parameter
in Definition~\ref{def:path:detection-path-emb}.
For instance, there is always a cartesian theory $\theory$
in a signature $\Sig$ with no relation symbols and
such that $\E \cong \Mod(\theory)$
(see loc.\ cit.).
It then follows from Example~\ref{ex:prelim:fact:struct}
that the embeddings of structures are exactly the monomorphisms
in $\Struct(\Sig)$, and thus also in $\E$,
since the reflection $\E \into \Struct(\Sig)$ preserves and reflects monos.

\item
\label{item:path:structemb:strong}
As we will see in Remark~\ref{rem:emb:strong}, 
given a category $\E \into \Struct(\Sig)$
as in Definition~\ref{def:path:detection-path-emb},
every strong monomorphism in $\E$ is an embedding of structures.

We have seen in Example~\ref{ex:prelim:fact:struct}
that the converse holds in the case of $\E = \Struct(\Sig)$,
i.e.\ that the strong monos of $\Struct(\Sig)$
are exactly the embeddings of structures.
This is also true in the case of
$\Pos \into \Struct(\sig)$
(with $\sig = \{\Leq\}$), see Example~\ref{ex:path:wooded:pos}.

However, there are cartesian theories $\theory$ in a signature $\Sig$
such that embeddings of $\Sig$-structures need not be strong monos in $\Mod(\theory)$,
even when their domains are finitely presentable.
E.g., in the signature of monoids, the monoid inclusion $\iota \colon (\NN,+,0) \emb (\ZZ,+,0)$
is an embedding of structures whose domain is finitely presentable,
but $\iota$ is not a strong monomorphism in the category of monoids.

%%%%%%%%%%%%%%%%%%%%%%%%%%%%%%%%%%%%%%%%%%%%%%%%%%%%%%%%%%%%%%%%%%%%%%%%%%%
\begin{fullproof}
%%%%%%%%%%%%%%%%%%%%%%%%%%%%%%%%%%%%%%%%%%%%%%%%%%%%%%%%%%%%%%%%%%%%%%%%%%%
We note the following.
\begin{enumerate}[(a)]
\item
The monoid $(\NN,+,0)$ is finitely presentable in $\Mod(\th M)$,
since it is the free monoid on one generator.

\item
$\iota$ is an embedding of structures in $\Mod(\th M)$,
since $\th M$ is an algebraic theory
while $\iota$ reflects equality.

\item
$\iota$ is epimorphic.
Indeed,
consider monoid morphisms
$h,k \colon (\ZZ,+,0) \to (M,\cdot,1)$
such that $h(n) = k(n)$ for all $n \in \NN$.
Then given $n \in \NN$ we have
$k(n) \cdot h(-n) = h(n) \cdot h(-n) = h(0) = 1$,
so that $k(-n) = h(-n)$ since in a monoid
inverses are unique whenever they exist.

\item
$\iota$ is not strong
since there is no diagonal filler $d$ as in
\[
\begin{tikzcd}
  \NN
  \arrow{r}{\id_\NN}
  \arrow{d}[left]{\iota}
& \NN
  \arrow{d}{\iota}
\\
  \ZZ
  \arrow{r}[below]{\id_\ZZ}
  \arrow{ur}[description]{d}
& \ZZ
\end{tikzcd}
\]
\qedhere
\end{enumerate}
%%%%%%%%%%%%%%%%%%%%%%%%%%%%%%%%%%%%%%%%%%%%%%%%%%%%%%%%%%%%%%%%%%%%%%%%%%%
\end{fullproof}
%%%%%%%%%%%%%%%%%%%%%%%%%%%%%%%%%%%%%%%%%%%%%%%%%%%%%%%%%%%%%%%%%%%%%%%%%%%

\item
\label{item:path:structemb:factorisation}
Let $\C$ be a wooded category with a stable factorisation system $(\Q,\M)$.
Suppose there exists a comonadic adjunction
$\Ladj \colon \C \inadj \E \cocolon R$,
with $\E = \Struct(\Sig)$ for some~$\Sig$,
such that for \emph{any} morphism $f$ in $\C$,
we have $f \in \M$
if, and only if, $\Ladj f$ is an embedding of structures in $\E$.
Then $\Ladj \adj R$ of course detects path embeddings in~$\Sig$.
Moreover, it is not difficult to see that $\Q=\{\text{epis}\}$ and
$\M=\{\text{strong monos}\}$.

%%%%%%%%%%%%%%%%%%%%%%%%%%%%%%%%%%%%%%%%%%%%%%%%%%%%%%%%%%%%%%%%%%%%%%%%%%%
\begin{fullproof}
%%%%%%%%%%%%%%%%%%%%%%%%%%%%%%%%%%%%%%%%%%%%%%%%%%%%%%%%%%%%%%%%%%%%%%%%%%%
By (the dual of) \cite[Proposition 20.12]{ahs06book},
if the adjunction $\Ladj \adj R$ is comonadic,
then the functor $\Ladj$ is conservative
(i.e.\ reflects isomorphisms).
Also, note that $\Ladj$ preserves epimorphisms since it is a left adjoint.

Consider an epimorphism $e \colon a \to b$ in $\C$,
and take a $(\Q,\M)$ factorisation of $e$, say $e = m \comp q$.
Since $\Ladj m$ is a strong mono
and $\Ladj e$ is epi in $\E$,
we have diagonal filler $d$, as in
\[
\begin{tikzcd}
  \Ladj a
  \arrow{d}[left]{\Ladj e}
  \arrow{r}{\Ladj q}
& \unit
  \arrow[rightarrowtail]{d}{\Ladj m}
\\
  \Ladj b
  \arrow{r}[below]{\id}
  \arrow{ur}[description]{d}
& \Ladj b
\end{tikzcd}
\]

\noindent
It follows that the mono $\Ladj m$ is also a split epi, and thus an iso.
But then $m$ is an iso since $\Ladj$ is conservative,
so that $e \in \Q$ by Lemma~\ref{lem:prelim:fact:base}.
In particular, $e \in \Q$ iff $e$ is an epimorphism of $\C$.

Then that $\M$ consists exactly of the strong monomorphisms
follows from the fact that 
$\M=\{m\mid \forall e\in \Q, e\pitchfork m\}$
(Definition~\ref{def:weak-f-s}).
%%%%%%%%%%%%%%%%%%%%%%%%%%%%%%%%%%%%%%%%%%%%%%%%%%%%%%%%%%%%%%%%%%%%%%%%%%%
\end{fullproof}
%%%%%%%%%%%%%%%%%%%%%%%%%%%%%%%%%%%%%%%%%%%%%%%%%%%%%%%%%%%%%%%%%%%%%%%%%%%

\end{enumerate}
%%%%%%%%%%%%%%%%%%%%%%%%%%%%%%%%%%%%%%%%%%%%%%%%%%%%%%%%%%%%%%%%%%%%%%%%%%%
\end{remark}
%%%%%%%%%%%%%%%%%%%%%%%%%%%%%%%%%%%%%%%%%%%%%%%%%%%%%%%%%%%%%%%%%%%%%%%%%%%

In all the examples of wooded adjunctions we have given so far,
the ``extensional'' category~$\E$ is merely a category of structures,
while Definition~\ref{def:path:detection-path-emb} only requires that~$\E$
be a category of models of a cartesian theory.
Example~\ref{ex:path:pos} below is a (trivial) instance of
Definition~\ref{def:path:detection-path-emb}
in which the extensional category is the category of models of a non-empty Horn theory.

%%%%%%%%%%%%%%%%%%%%%%%%%%%%%%%%%%%%%%%%%%%%%%%%%%%%%%%%%%%%%%%%%%%%%%%%%%%
\begin{example}
\label{ex:path:pos}
%%%%%%%%%%%%%%%%%%%%%%%%%%%%%%%%%%%%%%%%%%%%%%%%%%%%%%%%%%%%%%%%%%%%%%%%%%%
Consider the self-adjunction
\[
\begin{tikzcd}
  \Pos
  \arrow[r, bend left=25, ""{name=U, below}, "\Id"{above}]
  \arrow[r, leftarrow, bend right=25, ""{name=D}, "\Id"{below}]
& \Pos.
  \arrow[phantom, "\textnormal{\footnotesize{$\bot$}}", from=U, to=D] 
\end{tikzcd}
\]

\noindent
We have seen in Example~\ref{ex:path:wooded:pos} that the lfp category $\Pos$
is wooded when equipped with the factorisation system $(\Q,\M)$,
where $\Q$ is the class of surjective morphisms
and $\M$ the class of order embeddings.
Hence, $\Id \colon \Pos \inadj \Pos \cocolon \Id$
is a finitely accessible wooded adjunction
which detects path embeddings. 
%%%%%%%%%%%%%%%%%%%%%%%%%%%%%%%%%%%%%%%%%%%%%%%%%%%%%%%%%%%%%%%%%%%%%%%%%%%
\end{example}
%%%%%%%%%%%%%%%%%%%%%%%%%%%%%%%%%%%%%%%%%%%%%%%%%%%%%%%%%%%%%%%%%%%%%%%%%%%

We return to Example~\ref{ex:path:pos} in~\S\ref{sec:fact} below,
where we discuss whether a (finitely accessible)
wooded adjunction of the form $\C \inadj \Mod(\theory)$
factors through the reflection $\Mod(\theory) \into \Struct(\Sig)$,
where $\Sig$ is the signature of $\theory$
(see Remark~\ref{rem:lfp:struct}).

%%%%%%%%%%%%%%%%%%%%%%%%%%%%%%%%%%%%%%%%%%%%%%%%%%%%%%%%%%%%%%%%%%%%%%%%%%%
\subsection{Main result}
\label{sec:path:results}
%%%%%%%%%%%%%%%%%%%%%%%%%%%%%%%%%%%%%%%%%%%%%%%%%%%%%%%%%%%%%%%%%%%%%%%%%%%
With Definition~\ref{def:path:detection-path-emb} at hand,
we can state our main result:

%%%%%%%%%%%%%%%%%%%%%%%%%%%%%%%%%%%%%%%%%%%%%%%%%%%%%%%%%%%%%%%%%%%%%%%%%%%
\begin{restatable}{theorem}{maintheorem}
\label{thm:path:main}
%%%%%%%%%%%%%%%%%%%%%%%%%%%%%%%%%%%%%%%%%%%%%%%%%%%%%%%%%%%%%%%%%%%%%%%%%%%
Let $\Sig$ be a signature and assume that
$\Ladj \colon \C \inadj \E \cocolon R$ 
is a finitely accessible wooded adjunction
that detects path embeddings in $\Sig$.
For all $M,N \in \E$,
\[
\begin{array}{l l l}
  \text{$M,N$ equivalent in $\Lang_{\infty}(\Sig)$}
& \longimp
& M \bisim_R N.
\end{array}
\]
%%%%%%%%%%%%%%%%%%%%%%%%%%%%%%%%%%%%%%%%%%%%%%%%%%%%%%%%%%%%%%%%%%%%%%%%%%%
\end{restatable}
%%%%%%%%%%%%%%%%%%%%%%%%%%%%%%%%%%%%%%%%%%%%%%%%%%%%%%%%%%%%%%%%%%%%%%%%%%%

We now comment on the previous result and some of its applications.
First of all, let us indicate how Ehrenfeucht-Fraïssé games fit into
the picture of Theorem~\ref{thm:path:main}.

%%%%%%%%%%%%%%%%%%%%%%%%%%%%%%%%%%%%%%%%%%%%%%%%%%%%%%%%%%%%%%%%%%%%%%%%%%%
\begin{example}
\label{ex:path:main:fo}
%%%%%%%%%%%%%%%%%%%%%%%%%%%%%%%%%%%%%%%%%%%%%%%%%%%%%%%%%%%%%%%%%%%%%%%%%%%
Let $k \leq \omega$ and
consider the adjunctions of Example~\ref{ex:path:bisim-FOk-equivalence}:
\[
\begin{tikzcd}
{\cat R^E_k(\sig^I)} \arrow[r, bend left=25, ""{name=U, below}, "\Ladj_k"{above}]
\arrow[r, leftarrow, bend right=25, ""{name=D}, "R_k"{below}]
& {\Struct(\sig^I)} \arrow[r, bend left=25, ""{name=U', below}, "H"{above}]
\arrow[r, leftarrow, bend right=25, ""{name=D'}, "J"{below}] & {\Struct(\sig)}
\arrow[phantom, "\textnormal{\footnotesize{$\bot$}}", from=U, to=D] 
\arrow[phantom, "\textnormal{\footnotesize{$\bot$}}", from=U', to=D'] 
\end{tikzcd}
\]

\noindent
The composite adjunction does \emph{not} detect path embeddings in $\sig$,
but the comonadic adjunction $L_k \dashv R_k$ detects path embeddings in $\sig^I$.

So, given $\sig$-structures $M$ and $N$, Theorem~\ref{thm:path:main}
implies that $J M \bisim_{R_k} J N$ whenever $J M$ and $J N$
are $\Lang_{\infty}(\sig^I)$-equivalent.
Since $J \colon \Struct(\sig) \to \Struct(\sig^I)$
interprets the relation~$I$ as the identity,
it follows that $M$ and $N$ are $R^I_k$-equivalent whenever they are
$\Lang_\infty(\sig)$-equivalent. (In fact, as we saw in Example~\ref{ex:path:bisim-FOk-equivalence}, $R^I_k$-equivalence
coincides with equivalence 
in the fragment of $\Lang_{\omega}(\sig)$ with quantifier rank at most $k$.)
%%%%%%%%%%%%%%%%%%%%%%%%%%%%%%%%%%%%%%%%%%%%%%%%%%%%%%%%%%%%%%%%%%%%%%%%%%%
\end{example}
%%%%%%%%%%%%%%%%%%%%%%%%%%%%%%%%%%%%%%%%%%%%%%%%%%%%%%%%%%%%%%%%%%%%%%%%%%%

Under the assumptions of Theorem~\ref{thm:path:main},
the relation of $R$-equivalence cannot distinguish between objects of $\E$
that satisfy the same sentences in $\Lang_{\infty}$.
In general, this prevents the relation $\bisim_{R}$ from collapsing to the isomorphism relation.

%%%%%%%%%%%%%%%%%%%%%%%%%%%%%%%%%%%%%%%%%%%%%%%%%%%%%%%%%%%%%%%%%%%%%%%%%%%
\begin{example}
\label{ex:path:main:karp}
%%%%%%%%%%%%%%%%%%%%%%%%%%%%%%%%%%%%%%%%%%%%%%%%%%%%%%%%%%%%%%%%%%%%%%%%%%%
Recall from Example~\ref{ex:path:karp}
that the linear orders $(\QQ,<)$ and $(\RR,<)$ are $\Lang_\infty$-equivalent
in the signature $\sig = \{\Lt\}$.
It follows from Theorem~\ref{thm:path:main}
that $(\QQ,<) \bisim_R (\RR,<)$
for any
$\Ladj \colon \C \inadj \Struct(\sig) \cocolon R$
that detects path embeddings in $\sig = \{\Lt\}$.
%%%%%%%%%%%%%%%%%%%%%%%%%%%%%%%%%%%%%%%%%%%%%%%%%%%%%%%%%%%%%%%%%%%%%%%%%%%
\end{example}
%%%%%%%%%%%%%%%%%%%%%%%%%%%%%%%%%%%%%%%%%%%%%%%%%%%%%%%%%%%%%%%%%%%%%%%%%%%

Note that while the isomorphism type of a finite structure is 
definable in finitary first-order logic $\Lang_\omega$,
this is no longer the case for the finite-variable fragments
(see e.g.~\cite[3.3.6(a)]{ef99book}).
Besides, Corollary 8 in~\cite{adw17lics} (on pebble games for finite-variable logics)
forces $\E$ to contain infinite objects in general,
on which $\Lang_{\infty}$ is strictly more expressive than $\Lang_{\omega}$.%
\footnote{Also, in contrast with $\Lang_\infty$,
not every isomorphism-class of finite structures is definable in 
finite-variable logic with infinitary propositional connectives,
%(for which infinitary logic is non-trivial, see e.g.~\cite{ef99book,libkin04book}).
see e.g~\cite[Example 3.3.13]{ef99book}.}

Example~\ref{ex:path:main:karp} has the following consequence concerning game comonads for
Monadic Second-Order Logic ($\MSO$).

%%%%%%%%%%%%%%%%%%%%%%%%%%%%%%%%%%%%%%%%%%%%%%%%%%%%%%%%%%%%%%%%%%%%%%%%%%%
\begin{example}
\label{ex:mso}
%%%%%%%%%%%%%%%%%%%%%%%%%%%%%%%%%%%%%%%%%%%%%%%%%%%%%%%%%%%%%%%%%%%%%%%%%%%
Recall that $\MSO$ in a given signature $\sig$ is 
the extension of
first-order logic in $\sig$
with monadic (i.e.\ one-place) set variables,
and quantification over them.
In particular,
every $M \in \Struct(\sig)$ can be extended to a model
of $\MSO(\sig)$ in which set variables are allowed to range
over all subsets of $M$
(such models are called \emph{standard}).

A comonadic approach to Ehrenfeucht-Fraïssé games for $\MSO$
was proposed in~\cite{jms22mso}, by viewing $\MSO$ in a signature~$\sig$
as a two-sorted first-order logic in an extension~$\Sig$ of~$\sig$
with one sort for (monadic) set variables,
and a membership relation between the two sorts.
In particular, the $\MSO(\sig)$-theory of a standard model
is contained in its $\Lang_\infty(\Sig)$-theory,
and so $\Lang_\infty(\Sig)$-equivalent standard models are also $\MSO(\sig)$-equivalent.

On the other hand, Theorem~\ref{thm:path:main} captures known differences
between equivalence in $\Lang_\infty(\sig)$ and in $\MSO(\sig)$:
it implies that for $\sig = \{\Lt\}$, there is no adjunction
\[
\begin{tikzcd}[column sep=large]
  \C
  \arrow[bend left,
    start anchor={[xshift=0ex, yshift=0ex]north east}]{r}[above]{\Ladj}
  \arrow[bend right, leftarrow,
    start anchor={[xshift=0ex, yshift=0ex]south east}]{r}[below]{R}
  \arrow[phantom]{r}[description, xshift=5pt]{\textnormal{\footnotesize{$\bot$}}}
& \Struct(\sig)
\end{tikzcd}
\]

\noindent
that detects path embeddings in $\sig$
and such that given $M,N \in \Struct(\sig)$, 
the standard models over $M$ and $N$ are $\MSO(\sig)$-equivalent 
exactly when $M$ and $N$ are $R$-equivalent.

Indeed, consider the linear orders $(\QQ,<)$ and $(\RR,<)$.
We know from Example~\ref{ex:path:main:karp} that $(\QQ,<) \bisim_R (\RR,<)$.
On the other hand, it is well known
(and easy to see from~\cite[\S 2.4]{rosenstein82book})
that $\MSO(\sig)$ can express Dedekind cuts and thus order-completeness.
Hence $(\QQ,<)$ and $(\RR,<)$ are not $\MSO(\sig)$-equivalent.%
\footnote{Moreover,
%it is well-known that
the $\MSO(\sig)$-theory of $(\QQ,<)$
is decidable \cite[proof of Theorem 2.1]{rabin69tams},
while the $\MSO(\sig)$-theory of $(\RR,<)$
is undecidable \cite[Theorem 7]{shelah75am}.}
%%%%%%%%%%%%%%%%%%%%%%%%%%%%%%%%%%%%%%%%%%%%%%%%%%%%%%%%%%%%%%%%%%%%%%%%%%%
\end{example}
%%%%%%%%%%%%%%%%%%%%%%%%%%%%%%%%%%%%%%%%%%%%%%%%%%%%%%%%%%%%%%%%%%%%%%%%%%%

%%%%%%%%%%%%%%%%%%%%%%%%%%%%%%%%%%%%%%%%%%%%%%%%%%%%%%%%%%%%%%%%%%%%%%%%%%%
\subsubsection{Outline of the following sections}
\label{sec:path:results:outline}
%%%%%%%%%%%%%%%%%%%%%%%%%%%%%%%%%%%%%%%%%%%%%%%%%%%%%%%%%%%%%%%%%%%%%%%%%%%
The proof of Theorem~\ref{thm:path:main}
will occupy us for much of the remainder of the paper. 
In the next section we provide some background on arboreal categories,
since many of our examples of wooded categories fall into this class.
In~\S\S\ref{sec:lfp}--\ref{sec:coste} we recall some technical material on lfp categories and functorial semantics.
In~\S\ref{sec:hintikka} we obtain a weaker version of Theorem~\ref{thm:path:main},
where the assumption of detection of path embeddings is replaced with an
assumption of ``definability of path embeddings'' in a wooded lfp category $\C$.
Then in \S\ref{sec:emb} we devise formulae for substructure embeddings
in an lfp category $\E$.
We complete the proof of Theorem~\ref{thm:path:main} in~\S\ref{sec:wc}
by showing that these formulae can be translated back along \emph{left adjoints},
so that path embeddings are definable in $\C$
whenever they are detected by $\Ladj \colon \C \inadj \E \cocolon R$.
This last step relies on Hodges' \emph{word-constructions}~\cite{hodges74,hodges75la}.
%%%%%%%%%%%%%%%%%%%%%%%%%%%%%%%%%%%%%%%%%%%%%%%%%%%%%%%%%%%%%%%%%%%%%%%%%%%
\section{Arboreal categories}
\label{sec:arboreal}
%%%%%%%%%%%%%%%%%%%%%%%%%%%%%%%%%%%%%%%%%%%%%%%%%%%%%%%%%%%%%%%%%%%%%%%%%%%

In this section we recall the notion of arboreal categories, a special class of wooded categories in which the back-and-forth equivalence relation can be described in terms of open map bisimilarity. We then give an example of a class of arboreal categories that does not arise from game comonads, namely the categories of presheaves over forest orders.

%%%%%%%%%%%%%%%%%%%%%%%%%%%%%%%%%%%%%%%%%%%%%%%%%%%%%%%%%%%%%%%%%%%%%%%%%%%
\subsection{The axioms of arboreal categories}
\label{sec:arboreal:def}
%%%%%%%%%%%%%%%%%%%%%%%%%%%%%%%%%%%%%%%%%%%%%%%%%%%%%%%%%%%%%%%%%%%%%%%%%%%
We recall from \cite{ar21icalp,ar21arboreal} the main definitions and
facts pertaining to the theory of arboreal categories.

%%%%%%%%%%%%%%%%%%%%%%%%%%%%%%%%%%%%%%%%%%%%%%%%%%%%%%%%%%%%%%%%%%%%%%%%%%%
\begin{definition}
\label{def:arboreal:material}
%%%%%%%%%%%%%%%%%%%%%%%%%%%%%%%%%%%%%%%%%%%%%%%%%%%%%%%%%%%%%%%%%%%%%%%%%%%
Let $\C$ be a well-powered category
equipped with a proper factorisation system $(\Q,\M)$,
and let $a$ be an object of $\C$.
\begin{enumerate}[(1)]
\item
We say that $a$ is \emph{connected} if,
for all non-empty sets of paths $\{P_i\mid i\in I\}$ admitting a coproduct $b$ in $\C$,
any arrow $a\to b$
factors through some coproduct map $P_i\to b$.%
\footnote{This differs from the usual notion, which asks
$\C\funct{a,-}$ to preserve all existing coproducts.
Note that initial objects are in general not connected in that sense,
while they are always connected in our sense.}

\item
Consider the diagram with vertex $a$ consisting of all path embeddings into $a$.
The morphisms between paths are those making the obvious triangles commute:
\[
\begin{tikzcd}[column sep=small, row sep=normal]
& a
&
\\
  P
  \arrow[bend left=20,rightarrowtail]{ur}
  \arrow[rightarrowtail]{rr}
&
& Q
  \arrow[bend right=20,rightarrowtail]{ul}
\end{tikzcd}
\]

\noindent
Choosing representatives in a suitable way,
this yields a cocone with vertex $a$ over the small diagram $\Path{a}$.
We say that $a$ is \emph{path-generated} provided this is a colimit cocone in $\C$.
\end{enumerate}
\end{definition}

%%%%%%%%%%%%%%%%%%%%%%%%%%%%%%%%%%%%%%%%%%%%%%%%%%%%%%%%%%%%%%%%%%%%%%%%%%%
\begin{definition}
\label{def:arboreal-cat}
%%%%%%%%%%%%%%%%%%%%%%%%%%%%%%%%%%%%%%%%%%%%%%%%%%%%%%%%%%%%%%%%%%%%%%%%%%%
An \emph{arboreal category} is a locally small and well powered category $\A$,
equipped with a stable factorisation system, satisfying the following conditions:
\begin{enumerate}[(i)]
\item
\label{ax:2-out-of-3}
For any paths $P,Q,Q'$ in $\A$,
if a composite $P\to Q \to Q'$ is a quotient then so is $P\to Q$
(``2-out-of-3 condition'').

\item
\label{ax:colimits}
$\A$ has all coproducts of sets of paths.

\item
\label{ax:connected}
Every path in $\A$ is connected.

\item
\label{ax:path-generated}
Every object of $\A$ is path-generated. 
\end{enumerate}

\noindent
The full subcategory of $\A$ defined by the paths is denoted by~$\pth\A$.
%%%%%%%%%%%%%%%%%%%%%%%%%%%%%%%%%%%%%%%%%%%%%%%%%%%%%%%%%%%%%%%%%%%%%%%%%%%
\end{definition}
%%%%%%%%%%%%%%%%%%%%%%%%%%%%%%%%%%%%%%%%%%%%%%%%%%%%%%%%%%%%%%%%%%%%%%%%%%%

Throughout, we assume that for any arboreal category $\A$, its subcategory $\pth\A$ is \emph{essentially small}
(i.e., equivalent to a small category). Note, in particular, that every arboreal category is wooded, since item~\ref{ax:colimits} in Definition~\ref{def:arboreal-cat} implies the existence of an initial object $\zero$.
Similarly to the case of wooded adjunctions, an \emph{arboreal adjunction} is an adjunction
$L \colon \A \inadj \E \cocolon R$
where $\A$ is an arboreal category.
An arboreal adjunction is \emph{finitely accessible} if it is
finitely accessible as a wooded adjunction.

%%%%%%%%%%%%%%%%%%%%%%%%%%%%%%%%%%%%%%%%%%%%%%%%%%%%%%%%%%%%%%%%%%%%%%%%%%%
\begin{example}
\label{ex:forests-trees-presheaves}
%%%%%%%%%%%%%%%%%%%%%%%%%%%%%%%%%%%%%%%%%%%%%%%%%%%%%%%%%%%%%%%%%%%%%%%%%%%
The category $\Tree$ (and thus also $\Forest$) is arboreal when equipped with the
(epimorphisms, monomorphisms) factorisation system.
This was shown in~\cite[Example 5.4]{ar21arboreal}
and is a special case of Theorem~\ref{thm:arboreal:presh} below;
see Remark~\ref{rem:arboreal:presh}.
%%%%%%%%%%%%%%%%%%%%%%%%%%%%%%%%%%%%%%%%%%%%%%%%%%%%%%%%%%%%%%%%%%%%%%%%%%%
\end{example}
%%%%%%%%%%%%%%%%%%%%%%%%%%%%%%%%%%%%%%%%%%%%%%%%%%%%%%%%%%%%%%%%%%%%%%%%%%%

All the examples of wooded categories that arise from game comonads given
in~\S\ref{sec:path}
are arboreal (cf.\ Examples~\ref{ex:path:R}, \ref{ex:path:R^E} and~\ref{ex:path:misc}).
An example of a wooded category that is not arboreal is $\Pos$,
since the latter is not path-generated,
cf.\ Examples~\ref{ex:prelim:coste:pos} and~\ref{ex:path:wooded:pos}.

An important property of arboreal categories is that, for any such category $\A$,
the assignment $a\mapsto\Path{a}$ extends to a functor $\Path \colon \A \to \Tree$
\cite[Theorem~3.11]{ar21arboreal}.
The action of $\Path$ on morphisms is defined as follows:
for all arrows $f\colon a\to b$ in $\A$,
$\Path{f}\colon \Path{a}\to\Path{b}$ is the forest morphism sending
a path embedding $m\colon P\emb a$
to the path embedding $\exists_f P\emb b$ obtained by taking the
(quotient, embedding) factorisation of $f \comp m\colon P \epi \exists_f P \emb b$.
This assignment is well-defined because factorisations are unique up to isomorphism,
and a quotient of a path is again a path
(Lemma~\ref{lem:path:base}).

When $\A$ is one of the categories of forest-ordered $\sig$-structures introduced in~\S\ref{sec:path}, such as $\cat R_k(\sig)$, $\cat R^E_k(\sig)$ or $\cat R^P_k(\sig)$, the functor $\Path$ is naturally isomorphic to the composition of the obvious forgetful functor $\A\to \Forest$ with the equivalence $\Forest\to\Tree$ that adds a root, corresponding to the embedding of the empty substructure. If $\A$ is the category $\Tree$, then $\Path$ is naturally isomorphic to the identity functor $\Tree\to\Tree$.

%%%%%%%%%%%%%%%%%%%%%%%%%%%%%%%%%%%%%%%%%%%%%%%%%%%%%%%%%%%%%%%%%%%%%%%%%%%
\begin{remark} %[Path density]
\label{rem:path-density}
%%%%%%%%%%%%%%%%%%%%%%%%%%%%%%%%%%%%%%%%%%%%%%%%%%%%%%%%%%%%%%%%%%%%%%%%%%%
Assuming conditions~\ref{ax:2-out-of-3}--\ref{ax:connected}
of Definition~\ref{def:arboreal-cat},
condition~\ref{ax:path-generated} is equivalent to the fact that the inclusion functor
$I \colon \pth\A \into \A$
is dense \cite[Lemma~5.1]{ar21arboreal}.
In particular, 
condition~\ref{item:path:finaccadj:lfp:pathdense}
of Lemma~\ref{lem:path:finaccadj:lfp} always holds
for arboreal categories, and so a cocomplete arboreal category is lfp whenever its paths are
finitely presentable.
%%%%%%%%%%%%%%%%%%%%%%%%%%%%%%%%%%%%%%%%%%%%%%%%%%%%%%%%%%%%%%%%%%%%%%%%%%%
\end{remark}
%%%%%%%%%%%%%%%%%%%%%%%%%%%%%%%%%%%%%%%%%%%%%%%%%%%%%%%%%%%%%%%%%%%%%%%%%%%

%%%%%%%%%%%%%%%%%%%%%%%%%%%%%%%%%%%%%%%%%%%%%%%%%%%%%%%%%%%%%%%%%%%%%%%%%%%
\subsection{Bisimilarity and open maps}
\label{sec:bisim}
%%%%%%%%%%%%%%%%%%%%%%%%%%%%%%%%%%%%%%%%%%%%%%%%%%%%%%%%%%%%%%%%%%%%%%%%%%%

A key property of arboreal categories, compared to wooded categories, is that under mild assumptions the relation of
back-and-forth equivalence coincides with 
the notion of \emph{bisimilarity via open maps}
introduced in~\cite{jnw96ic}.

%%%%%%%%%%%%%%%%%%%%%%%%%%%%%%%%%%%%%%%%%%%%%%%%%%%%%%%%%%%%%%%%%%%%%%%%%%%
\begin{definition}
\label{def:arboreal:open}
%%%%%%%%%%%%%%%%%%%%%%%%%%%%%%%%%%%%%%%%%%%%%%%%%%%%%%%%%%%%%%%%%%%%%%%%%%%
An arrow $f \colon a \to b$ in an arboreal category $\A$ is \emph{open}
if it satisfies the following path-lifting property:
given any commutative square
\[
\begin{tikzcd}
  P
  \arrow{r}
  \arrow{d}
& a
  \arrow{d}{f}
\\
  Q
  \arrow{r}
    \arrow[dashed]{ur}
& b
\end{tikzcd}
\]

\noindent
with $P,Q$ paths, there is a diagonal filler $Q\to a$.\footnote{This definition of open morphism differs from, but is equivalent to, the one used in~\cite{ar21arboreal}; see Lemma~4.1 in \emph{op.\ cit.}} 
Two objects $c,d\in\A$ are \emph{arboreal-bisimilar}
if there exists a span of open morphisms
$c \leftarrow \cdot \rightarrow d$ connecting them.
%%%%%%%%%%%%%%%%%%%%%%%%%%%%%%%%%%%%%%%%%%%%%%%%%%%%%%%%%%%%%%%%%%%%%%%%%%%
\end{definition}
%%%%%%%%%%%%%%%%%%%%%%%%%%%%%%%%%%%%%%%%%%%%%%%%%%%%%%%%%%%%%%%%%%%%%%%%%%%

%%%%%%%%%%%%%%%%%%%%%%%%%%%%%%%%%%%%%%%%%%%%%%%%%%%%%%%%%%%%%%%%%%%%%%%%%%%
\begin{example}
\label{ex:arboreal:bisim:tree}
%%%%%%%%%%%%%%%%%%%%%%%%%%%%%%%%%%%%%%%%%%%%%%%%%%%%%%%%%%%%%%%%%%%%%%%%%%%
By~\cite[Proposition~4.2]{ar21arboreal}, open morphisms in $\Tree$ coincide with bounded ones. Moreover, it is well known that, for any two trees $a$ and $b$, there exists a span
of bounded morphisms
$a \leftarrow \cdot \rightarrow b$ 
precisely when $a$ and $b$ are bisimilar in the usual sense (i.e.\ as Kripke frames, where the accessibility relation is given by the covering relation). Thus, arboreal-bisimilarity in $\Tree$ coincides with the usual notion of bisimilarity.
%%%%%%%%%%%%%%%%%%%%%%%%%%%%%%%%%%%%%%%%%%%%%%%%%%%%%%%%%%%%%%%%%%%%%%%%%%%
\end{example}
%%%%%%%%%%%%%%%%%%%%%%%%%%%%%%%%%%%%%%%%%%%%%%%%%%%%%%%%%%%%%%%%%%%%%%%%%%%

%%%%%%%%%%%%%%%%%%%%%%%%%%%%%%%%%%%%%%%%%%%%%%%%%%%%%%%%%%%%%%%%%%%%%%%%%%%
\begin{theorem}[{\cite[Theorem~6.12]{ar21arboreal}}]
\label{th:games-vs-bisimilarity}
%%%%%%%%%%%%%%%%%%%%%%%%%%%%%%%%%%%%%%%%%%%%%%%%%%%%%%%%%%%%%%%%%%%%%%%%%%%
Let $a,b$ be objects of an arboreal category admitting a product.
Then $a$ and $b$ are back-and-forth equivalent (in the sense of Definition~\ref{def:path:bfe})
if, and only if, 
they are arboreal-bisimilar.
%%%%%%%%%%%%%%%%%%%%%%%%%%%%%%%%%%%%%%%%%%%%%%%%%%%%%%%%%%%%%%%%%%%%%%%%%%%
\end{theorem}
%%%%%%%%%%%%%%%%%%%%%%%%%%%%%%%%%%%%%%%%%%%%%%%%%%%%%%%%%%%%%%%%%%%%%%%%%%%

Recall from Example~\ref{ex:forests-trees-presheaves} that the category $\Tree$ is arboreal, and it is complete because it is equivalent to the presheaf category $\presh{\NN}$. It follows from Example~\ref{ex:arboreal:bisim:tree} and Theorem~\ref{th:games-vs-bisimilarity} that two trees are back-and-forth equivalent if, and only if, they are bisimilar in the usual sense.
In general, since the functor $\Path\colon \A\to\Tree$ sends open morphisms to bounded ones by~\cite[Proposition~4.2]{ar21arboreal},
if $a,b\in \A $ are back-and-forth equivalent objects of an arboreal category with binary products, then the trees
$\Path a$ and $\Path b$ are bisimilar---but the converse need not hold.

%%%%%%%%%%%%%%%%%%%%%%%%%%%%%%%%%%%%%%%%%%%%%%%%%%%%%%%%%%%%%%%%%%%%%%%%%%%
\begin{remark}
\label{rem:arboral:bisim}
%%%%%%%%%%%%%%%%%%%%%%%%%%%%%%%%%%%%%%%%%%%%%%%%%%%%%%%%%%%%%%%%%%%%%%%%%%%
The notion of arboreal bisimilarity in Definition~\ref{def:arboreal:open} makes sense, more generally, in wooded categories. In fact, a straightforward adaptation of \cite[Proposition~6.6]{ar21arboreal} shows that any two arboreal-bisimilar objects of a wooded category $\cat C$ are
back-and-forth equivalent. However, the reverse implication need not hold unless $\cat C$ is arboreal.
%
%%%%%%%%%%%%%%%%%%%%%%%%%%%%%%%%%%%%%%%%%%%%%%%%%%%%%%%%%%%%%%%%%%%%%%%%%%%
\end{remark}
%%%%%%%%%%%%%%%%%%%%%%%%%%%%%%%%%%%%%%%%%%%%%%%%%%%%%%%%%%%%%%%%%%%%%%%%%%%

%%%%%%%%%%%%%%%%%%%%%%%%%%%%%%%%%%%%%%%%%%%%%%%%%%%%%%%%%%%%%%%%%%%%%%%%%%%
\subsection{Presheaves over a forest}
\label{sec:arboreal:presh}
%%%%%%%%%%%%%%%%%%%%%%%%%%%%%%%%%%%%%%%%%%%%%%%%%%%%%%%%%%%%%%%%%%%%%%%%%%%
Recall that a presheaf category admits a unique proper factorisation system, namely the (epimorphisms, monomorphisms) factorisation system, and the latter is stable (see Example~\ref{ex:presheaf-cats-stable-fact-sys} and Footnote~\ref{footnote:unique-fs}). This section is devoted to the following result:

%%%%%%%%%%%%%%%%%%%%%%%%%%%%%%%%%%%%%%%%%%%%%%%%%%%%%%%%%%%%%%%%%%%%%%%%%%%
\begin{theorem}
\label{thm:arboreal:presh}
%%%%%%%%%%%%%%%%%%%%%%%%%%%%%%%%%%%%%%%%%%%%%%%%%%%%%%%%%%%%%%%%%%%%%%%%%%%
Let $\forest$ be a forest, seen as a (posetal) category.
The presheaf category~$\presh{\forest}$ is arboreal,
 and the paths in $\presh{\forest}$ are precisely the representable functors
along with the initial object.
\end{theorem}

%%%%%%%%%%%%%%%%%%%%%%%%%%%%%%%%%%%%%%%%%%%%%%%%%%%%%%%%%%%%%%%%%%%%%%%%%%%
\begin{remark}
\label{rem:arboreal:presh}
%%%%%%%%%%%%%%%%%%%%%%%%%%%%%%%%%%%%%%%%%%%%%%%%%%%%%%%%%%%%%%%%%%%%%%%%%%%
Letting $\forest$ in Theorem~\ref{thm:arboreal:presh} be the countable chain $\NN$,
we see that the categories $\Tree$ and $\Forest$ are arboreal.
\end{remark}

Theorem~\ref{thm:arboreal:presh} provides us with examples of
(finitely accessible) arboreal adjunctions of a different nature
than those discussed in~\S\ref{sec:path}.
For instance, the following example 
is inspired by the \emph{Diaconescu cover}
of a topos, as described in~\cite[\S IX.9]{mm92book}.

%%%%%%%%%%%%%%%%%%%%%%%%%%%%%%%%%%%%%%%%%%%%%%%%%%%%%%%%%%%%%%%%%%%%%%%%%%%
\begin{example}
\label{ex:Diaconescu-cover}
%%%%%%%%%%%%%%%%%%%%%%%%%%%%%%%%%%%%%%%%%%%%%%%%%%%%%%%%%%%%%%%%%%%%%%%%%%%
Given a small category $\cat D$,
let $\forest(\cat D)$ be the forest whose nodes
are finite sequences $s$ of composable arrows
\[
\begin{tikzcd}
  a_0
  \arrow{r}{k_1}
& a_1
  \arrow[dashed]{r}
& a_{n-1}
  \arrow{r}{k_{n}}
& a_n
\end{tikzcd}
\]

\noindent
in $\cat D$, and where $s \leq t$
exactly when $s$ is a prefix of $t$, say
\[
\begin{tikzcd}
  a_0
  \arrow[dashed]{r}{s}
& a_n
  \arrow{r}{k_{n+1}}
& a_{n+1}
  \arrow[dashed]{r}
& a_{n+m-1}
  \arrow{r}{k_{n+m}}
& a_{n+m}.
\end{tikzcd}
\]

\noindent
Let $\pi$ be the functor $\forest(\cat D) \to \cat D$
such that, for $s \leq t$ as above,
$\pi(s \leq t) \colon a_n \to a_{n+m}$ is the composition
$k_{n+m} \comp \dots \comp k_{n+1}$.
The functor $\ladj\pi \colon \presh{\cat D} \to \presh{\forest(\cat D)}$, 
defined by precomposing with~$\pi$,
has adjoints on both sides
(cf. e.g.\ \cite[Example A4.1.4]{johnstone02book}).
In view of Example~\ref{ex:presheaf-cats-lfp} and
Theorem~\ref{thm:arboreal:presh},
this yields a finitely accessible arboreal adjunction between
$\presh{\cat D}$ and the arboreal category $\presh{\forest(\cat D)}$. 
%%%%%%%%%%%%%%%%%%%%%%%%%%%%%%%%%%%%%%%%%%%%%%%%%%%%%%%%%%%%%%%%%%%%%%%%%%%%
\end{example}
%%%%%%%%%%%%%%%%%%%%%%%%%%%%%%%%%%%%%%%%%%%%%%%%%%%%%%%%%%%%%%%%%%%%%%%%%%%%

We elaborate on Example~\ref{ex:Diaconescu-cover} in
Example~\ref{ex:hintikka:mono:forest} (\S\ref{sec:hintikka:finaccadj})
and Example~\ref{ex:fact:Diaconescu-cover} (\S\ref{sec:fact}) below.

In the remainder of this section, we offer a proof of Theorem~\ref{thm:arboreal:presh}.
To this end, fix an arbitrary forest $\forest$
(regarded as a posetal category) and consider the presheaf category~$\presh{\forest}$.
Recall that an object of $\presh{\forest}$ is \emph{representable} if it is
naturally isomorphic to one of the form $\yo a$ for some $a\in \forest$, where 
\[
\begin{array}{*{5}{l}}
  \yo
& \colon
& \forest
& \longto
& \presh{\forest}
\end{array}
\]

\noindent
is the Yoneda embedding.

We start by characterising the paths in $\presh{\forest}$:

%%%%%%%%%%%%%%%%%%%%%%%%%%%%%%%%%%%%%%%%%%%%%%%%%%%%%%%%%%%%%%%%%%%%%%%%%%%
\begin{lemma}
\label{l:paths-in-presh-forest}
%%%%%%%%%%%%%%%%%%%%%%%%%%%%%%%%%%%%%%%%%%%%%%%%%%%%%%%%%%%%%%%%%%%%%%%%%%%
An object of $\presh{\forest}$ is a path if, and only if, it is either initial or representable.
%%%%%%%%%%%%%%%%%%%%%%%%%%%%%%%%%%%%%%%%%%%%%%%%%%%%%%%%%%%%%%%%%%%%%%%%%%%
\end{lemma}
%%%%%%%%%%%%%%%%%%%%%%%%%%%%%%%%%%%%%%%%%%%%%%%%%%%%%%%%%%%%%%%%%%%%%%%%%%%

%%%%%%%%%%%%%%%%%%%%%%%%%%%%%%%%%%%%%%%%%%%%%%%%%%%%%%%%%%%%%%%%%%%%%%%%%%%
\begin{proof}
%%%%%%%%%%%%%%%%%%%%%%%%%%%%%%%%%%%%%%%%%%%%%%%%%%%%%%%%%%%%%%%%%%%%%%%%%%%

Clearly, the initial object $\zero$ of $\presh{\forest}$ (i.e.\ the empty presheaf) is a path. To see that representable functors are paths, it suffices to show that $\yo a$ is a path whenever $a\in \forest$.
Let 
\[
a_1 \prec \dots \prec a_n=a
\]
be the finite chain of predecessors of $a$ in $\forest$.
Direct inspection shows that the poset of subobjects of $\yo a$
can be identified with the finite chain 
\[
  \zero
  \prec
  \yo a_1
  \prec
  \dots
  \prec
  \yo a.
%  ~,
\]
(Recall that, in a topos, the unique arrow from the initial object $\zero$ is monic.)
Therefore, $\yo a$ is a path.

Conversely, suppose that $P$ is a path in $\presh{\forest}$ that is not initial. As the Yoneda embedding $\yo \colon \forest \to \presh{\forest}$ is dense,
$P$ is the colimit of the canonical diagram
\[
\begin{array}{*{7}{l}}
  D
& \colon
& \yo/{P}
& \stackrel{\pi}\longto
& \forest
& \stackrel{\yo}\longto
& \presh{\forest}.
\end{array}
\]

\noindent
This diagram is non-empty because $P$ is not initial.
Suppose for a moment that every morphism $\yo a\to P$, with $a\in\forest$, is monic.
Then, since $P$ is a path, there is an object $i$ of $\yo/{P}$ such that,
for any other $j$ in $\yo/{P}$,
the colimit map $D(j)\to P$ factors through the colimit map $D(i)\to P$.
It follows that $P\cong \yo a$ where $a \deq \pi(i)$, and so $P$ is representable.

It remains to show that each morphism $\yo a\to P$ is monic. In turn, this follows at once from the fact that $\forest$ is a posetal category: any component of a natural transformation $f\colon \yo a\to P$ is injective because its domain has at most one element, and so $f$ is monic.
%%%%%%%%%%%%%%%%%%%%%%%%%%%%%%%%%%%%%%%%%%%%%%%%%%%%%%%%%%%%%%%%%%%%%%%%%%%
\end{proof}
%%%%%%%%%%%%%%%%%%%%%%%%%%%%%%%%%%%%%%%%%%%%%%%%%%%%%%%%%%%%%%%%%%%%%%%%%%%

To complete the proof of Theorem~\ref{thm:arboreal:presh}, it remains to prove that $\presh{\forest}$ satisfies conditions~\ref{ax:2-out-of-3}--\ref{ax:path-generated}
in Definition~\ref{def:arboreal-cat}.

%%%%%%%%%%%%%%%%%%%%%%%%%%%%%%%%%%%%%%%%%%%%%%%%%%%%%%%%%%%%%%%%%%%%%%%%%%%
\begin{lemma}[Condition~\ref{ax:2-out-of-3}]
\label{l:presheaf-2-3}
%%%%%%%%%%%%%%%%%%%%%%%%%%%%%%%%%%%%%%%%%%%%%%%%%%%%%%%%%%%%%%%%%%%%%%%%%%%
$\presh{\forest}$ satisfies the 2-out-of-3 condition.
%%%%%%%%%%%%%%%%%%%%%%%%%%%%%%%%%%%%%%%%%%%%%%%%%%%%%%%%%%%%%%%%%%%%%%%%%%%
\end{lemma}
%%%%%%%%%%%%%%%%%%%%%%%%%%%%%%%%%%%%%%%%%%%%%%%%%%%%%%%%%%%%%%%%%%%%%%%%%%%

%%%%%%%%%%%%%%%%%%%%%%%%%%%%%%%%%%%%%%%%%%%%%%%%%%%%%%%%%%%%%%%%%%%%%%%%%%%
\begin{proof}
%%%%%%%%%%%%%%%%%%%%%%%%%%%%%%%%%%%%%%%%%%%%%%%%%%%%%%%%%%%%%%%%%%%%%%%%%%%
Consider paths $P,Q,Q'$ in $\presh{\forest}$ and assume that the composition
\[
\begin{array}{*{5}{l}}
  P
& \stackrel{\alpha}\longto
& Q
& \stackrel{\beta}\longto
& Q'
\end{array}
\]

\noindent
is an epimorphism.
It suffices to show that, for any $a\in \forest$,
the component $\alpha_a\colon P(a)\to Q(a)$ is surjective.
By Lemma~\ref{l:paths-in-presh-forest}, the sets $P(a)$ and $Q(a)$ have at most one element.
Hence, we must show that, for any $a\in \forest$,
$Q(a)\neq\emptyset$ entails $P(a)\neq\emptyset$.
We prove the contrapositive.
If $P(a)=\emptyset$ then, since $(\beta\comp \alpha)_a$ is an epimorphism,
also $Q'(a)=\emptyset$.
Since there is no function from a non-empty set to the empty set,
it must be $Q(a)=\emptyset$.
%%%%%%%%%%%%%%%%%%%%%%%%%%%%%%%%%%%%%%%%%%%%%%%%%%%%%%%%%%%%%%%%%%%%%%%%%%%
\end{proof}
%%%%%%%%%%%%%%%%%%%%%%%%%%%%%%%%%%%%%%%%%%%%%%%%%%%%%%%%%%%%%%%%%%%%%%%%%%%

Condition~\ref{ax:colimits} of Definition~\ref{def:arboreal-cat}
(stating that $\presh{\forest}$ has coproducts of sets of paths)
trivially follows from the cocompleteness of $\presh{\forest}$.

%%%%%%%%%%%%%%%%%%%%%%%%%%%%%%%%%%%%%%%%%%%%%%%%%%%%%%%%%%%%%%%%%%%%%%%%%%%
\begin{lemma}[Condition~\ref{ax:connected}]
\label{l:paths-in-L-connected}
%%%%%%%%%%%%%%%%%%%%%%%%%%%%%%%%%%%%%%%%%%%%%%%%%%%%%%%%%%%%%%%%%%%%%%%%%%%
Every path in $\presh{\forest}$ is connected.
\end{lemma}

%%%%%%%%%%%%%%%%%%%%%%%%%%%%%%%%%%%%%%%%%%%%%%%%%%%%%%%%%%%%%%%%%%%%%%%%%%%
\begin{proof}
%%%%%%%%%%%%%%%%%%%%%%%%%%%%%%%%%%%%%%%%%%%%%%%%%%%%%%%%%%%%%%%%%%%%%%%%%%%
The initial object of $\presh{\forest}$ is trivially connected.
If $P$ is a non-initial path in $\presh{\forest}$,
by Lemma~\ref{l:paths-in-presh-forest}
we can assume
that $P$ is of the form $\yo a$ for some $a\in \forest$.
For every non-empty set of paths $\{P_i\mid i\in I\}$, we have:
\begin{align*}
  \presh{\forest}\funct{\yo a, \mathord{\coprod}_{i\in I}{P_i}}
& \cong
  \big(\coprod_{i\in I}{P_i}\big)(a)
  \tag{Yoneda Lemma}
\\
& \cong
  \coprod_{i\in I}{P_i(a)}
  \tag{colimits computed in $\Set$}
\\
& \cong
  \coprod_{i\in I}{\presh{\forest}\funct{\yo a, P_i}}.
  \tag{Yoneda Lemma}
\end{align*}

\noindent
Direct inspection shows that the ensuing bijection
\[
  \coprod_{i\in I}{\presh{\forest}\funct{\yo a, P_i}}
  \longto
  \presh{\forest}\funct{\yo a, \mathord{\coprod}_{i\in I}{P_i}}
\]

\noindent
is the canonical map that, on a summand
$\presh{\forest}\funct{\yo a, P_j}$,
sends an arrow $\yo a\to P_j$ to its composition with the coproduct map
$P_j \to \coprod_{i\in I}{P_i}$. Therefore, $\yo a$ is connected.
\end{proof}

%%%%%%%%%%%%%%%%%%%%%%%%%%%%%%%%%%%%%%%%%%%%%%%%%%%%%%%%%%%%%%%%%%%%%%%%%%%
\begin{lemma}[Condition~\ref{ax:path-generated}]
\label{l:L-path-gen}
%%%%%%%%%%%%%%%%%%%%%%%%%%%%%%%%%%%%%%%%%%%%%%%%%%%%%%%%%%%%%%%%%%%%%%%%%%%
Every object of $\presh{\forest}$ is path-generated.
%%%%%%%%%%%%%%%%%%%%%%%%%%%%%%%%%%%%%%%%%%%%%%%%%%%%%%%%%%%%%%%%%%%%%%%%%%%
\end{lemma}
%%%%%%%%%%%%%%%%%%%%%%%%%%%%%%%%%%%%%%%%%%%%%%%%%%%%%%%%%%%%%%%%%%%%%%%%%%%

%%%%%%%%%%%%%%%%%%%%%%%%%%%%%%%%%%%%%%%%%%%%%%%%%%%%%%%%%%%%%%%%%%%%%%%%%%%
\begin{proof}
%%%%%%%%%%%%%%%%%%%%%%%%%%%%%%%%%%%%%%%%%%%%%%%%%%%%%%%%%%%%%%%%%%%%%%%%%%%
By Remark~\ref{rem:path-density},
it suffices to show that the full subcategory
$\pth{\presh{\forest}}$ of $\presh{\forest}$ defined by the paths is dense.
Representables in $\presh{\forest}$ form a dense subcategory;
as $\pth{\presh{\forest}}$ can be identified by Lemma~\ref{l:paths-in-presh-forest} with
the subcategory of representables, along with the initial object, it is also dense.
%%%%%%%%%%%%%%%%%%%%%%%%%%%%%%%%%%%%%%%%%%%%%%%%%%%%%%%%%%%%%%%%%%%%%%%%%%%
\end{proof}
%%%%%%%%%%%%%%%%%%%%%%%%%%%%%%%%%%%%%%%%%%%%%%%%%%%%%%%%%%%%%%%%%%%%%%%%%%%

%%%%%%%%%%%%%%%%%%%%%%%%%%%%%%%%%%%%%%%%%%%%%%%%%%%%%%%%%%%%%%%%%%%%%%%%%%%
\section{Gabriel-Ulmer duality}
\label{sec:lfp}
%%%%%%%%%%%%%%%%%%%%%%%%%%%%%%%%%%%%%%%%%%%%%%%%%%%%%%%%%%%%%%%%%%%%%%%%%%%

In the present work, 
we rely crucially on Gabriel-Ulmer duality for
locally finitely presentable (lfp) categories,
which has two possible readings: an algebraic one and a logical one.
In its algebraic reading,
Gabriel-Ulmer duality is
a dual biequivalence between lfp categories and
\emph{finitely complete small categories}.
Applying the duality twice turns an lfp category into an equivalent category
consisting of all finite-limit-preserving functors
\[
\cat C \longto \Set,
\]

\noindent
where $\cat C$ is small and finitely complete.
In its logical reading, Gabriel-Ulmer duality is a syntax-semantics duality, with
finitely complete small categories representing certain first-order theories,
namely the cartesian theories of \S\ref{sec:prelim:coste},
and lfp categories being categories of models of cartesian theories.
In this section we focus on the algebraic view,
leaving the discussion of the logical view to~\S\ref{sec:coste}.

We begin in~\S\ref{sec:lfp:alglatt} by recalling
the posetal restriction of Gabriel-Ulmer duality,
namely the
duality between (meet-)semilattices and algebraic lattices~\cite{Mislove-scattered1984}.
The full duality between lfp categories and small finitely complete
categories is presented in~\S\ref{ss:lex-and-lfp}.%
\footnote{In its original form~\cite{gu71lnm},
this duality involves small $\kappa$-complete categories and
locally $\kappa$-presentable categories, for~$\kappa$ a regular cardinal.
We shall focus on the case $\kappa = \aleph_0$.}

%%%%%%%%%%%%%%%%%%%%%%%%%%%%%%%%%%%%%%%%%%%%%%%%%%%%%%%%%%%%%%%%%%%%%%%%%%%
\subsection{Semilattices and algebraic lattices}
\label{sec:lfp:alglatt}
%%%%%%%%%%%%%%%%%%%%%%%%%%%%%%%%%%%%%%%%%%%%%%%%%%%%%%%%%%%%%%%%%%%%%%%%%%%
A \emph{(meet-)semilattice} is a poset in which every finite subset has an infimum,
and a morphism of semilattices is a map that preserves finite infima.
We denote the ensuing category by $\SL$. 

In the same way that Boolean algebras correspond to classical propositional theories
via the Lindenbaum-Tarski construction,
semilattices correspond to theories in the $\{\top,\wedge\}$-fragment of
propositional logic.
From this viewpoint, a \emph{model} of a semilattice $L$ is a morphism of semilattices $h\colon L\to 2$
into the two-element semilattice.
The set $\SL[L,2]$ of all models of $L$ is a poset with respect to
the pointwise order,
and is order-isomorphic to the poset $\Filter L$ of filters on $L$
(ordered by set-theoretic inclusion)
via the map $h\mapsto h^{-1}(1)$,
where $1$ denotes the infimum of the empty set
(i.e.\ the largest element of the semilattice).

%%%%%%%%%%%%%%%%%%%%%%%%%%%%%%%%%%%%%%%%%%%%%%%%%%%%%%%%%%%%%%%%%%%%%%%%%%%
\begin{remark}
\label{rem:Yoneda-factors-semilattices}
%%%%%%%%%%%%%%%%%%%%%%%%%%%%%%%%%%%%%%%%%%%%%%%%%%%%%%%%%%%%%%%%%%%%%%%%%%%
Regarding a semilattice $L$ as a category, the 
functor category $[L,2]$ 
can be identified with the poset $\Up L$ of upwards closed subsets of $L$.
The (contravariant) Yoneda embedding $L^\op \to \funct{L,\Set}$ factors through the embedding $L^\op \to \Up L$ that
sends $a$ to $\up a$. In turn, the latter factors through the inclusion $\Filter L\into \Up L$.
In fact, the image of the Yoneda embedding consists of the \emph{principal} filters on $L$.
\end{remark}

The intersection of a set of filters is again a filter,
hence the poset $\Filter L$ is a complete lattice.
Further, \emph{directed} suprema in $\Filter L$ are given by set-theoretic unions.
It follows that, for all morphisms of semilattices $k\colon L\to M$,
the map $\Filter k\colon \Filter M \to \Filter L$
defined by $F\mapsto k^{-1}(F)$ preserves arbitrary infima and directed suprema.

In order to characterise the lattices of the form $\Filter L$,
recall that an element $x$ of a poset $X$ is \emph{compact}
if, for all directed subsets $S\subseteq X$ admitting a supremum,
$x\leq \bigvee{S}$ implies $x\leq y$ for some $y\in S$. 
Regarding the poset $X$ as a category,
an element $x\in X$ is compact precisely when it is finitely presentable.
Principal filters are (exactly the) compact elements of $\Filter L$,
thus each element of $\Filter L$ is a directed supremum of compact ones.
That is, $\Filter L$ is an \emph{algebraic lattice}:

%%%%%%%%%%%%%%%%%%%%%%%%%%%%%%%%%%%%%%%%%%%%%%%%%%%%%%%%%%%%%%%%%%%%%%%%%%%
\begin{definition}
\label{def:lfp:alglatt}
%%%%%%%%%%%%%%%%%%%%%%%%%%%%%%%%%%%%%%%%%%%%%%%%%%%%%%%%%%%%%%%%%%%%%%%%%%%
A poset $X$ is an \emph{algebraic lattice} if it satisfies the following conditions:
\begin{enumerate}[(i)]
\item $X$ is a complete lattice.
\item Every element of $X$ is the supremum of a directed set of compact elements.
\end{enumerate}

\noindent
We write $\Alg$ for the category of algebraic lattices and maps that preserve
directed suprema and all infima.
%%%%%%%%%%%%%%%%%%%%%%%%%%%%%%%%%%%%%%%%%%%%%%%%%%%%%%%%%%%%%%%%%%%%%%%%%%%
\end{definition}
%%%%%%%%%%%%%%%%%%%%%%%%%%%%%%%%%%%%%%%%%%%%%%%%%%%%%%%%%%%%%%%%%%%%%%%%%%%

%%%%%%%%%%%%%%%%%%%%%%%%%%%%%%%%%%%%%%%%%%%%%%%%%%%%%%%%%%%%%%%%%%%%%%%%%%%
\begin{remark}
%%%%%%%%%%%%%%%%%%%%%%%%%%%%%%%%%%%%%%%%%%%%%%%%%%%%%%%%%%%%%%%%%%%%%%%%%%%
A poset (regarded as a small category)
is lfp if, and only if, it is an algebraic lattice
\cite[Example 1.10(5)]{ar94book}.
Moreover, the morphisms of $\Alg$ correspond precisely to the morphisms of lfp categories.
In fact, since a cocomplete small category
is a preorder with all suprema
\cite[Proposition V.2.3]{maclane98book},
and is thus equivalent (as a category) to a complete lattice, every small lfp category is equivalent to an algebraic lattice.
%%%%%%%%%%%%%%%%%%%%%%%%%%%%%%%%%%%%%%%%%%%%%%%%%%%%%%%%%%%%%%%%%%%%%%%%%%%
\end{remark}
%%%%%%%%%%%%%%%%%%%%%%%%%%%%%%%%%%%%%%%%%%%%%%%%%%%%%%%%%%%%%%%%%%%%%%%%%%%

Given an algebraic lattice $X$,
let $K X$ be the set of compact elements of $X$ with the induced order.
Note that $K X$ admits finite suprema (computed in $X$),
hence $(K X)^\op$ is a (meet-)semilattice.
Now, consider a morphism $f\colon X\to Y$ in $\Alg$.
Since $f$ preserves all infima, it has a left adjoint $\ladj f$.
In turn, $\ladj f$ preserves compact elements because $f$ preserves directed suprema.
If $Kf\colon K Y\to K X$ denotes the restriction of $\ladj f$ to the compact elements,
its order-dual $(Kf)^\op \colon (K Y)^\op\to (K X)^\op$ is a morphism of semilattices.
We thus obtain a pair of functors
\[
\Filter(-)\colon \SL^\op \leftrightarrows \Alg \cocolon K(-)^\op.
\]

These functors are quasi-inverse each other,
yielding a dual equivalence between $\SL$ and $\Alg$, see~\cite{Mislove-scattered1984}.
In particular, for every semilattice $L$ and algebraic lattice $X$ there are natural
isomorphisms $L\cong (K(\Filter L))^\op$ and $X\cong \Filter (K(X)^\op)$. 

\begin{remark}
For any algebraic lattice $X$, the set $\Alg[X,2]$ equipped with the pointwise order is isomorphic to $(K X)^\op$. In fact, $K(-)^\op\colon \Alg\to\SL^\op$ is naturally isomorphic to (the appropriate co-restriction of) the functor $\Alg[-,2]$. In other words, the duality between semilattices and algebraic lattices is induced by the dualising object $2$, regarded as either a semilattice or an algebraic lattice.
\end{remark}

%%%%%%%%%%%%%%%%%%%%%%%%%%%%%%%%%%%%%%%%%%%%%%%%%%%%%%%%%%%%%%%%%%%%%%%%%%%
\subsection{Finitely complete and locally finitely presentable categories}
\label{ss:lex-and-lfp}
%%%%%%%%%%%%%%%%%%%%%%%%%%%%%%%%%%%%%%%%%%%%%%%%%%%%%%%%%%%%%%%%%%%%%%%%%%%
Let $\cat D$ be a finitely complete small category. 
Recall that we denote by $\lex\funct{\cat D,\Set}$ the full subcategory of the functor
category $\funct{\cat D,\Set}$ consisting of the lex-morphisms,
i.e.\ of those functors $\cat D \to \Set$
which preserve finite limits.

%%%%%%%%%%%%%%%%%%%%%%%%%%%%%%%%%%%%%%%%%%%%%%%%%%%%%%%%%%%%%%%%%%%%%%%%%%%
\begin{remark}
%%%%%%%%%%%%%%%%%%%%%%%%%%%%%%%%%%%%%%%%%%%%%%%%%%%%%%%%%%%%%%%%%%%%%%%%%%%
A poset, regarded as a (small) category,
is finitely complete if and only if it is a semilattice.
A map between semilattices is a morphism of semilattices precisely when it is a lex-morphism.
%%%%%%%%%%%%%%%%%%%%%%%%%%%%%%%%%%%%%%%%%%%%%%%%%%%%%%%%%%%%%%%%%%%%%%%%%%%
\end{remark}
%%%%%%%%%%%%%%%%%%%%%%%%%%%%%%%%%%%%%%%%%%%%%%%%%%%%%%%%%%%%%%%%%%%%%%%%%%%

Generalising Remark~\ref{rem:Yoneda-factors-semilattices},
since representable functors preserve limits,
the contravariant Yoneda embedding factors as 
\[
\begin{tikzcd}
  \cat D^\op
  \arrow{r}{\yoneda}
& \lex\funct{\cat D,\Set}
  \arrow[hook]{r}{\incl}
& \funct{\cat D,\Set}.
\end{tikzcd}
\]

Limits and colimits in $\funct{\cat D,\Set}$ exist and are computed pointwise;
since limits commute with (finite) limits,
$\lex\funct{\cat D,\Set}$ is closed in $\funct{\cat D,\Set}$ under limits.
Similarly, as filtered colimits in $\Set$ commute with finite limits,
if $f \in \funct{\cat D, \Set}$ is a filtered colimit 
of finite-limit-preserving functors,
then $f$ itself preserves finite limits.
Since $\incl \colon \lex\funct{\cat D,\Set} \into \funct{\cat D,\Set}$
is a full inclusion, we get that
filtered colimits in $\lex\funct{\cat D,\Set}$ exist and are preserved by~$\incl$.
In other words,
$\lex\funct{\cat D,\Set}$ is closed in $\funct{\cat D,\Set}$ under filtered colimits.
Further, $\lex\funct{\cat D,\Set}$ admits also finite colimits,
although they are not computed pointwise
(see e.g.\ \cite[Theorem~VI.1.6]{johnstone82book},
where it is also proved that
$\yoneda \colon \cat D^\op \to \lex\funct{\cat D,\Set}$ preserves finite colimits).
Hence, $\lex\funct{\cat D,\Set}$ is a cocomplete category.

Recall that the categorical analogue of compact elements are finitely presentable objects. 
By the Yoneda lemma, all representable functors are finitely presentable in
$\funct{\cat D,\Set}$,
hence in $\lex\funct{\cat D,\Set}$ because the inclusion $\incl$ preserves
filtered colimits.\footnote{\label{footnote:lfp:splitting}%
In fact, since idempotents split in $\cat D$,
the finitely presentable objects of $\lex\funct{\cat D,\Set}$
are exactly the representable ones, 
see~\cite[Proposition 6.5.4 and Lemma 6.5.6]{borceux94vol1}.}\
Since every object $F$ of $\funct{\cat D,\Set}$ is the colimit of the
canonical diagram of representable functors 
\[
\begin{tikzcd}
  (\yoneda/F)
  \arrow{r}{\pi}
& \cat D^\op
  \arrow{r}{\yoneda}
& \lex\funct{\cat D,\Set}
  \arrow[hook]{r}{\incl}
& \funct{\cat D,\Set},
\end{tikzcd}
\]

\noindent
and the latter is filtered whenever $F$ is a lex-morphism
(cf.\ e.g.\ \cite[Proposition~VI.1.3]{johnstone82book}),
every object of $\lex\funct{\cat D,\Set}$ is a filtered colimit
of representable (hence finitely presentable) ones.
Thus, $\lex\funct{\cat D,\Set}$ is a locally finitely presentable category.

Conversely, given an lfp category $\E$,
its full subcategory $\fp\E$ defined by
the finitely presentable objects
is essentially small and closed under finite colimits,
see e.g.\ \cite[Proposition~1.3 and Remark~1.9]{ar94book}.
Replacing $\fp\E$ with a small skeleton if necessary,
$\fp\E^\op$ is a small finitely complete category,
and there is an equivalence of categories
\[
\begin{array}{l l l}
  \E
& \simeq
& \lex\funct{\fp\E^\op,\Set},
\end{array}
\] 

\noindent
see e.g.\ \cite[Theorem~1.46]{ar94book}.
In particular, a category is lfp if, and only if, it is equivalent to one of the form
$\lex\funct{\cat D,\Set}$ for some finitely complete small category $\cat D$.
As a result, each lfp category is complete and well-powered.

We now turn to morphisms of lfp categories.
Recall from Definition~\ref{def:lfp}
that they are the limit-preserving finitary functors between lfp categories.
Every lex-morphism $f\colon \cat D \to \cat D'$ induces a functor 
\begin{equation}
\label{eq:f-upper-star}
  \ladj f\colon \lex\funct{\cat D',\Set} \to \lex\funct{\cat D,\Set},
  \
  \ladj f G \deq G\comp f.
\end{equation}

\noindent
The functor $-\comp f\colon \funct{\cat D',\Set} \to \funct{\cat D,\Set}$
preserves all limits and colimits because it is both right and left adjoint
(cf.\ e.g.\ \cite[Example A4.1.4]{johnstone02book}).
Hence, $\ladj f$ preserves limits and filtered colimits
since $\lex\funct{\cat D,\Set}$ and
$\lex\funct{\cat D',\Set}$ are closed under limits and filtered colimits in
the respective functor categories.
That is, $\ladj f$ is a morphism of lfp categories.

We shall now see that every morphism of lfp categories
$F \colon \cat E \to \cat F$
arises from one of the form 
$\ladj f$ for a lex-morphism
$f\colon \fp{\cat F}^\op \to \fp{\cat E}^\op$.
To this end, we recall the following known characterisation of morphisms of lfp categories:

%%%%%%%%%%%%%%%%%%%%%%%%%%%%%%%%%%%%%%%%%%%%%%%%%%%%%%%%%%%%%%%%%%%%%%%%%%%
\begin{proposition}
\label{p:lfp-morphisms-characterisation}
%%%%%%%%%%%%%%%%%%%%%%%%%%%%%%%%%%%%%%%%%%%%%%%%%%%%%%%%%%%%%%%%%%%%%%%%%%%
A functor between lfp categories is a morphism of lfp categories
if, and only if, it has a left adjoint that preserves finitely presentable objects.
%%%%%%%%%%%%%%%%%%%%%%%%%%%%%%%%%%%%%%%%%%%%%%%%%%%%%%%%%%%%%%%%%%%%%%%%%%%
\end{proposition}
%%%%%%%%%%%%%%%%%%%%%%%%%%%%%%%%%%%%%%%%%%%%%%%%%%%%%%%%%%%%%%%%%%%%%%%%%%%

%%%%%%%%%%%%%%%%%%%%%%%%%%%%%%%%%%%%%%%%%%%%%%%%%%%%%%%%%%%%%%%%%%%%%%%%%%%
\begin{fullproof}
%%%%%%%%%%%%%%%%%%%%%%%%%%%%%%%%%%%%%%%%%%%%%%%%%%%%%%%%%%%%%%%%%%%%%%%%%%%
By the Adjoint Functor Theorem for lfp categories (cf.\ \cite[1.66]{ar94book}),
any limit-preserving finitary functor between lfp categories is right adjoint.
Further, a right adjoint between lfp categories is finitary precisely when
its left adjoint preserves finitely presentable objects
(cf.\ the proof in \emph{loc.\ cit.}).
%%%%%%%%%%%%%%%%%%%%%%%%%%%%%%%%%%%%%%%%%%%%%%%%%%%%%%%%%%%%%%%%%%%%%%%%%%%
\end{fullproof} 
%%%%%%%%%%%%%%%%%%%%%%%%%%%%%%%%%%%%%%%%%%%%%%%%%%%%%%%%%%%%%%%%%%%%%%%%%%%

%%%%%%%%%%%%%%%%%%%%%%%%%%%%%%%%%%%%%%%%%%%%%%%%%%%%%%%%%%%%%%%%%%%%%%%%%%%
\opt{full}{\begin{remark} % BEGIN FULL
\label{rem:lfp:ladj}
%%%%%%%%%%%%%%%%%%%%%%%%%%%%%%%%%%%%%%%%%%%%%%%%%%%%%%%%%%%%%%%%%%%%%%%%%%%
Let $f \in \lex\funct{\cat D,\cat D'}$.
As $\lex\funct{\cat D',\Set}$ is cocomplete,
the restriction of the functor
$- \comp f$ to $\lex\funct{\cat D',\Set} \to \funct{\cat D,\Set}$
has a left adjoint $\eadj f$
that sends $F$ to the left Kan extension $\Lan_{\yoneda}(\yoneda f)(F)=\Lan_f{F}$
\cite[Theorem I.5.2]{mm92book}.
The left adjoint of $\ladj f$ is then given by the restriction
of $\eadj f$ to $\lex\funct{\cat D,\Set}$.
%%%%%%%%%%%%%%%%%%%%%%%%%%%%%%%%%%%%%%%%%%%%%%%%%%%%%%%%%%%%%%%%%%%%%%%%%%%
\end{remark}} % END FULL
%%%%%%%%%%%%%%%%%%%%%%%%%%%%%%%%%%%%%%%%%%%%%%%%%%%%%%%%%%%%%%%%%%%%%%%%%%%

If $F\colon \cat E \to \cat F$ is a morphism of lfp categories then, 
by Proposition~\ref{p:lfp-morphisms-characterisation}, it has a left adjoint
that restricts to a lex-morphism $f\colon \fp{\cat F}^\op \to \fp{\cat E}^\op$.
Then $F$ corresponds to $\ladj f$
via the equivalences $\cat E \cong \lex\funct{\fp{\cat E}^\op,\Set}$
and $\cat F \cong \lex\funct{\fp{\cat F}^\op,\Set}$.
In particular, every morphism of lfp categories
$\lex\funct{\cat D',\Set} \to \lex\funct{\cat D,\Set}$
is naturally isomorphic to one of the form $\ladj f$ for a lex-morphism
$f\colon \cat D \to \cat D'$. 

%%%%%%%%%%%%%%%%%%%%%%%%%%%%%%%%%%%%%%%%%%%%%%%%%%%%%%%%%%%%%%%%%%%%%%%%%%%
\begin{remark}
\label{rem:finitely-accessible-adj}
%%%%%%%%%%%%%%%%%%%%%%%%%%%%%%%%%%%%%%%%%%%%%%%%%%%%%%%%%%%%%%%%%%%%%%%%%%%
In Definition~\ref{def:path:finaccadj},
we called an adjunction
$L\dashv R$ between lfp categories \emph{finitely accessible} when $R$ is finitary.
By Proposition~\ref{p:lfp-morphisms-characterisation},
the morphisms of lfp categories can be characterised 
as the 
(right adjoints in) finitely accessible adjunctions.
%%%%%%%%%%%%%%%%%%%%%%%%%%%%%%%%%%%%%%%%%%%%%%%%%%%%%%%%%%%%%%%%%%%%%%%%%%%
\end{remark}
%%%%%%%%%%%%%%%%%%%%%%%%%%%%%%%%%%%%%%%%%%%%%%%%%%%%%%%%%%%%%%%%%%%%%%%%%%%

Write $\LFP$ for the 2-category whose 0-cells (objects) are
locally finitely presentable categories, whose 1-cells (morphisms)
are morphisms of lfp categories, 
and whose 2-cells are natural transformations.
Moreover, write $\Lex$ for the 2-category of small finitely complete categories,
lex-morphisms and natural transformations.
In view of the above,
there is a 2-functor from $\Lex^\op$ (reverse the 1-cells) to
$\LFP$ sending $\cat D$ to $\lex\funct{\cat D,\Set}$. 
\emph{Gabriel-Ulmer duality} states that the latter is a
biequivalence $\Lex^\op \to \LFP$
(with left biadjoint $\E \mapsto \fp\E^\op$).
This was established at the level of objects by Gabriel and Ulmer 
in~\cite{gu71lnm}
and later extended to a fully fledged duality, cf.~\cite{mp87tams,ap98ja}.
%%%%%%%%%%%%%%%%%%%%%%%%%%%%%%%%%%%%%%%%%%%%%%%%%%%%%%%%%%%%%%%%%%%%%%%%%%%
%\end{remark}
%%%%%%%%%%%%%%%%%%%%%%%%%%%%%%%%%%%%%%%%%%%%%%%%%%%%%%%%%%%%%%%%%%%%%%%%%%%
%%%%%%%%%%%%%%%%%%%%%%%%%%%%%%%%%%%%%%%%%%%%%%%%%%%%%%%%%%%%%%%%%%%%%%%%%%%
\section{Functorial semantics for cartesian theories}
\label{sec:coste}
%%%%%%%%%%%%%%%%%%%%%%%%%%%%%%%%%%%%%%%%%%%%%%%%%%%%%%%%%%%%%%%%%%%%%%%%%%%

Building on the idea of functorial semantics,
logical characterisations of lfp categories
have been proposed, see e.g.~\cite{ar94book}.
We recall the approach of Coste~\cite{coste76benabou},
according to which lfp categories correspond exactly
to categories of models of cartesian theories
(in the sense of \S\ref{sec:prelim:coste}).
We mostly follow the presentation in~\cite[\S D1]{johnstone02book}.

%%%%%%%%%%%%%%%%%%%%%%%%%%%%%%%%%%%%%%%%%%%%%%%%%%%%%%%%%%%%%%%%%%%%%%%%%%%
\subsection{Syntactic categories and functorial semantics}
\label{sec:coste:synt}
%%%%%%%%%%%%%%%%%%%%%%%%%%%%%%%%%%%%%%%%%%%%%%%%%%%%%%%%%%%%%%%%%%%%%%%%%%%

A basic idea of categorical logic
is that a theory $\theory$ should be seen as a category,
its \emph{syntactic category},
and models of $\theory$ as certain functors defined on the latter category.
Theorem~\ref{thm:coste:synt} and
Corollary~\ref{cor:lfp-iff-models-of-T}
below establish these equivalences for cartesian theories and their $\Set$-based models.

%%%%%%%%%%%%%%%%%%%%%%%%%%%%%%%%%%%%%%%%%%%%%%%%%%%%%%%%%%%%%%%%%%%%%%%%%%%
\begin{definition}[Syntactic category]
\label{def:coste:synt}
%%%%%%%%%%%%%%%%%%%%%%%%%%%%%%%%%%%%%%%%%%%%%%%%%%%%%%%%%%%%%%%%%%%%%%%%%%%
The \emph{syntactic category} $\cat T$ of a
cartesian
theory $\theory$
is defined as follows.
\begin{description}

\item[Objects]
are
formulae-in-context
$\Env \sorting \varphi$
up to renaming of variables
(i.e., $\alpha$-equivalence),
denoted by
$\OI{\Env \mid \varphi}$.

\item[Morphisms]
from
$\OI{\vec x \mid \varphi}$
to
$\OI{\vec y \mid \psi}$
are $\theory$-provable equivalence classes of $\theory$-provable functional
relations from
$\OI{\vec x \mid \varphi}$ to
$\OI{\vec y \mid \psi}$,
i.e.\ $\theory$-provable equivalence classes of formulae-in-context
$\vec x,\vec y \sorting \theta$
that are cartesian over $\theory$ and such that $\theory$ proves the sequents:
\[
\begin{array}{r l l}
  \theta
& \thesis_{\vec x,\vec y}
& \varphi \land \psi
\\

  \theta
  \land
  \theta[\vec z/\vec y]
& \thesis_{\vec x,\vec y,\vec z}
& \vec z \Eq \vec y
\\

  \varphi
& \thesis_{\vec x}
& (\exists \vec y)\theta
\end{array}
\]

\noindent
We denote by $\MI{\vec x,\vec y \mid \theta}$ the equivalence class
of $\vec x,\vec y \sorting \theta$. 
\end{description}
\end{definition}

\noindent
Note that the last sequent above is cartesian over the penultimate one.
It is clear that $\cat T$ is a category.
Moreover:

%%%%%%%%%%%%%%%%%%%%%%%%%%%%%%%%%%%%%%%%%%%%%%%%%%%%%%%%%%%%%%%%%%%%%%%%%%%
\begin{lemma}
\label{lem:coste:synt:lim}
%%%%%%%%%%%%%%%%%%%%%%%%%%%%%%%%%%%%%%%%%%%%%%%%%%%%%%%%%%%%%%%%%%%%%%%%%%%
The syntactic category 
of a cartesian theory 
is finitely complete.
%%%%%%%%%%%%%%%%%%%%%%%%%%%%%%%%%%%%%%%%%%%%%%%%%%%%%%%%%%%%%%%%%%%%%%%%%%%
\end{lemma}
%%%%%%%%%%%%%%%%%%%%%%%%%%%%%%%%%%%%%%%%%%%%%%%%%%%%%%%%%%%%%%%%%%%%%%%%%%%

%%%%%%%%%%%%%%%%%%%%%%%%%%%%%%%%%%%%%%%%%%%%%%%%%%%%%%%%%%%%%%%%%%%%%%%%%%%
\begin{proof}[Sketch of proof]
%%%%%%%%%%%%%%%%%%%%%%%%%%%%%%%%%%%%%%%%%%%%%%%%%%%%%%%%%%%%%%%%%%%%%%%%%%%
The terminal object is given by the formula $\True$ in the empty context,
and binary products by binary conjunctions.
Equalisers are obtained by means of binary conjunctions and
(provably unique) existential quantifiers. 
See~\cite[Lemma D1.4.2]{johnstone02book}.
%%%%%%%%%%%%%%%%%%%%%%%%%%%%%%%%%%%%%%%%%%%%%%%%%%%%%%%%%%%%%%%%%%%%%%%%%%%
\end{proof}
%%%%%%%%%%%%%%%%%%%%%%%%%%%%%%%%%%%%%%%%%%%%%%%%%%%%%%%%%%%%%%%%%%%%%%%%%%%

Let $\theory$ be a cartesian theory.
Given a model $M \in \Mod(\theory)$,
we can associate with each object
$\OI{\vec x \mid \varphi}$ of the syntactic category $\cat T$
the set $\I{\vec x \mid \varphi}_M \in \Set$.
Further, each morphism
\(
  \MI{\vec x,\vec y \mid \theta}
  \colon
  \OI{\vec x \mid \varphi}
  \to
  \OI{\vec y \mid \psi}
\)
induces a functional relation
\begin{equation}
\label{eq:coste:graph}
\begin{tikzcd}
  \I{\vec x,\vec y \mid \theta}_M
  \arrow[hook]{r}
& \I{\vec x \mid \varphi}_M
  \times
  \I{\vec y \mid \psi}_M
\end{tikzcd}
\end{equation}

\noindent
and is thus the graph of a function
$\I{\vec x \mid \varphi}_M \to \I{\vec y \mid \psi}_M$.
These assignments determine a functor $F_M \colon \cat T \to \Set$
defined on the syntactic category of $\theory$.
It is not difficult to see that $F_M$ is a lex-morphism.
In fact,

%%%%%%%%%%%%%%%%%%%%%%%%%%%%%%%%%%%%%%%%%%%%%%%%%%%%%%%%%%%%%%%%%%%%%%%%%%%
\begin{theorem}[{\cite[Theorem D1.4.7]{johnstone02book}}]
\label{thm:coste:synt}
%%%%%%%%%%%%%%%%%%%%%%%%%%%%%%%%%%%%%%%%%%%%%%%%%%%%%%%%%%%%%%%%%%%%%%%%%%%
For each cartesian theory $\theory$,
the assignment
\[
\Mod(\theory) \to \lex\funct{\cat T, \Set}, \ \ M \mapsto F_M
\]
\noindent
is (part of) an equivalence of categories.
%%%%%%%%%%%%%%%%%%%%%%%%%%%%%%%%%%%%%%%%%%%%%%%%%%%%%%%%%%%%%%%%%%%%%%%%%%%
\end{theorem}
%%%%%%%%%%%%%%%%%%%%%%%%%%%%%%%%%%%%%%%%%%%%%%%%%%%%%%%%%%%%%%%%%%%%%%%%%%%

In the above discussion of Theorem~\ref{thm:coste:synt},
we said nothing about morphisms.
In the following lemma, we collect some basic facts relating
homomorphisms in $\Mod(\theory)$ to
natural transformations in $\lex\funct{\cat T,\Set}$.

%%%%%%%%%%%%%%%%%%%%%%%%%%%%%%%%%%%%%%%%%%%%%%%%%%%%%%%%%%%%%%%%%%%%%%%%%%%
\begin{lemma}
\label{lemma:coste:synt:homnat}
%%%%%%%%%%%%%%%%%%%%%%%%%%%%%%%%%%%%%%%%%%%%%%%%%%%%%%%%%%%%%%%%%%%%%%%%%%%
Let $\ell \colon N \to M$ be a morphism in $\Mod(\theory)$
and let $\lambda \colon F_N \to F_M$
be the corresponding natural transformation,
under the equivalence in Theorem~\ref{thm:coste:synt}.
Then
\begin{enumerate}[(1)]
\item
\label{item:synt:lexmod:homnat:sort}
$M(\sort) = \I{y:\sort \mid \True}_M = F_M\OI{y:\sort \mid \True}$
and
$\ell^\sort = \lambda_{\OI{y : \sort \mid \True}}$
for each sort $\sort \in \Sort(\Sig)$.
\setcounter{SplitEnum}{\value{enumi}}
\end{enumerate}

\noindent
Furthermore,
given $\vec y = y_1 : \sort_1,\dots,y_n : \sort_n$,
\begin{enumerate}[(1)]
\setcounter{enumi}{\value{SplitEnum}}
\setcounter{StartEnum}{\value{SplitEnum}}
\item
\(
  F_M\OI{\vec y \mid \True}
  =
  \prod_{i=1}^{n}
  F_M\OI{y_i : \sort_i \mid \True}
\),

\item
\(
  \lambda_{\OI{\vec y \mid \True}}
  =
  \prod_{i=1}^{n}
  \lambda_{\OI{y_i : \sort_i \mid \True}}
\).
\end{enumerate}

\noindent
In particular,
assuming $\vec a \in \I{\vec y \mid \True}_N$, we have
\[
\begin{array}{*{5}{l}}
  (\ell^{\sort_1}(a_1),\dots,\ell^{\sort_n}(a_n))
& =
& \lambda_{\OI{\vec y \mid \True}}(\vec a)
& \in
& \I{\vec y \mid \True}_M.
\end{array}
\]
%%%%%%%%%%%%%%%%%%%%%%%%%%%%%%%%%%%%%%%%%%%%%%%%%%%%%%%%%%%%%%%%%%%%%%%%%%%
\end{lemma}
%%%%%%%%%%%%%%%%%%%%%%%%%%%%%%%%%%%%%%%%%%%%%%%%%%%%%%%%%%%%%%%%%%%%%%%%%%%

%%%%%%%%%%%%%%%%%%%%%%%%%%%%%%%%%%%%%%%%%%%%%%%%%%%%%%%%%%%%%%%%%%%%%%%%%%%
\begin{fullproof}
%%%%%%%%%%%%%%%%%%%%%%%%%%%%%%%%%%%%%%%%%%%%%%%%%%%%%%%%%%%%%%%%%%%%%%%%%%%
Item~\ref{item:synt:lexmod:homnat:sort}
follows from the definitions of
$F_M$ and $\lambda$ from $M$ and $\ell$.
\begin{enumerate}[(1)]
\setcounter{enumi}{\value{StartEnum}}

\item
By Lemma~\ref{lem:coste:synt:lim},
in $\cat T$
we have
\[
\begin{array}{l l l}
  \OI{\vec y \mid \True}
& \cong
& \OI{y_1 : \sort_1 \mid \True}
  \times
  \dots
  \times
  \OI{y_n : \sort_n \mid \True}
\end{array}
\]

\noindent
Hence
\[
\begin{array}{l l l}
  F_N\OI{\vec y \mid \True}
& \cong
& F_N\OI{y_1 : \sort_1 \mid \True}
  \times
  \dots
  \times
  F_N\OI{y_n : \sort_n \mid \True}
\end{array}
\]

\noindent
since $F_N$ preserves finite limits.

\item
By naturality of $\lambda \colon F_N \to F_M$,
for $i = 1,\dots,n$ we have
\[
\begin{tikzcd}
  F_N\OI{\vec y \mid \True}
  \arrow{r}[above]{F_N(\pi_i)}
  \arrow{d}[left]{\lambda_{\OI{\vec y \mid \True}}}
& F_N\OI{y_i : \sort_i \mid \True} 
  \arrow{d}[right]{\lambda_{\OI{y_i : \sort_i \mid \True}}}
\\
  F_M\OI{\vec y \mid \True}
  \arrow{r}[above]{F_M(\pi_i)}
& F_M\OI{y_i : \sort_i \mid \True}   
\end{tikzcd}
\]

\noindent
where
$\pi_i \colon \OI{\vec y \mid \True} \to \OI{y_i : \sort_i \mid \True}$
is the $i$-th projection in $\cat T$.
But since $F_N$ and $F_M$ preserve finite limits,
the arrows
$F_N(\pi_i)$ and $F_M(\pi_i)$ are also projections,
and the result follows.
\qedhere
\end{enumerate}
\end{fullproof}

Theorem~\ref{thm:coste:synt},
combined with Example~\ref{ex:prelim:coste:funct}
and Gabriel-Ulmer duality (\S\ref{ss:lex-and-lfp}),
entails the following result, which amounts to a logical reading of Gabriel-Ulmer duality.

%%%%%%%%%%%%%%%%%%%%%%%%%%%%%%%%%%%%%%%%%%%%%%%%%%%%%%%%%%%%%%%%%%%%%%%%%%%
\begin{corollary}
\label{cor:lfp-iff-models-of-T}
%%%%%%%%%%%%%%%%%%%%%%%%%%%%%%%%%%%%%%%%%%%%%%%%%%%%%%%%%%%%%%%%%%%%%%%%%%%
The following statements hold:
\begin{enumerate}[(1)]
\item
\label{finitely-complete-synt-cat}
A small category is finitely complete
if, and only if,
it is equivalent to the syntactic category of a cartesian theory. 

\item
\label{lfp-model-cat}
A category is lfp
if, and only if,
it is equivalent to $\Mod(\theory)$ for some cartesian theory~$\theory$.
\end{enumerate}
%%%%%%%%%%%%%%%%%%%%%%%%%%%%%%%%%%%%%%%%%%%%%%%%%%%%%%%%%%%%%%%%%%%%%%%%%%%
\end{corollary}
%%%%%%%%%%%%%%%%%%%%%%%%%%%%%%%%%%%%%%%%%%%%%%%%%%%%%%%%%%%%%%%%%%%%%%%%%%%

%%%%%%%%%%%%%%%%%%%%%%%%%%%%%%%%%%%%%%%%%%%%%%%%%%%%%%%%%%%%%%%%%%%%%%%%%%%
\begin{fullproof}
%%%%%%%%%%%%%%%%%%%%%%%%%%%%%%%%%%%%%%%%%%%%%%%%%%%%%%%%%%%%%%%%%%%%%%%%%%%
\ref{finitely-complete-synt-cat}
One direction follows from Lemma~\ref{lem:coste:synt:lim}.
Conversely, if $\cat C$ is small and finitely complete,
by Example~\ref{ex:prelim:coste:funct} there is a theory $\theory$ such that
$\lex\funct{\cat C,\Set}\cong \Mod(\theory)$.
An application of Theorem~\ref{thm:coste:synt} yields
$\lex\funct{\cat C,\Set}\cong \lex\funct{\cat T, \Set}$ and so,
by Gabriel-Ulmer duality, $\cat C \cong \cat T$.
(Alternatively, observe that the equivalence
$\lex\funct{\cat C,\Set}\cong \lex\funct{\cat T, \Set}$
restricts to an equivalence between the full subcategories
on the finitely presentable objects, which in turn gives $\cat C \cong \cat T$.)

\ref{lfp-model-cat}
This follows from Example~\ref{ex:prelim:coste:funct} and Theorem~\ref{thm:coste:synt},
combined with the fact that, by Gabriel-Ulmer duality,
a category is lfp if and only if it is equivalent to one of the form
$\lex\funct{\cat C,\Set}$ for some small finitely complete category~$\cat C$.
%%%%%%%%%%%%%%%%%%%%%%%%%%%%%%%%%%%%%%%%%%%%%%%%%%%%%%%%%%%%%%%%%%%%%%%%%%%
\end{fullproof}
%%%%%%%%%%%%%%%%%%%%%%%%%%%%%%%%%%%%%%%%%%%%%%%%%%%%%%%%%%%%%%%%%%%%%%%%%%%

In particular, for any signature $\Sig$, Theorem~\ref{thm:coste:synt} entails that 
$\Struct(\Sig)$ is lfp as
\begin{equation}
\label{eq:coste:struct}
\begin{array}{l l l}
  \Struct(\Sig)
& \cong
& \lex\funct{\cat T(\Sig),\Set},
\end{array}
\end{equation}
where $\cat T(\Sig)$ is the syntactic category 
of the cartesian theory $\theory(\Sig)$ in $\Sig$ with no axioms.
Limits and filtered colimits in $\Struct(\Sig)$ are simple: they are computed in $\Set$;
see e.g.\ \cite[Remark~5.1]{ar94book}.
The same is true for categories of models of cartesian theories:
if $\theory$ is any cartesian theory in the signature $\Sig$,
then the identity-on-objects functor $\inclth\colon \cat T(\Sig)\to \cat T$
is a lex-morphism, and the corresponding morphism of lfp categories 
\[
\ladj\inclth\colon \lex\funct{\cat T,\Set}\to \lex\funct{\cat T(\Sig),\Set}
\] 

\noindent
coincides (via Theorem~\ref{thm:coste:synt}) to the inclusion functor
$\Mod(\theory) \into \Struct(\Sig)$. Thus,

%%%%%%%%%%%%%%%%%%%%%%%%%%%%%%%%%%%%%%%%%%%%%%%%%%%%%%%%%%%%%%%%%%%%%%%%%%%
\begin{lemma}
\label{lem:coste:mod:filtcolim}
%%%%%%%%%%%%%%%%%%%%%%%%%%%%%%%%%%%%%%%%%%%%%%%%%%%%%%%%%%%%%%%%%%%%%%%%%%%
For any cartesian theory $\theory$,
the inclusion functor $\Mod(\theory) \into \Struct(\Sig)$ is a morphism
of lfp categories.
In particular, $\Mod(\theory)$ is closed in $\Struct(\Sig)$ under limits and
filtered colimits.%
\footnote{In fact, for any first-order theory $\theory$,
$\Mod(\theory)$ is closed in $\Struct(\Sig)$ under filtered colimits,
and if it is closed in $\Struct(\Sig)$ under finite limits
then it is closed under arbitrary limits.
However, the proof of these facts is considerably more difficult
than in the case of cartesian theories; cf.\ \cite[Theorems~5.20 and~5.23]{ar94book}.}
%%%%%%%%%%%%%%%%%%%%%%%%%%%%%%%%%%%%%%%%%%%%%%%%%%%%%%%%%%%%%%%%%%%%%%%%%%%
\end{lemma}
%%%%%%%%%%%%%%%%%%%%%%%%%%%%%%%%%%%%%%%%%%%%%%%%%%%%%%%%%%%%%%%%%%%%%%%%%%%

%%%%%%%%%%%%%%%%%%%%%%%%%%%%%%%%%%%%%%%%%%%%%%%%%%%%%%%%%%%%%%%%%%%%%%%%%%%
\subsection{Finitely presentable models}
\label{sec:coste:fp}
%%%%%%%%%%%%%%%%%%%%%%%%%%%%%%%%%%%%%%%%%%%%%%%%%%%%%%%%%%%%%%%%%%%%%%%%%%%

Syntactic categories, combined with the Yoneda lemma,
allow us to think of finitely presentable objects in lfp categories
as being presented by ``generators and relations''.
This extends
Example~\ref{ex:lfp:struct} and Fact~\ref{fact:lfp:struct}
to categories of models of cartesian theories.

Fix an arbitrary cartesian theory $\theory$,
and write $\cat T$ for its syntactic category.

%%%%%%%%%%%%%%%%%%%%%%%%%%%%%%%%%%%%%%%%%%%%%%%%%%%%%%%%%%%%%%%%%%%%%%%%%%%
\begin{notation}
%%%%%%%%%%%%%%%%%%%%%%%%%%%%%%%%%%%%%%%%%%%%%%%%%%%%%%%%%%%%%%%%%%%%%%%%%%%
Given $\OI{\vec x \mid \varphi}\in \cat T$,
we write $\FG{\vec x \mid \varphi} \in \Mod(\theory)$
for the model 
corresponding
(under Theorem~\ref{thm:coste:synt})
to the representable functor
$\yoneda \OI{\vec x \mid \varphi} \in \lex\funct{\cat T,\Set}$.
%%%%%%%%%%%%%%%%%%%%%%%%%%%%%%%%%%%%%%%%%%%%%%%%%%%%%%%%%%%%%%%%%%%%%%%%%%%
\end{notation}
%%%%%%%%%%%%%%%%%%%%%%%%%%%%%%%%%%%%%%%%%%%%%%%%%%%%%%%%%%%%%%%%%%%%%%%%%%%

Recall from \S\ref{ss:lex-and-lfp}
that the finitely presentable objects of $\lex\funct{\cat T,\Set}$
are exactly the representable functors.
Hence, the finitely presentable objects of $\Mod(\theory)$
are, up to isomorphism,
exactly those of the form $\FG{\vec x \mid \varphi}$.

%%%%%%%%%%%%%%%%%%%%%%%%%%%%%%%%%%%%%%%%%%%%%%%%%%%%%%%%%%%%%%%%%%%%%%%%%%%
\begin{remark}
\label{rem:quantifier-free}
%%%%%%%%%%%%%%%%%%%%%%%%%%%%%%%%%%%%%%%%%%%%%%%%%%%%%%%%%%%%%%%%%%%%%%%%%%%
Up to isomorphism, we can always
assume that a finitely presentable model is of the form
$\FG{\vec x \mid \varphi}$
where $\varphi$ is \emph{quantifier-free}
(that is, a finite conjunction of atomic formulae).
See~\cite[Lemma D1.4.4(ii)]{johnstone02book}.
%%%%%%%%%%%%%%%%%%%%%%%%%%%%%%%%%%%%%%%%%%%%%%%%%%%%%%%%%%%%%%%%%%%%%%%%%%%
\end{remark}
%%%%%%%%%%%%%%%%%%%%%%%%%%%%%%%%%%%%%%%%%%%%%%%%%%%%%%%%%%%%%%%%%%%%%%%%%%%

%%%%%%%%%%%%%%%%%%%%%%%%%%%%%%%%%%%%%%%%%%%%%%%%%%%%%%%%%%%%%%%%%%%%%%%%%%%
\begin{remark}
\label{rem:coste:fp:induced-lfp:fp}
%%%%%%%%%%%%%%%%%%%%%%%%%%%%%%%%%%%%%%%%%%%%%%%%%%%%%%%%%%%%%%%%%%%%%%%%%%%
In view of Remark~\ref{rem:prelim:coste:compl},
if $\Mod(\theory) = \Mod(\theorybis)$
for some other cartesian theory $\theorybis$ in $\Sig$,
then $\theory$ and $\theorybis$ prove the same sequents.
In particular, 
a formula-in-context $\vec x \sorting \varphi$ is cartesian over $\theory$
exactly when it is cartesian over $\theorybis$,
and the respective syntactic categories of $\theory$ and $\theorybis$ are the same.

It follows that given $\E \into \Struct(\Sig)$ such that
$\E = \Mod(\theorybis)$ for some $\theorybis$,
the finitely presentable models $\FG{\vec x \mid \varphi} \in \E$
are independent from the choice of the theory
$\theorybis$ such that $\E = \Mod(\theorybis)$.
%%%%%%%%%%%%%%%%%%%%%%%%%%%%%%%%%%%%%%%%%%%%%%%%%%%%%%%%%%%%%%%%%%%%%%%%%%%
\end{remark}
%%%%%%%%%%%%%%%%%%%%%%%%%%%%%%%%%%%%%%%%%%%%%%%%%%%%%%%%%%%%%%%%%%%%%%%%%%%

%%%%%%%%%%%%%%%%%%%%%%%%%%%%%%%%%%%%%%%%%%%%%%%%%%%%%%%%%%%%%%%%%%%%%%%%%%%
\subsubsection{The Yoneda lemma}
\label{sec:coste:fp:yoneda}
%%%%%%%%%%%%%%%%%%%%%%%%%%%%%%%%%%%%%%%%%%%%%%%%%%%%%%%%%%%%%%%%%%%%%%%%%%%
Consider a model $M \in \Mod(\theory)$ and recall 
from Theorem~\ref{thm:coste:synt} that the functor $F_M \in \lex\funct{\cat T,\Set}$
associated to $M$ 
takes $\OI{\vec x \mid \varphi} \in \cat T$ to $\I{\vec x \mid \varphi}_M$.
The Yoneda lemma provides a bijection
\begin{equation}
\label{eq:Yoneda-interpretations}
\begin{array}{*{5}{l}}
  \I{\vec x \mid \varphi}_M
& =
& F_M(\OI{\vec x \mid \varphi})
& \cong
& \lex\funct{\cat T,\Set}\funct{\yoneda\OI{\vec x \mid \varphi},\, F_M}.
\end{array}
\end{equation}

\noindent
In one direction,
given $\vec a \in \I{\vec x \mid \varphi}_M$,
the unique natural transformation
\begin{equation}
\label{eq:lambda-vec-a}
\begin{array}{*{5}{l}}
  F_M(-)(\vec a)
& \colon
& \yoneda\OI{\vec x \mid \varphi}
& \longto
& F_M
\end{array}
\end{equation}

\noindent
corresponding to
$\vec a$
has component at $\OI{\vec y \mid \psi} \in \cat T$ defined by
\[
  \cat T\funct{\OI{\vec x \mid \varphi}, \OI{\vec y \mid \psi}}
  \to F_M(\OI{\vec y \mid \psi}),
  \quad
  \theta \mapsto F_M(\theta)(\vec a).
\]

\noindent
Conversely, 
given an arbitrary $\lambda \colon \yoneda\OI{\vec x \mid \varphi} \to F_M$,
set
$\vec a \deq \lambda_{\OI{\vec x \mid \varphi}}(\id) \in \I{\vec x \mid \varphi}_M$.
Then we have $\lambda = F_M(-)(\vec a)$.

Combining
eq.\ \eqref{eq:Yoneda-interpretations}
with the bijection 
\begin{equation}
\label{eq:fp-models-nat-bij}
\begin{array}{l l l}
  \lex\funct{\cat T,\Set}\funct{\yoneda\OI{\vec x \mid \varphi},\, F_M}
& \cong
& \Mod(\theory)\funct{\FG{\vec x \mid \varphi} ,\, M}
\end{array}
\end{equation}
from Theorem~\ref{thm:coste:synt}
yields the following
generalisation of Fact~\ref{fact:lfp:struct}.

%%%%%%%%%%%%%%%%%%%%%%%%%%%%%%%%%%%%%%%%%%%%%%%%%%%%%%%%%%%%%%%%%%%%%%%%%%%
\begin{lemma}
\label{lem:coste:lp:hom}
%%%%%%%%%%%%%%%%%%%%%%%%%%%%%%%%%%%%%%%%%%%%%%%%%%%%%%%%%%%%%%%%%%%%%%%%%%%
There is a bijection
\[
\begin{array}{l l l}
  \I{\vec x \mid \varphi}_M
& \cong
& \Mod(\theory)\funct{\FG{\vec x \mid \varphi}, M}.
\end{array}
\]
%%%%%%%%%%%%%%%%%%%%%%%%%%%%%%%%%%%%%%%%%%%%%%%%%%%%%%%%%%%%%%%%%%%%%%%%%%%
\end{lemma}
%%%%%%%%%%%%%%%%%%%%%%%%%%%%%%%%%%%%%%%%%%%%%%%%%%%%%%%%%%%%%%%%%%%%%%%%%%%

As we will often toggle between the set $\I{\vec x \mid \varphi}_M$
and the hom-sets
$\Mod(\theory)\funct{\FG{\vec x \mid \varphi},M}$
and
$\lex\funct{\cat T,\Set}\funct{\OI{\vec x \mid \varphi}, F_M}$,
it is convenient to
introduce the following terminology:

%%%%%%%%%%%%%%%%%%%%%%%%%%%%%%%%%%%%%%%%%%%%%%%%%%%%%%%%%%%%%%%%%%%%%%%%%%%
\begin{definition}
\label{def:coste:fp:hom}
%%%%%%%%%%%%%%%%%%%%%%%%%%%%%%%%%%%%%%%%%%%%%%%%%%%%%%%%%%%%%%%%%%%%%%%%%%%
Given a morphism $h \colon \FG{\vec x \mid \varphi} \to M$ in $\Mod(\theory)$
and $\vec a \in M$ of the same sorts as $\vec x$,
we say that \emph{$h$ takes $\vec x$ to $\vec a$},
or that \emph{$h$ is induced by $\vec a$},
when 
$h$ is the image of the natural transformation $F_M(-)(\vec a)$
in eq.~\eqref{eq:lambda-vec-a} under the
bijection in eq.~\eqref{eq:fp-models-nat-bij}.
%%%%%%%%%%%%%%%%%%%%%%%%%%%%%%%%%%%%%%%%%%%%%%%%%%%%%%%%%%%%%%%%%%%%%%%%%%%
\end{definition}
%%%%%%%%%%%%%%%%%%%%%%%%%%%%%%%%%%%%%%%%%%%%%%%%%%%%%%%%%%%%%%%%%%%%%%%%%%%

%%%%%%%%%%%%%%%%%%%%%%%%%%%%%%%%%%%%%%%%%%%%%%%%%%%%%%%%%%%%%%%%%%%%%%%%%%%
\subsubsection{Generators}
\label{sec:coste:fp:gen}
%%%%%%%%%%%%%%%%%%%%%%%%%%%%%%%%%%%%%%%%%%%%%%%%%%%%%%%%%%%%%%%%%%%%%%%%%%%
The terminology in Definition~\ref{def:coste:fp:hom} suggests
that the model $\FG{\vec x \mid \varphi}$ is ``generated''
by the variables $\vec x = x_1,\dots,x_n$.
We will now make this intuition precise.

The bijection in eq.~\eqref{eq:Yoneda-interpretations}
yields the following description of finitely presentable models:
for each sort $\sort$,
\[
\begin{array}{l l l}
  \FG{\vec x \mid \varphi}(\sort)
& =
& \I{y : \sort \mid \True}_{\FG{\vec x \mid \varphi}}
\\

& \cong
& \lex\funct{\cat T,\Set}
  \funct{\yoneda \OI{y:\sort \mid \True}, F_{\FG{\vec x \mid \varphi}}}
\\

& \cong
& \cat T\funct{\OI{\vec x \mid \varphi},\OI{y :\sort \mid \True}}.
\end{array}
\]

\noindent Assuming that $\vec x = x_1,\dots,x_n$,
the next lemma shows that
we can think of $\FG{\vec x \mid \varphi}$ as being generated by elements
$\const x_1,\dots,\const x_n \in \FG{\vec x \mid \varphi}$
corresponding to the projections
$\varpi_i \colon \OI{\vec x \mid \varphi} \to \OI{x_i \mid \True}$ in $\cat T$,
akin to the way in which generators of free algebras are modelled by
projections in categorical algebra.
This implies that
homomorphisms with domain $\FG{\vec x \mid \varphi}$
are completely determined by their values on 
the $\const x_i$'s, which will be crucial to our use of word-constructions in
\S\ref{sec:wc}.

%%%%%%%%%%%%%%%%%%%%%%%%%%%%%%%%%%%%%%%%%%%%%%%%%%%%%%%%%%%%%%%%%%%%%%%%%%%
\begin{lemma}
\label{lem:coste:fp:const}
%%%%%%%%%%%%%%%%%%%%%%%%%%%%%%%%%%%%%%%%%%%%%%%%%%%%%%%%%%%%%%%%%%%%%%%%%%%
Let $\OI{\vec x \mid \varphi} \in \cat T$
with $\vec x = x_1,\dots,x_n$.
There exist $\const x_1,\ldots,\const x_n \in \FG{\vec x \mid \varphi}$
such that, for any $M \in \Mod(\theory)$,
any homomorphism $h \colon \FG{\vec x \mid \varphi} \to M$
and any $\vec a \in \I{\vec x \mid \varphi}_M$,
\[
\begin{array}{l !{\quad\Longleftrightarrow\quad} l}
  \text{$h$ takes $\vec x$ to $\vec a$}
& \text{$h(\const x_i) = a_i$ for all $i = 1,\dots,n$.}
\end{array}
\]
%%%%%%%%%%%%%%%%%%%%%%%%%%%%%%%%%%%%%%%%%%%%%%%%%%%%%%%%%%%%%%%%%%%%%%%%%%%
\end{lemma}
%%%%%%%%%%%%%%%%%%%%%%%%%%%%%%%%%%%%%%%%%%%%%%%%%%%%%%%%%%%%%%%%%%%%%%%%%%%

%%%%%%%%%%%%%%%%%%%%%%%%%%%%%%%%%%%%%%%%%%%%%%%%%%%%%%%%%%%%%%%%%%%%%%%%%%%
\begin{proof}
%%%%%%%%%%%%%%%%%%%%%%%%%%%%%%%%%%%%%%%%%%%%%%%%%%%%%%%%%%%%%%%%%%%%%%%%%%%
Fix an arbitrary $i\in \{1,\ldots,n\}$.
Assume $y$ is a variable of sort $\sort_i$
and consider the $\cat T$-morphism
$\varpi_i \colon \OI{\vec x \mid \varphi} \to \OI{y : \sort_i \mid \True}$
given by the formula 
\[
\begin{array}{l l l}
  \vec x, y
& \sorting
& \varphi ~\land~ x_i \Eq y
  ~.
\end{array}
\]

\noindent
We define $\const x_i\in \FG{\vec x \mid \varphi}(\sort_i)$, of sort $\sort_{i}$, to be the image of $\varpi_i$ under
the bijection
\[
\begin{array}{*{7}{l}}
\cat T\funct{\OI{\vec x \mid \varphi},\OI{y : \sort_i \mid \True}}
& \cong
& \FG{\vec x \mid \varphi}(\sort_i).
\end{array}
\]

Now, let $h \colon \FG{\vec x \mid \varphi} \to M$
and $\vec a \in \I{\vec x \mid \varphi}_M$
such that $h(\const x_i) = a_i$ for all $i \in \{1,\dots,n\}$,
and suppose that
$h$ takes $\vec x$ to $\vec b \in \I{\vec x \mid \varphi}_M$.
The function
\[
\begin{array}{*{5}{l}}
  F_M(\varpi_i)
& :
& \I{\vec x \mid \varphi}_M
& \longto
& M(\sort_i)
  =
  \I{y : \sort_i \mid \True}_M
\end{array}
\]

\noindent
has graph
$\I{\vec x,y \mid \varphi ~\land~ x_i \Eq y}_M$
and is thus the $i$th projection.
Hence,
the natural transformation
$F_M(-)(\vec b) \colon \yoneda\OI{\vec x \mid \varphi} \to F_M$
takes
$\varpi_i \in \cat T\funct{\OI{\vec x \mid \varphi},\OI{y : \sort_i \mid \True}}$
to
\[
\begin{array}{*{7}{l}}
  F_M(\varpi_i)(\vec b)
& =
& b_i
& \in
& M(\sort_i)
  ~.
\end{array}
\]
By Lemma~\ref{lemma:coste:synt:homnat}
and the definition of $\const x_i$, we have
$h(\const x_i) = F_M(\varpi_i)(\vec b) = b_i$.
The statement then follows because $h(\const x_i) = a_i$ if, and only if, $a_i = b_i$.
%%%%%%%%%%%%%%%%%%%%%%%%%%%%%%%%%%%%%%%%%%%%%%%%%%%%%%%%%%%%%%%%%%%%%%%%%%%
\end{proof}
%%%%%%%%%%%%%%%%%%%%%%%%%%%%%%%%%%%%%%%%%%%%%%%%%%%%%%%%%%%%%%%%%%%%%%%%%%%

%%%%%%%%%%%%%%%%%%%%%%%%%%%%%%%%%%%%%%%%%%%%%%%%%%%%%%%%%%%%%%%%%%%%%%%%%%%
\begin{remark}
\label{rem:finite-carrier-fp-purely-rel-empty-th}
%%%%%%%%%%%%%%%%%%%%%%%%%%%%%%%%%%%%%%%%%%%%%%%%%%%%%%%%%%%%%%%%%%%%%%%%%%%
Let $\sig$ be a (mono-sorted) purely relational signature.
Then it is possible to show that in $\Struct(\sig)$ the structure
$\FG{\vec x \mid \True}$, with $\vec x = x_1,\dots,x_n$,
has carrier $\{\const x_1,\dots,\const x_n\}$
(where the $\const x_{i}$'s are pairwise distinct).
Now, for any formula-in-context $\OI{\vec x \mid \varphi}$
there exists an epimorphism $\FG{\vec x \mid \True}\to \FG{\vec x \mid \varphi}$
in $\Struct(\sig)$, see Remark~\ref{rem:emb:epi} below.
This induces a surjection at the level of carriers,
and so the carrier of $\FG{\vec x \mid \varphi}$ is $\{\const x_1,\dots,\const x_n\}$
(but this time the $\const x_{i}$'s need not be pairwise distinct).
This extends to many-sorted relational signatures,
but is typically no longer true for signatures that
are not purely relational,
or if we consider a non-empty (cartesian) theory.
%%%%%%%%%%%%%%%%%%%%%%%%%%%%%%%%%%%%%%%%%%%%%%%%%%%%%%%%%%%%%%%%%%%%%%%%%%%
\end{remark}
%%%%%%%%%%%%%%%%%%%%%%%%%%%%%%%%%%%%%%%%%%%%%%%%%%%%%%%%%%%%%%%%%%%%%%%%%%%

%%%%%%%%%%%%%%%%%%%%%%%%%%%%%%%%%%%%%%%%%%%%%%%%%%%%%%%%%%%%%%%%%%%%%%%%%%%%
\begin{fullproof}
%%%%%%%%%%%%%%%%%%%%%%%%%%%%%%%%%%%%%%%%%%%%%%%%%%%%%%%%%%%%%%%%%%%%%%%%%%%%
We first show that the carrier of
$\FG{\vec x \mid \True}$ is $\{\const x_1,\dots,\const x_n\}$.

It follows from Lemma~\ref{lem:coste:fp:const} that the carrier of
$\FG{\vec x \mid \True}$ contains $\{\const x_1,\dots,\const x_n\}$.
For the converse inclusion,
let $\cat T$ be the syntactic category of $\theory(\sig)$,
the empty theory in $\sig$
(see eq.~\eqref{eq:coste:struct}, \S\ref{sec:coste:synt}).
Recall that $\FG{\vec x \mid \True}$ has carrier
\[
  \cat T\funct{
  \OI{\vec x \mid \True}
  \,,\,
  \OI{y \mid \True}
  }
\]

\noindent
Let $\MI{\vec x,y \mid \theta}$ be a $\cat T$-morphism
from $\OI{\vec x \mid \True}$ to $\OI{y \mid \True}$.
Hence $\theory(\sig)$ proves
\[
\begin{array}{l !{\qquad\text{and}\qquad} l}
  \theta \land \theta[z/y]
  \thesis_{\vec x,y,z}
  y \Eq z

& \thesis_{\vec x}
  (\exists y)\theta
\end{array}
\]

We show that there is some $j\in \{1,\dots,n\}$
such that $(\vec x,y \sorting \theta)$ is $\theory(\sig)$-equivalent to
\begin{equation}
\label{eq:emb:formemb:fin}
\begin{array}{l l l}
  \vec x,y
& \sorting
& y \Eq x_j
\end{array}
\end{equation}

\noindent
In view of (the proof of) Lemma~\ref{lem:coste:fp:const},
this will give the result.
We rely on completeness (Remark~\ref{rem:prelim:coste:compl}).

We first determine the $j \in \{1,\dots,n\}$.
Let $N$ be the $\sig$-structure $\{\const x_1,\dots,\const x_n\}$,
with no relation assumed on the $\const x_i$'s.%
\footnote{This structure ought to be $\FG{\vec x \mid \True}$,
but we do not need to be that precise.}
Since $\theory(\sig)$ proves $(\exists y)\theta$,
it follows 
that there is some $j \in \{1,\dots,n\}$ such that 
$N \models \theta(\const x_1,\dots,\const x_n,\const x_j)$.

We now show that every model of $\theta$ is a model of \eqref{eq:emb:formemb:fin}.
Let $M$ be a $\sig$-structure, and let
$(\vec a,b) \in \I{\vec x,y \mid \theta}_M$.
It follows from the definition of $N$ above that the function
$N \to M$ which takes $\const x_i$ to $a_i$ is a homomorphism.
Since $\theta$ is a positive existential formula,
we get that $M \models \theta(\vec a,a_j)$
(see e.g.~\cite[Theorem 2.4.3(a)]{hodges93book}).
Hence $b = a_j$ since $(\exists y)\theta$ is cartesian.

Conversely, we show that every model of \eqref{eq:emb:formemb:fin}
is a model of $\theta$.
Let $M$ be a $\sig$-structure, and let $\vec a \in M$.
Since $\theory(\sig)$ proves $(\exists y)\theta$,
we get $(\vec a,b) \in \I{\vec x,y \mid \theta}_M$
for some $b \in M$.
Then, reasoning similarly as above yields $b = a_j$.

Hence, the carrier of $\FG{\vec x \mid \True}$ is $\{\const x_1,\dots,\const x_n\}$.
It remains to show that the $\const x_i$'s are pairwise distinct.
In view of (the proof of) Lemma~\ref{lem:coste:fp:const},
this amounts to showing that for a given $\theta$ as above,
there is a \emph{unique} $j$ as in~\eqref{eq:emb:formemb:fin}.
But $j$ is unique since $(\exists y)\theta$ is cartesian.
%%%%%%%%%%%%%%%%%%%%%%%%%%%%%%%%%%%%%%%%%%%%%%%%%%%%%%%%%%%%%%%%%%%%%%%%%%%%
\end{fullproof}
%%%%%%%%%%%%%%%%%%%%%%%%%%%%%%%%%%%%%%%%%%%%%%%%%%%%%%%%%%%%%%%%%%%%%%%%%%%%

%%%%%%%%%%%%%%%%%%%%%%%%%%%%%%%%%%%%%%%%%%%%%%%%%%%%%%%%%%%%%%%%%%%%%%%%%%%
\subsubsection{Homomorphisms}
%%%%%%%%%%%%%%%%%%%%%%%%%%%%%%%%%%%%%%%%%%%%%%%%%%%%%%%%%%%%%%%%%%%%%%%%%%%
Fix a model $M \in \Mod(\theory)$.
We discuss
a logical description of arrows in the comma category $\fp{\Mod(\theory)}/M$
that we will use repeatedly.
Let
\(
  h
  \colon
  \FG{\vec y \mid \psi}
  \to
  \FG{\vec x \mid \varphi}
\)
be a homomorphism in $\Mod(\theory)$
with corresponding $\cat T$-morphism
$\MI{\vec x, \vec y \mid \theta(\vec x, \vec y)}$, 
and let $k \colon \FG{\vec x \mid \varphi} \to M$ be an arrow in $\Mod(\theory)$:
\[
\begin{tikzcd}
  \OI{\vec y \mid \psi}
  \arrow[leftarrow]{rr}{\MI{\vec x, \vec y \mid \theta(\vec x, \vec y)}}
&
& \OI{\vec x \mid \varphi}
\\
  \FG{\vec y \mid \psi}
  \arrow{rr}{h}
  \arrow[dotted]{dr}[below, xshift=-5pt]{k \comp h}
&
& \FG{\vec x \mid \varphi}
  \arrow{dl}{k}
\\
& M
\end{tikzcd}
\]

Lemma~\ref{lem:coste:fp:triangle} below says that
if $k$ takes $\vec x$ to $\vec a \in \I{\vec x \mid \varphi}_M$,
then the composition $k \comp h$ takes $\vec y$ to the
unique $\vec b \in \I{\vec y \mid \psi}_M$ such that
$M \models \theta(\vec a,\vec b)$.

%%%%%%%%%%%%%%%%%%%%%%%%%%%%%%%%%%%%%%%%%%%%%%%%%%%%%%%%%%%%%%%%%%%%%%%%%%%
\begin{lemma}
\label{lem:coste:fp:triangle}
%%%%%%%%%%%%%%%%%%%%%%%%%%%%%%%%%%%%%%%%%%%%%%%%%%%%%%%%%%%%%%%%%%%%%%%%%%%
Given $h,\theta$ as above,
let $\vec a \in \I{\vec x \mid \varphi}_M$ and $\vec b \in \I{\vec y \mid \psi}_M$.
Then $M \models \theta(\vec a, \vec b)$
if, and only if,
the triangle below commutes,
where $k$ 
takes $\vec x$ to $\vec a$
and $\ell$
takes $\vec y$ to $\vec b$.
\[
\begin{tikzcd}
  \FG{\vec y \mid \psi}
  \arrow{rr}{h}
  \arrow{dr}[below, xshift=-5pt]{\ell}
&
& \FG{\vec x \mid \varphi}
  \arrow{dl}{k}
\\
& M
\end{tikzcd}
\]
%%%%%%%%%%%%%%%%%%%%%%%%%%%%%%%%%%%%%%%%%%%%%%%%%%%%%%%%%%%%%%%%%%%%%%%%%%%
\end{lemma}
%%%%%%%%%%%%%%%%%%%%%%%%%%%%%%%%%%%%%%%%%%%%%%%%%%%%%%%%%%%%%%%%%%%%%%%%%%%

%%%%%%%%%%%%%%%%%%%%%%%%%%%%%%%%%%%%%%%%%%%%%%%%%%%%%%%%%%%%%%%%%%%%%%%%%%%
\begin{proof}
%%%%%%%%%%%%%%%%%%%%%%%%%%%%%%%%%%%%%%%%%%%%%%%%%%%%%%%%%%%%%%%%%%%%%%%%%%%
We apply Theorem~\ref{thm:coste:synt}.
With the notations of eq.\ \eqref{eq:lambda-vec-a},
we have $\ell = k \comp h$ if, and only if,
the following diagram commutes.
\begin{equation}
\label{eq:coste:fp:triangle}
\begin{tikzcd}
  \yoneda\OI{\vec y \mid \psi}
  \arrow{rr}{\yoneda \MI{\vec x,\vec y \mid \theta(\vec x,\vec y)}}
  \arrow{dr}[below, xshift=-15pt]{F_M(-)(\vec b)}
&
& \yoneda\OI{\vec x \mid \varphi}
  \arrow{dl}{F_M(-)(\vec a)}
\\
& F_M
\end{tikzcd}
\end{equation}

\noindent
By functoriality of $F_M$, the rightmost path in \eqref{eq:coste:fp:triangle} is
\[
\begin{array}{l l l}
  F_M((-) \comp \MI{\vec x,\vec y \mid \theta(\vec x, \vec y)})(\vec a)
& =
& F_M(-)(F_M(\MI{\vec x,\vec y \mid \theta(\vec x,\vec y)})(\vec a)).
\end{array}
\]

\noindent
Hence, the triangle in~\eqref{eq:coste:fp:triangle} commutes if and only if
\[
\begin{array}{l l l}
  \vec b
& =
& F_M(\MI{\vec x, \vec y \mid \theta(\vec x,\vec y)})(\vec a),
\end{array}
\]

\noindent
which is in turn equivalent to
$(\vec a, \vec b) \in \I{\vec x,\vec y \mid \theta(\vec x,\vec y)}_M$
(recall from eq.~\eqref{eq:coste:graph} above that
$F_M(\MI{\vec x, \vec y \mid \theta(\vec x,\vec y)})$
is the function 
$\I{\vec x \mid \varphi}_M \to \I{\vec y \mid \psi}_M$
of graph
$\I{\vec x, \vec y \mid \theta(\vec x,\vec y)}_M$).
%%%%%%%%%%%%%%%%%%%%%%%%%%%%%%%%%%%%%%%%%%%%%%%%%%%%%%%%%%%%%%%%%%%%%%%%%%%
\end{proof}
%%%%%%%%%%%%%%%%%%%%%%%%%%%%%%%%%%%%%%%%%%%%%%%%%%%%%%%%%%%%%%%%%%%%%%%%%%%

The following simple and expected fact is an
immediate consequence of Lemma~\ref{lem:coste:fp:triangle}.

%%%%%%%%%%%%%%%%%%%%%%%%%%%%%%%%%%%%%%%%%%%%%%%%%%%%%%%%%%%%%%%%%%%%%%%%%%%
\begin{remark}
\label{rem:coste:fp:coprod}
%%%%%%%%%%%%%%%%%%%%%%%%%%%%%%%%%%%%%%%%%%%%%%%%%%%%%%%%%%%%%%%%%%%%%%%%%%%
Let $\vec x = x_1,\dots,x_n$, with $x_j$ of sort $\sort_j$.
Then $\FG{\vec x \mid \True}$ is the coproduct of 
$\FG{x_1 : \sort_1 \mid \True}, \dots, \FG{x_n : \sort_n \mid \True}$
in $\Mod(\theory)$.
%%%%%%%%%%%%%%%%%%%%%%%%%%%%%%%%%%%%%%%%%%%%%%%%%%%%%%%%%%%%%%%%%%%%%%%%%%%
\end{remark}
%%%%%%%%%%%%%%%%%%%%%%%%%%%%%%%%%%%%%%%%%%%%%%%%%%%%%%%%%%%%%%%%%%%%%%%%%%%

%%%%%%%%%%%%%%%%%%%%%%%%%%%%%%%%%%%%%%%%%%%%%%%%%%%%%%%%%%%%%%%%%%%%%%%%%%%
\begin{fullproof}
%%%%%%%%%%%%%%%%%%%%%%%%%%%%%%%%%%%%%%%%%%%%%%%%%%%%%%%%%%%%%%%%%%%%%%%%%%%
We show that $\FG{\vec x \mid \True}$ is the coproduct of the
$\FG{x_j : \sort_j \mid \True}$,
with coprojections $\iota_j$ induced from $\cat T$-morphisms
\(
  \varpi_j
  \colon
  \OI{\vec x \mid \True}
  \to
  \OI{x_j : \sort_j \mid \True}
\)
as in the proof of Lemma~\ref{lem:coste:fp:const}.

Let $M \in \Mod(\theory)$, together with morphisms
$h_j \colon \FG{x_j \mid \sort_j} \to M$ for each $j=1,\dots,n$.
Using Lemma~\ref{lem:coste:fp:const},
let $c_j \deq h_j(\const x_j) \in M(\sort_j)$.
Let $k \colon \FG{\vec x \mid \True} \to M$ take
$x_j$ to $c_j$.
It then follows from Lemma~\ref{lem:coste:fp:triangle}
that $k \comp \iota_j = h_j$ for each $j = 1,\dots,n$.
Moreover, given 
$k' \colon \FG{\vec x \mid \True} \to M$
such that $k' \comp \iota_j = h_j$ for each $j = 1,\dots,n$,
we have $k'(\const x_j) = (k' \comp \iota_j)(\const x_j) = c_j = k(\const x_j)$.
Hence $k' = k$
by Lemma~\ref{lem:coste:fp:const}.
%%%%%%%%%%%%%%%%%%%%%%%%%%%%%%%%%%%%%%%%%%%%%%%%%%%%%%%%%%%%%%%%%%%%%%%%%%%
\end{fullproof}
%%%%%%%%%%%%%%%%%%%%%%%%%%%%%%%%%%%%%%%%%%%%%%%%%%%%%%%%%%%%%%%%%%%%%%%%%%%

Lemma~\ref{lem:coste:fp:triangle}
will often be used in the following form: 
%%%%%%%%%%%%%%%%%%%%%%%%%%%%%%%%%%%%%%%%%%%%%%%%%%%%%%%%%%%%%%%%%%%%%%%%%%%
\begin{corollary}
\label{cor:coste:fp:triangle}
%%%%%%%%%%%%%%%%%%%%%%%%%%%%%%%%%%%%%%%%%%%%%%%%%%%%%%%%%%%%%%%%%%%%%%%%%%%
Let $\OI{\vec x \mid \varphi}, \OI{\vec y \mid \psi} \in \cat T$
and let $\vec a,\vec b \in M$ be of the same sorts as~$\vec x$ and~$\vec y$, respectively.
The following statements are equivalent:
\begin{enumerate}[(i)]
\item
There exists a homomorphism $h$ making the triangle below commute,
\[
\begin{tikzcd}[row sep=tiny]
  \FG{\vec y \mid \psi}
  \arrow{rr}{h}
  \arrow{dr}[below, xshift=-5pt]{\ell}
&
& \FG{\vec x \mid \varphi}
  \arrow{dl}{k}
\\
& M
\end{tikzcd}
\]
where $k$
takes $\vec x$ to $\vec a$
and $\ell$
takes $\vec y$ to $\vec b$.

\item
\label{item:coste:fp:triangle:form}
\(
M
\models
\varphi(\vec a)
\land
\psi(\vec b)
\land
\bigvee
\{
\theta(\vec a, \vec b)
\mid
\theta \in \cat T\funct{\OI{\vec x \mid \varphi},\OI{\vec y \mid \psi}}
\}
\).
(Note that the latter formula may be infinite.)
\end{enumerate}
%%%%%%%%%%%%%%%%%%%%%%%%%%%%%%%%%%%%%%%%%%%%%%%%%%%%%%%%%%%%%%%%%%%%%%%%%%%
\end{corollary}
%%%%%%%%%%%%%%%%%%%%%%%%%%%%%%%%%%%%%%%%%%%%%%%%%%%%%%%%%%%%%%%%%%%%%%%%%%%

%%%%%%%%%%%%%%%%%%%%%%%%%%%%%%%%%%%%%%%%%%%%%%%%%%%%%%%%%%%%%%%%%%%%%%%%%%%
\subsection{Interpretations and a transfer result}
\label{sec:coste:interp}
%%%%%%%%%%%%%%%%%%%%%%%%%%%%%%%%%%%%%%%%%%%%%%%%%%%%%%%%%%%%%%%%%%%%%%%%%%%

Recall from Corollary~\ref{cor:lfp-iff-models-of-T}
that Gabriel-Ulmer duality can be seen as a syntax-semantics duality.
Accordingly, morphisms of lfp categories induce interpretations of formulae,
as we now explain.

Consider cartesian theories $\theory$ and $\theorybis$
with syntactic categories $\cat T$ and $\cat U$, respectively. 
By Gabriel-Ulmer duality,
for each morphism of lfp categories
$R \colon \Mod(\theorybis) \to \Mod(\theory)$
there is a lex-morphism
$\inclth \colon \cat T \to \cat U$
such that
$R$ corresponds, via Theorem~\ref{thm:coste:synt},
to the functor
$\ladj\inclth \colon \lex\funct{\cat U,\Set} \to \lex\funct{\cat T,\Set}$.
Given a model $M \in \Mod(\theorybis)$ and
$\OI{\vec x \mid \varphi} \in \cat T$,
the formula $\vec x \sorting \varphi$ is interpreted in $R M$
as
\[
\begin{array}{*{7}{l}}
  \I{\vec x \mid \varphi}_{R M}
& =
& F_{R M}\OI{\vec x \mid \varphi}
& \cong
& F_M \inclth\OI{\vec x \mid \varphi} 
& =
& \I{\inclth\OI{\vec x \mid \varphi}}_M.
\end{array}
\]

\noindent
In particular, for each sort $\sort$ of the signature $\Sig$ of $\theory$,
\begin{equation}
\label{eq:coste:iso-RM-sortwise}
\begin{array}{l l l}
  (R M)(\sort)
& \cong
& \I{\inclth\OI{y : \sort \mid \True}}_M.
\end{array}
\end{equation}

It follows that $R$ induces an interpretation $(-)^{\inclth}$ 
of (possibly infinitary) first-order formulae in 
$\Sig$
as formulae in the signature $\Sigbis$ of $\theorybis$.
Let $\psi(\vec y)$ be a formula in $\Sig$, with $\vec y = y_1,\dots,y_n$.
We define a formula
$\psi^{\inclth}(\vec x_1,\dots,\vec x_n)$ in~$\Sigbis$,
by induction on $\psi$, as follows.
If $\psi$ is atomic then
$\psi^{\inclth}$ is
$\varphi$
where
$\inclth\OI{\vec y \mid \psi} = \OI{\vec x \mid \varphi}$, and
$(-)^{\inclth}$ commutes with propositional connectives. 
If
$\psi = (\exists y : \sort)\psi'(\vec y,y)$,
write
$\OI{\vec x \mid \varphi}$
for
${\inclth\OI{y : \sort \mid \True}}$,
and define
$\psi^{\inclth}(\vec x_1,\dots,\vec x_n)$
as
\[
  (\exists \vec x)
  \left(
  \varphi(\vec x)
  ~\land~
  (\psi')^{\inclth}
  (\vec x_1,\dots,\vec x_n,\vec x)
  \right).
\]
The case of 
$\psi = (\forall y)\psi'(\vec y,y)$
is similar.

For the next result,
given a formula $\psi(y_1,\dots,y_n)$
in $\Sig$ with $y_i$ of sort $\sort_i$, for each $i = 1,\dots,n$ we let
$\OI{\vec x_i \mid \varphi_i} \deq \inclth\OI{y_i : \sort_i \mid \True}$.
Further, for any model
$M$ of $\theorybis$
and sort~$\sort$ of~$\Sig$,
let
\(
  \chi^{\sort}_M
  :
  \I{\inclth\OI{y : \sort \mid \True}}_M
  \to
  (R M)(\sort)
\)
be the inverse of the bijection in eq.~\eqref{eq:coste:iso-RM-sortwise}.

%%%%%%%%%%%%%%%%%%%%%%%%%%%%%%%%%%%%%%%%%%%%%%%%%%%%%%%%%%%%%%%%%%%%%%%%%%%
\begin{theorem}
\label{thm:coste:interp}
%%%%%%%%%%%%%%%%%%%%%%%%%%%%%%%%%%%%%%%%%%%%%%%%%%%%%%%%%%%%%%%%%%%%%%%%%%%
In the above situation,
for all tuples
$\vec a_1,\dots,\vec a_n \in M$
with
$\vec a_i \in \I{\vec x_i \mid \varphi_i}_M$ for each $i = 1,\dots,n$,
\[
\begin{array}{l l l}
  R M
  \models
  \psi\left( \chi^{\sort_1}_M(\vec a_1) ,\dots, \chi^{\sort_n}_{M}(\vec a_n) \right)
& \longiff
& M \models \psi^{\inclth}(\vec a_1,\dots,\vec a_n).
\end{array}
\]
\end{theorem}

%%%%%%%%%%%%%%%%%%%%%%%%%%%%%%%%%%%%%%%%%%%%%%%%%%%%%%%%%%%%%%%%%%%%%%%%%%%
\begin{remark}
\label{rem:coste:interp:card}
%%%%%%%%%%%%%%%%%%%%%%%%%%%%%%%%%%%%%%%%%%%%%%%%%%%%%%%%%%%%%%%%%%%%%%%%%%%
Note that $(-)^{\inclth}$ interprets
$\Lang_\kappa(\Sig)$ in $\Lang_\kappa(\Sigbis)$,
for each regular cardinal~$\kappa$.
In particular, it interprets $\Lang_\omega(\Sig)$ in $\Lang_\omega(\Sigbis)$.
This is a crucial difference between interpretations and word-constructions
(\S\ref{sec:wc:wc}).
\end{remark}

%%%%%%%%%%%%%%%%%%%%%%%%%%%%%%%%%%%%%%%%%%%%%%%%%%%%%%%%%%%%%%%%%%%%%%%%%%%
\section{Hintikka formulae for back-and-forth games}
\label{sec:hintikka}
%%%%%%%%%%%%%%%%%%%%%%%%%%%%%%%%%%%%%%%%%%%%%%%%%%%%%%%%%%%%%%%%%%%%%%%%%%%

We now return to the main goal of this paper: to show
that, under mild additional assumptions,
given a finitely accessible wooded adjunction $R\colon \E \to \C$,
the relation of $R$-back-and-forth-equivalence
induced on the extensional category $\E$
is definable in the infinitary first-order logic $\Lang_{\infty}$
(see Theorem~\ref{thm:path:main}). 

This is achieved in two steps.
First, we give a logical description of the games $\G(a,b)$ in the wooded category~$\C$.
Then, using the tools of \S\ref{sec:coste:interp},
we transfer this description along (the Gabriel-Ulmer dual of) $R$
to obtain a characterisation in the signature of~$\E$.
In this section, we assume the existence of formulae defining path embeddings in~$\C$;
we will see in \S\S\ref{sec:emb}--\ref{sec:wc} how to enforce this condition
under the assumptions of Theorem~\ref{thm:path:main}.

To motivate the logical perspective on the games $\G(a,b)$,
recall that in the classical setting of finite Ehrenfeucht-Fraïssé games,
the number of rounds corresponds to the quantifier rank of formulae in $\Lang_{\omega}$.
This correspondence can be established by means of \emph{$m$-Hintikka formulae},
see e.g.\ \cite[\S2.2]{ef99book}. 
Given a relational structure $M$ and a tuple $\vec a$ of elements from $M$,
the $0$-Hintikka formula $\varphi^0_{\vec a}$ describes the isomorphism type of $\vec a$.
For $m>0$, the $m$-Hintikka formula $\varphi^m_{\vec a}$
encodes the isomorphism types to which $\vec a$ can be extended in $m$ steps,
adding one element at each step.

In the infinite Ehrenfeucht-Fraïssé game,
(an infinitary variant of) Hintikka formulae yield a correspondence
between the quantifier rank of formulae in $\Lang_{\infty}$ and the
\emph{rank of positions} in the game; see e.g.\ \cite[\S3.5]{hodges93book}.
We recall the latter notion in~\S\ref{sec:hintikka:rank},
tailored to the case of the games $\G(a,b)$.
We then proceed in~\S\ref{sec:hintikka:form}
to construct ``Hintikka formulae'' 
for these games.
This leads to Theorem~\ref{thm:hintikka:wooded:bfe},
the first main result of this paper,
which states that given an lfp wooded category $\C$,
if the path embeddings in $\C$
are definable in some signature $\Sigbis$
associated with $\C$,
then
\[
\begin{array}{l l l}
  \text{$a,b \in \C$ equivalent in $\Lang_{\infty}(\Sigbis)$}
& \longimp
& \text{Duplicator wins $\G(a,b)$}.
\end{array}
\]

Finally, in \S\ref{sec:hintikka:finaccadj}
we transfer this result along the interpretation induced
by a finitely accessible wooded adjunction $R \colon \E \to \C$.

%%%%%%%%%%%%%%%%%%%%%%%%%%%%%%%%%%%%%%%%%%%%%%%%%%%%%%%%%%%%%%%%%%%%%%%%%%%%
\subsection{Ranks of positions}
\label{sec:hintikka:rank}
%%%%%%%%%%%%%%%%%%%%%%%%%%%%%%%%%%%%%%%%%%%%%%%%%%%%%%%%%%%%%%%%%%%%%%%%%%%
The following development mirrors~\cite[\S 3.4]{hodges93book}
(see also~\cite[20.2]{kechris95book}).
Let $\C$ be a wooded category.
Given a game $\G(a,b)$ as in 
Definition~\ref{def:games}
we define for each ordinal $\ord$
a set $\Rk(\ord)$ of \emph{positions of rank $\geq \ord$}
in $\G(a,b)$.
The definition is by induction on $\ord$:
\begin{itemize}
\item
$(m,n) \in \Rk(0)$ if $(m,n) \in \W(a,b)$;

\item
$(m,n) \in \Rk(\ord+1)$
if
$(m,n) \in \W(a,b)$
and
for each Spoiler move from position $(m,n)$,
there is a Duplicator move 
arriving at
a position in $\Rk(\ord)$;

\item
$(m,n) \in \Rk(\ord)$, for $\ord$ a limit ordinal,
if ${(m,n) \in \Rk(\ordbis)}$ for all $\ordbis < \ord$.
\end{itemize}

\noindent
Note that 
$\Rk(\ord) \subseteq \W(a,b)$
for all ordinals $\ord$
(this must be enforced 
in the clause for successor ordinals, as there may be
no possible Spoiler moves from position $(m,n)$).
Moreover,

%%%%%%%%%%%%%%%%%%%%%%%%%%%%%%%%%%%%%%%%%%%%%%%%%%%%%%%%%%%%%%%%%%%%%%%%%%%
\begin{lemma}
\label{lem:hintikka:rank:mon}
%%%%%%%%%%%%%%%%%%%%%%%%%%%%%%%%%%%%%%%%%%%%%%%%%%%%%%%%%%%%%%%%%%%%%%%%%%%
$\Rk(\ord) \subseteq \Rk(\ordbis)$ for all $\ordbis < \ord$.
%%%%%%%%%%%%%%%%%%%%%%%%%%%%%%%%%%%%%%%%%%%%%%%%%%%%%%%%%%%%%%%%%%%%%%%%%%%
\end{lemma}
%%%%%%%%%%%%%%%%%%%%%%%%%%%%%%%%%%%%%%%%%%%%%%%%%%%%%%%%%%%%%%%%%%%%%%%%%%%

%%%%%%%%%%%%%%%%%%%%%%%%%%%%%%%%%%%%%%%%%%%%%%%%%%%%%%%%%%%%%%%%%%%%%%%%%%%
\begin{fullproof}
%%%%%%%%%%%%%%%%%%%%%%%%%%%%%%%%%%%%%%%%%%%%%%%%%%%%%%%%%%%%%%%%%%%%%%%%%%%
By induction on $\ord$,
we show that for all position $(m,n)$,
if $(m,n) \in \Rk(\ord)$ then $(m,n) \in \Rk(\ordbis)$
for all $\ordbis < \ord$.
The cases of $0$ and of limit ordinals are trivial.
So we only consider the case of a successor ordinal $\ord+1$.
Let $(m,n) \in \Rk(\ord+1)$.
We reason by cases on $\ord$.
\begin{description}
\item[Case of $\ord = 0$]

Trivial, since $\Rk(0) = \W(a,b)$ by definition.

\item[Case of $\ord = \ordbis + 1$]

Since $(m,n) \in \Rk(\ord+1)$,
after each Spoiler move, there is a Duplicator move reaching
a position $(m',n') \in \Rk(\ord)$.
By induction hypothesis, all such $(m',n')$ are in $\Rk(\ordbis)$.
Hence $(m,n) \in \Rk(\ordbis+1)$ by definition,
and the result follows from the induction hypothesis.

\item[Case of $\ord$ a limit ordinal]

Since $(m,n) \in \Rk(\ord+1)$,
after each Spoiler move, there is a Duplicator move reaching
a position $(m',n') \in \Rk(\ord)$.
By definition,
we have $(m',n') \in \Rk(\ordbis)$
for all $\ordbis < \ord$.
It follows that for all $\ordbis < \ord$,
after all Spoiler move from position $(m,n)$,
there is a Duplicator move reaching
a position $(m',n') \in \Rk(\ordbis)$.
Hence $(m,n) \in \Rk(\ordbis+1)$ for all $\ordbis < \ord$.
Since $\ord$ is a limit ordinal, we have
$\ordbis + 1 < \ord$ whenever $\ordbis < \ord$,
so that the induction hypothesis implies
$(m,n) \in \Rk(\ordbis)$ for all $\ordbis < \ord$.
Hence $(m,n) \in \Rk(\ord)$ by definition.
\qedhere
\end{description}
%%%%%%%%%%%%%%%%%%%%%%%%%%%%%%%%%%%%%%%%%%%%%%%%%%%%%%%%%%%%%%%%%%%%%%%%%%%
\end{fullproof}
%%%%%%%%%%%%%%%%%%%%%%%%%%%%%%%%%%%%%%%%%%%%%%%%%%%%%%%%%%%%%%%%%%%%%%%%%%%

We thus have a decreasing sequence of
sets $(\Rk(\ord))_{\ord} \sle \Path a \times \Path b$,
and this sequence is eventually stationary.
This relies on the following observation:

%%%%%%%%%%%%%%%%%%%%%%%%%%%%%%%%%%%%%%%%%%%%%%%%%%%%%%%%%%%%%%%%%%%%%%%%%%%
\begin{lemma}
%%%%%%%%%%%%%%%%%%%%%%%%%%%%%%%%%%%%%%%%%%%%%%%%%%%%%%%%%%%%%%%%%%%%%%%%%%%
Let $\ord_0$ be an ordinal such that $\Rk(\ord_0) \sle \Rk(\ord_0+1)$.
Then $\Rk(\ord) = \Rk(\ord_0)$ for all $\ord \geq \ord_0$.
\end{lemma}

%%%%%%%%%%%%%%%%%%%%%%%%%%%%%%%%%%%%%%%%%%%%%%%%%%%%%%%%%%%%%%%%%%%%%%%%%%%
\begin{fullproof}
%%%%%%%%%%%%%%%%%%%%%%%%%%%%%%%%%%%%%%%%%%%%%%%%%%%%%%%%%%%%%%%%%%%%%%%%%%%
The proof is by induction on $\ord \geq \ord_0$.
The result is trivial if $\ord = \ord_0$.
We consider the two possible cases for $\ord > \ord_0$.
Note that $\Rk(\ord_0) = \Rk(\ord_0 +1)$
by Lemma~\ref{lem:hintikka:rank:mon}.
\begin{description}
\item[Case of $\ord = \ordbis + 1$]

Let $(m,n) \in \Rk(\ord_0)$.
We have to show that $(m,n) \in \Rk(\ordbis+1)$,
i.e.\ that for every Spoiler move from position $(m,n)$,
Duplicator has a move reaching a position $(m',n') \in \Rk(\ordbis)$.
But since $(m,n) \in \Rk(\ord_0 + 1)$,
for every Spoiler move from position $(m,n)$,
Duplicator has a move reaching a position $(m',n') \in \Rk(\ord_0)$.
Hence, by induction hypothesis,
for every Spoiler move from position $(m,n)$,
Duplicator has a move reaching a position $(m',n') \in \Rk(\ordbis)$,
and we are done.

\item[Case of $\ord$ a limit ordinal]

Given $(m,n) \in \Rk(\ord_0)$,
by induction hypothesis we have
$(m,n) \in \Rk(\ordbis)$ for all $\ordbis < \ord$.
Hence $(m,n) \in \Rk(\ord)$ by definition and we are done.
\qedhere
\end{description}
%%%%%%%%%%%%%%%%%%%%%%%%%%%%%%%%%%%%%%%%%%%%%%%%%%%%%%%%%%%%%%%%%%%%%%%%%%%
\end{fullproof}
%%%%%%%%%%%%%%%%%%%%%%%%%%%%%%%%%%%%%%%%%%%%%%%%%%%%%%%%%%%%%%%%%%%%%%%%%%%

The \emph{rank} of $\G(a,b)$,
denoted by $\ord_{a,b}$,
is the least ordinal $\ord$
such that $\Rk(\ord) \sle \Rk(\ord+1)$.
The next result, a variant of~\cite[Lemma 3.4.1]{hodges93book},
implies that positions of rank $\geq \ord_{a,b}$ are winning.

%%%%%%%%%%%%%%%%%%%%%%%%%%%%%%%%%%%%%%%%%%%%%%%%%%%%%%%%%%%%%%%%%%%%%%%%%%%
\begin{lemma}
\label{lem:hintikka:rank:win}
%%%%%%%%%%%%%%%%%%%%%%%%%%%%%%%%%%%%%%%%%%%%%%%%%%%%%%%%%%%%%%%%%%%%%%%%%%%
Let $\ord$ be an ordinal such that $\Rk(\ord) \sle \Rk(\ord+1)$.
For every $(m,n) \in \Rk(\ord)$,
Duplicator has a winning strategy from position $(m,n)$.
\end{lemma}

%%%%%%%%%%%%%%%%%%%%%%%%%%%%%%%%%%%%%%%%%%%%%%%%%%%%%%%%%%%%%%%%%%%%%%%%%%%
\begin{fullproof}
%%%%%%%%%%%%%%%%%%%%%%%%%%%%%%%%%%%%%%%%%%%%%%%%%%%%%%%%%%%%%%%%%%%%%%%%%%%
Fix some $(m,n) \in \Rk(\ord)$.
Since $(m,n) \in \Rk(\ord+1)$,
for every Spoiler move from position $(m,n)$,
Duplicator can reach a position $(m',n') \in \Rk(\ord)$.
Iterating this process induces a winning strategy for Duplicator,
since $\Rk(\ord) \sle \W(a,b)$.
%%%%%%%%%%%%%%%%%%%%%%%%%%%%%%%%%%%%%%%%%%%%%%%%%%%%%%%%%%%%%%%%%%%%%%%%%%%
\end{fullproof}
%%%%%%%%%%%%%%%%%%%%%%%%%%%%%%%%%%%%%%%%%%%%%%%%%%%%%%%%%%%%%%%%%%%%%%%%%%%

%%%%%%%%%%%%%%%%%%%%%%%%%%%%%%%%%%%%%%%%%%%%%%%%%%%%%%%%%%%%%%%%%%%%%%%%%%%
\begin{remark}
%%%%%%%%%%%%%%%%%%%%%%%%%%%%%%%%%%%%%%%%%%%%%%%%%%%%%%%%%%%%%%%%%%%%%%%%%%%
In full generality, Lemma~\ref{lem:hintikka:rank:win}
requires the \emph{Axiom of Dependent Choices}%
\footnote{This axiom is of intermediate 
strength between the axiom of \emph{Countable Choice}
and the full \emph{Axiom of Choice}
cf.\ e.g.~\cite[\S 5]{jech06set} and~\cite[\S 20B]{kechris95book}.}
But no choice axiom is required (with respect to $\ZF$)
if $\G(a,b)$ has at most countably many positions.
\end{remark}

As an aside, we note that
for each game $\G(a,b)$
we can exhibit an ordinal as in Lemma~\ref{lem:hintikka:rank:win}.
This can be seen by adapting (the proof of) \cite[Theorem~3.4.2]{hodges93book}:
let $S$ be the set of positions $(m,n)$ of $\G(a,b)$
such that $(m,n) \notin \Rk(\ordbis)$ for some ordinal $\ordbis$.
Since $S$ is a set,
\[
\begin{array}{l l l}
  \ord_0
& \deq
& \bigvee \left\{ 
  \bigwedge \left\{ \text{$\ordbis$ ordinal} \mid (m,n) \notin \Rk(\ordbis) \right\}
  \mid (m,n) \in S \right\}
\end{array}
\]

\noindent
is a well-defined ordinal.

%%%%%%%%%%%%%%%%%%%%%%%%%%%%%%%%%%%%%%%%%%%%%%%%%%%%%%%%%%%%%%%%%%%%%%%%%%%
\begin{lemma}
\label{lem:hintikka:rank:ord}
%%%%%%%%%%%%%%%%%%%%%%%%%%%%%%%%%%%%%%%%%%%%%%%%%%%%%%%%%%%%%%%%%%%%%%%%%%%
With $\ord_0$ as above, we have $\Rk(\ord_0) \sle \Rk(\ord_0+1)$.
%%%%%%%%%%%%%%%%%%%%%%%%%%%%%%%%%%%%%%%%%%%%%%%%%%%%%%%%%%%%%%%%%%%%%%%%%%%
\end{lemma}
%%%%%%%%%%%%%%%%%%%%%%%%%%%%%%%%%%%%%%%%%%%%%%%%%%%%%%%%%%%%%%%%%%%%%%%%%%%

%%%%%%%%%%%%%%%%%%%%%%%%%%%%%%%%%%%%%%%%%%%%%%%%%%%%%%%%%%%%%%%%%%%%%%%%%%%
\begin{fullproof}
%%%%%%%%%%%%%%%%%%%%%%%%%%%%%%%%%%%%%%%%%%%%%%%%%%%%%%%%%%%%%%%%%%%%%%%%%%%
If $(m,n) \in \Rk(\ord_0)$, then $(m,n) \notin S$
since $(m,n) \in S$ would imply 
$(m,n) \notin \Rk(\ordbis)$ for some $\ordbis \leq \ord_0$,
contradicting Lemma~\ref{lem:hintikka:rank:mon}.
Hence $(m,n) \in \Rk(\ord_0)$ implies that there is no ordinal $\ordbis$
such that $(m,n) \notin \Rk(\ordbis)$;
in particular $(m,n) \in \Rk(\ord_0+1)$.
%%%%%%%%%%%%%%%%%%%%%%%%%%%%%%%%%%%%%%%%%%%%%%%%%%%%%%%%%%%%%%%%%%%%%%%%%%%
\end{fullproof}
%%%%%%%%%%%%%%%%%%%%%%%%%%%%%%%%%%%%%%%%%%%%%%%%%%%%%%%%%%%%%%%%%%%%%%%%%%%

We then arrive at the following characterisation,
based on~\cite[Theorem~3.4.2]{hodges93book}.

%%%%%%%%%%%%%%%%%%%%%%%%%%%%%%%%%%%%%%%%%%%%%%%%%%%%%%%%%%%%%%%%%%%%%%%%%%%
\begin{theorem}
\label{thm:hintikka:rank}
%%%%%%%%%%%%%%%%%%%%%%%%%%%%%%%%%%%%%%%%%%%%%%%%%%%%%%%%%%%%%%%%%%%%%%%%%%%
Let $(m,n)$ be an arbitrary position in the game $\G(a,b)$.
The following statements are equivalent:
\begin{enumerate}[(i)]
\item
\label{item:hintikka:rank:ord}
$(m,n) \in \Rk(\ord)$ for all ordinals $\ord$.

\item
\label{item:hintikka:rank:rank}
$(m,n) \in \Rk(\ord_{a,b})$.

\item
\label{item:hintikka:rank:win}
Duplicator has a winning strategy from position $(m,n)$.
\end{enumerate}
%%%%%%%%%%%%%%%%%%%%%%%%%%%%%%%%%%%%%%%%%%%%%%%%%%%%%%%%%%%%%%%%%%%%%%%%%%%
\end{theorem}
%%%%%%%%%%%%%%%%%%%%%%%%%%%%%%%%%%%%%%%%%%%%%%%%%%%%%%%%%%%%%%%%%%%%%%%%%%%

%%%%%%%%%%%%%%%%%%%%%%%%%%%%%%%%%%%%%%%%%%%%%%%%%%%%%%%%%%%%%%%%%%%%%%%%%%%
\begin{proof}
%%%%%%%%%%%%%%%%%%%%%%%%%%%%%%%%%%%%%%%%%%%%%%%%%%%%%%%%%%%%%%%%%%%%%%%%%%%
The implication
\(
  \text{\ref{item:hintikka:rank:ord}}
  \imp
  \text{\ref{item:hintikka:rank:rank}}
\)
follows from Lemma~\ref{lem:hintikka:rank:ord}
and
\(
  \text{\ref{item:hintikka:rank:rank}}
  \imp
  \text{\ref{item:hintikka:rank:win}}
\)
is given by Lemma~\ref{lem:hintikka:rank:win}.
It remains to prove
\(
  \text{\ref{item:hintikka:rank:win}}
  \imp
  \text{\ref{item:hintikka:rank:ord}}
\).
We show by induction on $\ord$ that 
for all position $(m,n)$ in $\G(a,b)$,
if Duplicator
has a winning strategy from position $(m,n)$, then $(m,n) \in \Rk(\ord)$.
\begin{description}
\item[Case of $\ord = 0$]

If Duplicator has a winning strategy from position $(m,n)$
then by definition we have $(m,n) \in \W(a,b)$.

\item[Case of $\ord = \ordbis +1$]

Assume that Duplicator has winning strategy from position $(m,n)$.
Then by induction hypothesis we have $(m,n) \in \Rk(\ordbis) \sle \W(a,b)$.
Moreover,
for every Spoiler move, this strategy leads to a winning position $(m',n')$,
and by induction hypothesis we have $(m',n') \in \Rk(\ordbis)$.
Hence $(m,n) \in \Rk(\ord)$ by definition.

\item[Case of $\ord$ limit]

By induction hypothesis.
\qedhere
\end{description}
%%%%%%%%%%%%%%%%%%%%%%%%%%%%%%%%%%%%%%%%%%%%%%%%%%%%%%%%%%%%%%%%%%%%%%%%%%%
\end{proof}
%%%%%%%%%%%%%%%%%%%%%%%%%%%%%%%%%%%%%%%%%%%%%%%%%%%%%%%%%%%%%%%%%%%%%%%%%%%

%%%%%%%%%%%%%%%%%%%%%%%%%%%%%%%%%%%%%%%%%%%%%%%%%%%%%%%%%%%%%%%%%%%%%%%%%%%
\begin{remark}
\label{rem:hintikka:rank}
%%%%%%%%%%%%%%%%%%%%%%%%%%%%%%%%%%%%%%%%%%%%%%%%%%%%%%%%%%%%%%%%%%%%%%%%%%%
Theorem~\ref{thm:hintikka:rank}
deserves a few comments:
\begin{enumerate}[(1)]
\item
\label{item:hintikka:rank:fin}
From a position of rank $\geq k$, with $k \in \NN$,
Duplicator has a strategy which is winning for $k$ rounds in $\G(a,b)$.
Hence,
the rank of $\G(a,b)$ is finite if the plays of
$\G(a,b)$ have bounded length
(cf.\ Remark~\ref{rem:path:games:length}).

For instance, 
when $k < \omega$,
the games played in the arboreal category $\cat R^E_k(\sig)$
(Examples~\ref{ex:path:R^E} and~\ref{ex:path:bisim-FOk-equivalence})
have finite rank because 
all plays have length at most~$k$.
Just observe that each path has cardinality at most $k$.

\item
From a position of rank $\omega$,
for all $k \in \NN$ Duplicator has a strategy
that is winning for $k$ rounds.
But they need not have a single strategy that
is winning for $k$ rounds for~all~$k \in \NN$.
\end{enumerate}
\end{remark}
%%%%%%%%%%%%%%%%%%%%%%%%%%%%%%%%%%%%%%%%%%%%%%%%%%%%%%%%%%%%%%%%%%%%%%%%%%%

%%%%%%%%%%%%%%%%%%%%%%%%%%%%%%%%%%%%%%%%%%%%%%%%%%%%%%%%%%%%%%%%%%%%%%%%%%%
\subsection{Hintikka formulae for wooded categories}
\label{sec:hintikka:form}
%%%%%%%%%%%%%%%%%%%%%%%%%%%%%%%%%%%%%%%%%%%%%%%%%%%%%%%%%%%%%%%%%%%%%%%%%%%
Let $\C$ be an lfp wooded category.
We are going to devise sufficient conditions on $\C$ so as to obtain
Hintikka formulae for the games $\G(a,b)$ for any $a,b \in \C$.
The main technical ingredients are the ordinal ranks discussed
in~\S\ref{sec:hintikka:rank}
and the material of~\S\ref{sec:coste:fp} on finitely presentable models
of cartesian theories.
Ultimately, we aim to prove Theorem~\ref{thm:hintikka:wooded:bfe},
which is of the form
\[
\begin{array}{l l l}
  \text{$a,b \in \C$ equivalent in $\Lang_{\infty}$}
& \longimp
& \text{Duplicator wins $\G(a,b)$}.
\end{array}
\]

As in~\S\ref{sec:path:results}, we assume
that the paths in $\C$ are finitely presentable.
However, in contrast to~\S\ref{sec:path:results},
there is (yet) no extensional category $\E$,
and we cannot enforce the assumption of ``detection of path embeddings''
of Theorem~\ref{thm:path:main} (see Definition~\ref{def:path:detection-path-emb}).
Instead, we will assume that the path embeddings in $\C$ are ``definable''
in a suitable sense.
In a similar way as in~\S\ref{sec:path:results},
we take as parameter the signature with respect to which path
embeddings are assumed to be definable.

%%%%%%%%%%%%%%%%%%%%%%%%%%%%%%%%%%%%%%%%%%%%%%%%%%%%%%%%%%%%%%%%%%%%%%%%%%%
\begin{definition}
\label{def:hintikka:def-path-emb}
%%%%%%%%%%%%%%%%%%%%%%%%%%%%%%%%%%%%%%%%%%%%%%%%%%%%%%%%%%%%%%%%%%%%%%%%%%%
Let $\Sigbis$ be a signature.
We say that 
the path embeddings of a wooded category
$\C$ are \emph{definable in $\Sigbis$} if the following conditions hold.
\begin{enumerate}[(i)]
\item
There is a cartesian theory $\theorybis$ in $\Sigbis$
such that $\C = \Mod(\theorybis)$.%
\footnote{Recall from Remark~\ref{rem:coste:fp:induced-lfp:fp}
that the choice of $\theorybis$ is irrelevant.}

\item
\label{item:hintikka:def-path-emb:path-fp}
Every path in $\C$ is finitely presentable (in $\C$).

\item
\label{item:hintikka:def-path-emb:formulae}
For each path $\FG{\vec y \mid \psi}$ of $\C$,
there is a (possibly infinite)
formula $\FEmb_{\C}\FG{\vec y | \psi}(\vec y)$ in~$\Sigbis$
such that for all $a \in \C$
and all $h \colon \FG{\vec y \mid \psi} \to a$
taking $\vec y$ to $\vec c \in \I{\vec y \mid \psi}_a$,
\[
\begin{array}{l !{\qquad\text{if, and only if,}\qquad} l}
  \text{$h$ is an embedding}
& a \models \FEmb_{\C}\FG{\vec y \mid \psi}(\vec c).
\end{array}
\]
\end{enumerate}
%%%%%%%%%%%%%%%%%%%%%%%%%%%%%%%%%%%%%%%%%%%%%%%%%%%%%%%%%%%%%%%%%%%%%%%%%%%
\end{definition}
%%%%%%%%%%%%%%%%%%%%%%%%%%%%%%%%%%%%%%%%%%%%%%%%%%%%%%%%%%%%%%%%%%%%%%%%%%%

We shall see in~\S\ref{sec:wc} that if a finitely accessible wooded
adjunction $R \colon \E \to \C$ detects path embeddings,
then $\C$ has definable path embeddings.
E.g., the finitely accessible arboreal adjunction
$\Ladj_k \colon {\cat R^E_k(\sig)} \inadj {\Struct(\sig)} \cocolon R_k$
in Example~\ref{ex:path:EFk-fin-acc-wooded-adj} detects path embeddings,
and similarly for the arboreal adjunctions
$R^P_k \colon \Struct(\sig) \to \cat R^P_k(\sig)$ and
$R^M_k\colon \Struct(\sig) \to \cat R^M_k(\sig)$ in Example~\ref{ex:path:misc:games}.

Fix a signature~$\Sigbis$ and a wooded category $\C$
whose path embeddings are definable in~$\Sigbis$.
We are going to devise Hintikka formulae for the games $\G(a,b)$, for any $a,b \in \C$.
Let $\cat U$ be the syntactic category of $\theorybis$.
By condition~\ref{item:hintikka:def-path-emb:path-fp}
in Definition~\ref{def:hintikka:def-path-emb},
for each path $Q$ of $\C$, there is some
$\OI{\vec z_Q \mid \varpi_{Q}} \in \cat U$ such that
$Q \cong \FG{\vec z_Q \mid \varpi_Q}$.
We fix a choice of such a
$\OI{\vec z_Q \mid \varpi_{Q}}$ for each $Q$. 
This allows us to extend the conventions of 
condition~\ref{item:hintikka:def-path-emb:formulae} in
Definition~\ref{def:hintikka:def-path-emb} by writing
\(
  \FEmb_{\C}[Q](\vec z_Q)
  =
  \FEmb_{\C}\FG{\vec z_Q \mid \varpi_Q}(\vec z_Q)
\).
Whenever it is convenient, we drop the subscript $Q$ in $\vec z_Q$.

%%%%%%%%%%%%%%%%%%%%%%%%%%%%%%%%%%%%%%%%%%%%%%%%%%%%%%%%%%%%%%%%%%%%%%%%%%%
\begin{figure}[!t]
%%%%%%%%%%%%%%%%%%%%%%%%%%%%%%%%%%%%%%%%%%%%%%%%%%%%%%%%%%%%%%%%%%%%%%%%%%%
\[
\begin{array}{r c l}
  \Theta[m,Q,0](\vec z)
& \deq
& \FIso[P,Q]

\\\\
  \Theta[m,Q,\ord](\vec z)
& \deq
& \bigwedge_{\ordbis < \ord} \Theta[m,Q,\ordbis](\vec z)
  \hfill\text{($\ord$ limit ordinal)}

\\\\
  \Theta[m,Q,\ord+1](\vec z)
& \deq
& \FIso[P,Q]
  ~\land~
  \FForth[m,Q,\ord+1](\vec z)
  ~\land~
  \FBack[m,Q,\ord+1](\vec z)

\\\\
  \FForth[m,Q,\ord+1](\vec z)
& \deq
\\
\multicolumn{3}{r}{
  \bigwedge_{\substack{m' \colon P' \emb a \\ m \prec m'}}
  \bigvee_{\substack{\ell \colon Q \emb Q' \\ \ell \prec \id_{Q'}}}
  (\exists \vec z')
  \bigg(
  \varpi_{Q'}(\vec z')
  ~\land~
  \vartheta_\ell(\vec z',\vec z)
  ~\land~
  \FEmb[Q'](\vec z')
  ~\land~
  \Theta[m',Q',\ord](\vec z')
  \bigg)
}

\\\\
  \FBack[m,Q,\ord+1](\vec z)
& \deq
\\
\multicolumn{3}{r}{
  \bigwedge_{\substack{\ell \colon Q \emb Q' \\ \ell \prec \id_{Q'}}}
  (\forall \vec z')
  \bigg(
  \big(
    \varpi_{Q'}(\vec z')
    ~\land~
    \vartheta_\ell(\vec z',\vec z)
    ~\land~
    \FEmb[Q'](\vec z')
  \big)
  ~\limp~
  \bigvee_{\substack{m' \colon P' \emb a \\ m \prec m'}}
  \Theta[m',Q',\ord](\vec z')
  \bigg)
}

\end{array}
\]
\caption{Hintikka formulae 
for the wooded category $\C$ (\S\ref{sec:hintikka:form}).
\label{fig:hintikka:arb}}
%%%%%%%%%%%%%%%%%%%%%%%%%%%%%%%%%%%%%%%%%%%%%%%%%%%%%%%%%%%%%%%%%%%%%%%%%%%
\end{figure}
%%%%%%%%%%%%%%%%%%%%%%%%%%%%%%%%%%%%%%%%%%%%%%%%%%%%%%%%%%%%%%%%%%%%%%%%%%%

Fix an object $a \in \C$.
For each ordinal $\ord$,
each
path embedding $m \colon P \emb a$,
and each path $Q$,
we define a formula
\[
\begin{array}{l l l}
  \Theta[m,Q,\ord](\vec z)
& =
& \Theta_{\C}[a,m,Q,\ord](\vec z_Q)
\end{array}
\]

\noindent
in the signature $\Sigbis$.
The definition is by induction on $\ord$
and is given in Figure~\ref{fig:hintikka:arb},
where
we use the following notations.
First, given an embedding between paths $\ell\colon Q\emb Q'$, we denote by 
\(
  \MI{\vec z',\vec z \mid \vartheta_\ell}
  \colon
  \OI{\vec z' \mid \varpi_{Q'}}
  \to
  \OI{\vec z \mid \varpi_{Q}}
\)
the $\cat U$-morphism corresponding to $\ell$ via the Yoneda embedding
$\yoneda \colon \cat U^\op \to \lex\funct{\cat U,\Set}$.
Moreover, $\FIso[P,Q]$ is the formula
$\True$ if $P \cong Q$,
and $\False$ otherwise.
Finally, for notational simplicity,
we have dropped references to~$\C$ and~$a$, 
since they remain fixed throughout the induction on $\ord$.
But note that the definitions do depend on $a$:
see the clauses for $\FForth[-,-,-]$ and $\FBack[-,-,-]$.

The next result says that the formulae
$\Theta_{\C}[a,-,-,\ord](\vec z)$ represent the positions of rank
at least $\ord$ in the games $\G(a,b)$, where $b$ ranges over the objects of $\C$.
This is proved by induction on $\ord$,
along the same lines as, e.g., \cite[Theorem~3.5.1]{hodges93book}.

%%%%%%%%%%%%%%%%%%%%%%%%%%%%%%%%%%%%%%%%%%%%%%%%%%%%%%%%%%%%%%%%%%%%%%%%%%%
\begin{proposition}
\label{prop:hintikka:wooded}
%%%%%%%%%%%%%%%%%%%%%%%%%%%%%%%%%%%%%%%%%%%%%%%%%%%%%%%%%%%%%%%%%%%%%%%%%%%
Assume that the path embeddings of $\C$ are definable in $\Sigbis$.
Let $a, b \in \C$ and $m \colon P \emb a$.
Further, let $n \colon Q \emb b$
be induced by $\vec c \in \I{\vec z \mid \varpi_{Q}}_b$.
The following statements are equivalent for all ordinals $\ord$:
\begin{enumerate}[(i)]
\item
$b \models \Theta_{\C}[a,m,Q,\ord](\vec c)$.

\item
$(m,n)$ is a position of rank $\geq \ord$
in the game $\G(a,b)$.
\end{enumerate}
%%%%%%%%%%%%%%%%%%%%%%%%%%%%%%%%%%%%%%%%%%%%%%%%%%%%%%%%%%%%%%%%%%%%%%%%%%%
\end{proposition}
%%%%%%%%%%%%%%%%%%%%%%%%%%%%%%%%%%%%%%%%%%%%%%%%%%%%%%%%%%%%%%%%%%%%%%%%%%%

%%%%%%%%%%%%%%%%%%%%%%%%%%%%%%%%%%%%%%%%%%%%%%%%%%%%%%%%%%%%%%%%%%%%%%%%%%%
\begin{proof}
%%%%%%%%%%%%%%%%%%%%%%%%%%%%%%%%%%%%%%%%%%%%%%%%%%%%%%%%%%%%%%%%%%%%%%%%%%%
We proceed by induction on~$\ord$.
The case of $\ord$ a limit ordinal follows trivially from the induction
hypothesis.
In the case of $\ord = 0$, recall from Definition~\ref{def:games}
that $(m,n) \in \W(a,b)$ precisely when $\dom(m) \cong \dom(n)$.
This is equivalent to $b \models \FIso[P,Q]$.

We now consider the case of a successor ordinal, say $\ord+1$.
We again have $(m,n) \in \W(a,b)$ if and only if $b \models \FIso[P,Q]$.
We then reason by cases on the next Spoiler move.

%%%%%%%%%%%%%%%%%%%%%%%%%%%%%%%%%%%%%%%%%%%%%%%%%%%%%%%%%%%%%%%%%%%%%%%%%%%
\begin{claim}
\label{claim:hintikka:wooded}
%%%%%%%%%%%%%%%%%%%%%%%%%%%%%%%%%%%%%%%%%%%%%%%%%%%%%%%%%%%%%%%%%%%%%%%%%%%
The following are equivalent:
\begin{itemize}
\item
$b \models \FForth[m,Q,\ord+1](\vec c)$.

\item
For all $m' \succ m$, there is some $n' \succ n$
such that $(m',n') \in \Rk(\ord)$.
\end{itemize}
%%%%%%%%%%%%%%%%%%%%%%%%%%%%%%%%%%%%%%%%%%%%%%%%%%%%%%%%%%%%%%%%%%%%%%%%%%%
\end{claim}
%%%%%%%%%%%%%%%%%%%%%%%%%%%%%%%%%%%%%%%%%%%%%%%%%%%%%%%%%%%%%%%%%%%%%%%%%%%

%%%%%%%%%%%%%%%%%%%%%%%%%%%%%%%%%%%%%%%%%%%%%%%%%%%%%%%%%%%%%%%%%%%%%%%%%%%
\begin{proof}[Proof of Claim~\ref{claim:hintikka:wooded}]
%%%%%%%%%%%%%%%%%%%%%%%%%%%%%%%%%%%%%%%%%%%%%%%%%%%%%%%%%%%%%%%%%%%%%%%%%%%
Assume that
$b \models \FForth[m,Q,\ord+1](\vec c)$
and consider some $m' \succ m$.
Then there exist some $\ell \colon Q \emb Q'$
and $\vec c' \in b$
such that $\ell \prec \id_{Q'}$ and
\[
\begin{array}{l l l}
  b
& \models
& \varpi_{Q'}(\vec c')
  ~\land~
  \vartheta_\ell(\vec c',\vec c)
  ~\land~
  \FEmb[Q'](\vec c')
  ~\land~
  \Theta[m',Q',\ord](\vec c').
\end{array}
\]

Since the tuple $\vec c'$ belongs to $\I{\vec z' \mid \varpi_{Q'}}_b$,
it induces a (unique) morphism $n' \colon Q' \to b$, 
which is an embedding because
$b \models \FEmb[Q'](\vec c')$.
Moreover, since
$b \models \Theta[m',Q',\ord](\vec c')$,
the induction hypothesis implies that $(m',n') \in \Rk(\ord)$.

It remains to show that $n' \succ n$.
Lemma~\ref{lem:coste:fp:triangle} 
yields $n = n' \comp \ell$, that is $n \leq n'$.
If $n = n'$ in $\Path{b}$ then $n = n' \comp i$ for some iso $i$.
Since $n'$ is a mono, we get that $\ell=i$ is also an iso,
contradicting that $\ell \prec \id_{Q'}$.
Hence $n < n'$.
Further, if $n''$ is an element of $\Path{b}$ satisfying $n \leq n''$ and $n'' \leq n'$,
then $n = n'' \comp k'$ and $n'' = n' \comp k$
for some embeddings~$k,k'$.
Again since $n'$ is monic, we get $k \comp k' = \ell$.
But as $\ell \prec \id_{Q'}$, either $k$ or $k'$ is an iso.
So we cannot have $n < n'' < n'$.

The converse implication is proved in a similar fashion.
%%%%%%%%%%%%%%%%%%%%%%%%%%%%%%%%%%%%%%%%%%%%%%%%%%%%%%%%%%%%%%%%%%%%%%%%%%%
\end{proof}
%%%%%%%%%%%%%%%%%%%%%%%%%%%%%%%%%%%%%%%%%%%%%%%%%%%%%%%%%%%%%%%%%%%%%%%%%%%

Similarly, we see that
$b \models \FBack[m,Q,\ord+1](\vec c)$
if, and only if,
\begin{itemize}
\item
for all $n' \succ n$, there is some $m' \succ m$
such that $(m',n') \in \Rk(\ord)$.
\end{itemize}
The result then follows from the definition of $\Rk(\ord+1)$.
%%%%%%%%%%%%%%%%%%%%%%%%%%%%%%%%%%%%%%%%%%%%%%%%%%%%%%%%%%%%%%%%%%%%%%%%%%%
\end{proof}
%%%%%%%%%%%%%%%%%%%%%%%%%%%%%%%%%%%%%%%%%%%%%%%%%%%%%%%%%%%%%%%%%%%%%%%%%%%

We then get that 
equivalence in $\Lang_\infty$ implies back-and-forth equivalence $\bisim$.
This is our first main result.
Recall from Definition~\ref{def:path:bfe}
that $a \bisim b$ when Duplicator has a winning strategy in $\G(a,b)$
from the initial position $(\bot_a,\bot_b)$,
where $\bot_a \colon \zero_a \emb a$ is the root of $\Path{a}$
and similarly for $\bot_b$.

%%%%%%%%%%%%%%%%%%%%%%%%%%%%%%%%%%%%%%%%%%%%%%%%%%%%%%%%%%%%%%%%%%%%%%%%%%%
\begin{theorem}
\label{thm:hintikka:wooded:bfe}
%%%%%%%%%%%%%%%%%%%%%%%%%%%%%%%%%%%%%%%%%%%%%%%%%%%%%%%%%%%%%%%%%%%%%%%%%%%
Assume that the path embeddings of $\C$ are definable in $\Sigbis$.
Then 
\[
\begin{array}{l l l}
  \text{$a,b \in \C$ equivalent in $\Lang_{\infty}(\Sigbis)$}
& \longimp
& a \bisim b.
\end{array}
\]
%%%%%%%%%%%%%%%%%%%%%%%%%%%%%%%%%%%%%%%%%%%%%%%%%%%%%%%%%%%%%%%%%%%%%%%%%%%
\end{theorem}
%%%%%%%%%%%%%%%%%%%%%%%%%%%%%%%%%%%%%%%%%%%%%%%%%%%%%%%%%%%%%%%%%%%%%%%%%%%

The proof of Theorem~\ref{thm:hintikka:wooded:bfe} relies on
Lemmas~\ref{lem:hintikka:wooded:bfe:zero} and~\ref{lem:hintikka:wooded:bfe:zeroiso}
below.
Let $\C$ be a wooded category whose path embeddings are definable in the
signature $\Sigbis$.
Hence, similarly as above, 
there is a cartesian theory $\theorybis$ in $\Sigbis$
such that $\C \cong \Mod(\theorybis)$.
Let $\cat U$ be the syntactic category of $\theorybis$.

%%%%%%%%%%%%%%%%%%%%%%%%%%%%%%%%%%%%%%%%%%%%%%%%%%%%%%%%%%%%%%%%%%%%%%%%%%%
\begin{lemma}
\label{lem:hintikka:wooded:bfe:zero}
%%%%%%%%%%%%%%%%%%%%%%%%%%%%%%%%%%%%%%%%%%%%%%%%%%%%%%%%%%%%%%%%%%%%%%%%%%%
For each $a \in \C$, there is a sentence $\zeta_a$
such that $\zero_a \cong \FG{\emptyset \mid \zeta_a}$.
%%%%%%%%%%%%%%%%%%%%%%%%%%%%%%%%%%%%%%%%%%%%%%%%%%%%%%%%%%%%%%%%%%%%%%%%%%%
\end{lemma}
%%%%%%%%%%%%%%%%%%%%%%%%%%%%%%%%%%%%%%%%%%%%%%%%%%%%%%%%%%%%%%%%%%%%%%%%%%%

%%%%%%%%%%%%%%%%%%%%%%%%%%%%%%%%%%%%%%%%%%%%%%%%%%%%%%%%%%%%%%%%%%%%%%%%%%%
\begin{proof}
%%%%%%%%%%%%%%%%%%%%%%%%%%%%%%%%%%%%%%%%%%%%%%%%%%%%%%%%%%%%%%%%%%%%%%%%%%%
Recall from Lemma~\ref{lem:path:base} that $\zero_a$
is obtained by taking a factorisation
$\zero \epi \zero_a \emb a$
of the unique map $\zero \to a$, where $\zero$ is initial in $\C$.
Since $\C$ is lfp, $\zero$ is the finitely presentable $\theorybis$-model
$\FG{\emptyset \mid \True}$, where $\emptyset$ is the empty context
(\S\ref{sec:prelim:struct}).
Moreover, $\zero$ and $\zero_a$ are paths by Lemma~\ref{lem:path:base}.
Since we assumed that paths in $\C$ are finitely presentable, $\zero_a$ is finitely presentable,
say $\zero_a \cong \FG{\vec x \mid \varphi}$.
By Lemma~\ref{lem:prelim:fact:base}
we thus get a quotient 
$e \colon \FG{\emptyset \mid \True} \epi \FG{\vec x \mid \varphi}$.
In particular, $e$ is epic in $\C$, and thus also in its full subcategory $\fp\C$.
It then follows from Theorem~\ref{thm:coste:synt}
and the Yoneda lemma
that $e$ is induced by a monic $\cat U$-morphism
\[
\begin{array}{*{5}{l}}
  \MI{\vec x \mid \theta}
& :
& \OI{\vec x \mid \varphi}
& \longto
& \OI{\emptyset \mid \True}.
\end{array}
\]

Since $\MI{\vec x \mid \theta}$ is monic,
it follows from~\cite[Lemma D1.4.4(iii)]{johnstone02book}
that $\theorybis$ proves the sequent
\[
\begin{array}{l l l}
  \theta
  ~\land~
  \theta[\vec x'/\vec x]
& \thesis_{\vec x, \vec x'}
& \vec x \Eq \vec x'.
\end{array}
\]

\noindent
Moreover, by construction of $\cat U$ (Definition~\ref{def:coste:synt}),
$\theorybis$ proves
\[
\begin{array}{l l l !{\qquad\text{and}\qquad} l l l}
  \theta
& \thesis_{\vec x}
& \varphi

& \varphi
& \thesis_{\vec x}
& \theta.
\end{array}
\]

\noindent
It follows that the
sentence $\zeta_a = (\exists \vec x)\varphi$ is cartesian over $\theorybis$.
Further, $\OI{\vec x \mid \varphi}$ is isomorphic to
$\OI{\emptyset \mid (\exists \vec x)\varphi}$ in $\cat U$ (cf.\ the proof of \cite[Lemma~D1.4.4(ii)]{johnstone02book}),
hence
$\FG{\vec x \mid \varphi}$ is isomorphic to
$\FG{\emptyset \mid (\exists \vec x)\varphi}$ in $\C$,
which concludes the proof.
%%%%%%%%%%%%%%%%%%%%%%%%%%%%%%%%%%%%%%%%%%%%%%%%%%%%%%%%%%%%%%%%%%%%%%%%%%%
\end{proof}
%%%%%%%%%%%%%%%%%%%%%%%%%%%%%%%%%%%%%%%%%%%%%%%%%%%%%%%%%%%%%%%%%%%%%%%%%%%

Lemma~\ref{lem:hintikka:wooded:bfe:zero} shows that
Hintikka formulae of the form $\Theta_{\C}[a,\bot_a,\zero_b,\ord]$
can be assumed to be sentences.
It also implies the following simple yet important fact
about wooded categories with definable path embeddings.

%%%%%%%%%%%%%%%%%%%%%%%%%%%%%%%%%%%%%%%%%%%%%%%%%%%%%%%%%%%%%%%%%%%%%%%%%%%
\begin{lemma}
\label{lem:hintikka:wooded:bfe:zeroiso}
%%%%%%%%%%%%%%%%%%%%%%%%%%%%%%%%%%%%%%%%%%%%%%%%%%%%%%%%%%%%%%%%%%%%%%%%%%%
If $a,b \in \C$ are equivalent in $\Lang_\infty(\Sigbis)$,
then $\zero_a \cong \zero_b$.
%%%%%%%%%%%%%%%%%%%%%%%%%%%%%%%%%%%%%%%%%%%%%%%%%%%%%%%%%%%%%%%%%%%%%%%%%%%
\end{lemma}
%%%%%%%%%%%%%%%%%%%%%%%%%%%%%%%%%%%%%%%%%%%%%%%%%%%%%%%%%%%%%%%%%%%%%%%%%%%

%%%%%%%%%%%%%%%%%%%%%%%%%%%%%%%%%%%%%%%%%%%%%%%%%%%%%%%%%%%%%%%%%%%%%%%%%%%
\begin{proof}
%%%%%%%%%%%%%%%%%%%%%%%%%%%%%%%%%%%%%%%%%%%%%%%%%%%%%%%%%%%%%%%%%%%%%%%%%%%
Assume that $a,b \in \C$ are equivalent in $\Lang_\infty(\Sigbis)$.
Since we have an embedding $\bot_a \colon \zero_a \emb a$,
it follows that
$a$ is a model of the sentences $\zeta_a$ and $\FEmb_{\C}\FG{\emptyset | \zeta_a}$.
Hence $b \models \zeta_a$,
which by Lemma~\ref{lem:coste:lp:hom} implies the existence of some
$h \colon \zero_a \to b$.
Moreover, $b \models \FEmb_{\C}\FG{\emptyset | \zeta_a}$,
so that $h$ is actually an embedding.
It follows that the unique map $\zero \to b$
factors as $\zero \epi \zero_a \emb b$.
But $\zero \to b$ also factors as
$\zero \epi \zero_b \emb b$.
Hence $\zero_a \cong \zero_b$
since, in a proper factorisation system,
factorisations are unique up to isomorphism.
%%%%%%%%%%%%%%%%%%%%%%%%%%%%%%%%%%%%%%%%%%%%%%%%%%%%%%%%%%%%%%%%%%%%%%%%%%%
\end{proof}
%%%%%%%%%%%%%%%%%%%%%%%%%%%%%%%%%%%%%%%%%%%%%%%%%%%%%%%%%%%%%%%%%%%%%%%%%%%

%%%%%%%%%%%%%%%%%%%%%%%%%%%%%%%%%%%%%%%%%%%%%%%%%%%%%%%%%%%%%%%%%%%%%%%%%%%
\begin{remark}
%%%%%%%%%%%%%%%%%%%%%%%%%%%%%%%%%%%%%%%%%%%%%%%%%%%%%%%%%%%%%%%%%%%%%%%%%%%
Note that when $\C$ is either of the wooded categories $\cat R^E_k(\sig)$ or
$\cat R^P_k(\sig)$ from Examples~\ref{ex:path:R^E}
and~\ref{ex:path:misc}\ref{item:path:misc:pebble},
the initial object $\zero$ has empty carrier and so $\zero_a \cong \zero_b$
for all $a,b\in \C$.
However, this is not the case for the wooded category $\cat R^M_k(\sig)$
in Example~\ref{ex:path:misc}\ref{item:path:misc:modal},
since $\zero$ has one element;
in that case, we do need Lemma~\ref{lem:hintikka:wooded:bfe:zeroiso}.
%%%%%%%%%%%%%%%%%%%%%%%%%%%%%%%%%%%%%%%%%%%%%%%%%%%%%%%%%%%%%%%%%%%%%%%%%%%
\end{remark}
%%%%%%%%%%%%%%%%%%%%%%%%%%%%%%%%%%%%%%%%%%%%%%%%%%%%%%%%%%%%%%%%%%%%%%%%%%%

We can now prove Theorem~\ref{thm:hintikka:wooded:bfe}.

%%%%%%%%%%%%%%%%%%%%%%%%%%%%%%%%%%%%%%%%%%%%%%%%%%%%%%%%%%%%%%%%%%%%%%%%%%%
\begin{proof}[Proof of Theorem~\ref{thm:hintikka:wooded:bfe}]
%%%%%%%%%%%%%%%%%%%%%%%%%%%%%%%%%%%%%%%%%%%%%%%%%%%%%%%%%%%%%%%%%%%%%%%%%%%
Let $a,b \in \C$ be equivalent in $\Lang_\infty(\Sigbis)$.
We want to show that 
Duplicator wins the game $\G(a,b)$ from position $(\bot_a,\bot_b)$.
By Theorem~\ref{thm:hintikka:rank},
this amounts to showing that $(\bot_a,\bot_b) \in \Rk(\ord)$
for all ordinals $\ord$.
We apply Proposition~\ref{prop:hintikka:wooded}
and show that, for every ordinal $\ord$,
we have
$b \models \Theta_{\C}[a,\bot_a,\zero_b,\ord]$,
where we have assumed that $\zero_b = \FG{\emptyset \mid \zeta_b}$
by Lemma~\ref{lem:hintikka:wooded:bfe:zero}.

Note that Duplicator wins the game $\G(a,a)$ from position $(\bot_a,\bot_a)$.
On the other hand, by Lemma~\ref{lem:hintikka:wooded:bfe:zeroiso}
we have $\zero_a \cong \zero_b$.
Hence, there is some $\bot'_a \colon \zero_b \emb a$
such that $\bot'_a \sim \bot_a$,
i.e.\ such that $\bot'_a$ and $\bot_a$
represent the same element of $\Path{a}$ (see~\S\ref{sec:prelim:facto}).
It follows that Duplicator wins $\G(a,a)$ from position $(\bot_a,\bot'_a)$.
Thus, for each ordinal $\ord$,
we have $a \models \Theta_{\C}[a,\bot_a,\zero_b,\ord]$
and so also $b \models \Theta_{\C}[a,\bot_a,\zero_b,\ord]$.
%%%%%%%%%%%%%%%%%%%%%%%%%%%%%%%%%%%%%%%%%%%%%%%%%%%%%%%%%%%%%%%%%%%%%%%%%%%
\end{proof}
%%%%%%%%%%%%%%%%%%%%%%%%%%%%%%%%%%%%%%%%%%%%%%%%%%%%%%%%%%%%%%%%%%%%%%%%%%%

%%%%%%%%%%%%%%%%%%%%%%%%%%%%%%%%%%%%%%%%%%%%%%%%%%%%%%%%%%%%%%%%%%%%%%%%%%%
\subsection{The case of finitely accessible wooded adjunctions}
\label{sec:hintikka:finaccadj}
%%%%%%%%%%%%%%%%%%%%%%%%%%%%%%%%%%%%%%%%%%%%%%%%%%%%%%%%%%%%%%%%%%%%%%%%%%%
Consider a finitely accessible wooded adjunction
\[
\begin{tikzcd}
  \C
  \arrow[bend left=25]{r}{\Ladj}
  \arrow[phantom]{r}[description]{\textnormal{\footnotesize{$\bot$}}}
& \E,
  \arrow[bend left=25]{l}{R}
\end{tikzcd}
\]

\noindent
where the path embeddings of $\C$ are definable in some signature.
Intuitively, we would like to transfer Theorem~\ref{thm:hintikka:wooded:bfe}
along the interpretation induced by the morphism of lfp categories $R \colon \E \to \C$
(see \S\ref{sec:coste:interp}).
To make this precise, we consider a cartesian theory $\theory$
such that the lfp category $\E$ is equivalent to $\Mod(\theory)$;
see Corollary~\ref{cor:lfp-iff-models-of-T}.

Combining Theorem~\ref{thm:hintikka:wooded:bfe}
with the transfer Theorem~\ref{thm:coste:interp} of \S\ref{sec:coste:interp}
yields at once the following result.
Recall that given $M, N \in \E$,
we have $M \bisim_R N$ when $R M \bisim R N$ in $\C$.

%%%%%%%%%%%%%%%%%%%%%%%%%%%%%%%%%%%%%%%%%%%%%%%%%%%%%%%%%%%%%%%%%%%%%%%%%%%
\begin{corollary}
\label{cor:hintikka:wooded:R-bfe}
%%%%%%%%%%%%%%%%%%%%%%%%%%%%%%%%%%%%%%%%%%%%%%%%%%%%%%%%%%%%%%%%%%%%%%%%%%%
Let $\Sig$ be a signature.
Consider a finitely accessible wooded adjunction $R \colon \E \to \C$,
where the path embeddings of $\C$ are definable in some signature,
and where $\E = \Mod(\theory)$ for some cartesian theory $\theory$ in $\Sig$.
Then
\[
\begin{array}{l l l}
  \text{$M,N \in \E$ equivalent in $\Lang_{\infty}(\Sig)$}
& \longimp
& M \bisim_R N.
\end{array}
\]
%%%%%%%%%%%%%%%%%%%%%%%%%%%%%%%%%%%%%%%%%%%%%%%%%%%%%%%%%%%%%%%%%%%%%%%%%%%
\end{corollary}
%%%%%%%%%%%%%%%%%%%%%%%%%%%%%%%%%%%%%%%%%%%%%%%%%%%%%%%%%%%%%%%%%%%%%%%%%%%

Note that Corollary~\ref{cor:hintikka:wooded:R-bfe}
does not explicitly refer to
the signature in which the path embeddings of the wooded category $\C$ are definable;
it only refers to the signature $\Sig$ of the extensional category $\E$.

While the road towards Corollary~\ref{cor:hintikka:wooded:R-bfe}
was 
relatively straightforward,
the proof of our main Theorem~\ref{thm:path:main}
is considerably more involved.
Before proceeding, we comment on Hintikka formulae in the context
of Corollary~\ref{cor:hintikka:wooded:R-bfe}.
Fix some $R \colon \E \to \C$ as in the statement of the latter result,
and let $\theory$ be such that $\E = \Mod(\theory)$.

Note that we may apply the transfer Theorem~\ref{thm:coste:interp}
directly to the Hintikka formulae of~\S\ref{sec:hintikka:form},
without the detour via Theorem~\ref{thm:hintikka:wooded:bfe}.
This direct approach requires a more explicit description.
Assume that the path embeddings of $\C$ are definable in a signature~$\Sigbis$.
In particular, there is a cartesian theory $\theorybis$ in $\Sigbis$
such that $\C = \Mod(\theorybis)$.
Let $\cat U$ and $\cat T$ be the syntactic categories of $\theorybis$
and $\theory$, respectively.
As discussed in~\S\ref{sec:coste:interp}, 
the morphism of lfp categories $R \colon \E \to \C$
corresponds via Theorem~\ref{thm:coste:synt} to a functor
$\ladj\inclth \colon \lex\funct{\cat T,\Set} \to \lex\funct{\cat U,\Set}$,
where $\inclth \colon \cat U \to \cat T$ is a lex-morphism.
Recall that the interpretation $(-)^{\inclth}$ of~\S\ref{sec:coste:interp} relies
on the fact that given $M \in \Mod(\theory)$ and $\OI{\vec y \mid \psi} \in \cat U$,
we have
\[
\begin{array}{l l l}
  \I{\vec y \mid \psi}_{R M}
& \cong
& \I{\inclth\OI{\vec y \mid \psi}}_M.
\end{array}
\]

Combining this with Proposition~\ref{prop:hintikka:wooded}, we get the following result.

%%%%%%%%%%%%%%%%%%%%%%%%%%%%%%%%%%%%%%%%%%%%%%%%%%%%%%%%%%%%%%%%%%%%%%%%%%%
\begin{corollary}
\label{cor:hintikka:interp}
%%%%%%%%%%%%%%%%%%%%%%%%%%%%%%%%%%%%%%%%%%%%%%%%%%%%%%%%%%%%%%%%%%%%%%%%%%%
Let $R \colon \E \to \C$ be as above.
Let $M, N \in \E$, 
and let $m \colon P \emb R M$ be a path embedding in $\C$.
Consider a path $Q$ with $Q \cong \FG{\vec z \mid \varpi_Q}$,
and write $\OI{\vec x \mid \varphi} \in \cat T$
for $\inclth \OI{\vec z \mid \varpi_Q}$.
Let $n \colon Q \emb R N$ be induced by 
$\vec c \in \I{\vec x \mid \varphi}_{N}$.
Then the following statements are equivalent for all ordinals $\ord$:
\begin{enumerate}[(i)]
\item
$N \models (\Theta_{\C}[R M,m,Q,\ord])^{\inclth}(\vec c)$.

\item
$(m,n)$ is a position of rank $\geq \ord$
in the game $\G(R M,R N)$.
\end{enumerate}
%%%%%%%%%%%%%%%%%%%%%%%%%%%%%%%%%%%%%%%%%%%%%%%%%%%%%%%%%%%%%%%%%%%%%%%%%%%
\end{corollary}
%%%%%%%%%%%%%%%%%%%%%%%%%%%%%%%%%%%%%%%%%%%%%%%%%%%%%%%%%%%%%%%%%%%%%%%%%%%

The more accurate description in Corollary~\ref{cor:hintikka:interp}
(compared to Corollary~\ref{cor:hintikka:wooded:R-bfe})
can be useful in the following situation.
Recall from Remark~\ref{rem:coste:interp:card} that
$(-)^{\inclth}$ interprets $\Lang_\kappa(\Sigbis)$ in $\Lang_\kappa(\Sig)$
for each regular cardinal $\kappa$.
In particular, 
$(\Theta_{\C}[R M,m,Q,\ord])^{\inclth}(\vec x)$
is a finite formula whenever so is
$\Theta_{\C}[R M,m,Q,\ord](\vec z)$.
Direct inspection of Figure~\ref{fig:hintikka:arb} reveals that
the latter formula is finite when all the following conditions are satisfied:
\begin{itemize}
\item
The ordinal $\ord$ is finite.
This can always be assumed when for each $N \in \E$,
the plays of $\G(R M, R N)$
have (finite) bounded length
(Remark~\ref{rem:hintikka:rank}.\ref{item:hintikka:rank:fin}),
as e.g.\ in the case of
the $k$-round Ehrenfeucht-Fraïssé 
game (see Example~\ref{ex:path:bisim-FOk-equivalence}).

\item
There are (up to isomorphism) finitely many
embeddings $m' \colon P' \emb R M$ such that $m \prec m'$.
In the Ehrenfeucht-Fraïssé games for finite $\sig$,
this condition is met when the $\sig$-structure $M$ is finite.

\item
For each path $Q$ in $\C$ there are (up to isomorphism)
finitely many paths $Q'$ and embeddings $\ell \colon Q \emb Q'$ with
$\ell \prec \id_{Q'}$.
This condition is met with $k$-round Ehrenfeucht-Fraïssé games
for \emph{finite} (purely relational) signatures $\sig$.

\item
The formulae $\FEmb_{\C}[Q](-)$
are finite
(see Definition~\ref{def:hintikka:def-path-emb}).
\end{itemize}

Recall from Lemma~\ref{l:forest-of-paths-wooded} that the subcategory $\pth\C$ of a wooded category $\C$ defined the paths and the embeddings between them can be identified, up to equivalence, with a forest order. We get

%%%%%%%%%%%%%%%%%%%%%%%%%%%%%%%%%%%%%%%%%%%%%%%%%%%%%%%%%%%%%%%%%%%%%%%%%%%
\begin{corollary}
\label{cor:hintikka:interp:fin}
%%%%%%%%%%%%%%%%%%%%%%%%%%%%%%%%%%%%%%%%%%%%%%%%%%%%%%%%%%%%%%%%%%%%%%%%%%%
Let $\Sig$ be a signature.
Consider a finitely accessible wooded adjunction $R \colon \E \to \C$
such that the path embeddings of $\C$ are definable in some signature $\Sigbis$,
and $\E = \Mod(\theory)$ for some cartesian theory $\theory$ in $\Sig$.
Assume the following:
\begin{enumerate}[(i)]
\item
\label{item:hintikka:interp:fin:path}
The forest $\pth\C$ is finitely branching.

\item
\label{item:hintikka:interp:fin:emb}
The formulae $\FEmb_{\C}[Q](-)$ (in the signature $\Sigbis$) are finite.

\setcounter{SplitEnum}{\value{enumi}}
\end{enumerate}

\noindent
Further, let $M, N \in \E$ satisfy the following conditions:
\begin{enumerate}[(i)]
\setcounter{enumi}{\value{SplitEnum}}
\item
\label{item:hintikka:interp:fin:ord}
The plays of $\G(R M, R N)$ have (finite) bounded length.

\item
\label{item:hintikka:interp:fin:struct}
For each path embedding $m \colon P \emb R M$ in $\C$,
there are (up to isomorphism) finitely many
embeddings $m' \colon P' \emb R M$ such that $m \prec m'$.
\end{enumerate}

\noindent
Then
\[
\begin{array}{l l l}
  \text{$M,N$ equivalent in $\Lang_{\omega}(\Sig)$}
& \longimp
& M \bisim_R N.
\end{array}
\]
%%%%%%%%%%%%%%%%%%%%%%%%%%%%%%%%%%%%%%%%%%%%%%%%%%%%%%%%%%%%%%%%%%%%%%%%%%%
\end{corollary}
%%%%%%%%%%%%%%%%%%%%%%%%%%%%%%%%%%%%%%%%%%%%%%%%%%%%%%%%%%%%%%%%%%%%%%%%%%%

In view of
Corollary~\ref{cor:hintikka:interp:fin},
it is interesting to identify sufficient conditions ensuring that the formulae
$\FEmb_{\C}[Q](-)$ of Definition~\ref{def:hintikka:def-path-emb} can be assumed
to be finite.
In~\S\ref{sec:wc:fin},
we will investigate this question under the assumptions
of Theorem~\ref{thm:path:main}.

The finitely accessible arboreal adjunction of Example~\ref{ex:Diaconescu-cover}
involving presheaves and forests 
is a simpler case where
condition~\ref{item:hintikka:interp:fin:emb} above holds, as shown in the next example:

%%%%%%%%%%%%%%%%%%%%%%%%%%%%%%%%%%%%%%%%%%%%%%%%%%%%%%%%%%%%%%%%%%%%%%%%%%%
\begin{example} 
\label{ex:hintikka:mono:forest}
%%%%%%%%%%%%%%%%%%%%%%%%%%%%%%%%%%%%%%%%%%%%%%%%%%%%%%%%%%%%%%%%%%%%%%%%%%%
Let $\A$ be the category of 
presheaves over a forest, say $\A = \presh{\forest}$.

We have seen that $\A$ is arboreal, thus wooded
(Theorem~\ref{thm:arboreal:presh}). 
Recall that, since $\A$ is a topos,
(epis, monos) is the unique proper factorisation system on $\A$,
and this factorisation system is stable (see \S\ref{sec:prelim:facto}).
Moreover, $\A$ is lfp.
In fact, $\A \cong \Mod(\theory(\forest^\op))$,
where $\theory(\forest^\op)$ is the cartesian theory
of functors $\forest^\op \to \Set$
(see Example~\ref{ex:prelim:coste:funct}).

We shall now see that the path embeddings of $\A$ are definable in the signature $\Sig(\forest^\op)$
of $\theory(\forest^\op)$.
First, recall that
the paths of $\A$ are finitely presentable
since they are precisely the representable functors, along with the initial object
(Lemma~\ref{l:paths-in-presh-forest}).
Second, we can obtain \emph{finite} formulae that
define the path embeddings of $\A$ in $\Sig(\forest^\op)$.
The case of initial paths is trivial,
since in a topos any arrow from the initial object is monic.
As for non-initial paths,
recall that $\Sig(\forest^\op)$ has one sort $p$ for each $p \in \forest$,
and one function symbol $f_{(q \leq p)} \colon p \to q$ for each $q \leq p$ in $\forest$.
The representable functor
$\forest\funct{-,p}$
corresponds
to the finitely presentable model
$\FG{x:p \mid \True}$. 
Then, for each $p \in \forest$ there is a
finite formula $\FMono[p](x)$ in $\Sig(\forest^\op)$
such that the following are equivalent
in $\Mod(\theory(\forest^\op))$:
\begin{enumerate}[(i)]
\item 
$h \colon \FG{x:p \mid \True} \to M$ is a monomorphism.

\item
\(
  M
  \models
  \FMono[p](c)
\),
where $h$ takes $x$ to $c \in M(p)$.
\end{enumerate}

\noindent
(See \S\ref{sec:emb:presh} below for details.)
Hence condition~\ref{item:hintikka:interp:fin:emb}
of Corollary~\ref{cor:hintikka:interp:fin} holds.
Furthermore, condition~\ref{item:hintikka:interp:fin:path}
of the latter
is equivalent to $\forest$ being finitely branching.
In any case, Theorem~\ref{thm:hintikka:wooded:bfe} applies
and so (identifying $\Mod(\theory(\forest^\op))$ with $\A$) we get
\[
\begin{array}{l l l}
  \text{$M,N \in \Mod(\theory(\forest^\op))$
  equivalent in $\Lang_{\infty}(\Sig(\forest^\op))$}
& \longimp
& \text{$M \bisim N$ in $\A$}.
\end{array}
\]

Arboreal categories of the form $\A = \presh{\forest}$ appear
in the finitely accessible arboreal adjunction of Example~\ref{ex:Diaconescu-cover}.
Let $\cat D$ be a small category, and let
$\forest = \forest(\cat D)$ be the forest consisting of 
finite sequences $s$ of composable arrows
\[
\begin{tikzcd}
  a_0
  \arrow{r}{k_1}
& a_1
  \arrow[dashed]{r}
& a_{n-1}
  \arrow{r}{k_{n}}
& a_n
\end{tikzcd}
\]

\noindent
in $\cat D$,
as in Example~\ref{ex:Diaconescu-cover}.
Via the functor $\pi \colon \forest \to \cat D$
(which takes $s$ as above to $a_n$),
we obtain a finitely accessible arboreal adjunction
\[
\begin{tikzcd}
  \presh{\forest}
  \arrow[bend left=25]{r}{\eadj\pi}
  \arrow[phantom]{r}[description]{\textnormal{\footnotesize{$\bot$}}}
& \presh{\cat D}
  \arrow[bend left=25]{l}{\ladj\pi}
\end{tikzcd}
\]

\noindent
to which Corollary~\ref{cor:hintikka:wooded:R-bfe} applies.
Thus
\[
\begin{array}{l l l}
  \text{$X,Y \in \presh{\cat D}$ equivalent in $\Lang_{\infty}(\Sig(\cat D^\op))$}
& \longimp
& X \bisim_{\ladj\pi} Y,
\end{array}
\]

\noindent
this time identifying $\presh{\cat D}$ with $\Mod(\theory(\cat D^\op))$.
\qed
%%%%%%%%%%%%%%%%%%%%%%%%%%%%%%%%%%%%%%%%%%%%%%%%%%%%%%%%%%%%%%%%%%%%%%%%%%%
\end{example}
%%%%%%%%%%%%%%%%%%%%%%%%%%%%%%%%%%%%%%%%%%%%%%%%%%%%%%%%%%%%%%%%%%%%%%%%%%%

%%%%%%%%%%%%%%%%%%%%%%%%%%%%%%%%%%%%%%%%%%%%%%%%%%%%%%%%%%%%%%%%%%%%%%%%%%%
\section{Formulae for embeddings}
\label{sec:emb}
%%%%%%%%%%%%%%%%%%%%%%%%%%%%%%%%%%%%%%%%%%%%%%%%%%%%%%%%%%%%%%%%%%%%%%%%%%%

Let $\Sig$ be a signature and
consider a finitely accessible wooded adjunction
\[
\begin{tikzcd}
  \C
  \arrow[bend left=25]{r}{\Ladj}
  \arrow[phantom]{r}[description]{\textnormal{\footnotesize{$\bot$}}}
& \E
  \arrow[bend left=25]{l}{R}
\end{tikzcd}
\]

\noindent
where the extensional category $\E$
is of the form $\Mod(\theory)$ for some cartesian theory~$\theory$ in~$\Sig$.
By Corollary~\ref{cor:hintikka:wooded:R-bfe},
if the path embeddings of the wooded category~$\C$ are definable
(in some signature), then
\begin{equation*}
\begin{array}{l l l}
  \text{$M,N \in \E$ equivalent in $\Lang_{\infty}(\Sig)$}
& \longimp
& M \bisim_R N.
\end{array}
\end{equation*}

Our main goal is now to establish the latter implication 
under the assumptions of Theorem~\ref{thm:path:main}, namely
that the adjunction $\Ladj \colon \C \inadj \E \cocolon R$ 
detects path embeddings in the sense of Definition~\ref{def:path:detection-path-emb}:
\begin{itemize}
\item
A morphism $f \colon P \to a$ in $\C$ is an embedding
if, and only if, $\Ladj f \colon \Ladj P \to \Ladj a$
is an embedding of $\Sig$-structures in $\E$.
\end{itemize}

To this end, in this section we show that the embeddings of structures 
in $\E$ are definable in the signature $\Sig$,
provided their domains are finitely presentable
(recall from Proposition~\ref{p:lfp-morphisms-characterisation}
that $\Ladj$ preserves finitely presentable objects since its right
adjoint $R$ is finitary).
We shall thus obtain formulae similar to those
of condition~\ref{item:hintikka:def-path-emb:formulae}
in Definition~\ref{def:hintikka:def-path-emb}:
\begin{itemize}
\item
For each finitely presentable model $\FG{\vec x \mid \varphi}$ of $\E$, 
there is a (possibly infinite)
formula $\FEmb\FG{\vec x | \varphi}(\vec x)$ in $\Sig$
such that for all $M \in \E$
and all $h \colon \FG{\vec x \mid \varphi} \to M$
taking $\vec x$ to $\vec b \in \I{\vec x \mid \varphi}_M$,
\[
\begin{array}{l !{\quad\longiff\quad} l}
  \text{$h$ is an embedding of $\Sig$-structures}
& M \models \FEmb\FG{\vec x \mid \varphi}(\vec b).
\end{array}
\]
\end{itemize}

But this will only lead us halfway between Corollary~\ref{cor:hintikka:wooded:R-bfe}
and Theorem~\ref{thm:path:main},
since given $P, a \in \C$,
we obtain a formula $\FEmb[\Ladj P]$ in $\Sig$ and over $\Ladj a \in \E$,
while Corollary~\ref{cor:hintikka:wooded:R-bfe}
asks for a formula in a signature for $\C$, and over $a \in \C$.
This second step will be made in~\S\ref{sec:wc} with the help
of Hodges' \emph{word-constructions}.

Returning to the formulae $\FEmb\FG{\vec x \mid \varphi}(\vec x)$ defining embeddings
of structures in $\E$,
let us first consider the special case where $\E=\Struct(\sig)$ for a \emph{finite}
(mono-sorted) purely relational signature~$\sig$.
In this case, it is well known that we can define \emph{finite} formulae
$\FEmb\FG{\vec x \mid \varphi}(\vec x)$ by considering appropriate conjunctions
of negated atoms, as we now recall.
To start with, note that  
the structure $\FG{\vec x \mid \varphi}$ in $\Struct(\sig)$,
with $\vec x = x_1,\dots,x_n$,
has carrier $\{\const x_1,\dots,\const x_n\}$
where the elements $\const x_i$'s may not be pairwise distinct
(cf.\ Remark~\ref{rem:finite-carrier-fp-purely-rel-empty-th}).
Now, consider a homomorphism 
\[
h \colon \FG{\vec x \mid \varphi} \to M
\]
in $\Struct(\sig)$ 
and assume that $h$ is induced by $\vec a \in \I{\vec x \mid \varphi}_M$. 
Recall from Lemma~\ref{lem:coste:fp:const} that $h(\const x_i) = a_i$
for each $i = 1,\dots,n$, and so $h \colon \FG{\vec x \mid \varphi} \to M$
is an embedding of $\sig$-structures precisely when for every atomic formula
$\vec x \sorting \atom$, we have
\[
\begin{array}{l l l}
  M \models \atom(a_1,\dots,a_n)
& \longimp
& \FG{\vec x \mid \varphi} \models \atom(\const x_1,\dots,\const x_n).
\end{array}
\]

\noindent
In other words, $h$ is an embedding of $\sig$-structures precisely when
$M \models \FEmb^{\sig}\FG{\vec x \mid \varphi}(\vec a)$,
where $\FEmb^{\sig}\FG{\vec x \mid \varphi}(\vec x)$ is the \emph{finite} formula
\begin{equation}
\label{eq:emb:sig}
  \FEmb^{\sig}\FG{\vec x \mid \varphi}(\vec x)
  ~\deq~
  \mathord{\bigwedge}_{\substack{
  \text{$\vec x \sorting \atom$ atomic}
  \\
  \FG{\vec x \mid \varphi} \models \lnot \atom(\const x_1,\dots,\const x_n)
  }}
  \lnot \atom(\vec x).
\end{equation}

The finiteness of the formula $\FEmb^{\sig}\FG{\vec x \mid \varphi}(\vec x)$
relies on the facts that $\sig$ is a finite relational signature and
the finitely presentable structure $\FG{\vec x \mid \varphi}$ has finite carrier.
If $\sig$ was an infinite (purely relational) signature,
the above formula $\FEmb^{\sig}\FG{\vec x \mid \varphi}(\vec x)$
would still define embeddings, but would be infinite.

However, this approach does not extend to 
arbitrary cartesian theories $\theory$, even in a finite relational signature $\sig$.
For example, $\theory$ could force a relation symbol to be functional,
and this could lead to the structure $\FG{\vec x \mid \varphi}$ being infinite.
What is more, the carrier of $\FG{\vec x \mid \varphi}$ does not admit a simple
description for non-empty theories~$\theory$.
This is the reason why, in the remainder of this section,
we need to adopt a different approach to define the formulae
$\FEmb\FG{\vec x \mid \varphi}(\vec x)$.

We handle the general case of $\E = \Mod(\theory)$ for $\theory$
an arbitrary cartesian theory in~\S\ref{sec:emb:at}.
This invloves some technicalities,
for which we heavily rely on the material of~\S\ref{sec:coste}.
In~\S\ref{sec:emb:presh}, we use the machinery of \S\ref{sec:emb:at}
to provide the missing proofs for Example~\ref{ex:hintikka:mono:forest}.

The apparent technical gap between the cases of $\Struct(\sig)$ and $\Mod(\theory)$
will be discussed in~\S\ref{sec:fact} below,
where we show that that even 
when $\theory$ is a cartesian theory in a finite relational signature $\sig$,
an adjunction $\C \inadj \Mod(\theory)$ as in Theorem~\ref{thm:path:main}
may not reduce to the simpler case of $\C \inadj \Struct(\sig)$.

%%%%%%%%%%%%%%%%%%%%%%%%%%%%%%%%%%%%%%%%%%%%%%%%%%%%%%%%%%%%%%%%%%%%%%%%%%%
\subsection{$\At$-reflecting homomorphisms}
\label{sec:emb:at}
%%%%%%%%%%%%%%%%%%%%%%%%%%%%%%%%%%%%%%%%%%%%%%%%%%%%%%%%%%%%%%%%%%%%%%%%%%%
Fix a cartesian theory $\theory$ in a signature $\Sig$,
and write $\cat T$ for the syntactic category of $\theory$.
Let $\At$ be a set of atomic formulae-in-context,
i.e.\ each element of $\At$ is a formula-in-context $\vec x \sorting \atom$
where $\atom$ is atomic (cf.~\S\ref{sec:prelim:struct}).
We say that a homomorphism $h \in \Mod(\theory)\funct{M,N}$
is \emph{$\At$-reflecting} when for all $(\vec x \sorting \atom) \in \At$,
\[
\begin{array}{l l l}
  h(\vec a) \in \I{\vec x \mid \atom}_N
& \longimp
& \vec a \in \I{\vec x \mid \atom}_M.
\end{array}
\]

%%%%%%%%%%%%%%%%%%%%%%%%%%%%%%%%%%%%%%%%%%%%%%%%%%%%%%%%%%%%%%%%%%%%%%%%%%%
\begin{example}
\label{ex:emb:at:mono}
%%%%%%%%%%%%%%%%%%%%%%%%%%%%%%%%%%%%%%%%%%%%%%%%%%%%%%%%%%%%%%%%%%%%%%%%%%%
If $\At$ consists of all (sorted) formulae of the form
$(x,y \sorting x \Eq y)$,
then the $\At$-reflecting homomorphisms of $\Mod(\theory)$
are exactly the monomorphisms.
Just observe that the full inclusion $\Mod(\theory) \into \Struct(\Sig)$
preserves and reflects monos
by Lemma~\ref{lem:coste:mod:filtcolim},
and the monos of $\Struct(\Sig)$
are exactly the homomorphisms that are sortwise injective 
(see Example~\ref{ex:prelim:fact:struct}).
%%%%%%%%%%%%%%%%%%%%%%%%%%%%%%%%%%%%%%%%%%%%%%%%%%%%%%%%%%%%%%%%%%%%%%%%%%%
\end{example}
%%%%%%%%%%%%%%%%%%%%%%%%%%%%%%%%%%%%%%%%%%%%%%%%%%%%%%%%%%%%%%%%%%%%%%%%%%%

%%%%%%%%%%%%%%%%%%%%%%%%%%%%%%%%%%%%%%%%%%%%%%%%%%%%%%%%%%%%%%%%%%%%%%%%%%%
\begin{example}
\label{ex:emb:at:emb}
%%%%%%%%%%%%%%%%%%%%%%%%%%%%%%%%%%%%%%%%%%%%%%%%%%%%%%%%%%%%%%%%%%%%%%%%%%%
If $\At$ consists of all atomic formulae-in-context,
then the $\At$-reflecting morphisms of $\Mod(\theory)$
are exactly the embeddings of $\Sig$-structures,
i.e.\ the homomorphisms whose image under the inclusion
$\Mod(\theory) \into \Struct(\Sig)$ is a strong monomorphism.

In fact, as we shall see in Remark \ref{rem:emb:strong},
every strong monomorphism $m$ in $\Mod(\theory)$ is an embedding of $\Sig$-structures.
But recall from
Remark~\ref{rem:path:structemb}.\ref{item:path:structemb:strong}
that the converse is not true,
even if we assume that the domain of $m$ is finitely presentable.
%%%%%%%%%%%%%%%%%%%%%%%%%%%%%%%%%%%%%%%%%%%%%%%%%%%%%%%%%%%%%%%%%%%%%%%%%%%
\end{example}
%%%%%%%%%%%%%%%%%%%%%%%%%%%%%%%%%%%%%%%%%%%%%%%%%%%%%%%%%%%%%%%%%%%%%%%%%%%

%%%%%%%%%%%%%%%%%%%%%%%%%%%%%%%%%%%%%%%%%%%%%%%%%%%%%%%%%%%%%%%%%%%%%%%%%%%
\begin{notation}
\label{not:formemb}
%%%%%%%%%%%%%%%%%%%%%%%%%%%%%%%%%%%%%%%%%%%%%%%%%%%%%%%%%%%%%%%%%%%%%%%%%%%
Given $\OI{\vec x \mid \varphi} \in \cat T$,
let
\[
\begin{array}{l l l}
  \iota_\varphi(\vec x,\vec x')
& \deq
& \varphi(\vec x) ~\land~ \vec x \Eq \vec x'
  ~.
\end{array}
\]

\noindent
This induces a morphism
\(
  \MI{\vec x, \vec x' \mid \iota_\varphi}
  \colon
  \OI{\vec x \mid \varphi}
  \to
  \OI{\vec x \mid \True}
\).
We write
$\imath_{\varphi} \colon \FG{\vec x \mid \True} \to \FG{\vec x \mid \varphi}$
for the homomorphism in $\Mod(\theory)$ corresponding,
under Theorem~\ref{thm:coste:synt},
to the natural transformation
$\yoneda \MI{\vec x, \vec x' \mid \iota_\varphi}$
(recall that $\yoneda$ is contravariant).
%%%%%%%%%%%%%%%%%%%%%%%%%%%%%%%%%%%%%%%%%%%%%%%%%%%%%%%%%%%%%%%%%%%%%%%%%%%
\end{notation}
%%%%%%%%%%%%%%%%%%%%%%%%%%%%%%%%%%%%%%%%%%%%%%%%%%%%%%%%%%%%%%%%%%%%%%%%%%%

%%%%%%%%%%%%%%%%%%%%%%%%%%%%%%%%%%%%%%%%%%%%%%%%%%%%%%%%%%%%%%%%%%%%%%%%%%%
\begin{remark} 
\label{rem:emb:epi}
%%%%%%%%%%%%%%%%%%%%%%%%%%%%%%%%%%%%%%%%%%%%%%%%%%%%%%%%%%%%%%%%%%%%%%%%%%%
The arrows
$\imath_{\varphi} \colon \FG{\vec x \mid \True} \to \FG{\vec x \mid \varphi}$
are epimorphisms in $\Mod(\theory)$.
This follows from the fact that the corresponding natural transformations
$\yoneda\iota_\varphi \colon \yoneda\OI{\vec y \mid \True}
\to
\yoneda\OI{\vec y \mid \varphi}$
are epimorphisms in $\lex\funct{\cat T,\Set}$.
That is, whenever $M \in \Mod(\theory)$ and
$\lambda,\kappa \colon \yoneda\OI{\vec y \mid \varphi} \to F_M$
are such that
\begin{equation}
\label{eq:emb:strong}
\begin{array}{*{3}{l}}
  \lambda \comp \yoneda\iota_\varphi
& =
& \kappa \comp \yoneda\iota_\varphi,
\end{array}
\end{equation}
we have $\lambda = \kappa$.
To see this, set
$\vec a= \lambda_{\OI{\vec y \mid \varphi}}(\id)$
and
$\vec b= \kappa_{\OI{\vec y \mid \varphi}}(\id)$.
With the notations of~\S\ref{sec:coste:fp} we have
$\lambda = F_M(-)(\vec a)$
and
$\kappa = F_M(-)(\vec b)$,
and thus eq.~\eqref{eq:emb:strong}
entails
\[
\begin{array}{l l l}
  F_M(\iota_\varphi)(\vec a)
& =
& F_M(\iota_\varphi)(\vec b).
\end{array}
\]
But 
$F_M(\iota_\varphi) \colon \I{\vec y \mid \varphi}_M \to \I{\vec y \mid \True}_M$
is
the function with graph
\[
\begin{tikzcd}
  \I{\vec y,\vec y' \mid \varphi(\vec y) ~\land~ \vec y \Eq \vec y'}_M
  \arrow[hook]{r}
& \I{\vec y \mid \varphi}_M
  \times
  \I{\vec y \mid \True}_M,
\end{tikzcd}
\]
hence $\vec a = \vec b$ and so $\lambda = \kappa$.
%%%%%%%%%%%%%%%%%%%%%%%%%%%%%%%%%%%%%%%%%%%%%%%%%%%%%%%%%%%%%%%%%%%%%%%%%%%
\end{remark}
%%%%%%%%%%%%%%%%%%%%%%%%%%%%%%%%%%%%%%%%%%%%%%%%%%%%%%%%%%%%%%%%%%%%%%%%%%%

%%%%%%%%%%%%%%%%%%%%%%%%%%%%%%%%%%%%%%%%%%%%%%%%%%%%%%%%%%%%%%%%%%%%%%%%%%%
\begin{proposition}
\label{prop:emb:formemb}
%%%%%%%%%%%%%%%%%%%%%%%%%%%%%%%%%%%%%%%%%%%%%%%%%%%%%%%%%%%%%%%%%%%%%%%%%%%
The following are equivalent
for all $h \in \Mod(\theory)\funct{N,M}$:
\begin{enumerate}[(i)]
\item 
\label{item:formemb:hom}
$h$ is $\At$-reflecting.

\item
\label{item:formemb:diag}
For each $(\vec y \sorting \atom) \in \At$
and each commutative square as in the left-hand side below,
\begin{equation}
\label{eq:formemb:diag:ass}
\begin{array}{c !{\qquad} c}

\begin{tikzcd}
  \FG{\vec y \mid \True}
  \arrow{r}[above]{\imath_\atom}
  \arrow{d}[left]{\ell}
& \FG{\vec y \mid \atom}
  \arrow{d}[right]{k}
\\
  N
  \arrow{r}[below]{h}
& M
\end{tikzcd}

&

\begin{tikzcd}
  \FG{\vec y \mid \True}
  \arrow{r}[above]{\imath_\atom}
  \arrow{d}[left]{\ell}
& \FG{\vec y \mid \atom}
  \arrow{dl}[right, xshift=5pt]{d}
\\
  N
\end{tikzcd}

\end{array}
\end{equation}

\noindent
there exists $d \colon \FG{\vec y \mid \atom} \to N$
such that the rightmost triangle above commutes.%
\footnote{In view of Remark~\ref{rem:emb:epi},
this is equivalent to saying that $d$ is a diagonal filler.}
\end{enumerate}
%%%%%%%%%%%%%%%%%%%%%%%%%%%%%%%%%%%%%%%%%%%%%%%%%%%%%%%%%%%%%%%%%%%%%%%%%%%
\end{proposition}
%%%%%%%%%%%%%%%%%%%%%%%%%%%%%%%%%%%%%%%%%%%%%%%%%%%%%%%%%%%%%%%%%%%%%%%%%%%

Before proving Proposition~\ref{prop:emb:formemb}, let us make a couple of comments.

%%%%%%%%%%%%%%%%%%%%%%%%%%%%%%%%%%%%%%%%%%%%%%%%%%%%%%%%%%%%%%%%%%%%%%%%%%%
\begin{remark}
\label{rem:emb:strong}
%%%%%%%%%%%%%%%%%%%%%%%%%%%%%%%%%%%%%%%%%%%%%%%%%%%%%%%%%%%%%%%%%%%%%%%%%%%
It follows from Proposition~\ref{prop:emb:formemb} that any strong monomorphism in
$\Mod(\theory)$ is an embedding of $\Sig$-structures.
Just observe that, if $h$ is a strong monomorphism and $\atom$ is any atomic formula,
each commutative square as in the left-hand side of eq.~\eqref{eq:formemb:diag:ass}
admits a diagonal filler because $\imath_\atom$ is an epimorphism
(see Remark~\ref{rem:emb:epi}).
%%%%%%%%%%%%%%%%%%%%%%%%%%%%%%%%%%%%%%%%%%%%%%%%%%%%%%%%%%%%%%%%%%%%%%%%%%%
\end{remark}
%%%%%%%%%%%%%%%%%%%%%%%%%%%%%%%%%%%%%%%%%%%%%%%%%%%%%%%%%%%%%%%%%%%%%%%%%%%

%%%%%%%%%%%%%%%%%%%%%%%%%%%%%%%%%%%%%%%%%%%%%%%%%%%%%%%%%%%%%%%%%%%%%%%%%%%
\begin{remark}[Comparison with $\omega$-pure morphisms]
\label{rem:emb:pure}
%%%%%%%%%%%%%%%%%%%%%%%%%%%%%%%%%%%%%%%%%%%%%%%%%%%%%%%%%%%%%%%%%%%%%%%%%%%
In the setting of Example~\ref{ex:emb:at:emb},
condition~\ref{item:formemb:diag}
in Proposition~\ref{prop:emb:formemb}
is reminiscent of the notion of ($\omega$-)\emph{pure morphism};
see e.g.~\cite[Proposition 5.15, \S 5A]{ar94book}.

Recall from~\cite[\S 2D]{ar94book}
that a homomorphism of structures $h \colon N \to M$
is \emph{$\omega$-pure}
if for every commutative square as in the
%l.-h.s.\
left-hand side
below
\[
\begin{array}{c !{\qquad} c}

\begin{tikzcd}
  \FG{\vec y \mid \psi}
  \arrow{r}
  \arrow{d}
& \FG{\vec x \mid \varphi}
  \arrow{d}
\\
  N
  \arrow{r}[below]{h}
& M
\end{tikzcd}

&

\begin{tikzcd}
  \FG{\vec y \mid \psi}
  \arrow{r}
  \arrow{d}
& \FG{\vec x \mid \varphi}
  \arrow{dl}[right, xshift=5pt]{d}
\\
  N
\end{tikzcd}

\end{array}
\]

\noindent
there is a
$d \colon \FG{\vec x \mid \varphi} \to N$
such that the
triangle in the right-hand side above commutes.

Every $\omega$-pure morphism is a regular monomorphism
\cite[Proposition 2.31]{ar94book}.
Moreover, by~\cite[Proposition 5.15]{ar94book},
a homomorphism $h$ as above is $\omega$-pure if and only if
it (preserves and) reflects \emph{positive-primitive}
formulae, i.e.\ formulae of the form
$(\exists \vec y)\varphi(\vec x,\vec y)$
where $\varphi$ is a finite conjunction of atomic formulae.
In fact, a direct inspection of the proof of~\cite[Proposition 5.15]{ar94book}
shows that one only has to consider commutative squares of the form
\[
\begin{tikzcd}
  \FG{\vec x \mid \True}
  \arrow{r}
  \arrow{d}
& \FG{\vec x, \vec y \mid \varphi}
  \arrow{d}
\\
  N
  \arrow{r}[below]{h}
& M
\end{tikzcd}
\]

\noindent
where $\varphi(\vec x,\vec y)$ is a (finite) conjunction of atomic formulae.

In the setting of Example~\ref{ex:emb:at:emb},
it follows from Proposition~\ref{prop:emb:formemb}
that $\omega$-pure morphisms are embeddings of structures.
But note that Proposition~\ref{prop:emb:formemb}
only considers commutative squares whose top
arrows are of the form
$\FG{\vec x \mid \True}  \to \FG{\vec x \mid \atom}$,
where the two structures 
$\FG{\vec x \mid \True}$ and $\FG{\vec x \mid \atom}$
have the same generators
and $\atom$ is atomic.
%%%%%%%%%%%%%%%%%%%%%%%%%%%%%%%%%%%%%%%%%%%%%%%%%%%%%%%%%%%%%%%%%%%%%%%%%%%
\end{remark}
%%%%%%%%%%%%%%%%%%%%%%%%%%%%%%%%%%%%%%%%%%%%%%%%%%%%%%%%%%%%%%%%%%%%%%%%%%%

The proof of Proposition~\ref{prop:emb:formemb}
is similar in spirit to that of~\cite[Proposition 5.15]{ar94book}
for $\omega$-pure morphisms in $\Struct(\Sig)$.
But in the case of $\Mod(\theory)$,
we have little direct information about the elements
of finitely presentable models $\FG{\vec x \mid \varphi}$.
Instead, we rely on Theorem~\ref{thm:coste:synt}
and the tools of~\S\ref{sec:coste:fp},
in particular Lemmas~\ref{lem:coste:lp:hom}
and~\ref{lem:coste:fp:triangle}.

Recall from~\S\ref{sec:coste:synt}
that given $L \in \Mod(\theory)$ and
\(
  \MI{\vec x,\vec y \mid \theta}
  \in
  \cat T\funct{\OI{\vec x \mid \varphi},\OI{\vec y \mid \psi}}
\),
the function
\(
  F_L(\MI{\vec x,\vec y \mid \theta})
  \colon
  \I{\vec x \mid \varphi}_L \to \I{\vec y \mid \psi}_L
\)
has graph
\[
\begin{tikzcd}
  \I{\vec x,\vec y \mid \theta}_L
  \arrow[hook]{r}
& \I{\vec x \mid \varphi}_L
  \times
  \I{\vec y \mid \psi}_L.
\end{tikzcd}
\]

\noindent
In particular,
for a morphism $\MI{\vec x,\vec x' \mid \iota_\atom}$
as in Notation~\ref{not:formemb}, we simply have
\[
\begin{array}{l l c c c}
  F_L(\MI{\vec x,\vec x' \mid \iota_\atom})
& :
& \I{\vec y \mid \atom}_L
& \longto
& \I{\vec y \mid \True}_L
\\

&
& \vec b
& \longmapsto
& \vec b
\end{array}
\]

%%%%%%%%%%%%%%%%%%%%%%%%%%%%%%%%%%%%%%%%%%%%%%%%%%%%%%%%%%%%%%%%%%%%%%%%%%%
\begin{proof}[Proof of Proposition~\ref{prop:emb:formemb}]
%%%%%%%%%%%%%%%%%%%%%%%%%%%%%%%%%%%%%%%%%%%%%%%%%%%%%%%%%%%%%%%%%%%%%%%%%%%
  $\text{\ref{item:formemb:hom}}
  \imp
  \text{\ref{item:formemb:diag}}
$.
Suppose that $h \colon N \to M$ is $\At$-reflecting
and consider
a commutative square as in the left-hand side of eq.~\eqref{eq:formemb:diag:ass}:
\[
\begin{tikzcd}
  \FG{\vec y \mid \True}
  \arrow{r}[above]{\imath_\atom}
  \arrow{d}[left]{\ell}
& \FG{\vec y \mid \atom}
  \arrow{d}[right]{k}
\\
  N
  \arrow{r}[below]{h}
& M
\end{tikzcd}
\]

\noindent
Assume that $\ell$ takes $\vec y$ to $\vec a \in \I{\vec y \mid \True}_N$,
and $k$ takes $\vec y$ to $\vec b \in \I{\vec y \mid \atom}_M$.
We have to provide 
a $d \colon \FG{\vec y \mid \atom} \to N$ such that
$d \comp \imath_\atom = \ell$.
By Lemma~\ref{lem:coste:fp:triangle},
this amounts to showing that
\[
\begin{array}{l l l}
  N
& \models
& (\exists \vec y)
  \iota_\atom(\vec y,\vec a).
\end{array}
\]

\noindent
By definition of the formula $\iota_\atom$
(Notation~\ref{not:formemb}),
this is equivalent to proving that
${\vec a \in \I{\vec y \mid \atom}_N}$.
Hence, 
since $\vec b \in \I{\vec y \mid \atom}_M$ and
$h$ is $\At$-reflecting, 
it suffices to show that $h(\vec a) = \vec b$.

We reformulate the above diagram
in terms of the natural transformations given by 
Theorem~\ref{thm:coste:synt}.
We already know that $\imath_\atom$ is induced by the natural
transformation
\[
\begin{array}{*{5}{l}}
  \yoneda\iota_\atom
& :
& \yoneda\OI{\vec y \mid \True}
& \longto
& \yoneda\OI{\vec y \mid \atom}.
\end{array}
\]

\noindent
If
$\eta \colon F_N \to F_M$
is the natural transformation corresponding to
$h \colon N \to M$,
we get the following commutative diagram.
\begin{equation}
\label{eq:formemb:diag:nat}
\begin{tikzcd}[column sep=large]
  \yoneda\OI{\vec y \mid \True}
  \arrow{r}[above]{\yoneda\iota_\atom}
  \arrow{d}[left]{F_N(-)(\vec a)}
& \yoneda\OI{\vec y \mid \atom}
  \arrow{d}[right]{F_M(-)(\vec b)}
\\
  F_N
  \arrow{r}[below]{\eta}
& F_M
\end{tikzcd}
\end{equation}

Reasoning similarly as in the proof of Lemma~\ref{lem:coste:fp:triangle},
note that the upper path in eq.~\eqref{eq:formemb:diag:nat}
coincides with $F_M(-)(F_M(\iota_\atom)(\vec b))$.
Now, taking components at $\OI{\vec y \mid \True}$
and evaluating at the identity on $\OI{\vec y \mid \True}$,
the upper path 
yields $F_M(\iota_\atom)(\vec b)$ ($= \vec b$),
while the lower path yields
$\eta_{\OI{\vec y \mid \True}}(\vec a)$.
Hence
$\eta_{\OI{\vec y \mid \True}}(\vec a) = \vec b$,
and
we get $h(\vec a) = \vec b$
since by Lemma~\ref{lemma:coste:synt:homnat}
we have
\[
\begin{array}{*{5}{l}}
  h(\vec a)
& =
& \eta_{\OI{\vec y \mid \True}}(\vec a)
& \in
& \I{\vec y \mid \True}_M.
\end{array}
\]

$\text{\ref{item:formemb:diag}}
  \imp
  \text{\ref{item:formemb:hom}}$.
Consider a homomorphism $h \colon N \to M$
and a formula $(\vec y \sorting \atom) \in \At$.
Let
$\vec a \in \I{\vec y \mid \True}_N$
be such that
$h(\vec a) \in \I{\vec y \mid \atom}_M$.
We have to show that $\vec a \in \I{\vec y \mid \atom}_N$.

Write $\vec b$ for $h(\vec a)$
and let
$\eta \colon F_N \to F_M$ be the natural transformation
corresponding to $h \colon N \to M$.
We first show that diagram \eqref{eq:formemb:diag:nat} above commutes
with the present values of $\vec a$, $\vec b$ and $\eta$.
Thanks to the Yoneda lemma, we just have to show
that the two paths below
\[
\begin{tikzcd}[column sep=large]
  \yoneda\OI{\vec y \mid \True}\OI{\vec y \mid \True}
  \arrow{r}[above]{(\yoneda\iota_\atom)_{\OI{\vec y \mid \True}}}
  \arrow{d}[left]{F_N(-)(\vec a)}
& \yoneda\OI{\vec y \mid \atom}\OI{\vec y \mid \True}
  \arrow{d}[right]{F_M(-)(\vec b)}
\\
  F_N\OI{\vec y \mid \True}
  \arrow{r}[below]{\eta_{\OI{\vec y \mid \True}}}
& F_M\OI{\vec y \mid \True}
\end{tikzcd}
\]

\noindent
are equal on $\id \in \cat T\funct{\OI{\vec y \mid \True}, \OI{\vec y \mid \True}}$.
The upper path yields $F_M(\iota_\atom)(\vec b) = \vec b$,
while the lower path yields $\eta_{\OI{\vec y \mid \True}}(\vec a)$.
Hence we are done since 
$\eta_{\OI{\vec y \mid \True}}(\vec a) = h(\vec a) = \vec b$
by Lemma~\ref{lemma:coste:synt:homnat}.

Now, assumption~\ref{item:formemb:diag} gives some
$\delta \colon \yoneda\OI{\vec y \mid \atom} \to F_N$
such that
$\delta \comp \yoneda\iota_\atom = F_N(-)(\vec a)$.
By 
Lemma~\ref{lem:coste:fp:triangle},
this means that there is a tuple
$\vec c \in \I{\vec y \mid \atom}_N$
such that
$N \models \iota_\atom(\vec c, \vec a)$,
that is
$\vec a = \vec c \in \I{\vec y \mid \atom}_N$.
%%%%%%%%%%%%%%%%%%%%%%%%%%%%%%%%%%%%%%%%%%%%%%%%%%%%%%%%%%%%%%%%%%%%%%%%%%%
\end{proof}
%%%%%%%%%%%%%%%%%%%%%%%%%%%%%%%%%%%%%%%%%%%%%%%%%%%%%%%%%%%%%%%%%%%%%%%%%%%

Combining Proposition~\ref{prop:emb:formemb}
with Corollary~\ref{cor:coste:fp:triangle}
we obtain formulae that
detect when a homomorphism
with finitely presentable domain is $\At$-reflecting:

%%%%%%%%%%%%%%%%%%%%%%%%%%%%%%%%%%%%%%%%%%%%%%%%%%%%%%%%%%%%%%%%%%%%%%%%%%%
\begin{figure}
%%%%%%%%%%%%%%%%%%%%%%%%%%%%%%%%%%%%%%%%%%%%%%%%%%%%%%%%%%%%%%%%%%%%%%%%%%%
\begin{multline*}
  \FAt\FG{\vec x \mid \varphi}(\vec x)
  \deq
\\
  \hfill
  \bigwedge_{(\vec y \sorting \atom) \in \At}
  \
  \bigwedge_{
    \theta \in \cat T\funct{\OI{\vec x \mid \varphi}, \OI{\vec y \mid \True}}
  }
  (\forall \vec y)
  \left(
  \big( \atom(\vec y) ~\land~ \theta(\vec x,\vec y) \big)
  ~\limp~
  \mathord{\bigvee}_{
    \delta \in \cat T\funct{\OI{\vec x \mid \varphi}, \OI{\vec y \mid \atom}}
  }
  \delta(\vec x,\vec y)
  \right)
\end{multline*}
\caption{Formulae for $\At$-reflecting homomorphisms 
with finitely presentable domains (Corollary~\ref{cor:formemb}). 
\label{fig:formemb}}
\end{figure}

%%%%%%%%%%%%%%%%%%%%%%%%%%%%%%%%%%%%%%%%%%%%%%%%%%%%%%%%%%%%%%%%%%%%%%%%%%%
\begin{corollary}
\label{cor:formemb}
%%%%%%%%%%%%%%%%%%%%%%%%%%%%%%%%%%%%%%%%%%%%%%%%%%%%%%%%%%%%%%%%%%%%%%%%%%%
In $\Mod(\theory)$,
assume that $h \colon \FG{\vec x \mid \varphi} \to M$
is induced by $\vec a \in \I{\vec x \mid \varphi}_M$.
The following statements are equivalent:
\begin{enumerate}[(i)]
\item 
$h$ is $\At$-reflecting.

\item
\(
  M
  \models
  \FAt\FG{\vec x \mid \varphi}(\vec a)
\)
where the (possibly infinite) formula
$\FAt\FG{\vec x \mid \varphi}(\vec x)$
is defined as in Figure~\ref{fig:formemb}.
\end{enumerate}
%%%%%%%%%%%%%%%%%%%%%%%%%%%%%%%%%%%%%%%%%%%%%%%%%%%%%%%%%%%%%%%%%%%%%%%%%%%
\end{corollary}
%%%%%%%%%%%%%%%%%%%%%%%%%%%%%%%%%%%%%%%%%%%%%%%%%%%%%%%%%%%%%%%%%%%%%%%%%%%

%%%%%%%%%%%%%%%%%%%%%%%%%%%%%%%%%%%%%%%%%%%%%%%%%%%%%%%%%%%%%%%%%%%%%%%%%%%
\subsubsection{Embeddings of structures}
\label{sec:emb:emb}
%%%%%%%%%%%%%%%%%%%%%%%%%%%%%%%%%%%%%%%%%%%%%%%%%%%%%%%%%%%%%%%%%%%%%%%%%%%
Corollary~\ref{cor:formemb} yields, in particular,
formulae that detect when a homomorphism in $\Mod(\theory)$
with finitely presentable domain is an embedding of $\Sig$-structures.
We state this result explicitly for future reference.

%%%%%%%%%%%%%%%%%%%%%%%%%%%%%%%%%%%%%%%%%%%%%%%%%%%%%%%%%%%%%%%%%%%%%%%%%%%
\begin{notation}
%%%%%%%%%%%%%%%%%%%%%%%%%%%%%%%%%%%%%%%%%%%%%%%%%%%%%%%%%%%%%%%%%%%%%%%%%%%
When $\At$ consists of all atomic formulae, we write
$\FEmb\FG{\vec x \mid \varphi}(\vec x)$
for the formula $\FAt\FG{\vec x \mid \varphi}(\vec x)$
of Corollary~\ref{cor:formemb} (see Figure~\ref{fig:formemb}).
%%%%%%%%%%%%%%%%%%%%%%%%%%%%%%%%%%%%%%%%%%%%%%%%%%%%%%%%%%%%%%%%%%%%%%%%%%%
\end{notation}
%%%%%%%%%%%%%%%%%%%%%%%%%%%%%%%%%%%%%%%%%%%%%%%%%%%%%%%%%%%%%%%%%%%%%%%%%%%

With this notation, we have

%%%%%%%%%%%%%%%%%%%%%%%%%%%%%%%%%%%%%%%%%%%%%%%%%%%%%%%%%%%%%%%%%%%%%%%%%%%
\begin{corollary}
\label{cor:emb:emb}
%%%%%%%%%%%%%%%%%%%%%%%%%%%%%%%%%%%%%%%%%%%%%%%%%%%%%%%%%%%%%%%%%%%%%%%%%%%
In $\Mod(\theory)$,
assume that $h \colon \FG{\vec x \mid \varphi} \to M$
is induced by $\vec a \in \I{\vec x \mid \varphi}_M$.
The following statements are equivalent:
\begin{enumerate}[(i)]
\item 
$h$ is an embedding of $\Sig$-structures.

\item
\(
  M
  \models
  \FEmb\FG{\vec x \mid \varphi}(\vec a)
\).
\end{enumerate}
%%%%%%%%%%%%%%%%%%%%%%%%%%%%%%%%%%%%%%%%%%%%%%%%%%%%%%%%%%%%%%%%%%%%%%%%%%%
\end{corollary}
%%%%%%%%%%%%%%%%%%%%%%%%%%%%%%%%%%%%%%%%%%%%%%%%%%%%%%%%%%%%%%%%%%%%%%%%%%%

%%%%%%%%%%%%%%%%%%%%%%%%%%%%%%%%%%%%%%%%%%%%%%%%%%%%%%%%%%%%%%%%%%%%%%%%%%%
\subsection{Presheaves over a forest}
\label{sec:emb:presh}
%%%%%%%%%%%%%%%%%%%%%%%%%%%%%%%%%%%%%%%%%%%%%%%%%%%%%%%%%%%%%%%%%%%%%%%%%%%
In general, the formulae $\FEmb\FG{\vec x \mid \varphi}(\vec x)$ in 
Corollary~\ref{cor:emb:emb} are infinite. 
In this section, we show that in the setting of
Example~\ref{ex:hintikka:mono:forest},
concerning arboreal categories of presheaves over a forest,
these formulae yield \emph{finite} formulae defining monomorphisms. 
Fix a forest $\forest$
and consider the arboreal category~$\presh\forest$.
Recall from Example~\ref{ex:prelim:coste:funct}
the cartesian theory $\theory(\forest^\op)$ 
of functors $\forest^\op \to \Set$,
so that
\begin{equation}
\label{eq:emb:presh:presh-mod}
\begin{array}{l l l}
  \Mod(\theory(\forest^\op))
& \cong
& \presh{\forest}.
\end{array}
\end{equation}

\noindent
The signature $\Sig(\forest^\op)$ of $\theory(\forest^\op)$ 
has one sort $p$ for each element $p \in \forest$,
and one function symbol $f_{(q \leq p)} \colon p \to q$
for each pair of elements $q \leq p$ in $\forest$.

We claim
that for each finitely presentable model $\FG{\vec x \mid \varphi}$ 
there exists a
finite formula $\FMono\FG{\vec x \mid \varphi}(\vec x)$ in $\Sig(\forest^\op)$
such that the following are equivalent
in $\Mod(\theory(\forest^\op))$:
\begin{enumerate}[(i)]
\item 
$h \colon \FG{\vec x \mid \varphi} \to M$ is a monomorphism.

\item
\(
  M
  \models
  \FMono\FG{\vec x \mid \varphi}(\vec c)
\),
where $h$ takes $\vec x$ to $\vec c$.
\end{enumerate}

Fix an arbitrary $p \in \forest$. To start with, we show that the presheaf $\yoneda p \in \presh\forest$ corresponds
under eq.~\eqref{eq:emb:presh:presh-mod} to the finitely presentable model
$\FG{x:p \mid \True}$.
Since $\yoneda p$ is finitely presentable in $\presh\forest$, it corresponds under eq.~\eqref{eq:emb:presh:presh-mod} to a finitely presentable object of $\Mod(\theory(\forest^\op))$, say $\FG{\vec y \mid \psi}$. 
Now, consider a model $M\in \Mod(\theory(\forest^\op))$. We can assume without loss of generality that $M$ is of the form $M_{X}$ for some presheaf $X\in \presh\forest$, where
$M_{X}(p) \coloneqq X(p)$ for each
$p \in \Sort(\Sig(\forest^\op))$. 
Then we have
\begin{align*}
\Mod(\theory(\forest^\op)) [\FG{\vec y \mid \psi}, M_{X}] &\cong \presh\forest [\yoneda p, X] \\
&\cong X(p) \\
&\cong M_{X}(p) \\
&\cong \I{x :p \mid \True}_{M_{X}} \\
& \cong \Mod(\theory(\forest^\op)) [\FG{x :p \mid \True}, M_{X}].
\end{align*}
This chain of natural isomorphisms induces an isomorphism of representable functors 
\[
\Mod(\theory(\forest^\op)) [\FG{\vec y \mid \psi}, - ] \cong \Mod(\theory(\forest^\op)) [\FG{x :p \mid \True}, -]
\]
and thus we get $\FG{\vec y \mid \psi} \cong \FG{x :p \mid \True}$, showing that the presheaf $\yoneda p \in \presh\forest$ corresponds
to the finitely presentable model
$\FG{x:p \mid \True}$.

Now, write $\cat T$ for the syntactic category of $\theory(\forest^\op)$. 
The following lemma is the key to obtaining the desired finite formula $\FMono\FG{\vec x \mid \varphi}$:

%%%%%%%%%%%%%%%%%%%%%%%%%%%%%%%%%%%%%%%%%%%%%%%%%%%%%%%%%%%%%%%%%%%%%%%%%%%
\begin{lemma}
\label{lem:finite-hom-sets-synt-cat}
%%%%%%%%%%%%%%%%%%%%%%%%%%%%%%%%%%%%%%%%%%%%%%%%%%%%%%%%%%%%%%%%%%%%%%%%%%%
The following statements hold in $\cat T$:
\begin{enumerate}[(i)]
\item\label{i:singleton-empty-repr} The hom-set $\cat T\funct{\OI{x:p \mid \True} \,,\, \OI{y:q \mid \True}}$
has at most one element, and is empty precisely when $q \not\leq p$ in $\forest$.
\item\label{i:finite-empty-gen} The hom-set $\cat T\funct{\OI{x_{1}:p_{1},\ldots, x_{n}:p_{n}\mid \varphi} \,,\, \OI{\vec y:q \mid \True}}$ is finite, and it is empty whenever $q\not\leq p_{i}$ in $\forest$ for all $i\in\{1,\ldots,n\}$.
\end{enumerate}
%%%%%%%%%%%%%%%%%%%%%%%%%%%%%%%%%%%%%%%%%%%%%%%%%%%%%%%%%%%%%%%%%%%%%%%%%%%
\end{lemma}
%%%%%%%%%%%%%%%%%%%%%%%%%%%%%%%%%%%%%%%%%%%%%%%%%%%%%%%%%%%%%%%%%%%%%%%%%%%

\begin{proof}
\ref{i:singleton-empty-repr}
Recalling from \S\ref{sec:coste:fp} that the finitely presentable model
$\FG{\vec x \mid \varphi}$
corresponds (under Theorem~\ref{thm:coste:synt}) to the representable functor
$\yoneda \OI{\vec x \mid \varphi} \in \lex\funct{\cat T,\Set}$,
the Yoneda lemma entails that, for each $q \in \forest$,
\begin{align*}
  \cat T\funct{\OI{x:p \mid \True} \,,\, \OI{y:q \mid \True}}
& \cong
  \lex\funct{\cat T,\Set} [\yoneda \OI{y:q \mid \True}, \yoneda \OI{x:p \mid \True}]
\\
& \cong \Mod(\theory(\forest^\op)) [\FG{y :q \mid \True}, \FG{x :p \mid \True}]
\\
& \cong \presh\forest [\yoneda q, \yoneda p]
\\
& \cong \forest\funct{q,p}.
\end{align*}

\noindent
Hence, the hom-set
$\cat T\funct{\OI{x:p \mid \True} \,,\, \OI{y:q \mid \True}}$
has at most one element, and is empty precisely when $q \not\leq p$ in $\forest$.

\ref{i:finite-empty-gen}
Consider the contexts $\vec x = x_{1}:p_{1},\ldots, x_{n}:p_{n}$
and $\vec y = y_{1},\ldots, y_{m}:q$.
We consider, dually, the hom-set
$\Mod(\theory(\forest^\op)) [\FG{\vec y:q \mid \True}, \FG{\vec x \mid \varphi}]$.
Recall from Notation~\ref{not:formemb} the homomorphism
\[
\imath_{\varphi}\colon \FG{\vec x \mid \True} \to \FG{\vec x \mid \varphi},
\]

\noindent
which induces a map
\[
  \imath_{\varphi}\circ - \colon
  \Mod(\theory(\forest^\op)) [\FG{\vec y:q \mid \True}, \FG{\vec x \mid \True}]
  \to
  \Mod(\theory(\forest^\op)) [\FG{\vec y:q \mid \True}, \FG{\vec x \mid \varphi}].
\]

\noindent
By Remark~\ref{rem:coste:fp:coprod},
$\FG{\vec y:q \mid \True}$ is isomorphic to the coproduct of $m$ copies of
$\FG{y:q \mid \True}$,
and we have seen above that $\FG{y:q \mid \True}$ corresponds to
the representable presheaf $\yoneda q$.
Since representable functors are projective,
and a coproduct of projective objects is again projective,
we see that $\FG{\vec y:q \mid \True}$ is projective in $\Mod(\theory(\forest^\op))$.
Thus the map $\imath_{\varphi}\circ - $ is surjective because
$\imath_{\varphi}$ is an epimorphism (see Remark~\ref{rem:emb:epi}),
and so it suffices to show that the hom-set
$\Mod(\theory(\forest^\op)) [\FG{\vec y:q \mid \True}, \FG{\vec x \mid \True}]$ 
is finite, and empty whenever
$q\not\leq p_{i}$ in $\forest$ for all $i\in\{1,\ldots,n\}$. We have 
\begin{align*}
  \Mod(\theory(\forest^\op))\funct{\FG{\vec y:q \mid \True}, \FG{\vec x \mid \True}}
& \cong
  \Mod(\theory(\forest^\op)) \funct{
  \coprod_{j=1}^{m}\FG{y:q \mid \True}, \coprod_{i=1}^{n}\FG{x_{i}:p_{i} \mid \True}}
\\
& \cong
  \left(
  \Mod(\theory(\forest^\op))
  \funct{ \FG{y:q \mid \True}, \coprod_{i=1}^{n}\FG{x_{i}:p_{i} \mid \True}}
  \right)^{m}
\\
& \cong
  \left(\coprod_{i=1}^{n}\Mod(\theory(\forest^\op))
  \funct{ \FG{y:q \mid \True}, \FG{x_{i}:p_{i} \mid \True}} \right)^{m}
\end{align*}

\noindent
where in the last step we used the fact that $\FG{y:q \mid \True}$ is a connected object
(recall from Lemma~\ref{l:paths-in-L-connected} and its proof that
the corresponding presheaf $\yoneda q$ is connected).
It follows from item~\ref{i:singleton-empty-repr} that the hom-set
$\Mod(\theory(\forest^\op)) [\FG{\vec y:q \mid \True}, \FG{\vec x \mid \True}]$
is finite, and it is empty precisely when
$q\not\leq p_{i}$ in $\forest$ for all $i\in\{1,\ldots,n\}$.
\end{proof}

%%%%%%%%%%%%%%%%%%%%%%%%%%%%%%%%%%%%%%%%%%%%%%%%%%%%%%%%%%%%%%%%%%%%%%%%%%%
\begin{figure}
%%%%%%%%%%%%%%%%%%%%%%%%%%%%%%%%%%%%%%%%%%%%%%%%%%%%%%%%%%%%%%%%%%%%%%%%%%%
\begin{multline*}
  \FMono\FG{\vec x \mid \varphi}(\vec x)
  \deq
\\
  \bigwedge_{q \in \downarrow\{p_{1},\ldots,p_{n}\}} \ \ 
  \bigwedge_{\theta \in \cat T\funct{\OI{\vec x \mid \varphi}, \OI{y,y':q \mid \True}}}
  (\forall y,y')
  \left(
  \left(y \Eq y'
 \wedge
  \theta(\vec x,y,y') \right)
  ~\limp~
  \bigvee_{
    \delta \in \cat T\funct{\OI{\vec x \mid \varphi}, \OI{y,y':q \mid y \Eq y'}}
  }
  \delta(\vec x,y,y')
  \right)
\end{multline*}
\caption{Formulae for monomorphisms (Proposition~\ref{prop:emb:presh:fp-domain}). 
\label{fig:emb:mono:forest:fp-domain}}
%%%%%%%%%%%%%%%%%%%%%%%%%%%%%%%%%%%%%%%%%%%%%%%%%%%%%%%%%%%%%%%%%%%%%%%%%%%
\end{figure}
%%%%%%%%%%%%%%%%%%%%%%%%%%%%%%%%%%%%%%%%%%%%%%%%%%%%%%%%%%%%%%%%%%%%%%%%%%%

%%%%%%%%%%%%%%%%%%%%%%%%%%%%%%%%%%%%%%%%%%%%%%%%%%%%%%%%%%%%%%%%%%%%%%%%%%%
\begin{proposition}
\label{prop:emb:presh:fp-domain}
%%%%%%%%%%%%%%%%%%%%%%%%%%%%%%%%%%%%%%%%%%%%%%%%%%%%%%%%%%%%%%%%%%%%%%%%%%%
In $\Mod(\theory(\forest^\op))$,
assume that $h \colon \FG{\vec x \mid \varphi} \to M$
is induced by $\vec c \in \I{\vec x \mid \varphi}_M$, where $\vec x = x_{1}:p_{1},\ldots, x_{n}:p_{n}$.
Then the following are equivalent:
\begin{enumerate}[(i)]
\item 
\label{item:emb:mono:forest:hom}
$h$ is a monomorphism.

\item
\label{item:emb:mono:forest:diag}
Condition~\ref{item:formemb:diag}
of Proposition~\ref{prop:emb:formemb},
restricted to formulae
$(y,y' \sorting \atom)$ of the form $(y,y' : q \sorting y \Eq y')$
with $q \leq p_{i}$ for some $i\in\{1,\ldots,n\}$, holds.

\item
\label{item:emb:mono:forest:form}
\(
  M
  \models
  \FMono\FG{\vec x \mid \varphi}(\vec c)
\),
where
$\FMono\FG{\vec x \mid \varphi}(\vec x)$ is the \emph{finite} formula in Figure~\ref{fig:emb:mono:forest:fp-domain}.
\end{enumerate}
%%%%%%%%%%%%%%%%%%%%%%%%%%%%%%%%%%%%%%%%%%%%%%%%%%%%%%%%%%%%%%%%%%%%%%%%%%%
\end{proposition}
%%%%%%%%%%%%%%%%%%%%%%%%%%%%%%%%%%%%%%%%%%%%%%%%%%%%%%%%%%%%%%%%%%%%%%%%%%%

%%%%%%%%%%%%%%%%%%%%%%%%%%%%%%%%%%%%%%%%%%%%%%%%%%%%%%%%%%%%%%%%%%%%%%%%%%%
\begin{proof}
%%%%%%%%%%%%%%%%%%%%%%%%%%%%%%%%%%%%%%%%%%%%%%%%%%%%%%%%%%%%%%%%%%%%%%%%%%%
\ref{item:emb:mono:forest:hom} $\Leftrightarrow$ \ref{item:emb:mono:forest:diag}
Combining Proposition~\ref{prop:emb:formemb} with Example~\ref{ex:emb:at:mono}
yields that condition~\ref{item:emb:mono:forest:hom}
is equivalent to condition~\ref{item:emb:mono:forest:diag}
generalised to all formulae 
$(y,y' \sorting \atom)$ of the form 
\[
(y,y' : q \sorting y \Eq y')
\]
with $q \in \forest$.
Hence, we have to show that we can restrict ourselves to the case where $q \leq p_{i}$ for some $i\in\{1,\ldots,n\}$.
In turn, this follows from the fact that, by Lemma~\ref{lem:finite-hom-sets-synt-cat}\ref{i:finite-empty-gen}, if $q \not\leq p_{i}$ for all $i\in\{1,\ldots,n\}$ then there exists no morphism $\ell\colon \FG{y,y' : q \mid \True}\to \FG{\vec x \mid \varphi}$, and so there is no commutative square as in the left-hand side of eq.~\eqref{eq:formemb:diag:ass}
(cf.\ condition~\ref{item:formemb:diag} of Proposition~\ref{prop:emb:formemb}).

\ref{item:emb:mono:forest:diag} $\Leftrightarrow$ \ref{item:emb:mono:forest:form} This follows from Corollary~\ref{cor:formemb} (cf.\ also Corollary~\ref{cor:coste:fp:triangle}).

It remains to prove that the formula $\FMono\FG{\vec x \mid \varphi}$ is finite.
The first conjunction is finite since $\forest$ is a forest,
and the second one (over $\theta$) is finite by Lemma~\ref{lem:finite-hom-sets-synt-cat}\ref{i:finite-empty-gen}. 
Finally, consider the disjunction over
$\delta \in \cat T\funct{\OI{\vec x \mid \varphi}, \OI{y,y':q \mid y \Eq y'}}$.
It follows from~\cite[Lemma D1.4.4]{johnstone02book}
that 
\[
  \MI{y,y',z,z' : q \mid y \Eq y' ~\land~ z \Eq y ~\land~ z' \Eq y'}
\]

\noindent
is a monomorphism
from
$\OI{y,y' : q \mid y \Eq y'}$
to
$\OI{z,z' : q \mid \True}$.
By precomposition,
this yields an injection
of
\[
  \cat T\funct{\OI{\vec x \mid \varphi}, \OI{y,y':q \mid y \Eq y'}}
\]

\noindent
into the finite set
\[
  \cat T\funct{\OI{\vec x \mid \varphi}, \OI{y,y':q \mid \True}}.
\]

\noindent
Hence the hom-set
$\cat T\funct{\OI{\vec x \mid \varphi}, \OI{y,y':q \mid y \Eq y'}}$
is finite.
%%%%%%%%%%%%%%%%%%%%%%%%%%%%%%%%%%%%%%%%%%%%%%%%%%%%%%%%%%%%%%%%%%%%%%%%%%%
\end{proof}
%%%%%%%%%%%%%%%%%%%%%%%%%%%%%%%%%%%%%%%%%%%%%%%%%%%%%%%%%%%%%%%%%%%%%%%%%%%

The previous proposition yields, in particular, finite formulae defining path embeddings in $\presh{\forest} \cong \Mod(\theory(\forest^\op))$.
Recall from Lemma~\ref{l:paths-in-presh-forest} that the non-initial paths
in~$\presh{\forest}$ are precisely the representable presheaves,
which correspond to the objects of the form ${\FG{x:p \mid \True}}$
in $\Mod(\theory(\forest^\op))$.
Thus, for each $p\in \forest$, Proposition~\ref{prop:emb:presh:fp-domain}
gives a finite formula 
\[
\FMono[p](x)\coloneqq \FMono\FG{x:p \mid \True}(x)
\]
such that the following statements are equivalent
for any morphism $h \colon \FG{x:p \mid \True} \to M$ in $\Mod(\theory(\forest^\op))$:
\begin{enumerate}[(i)]
\item $h \colon \FG{x:p \mid \True} \to M$ is a monomorphism.
\item
\(
  M
  \models
  \FMono[p](c)
\),
where $h$ takes $x$ to $c \in M(p)$.
\end{enumerate}

%%%%%%%%%%%%%%%%%%%%%%%%%%%%%%%%%%%%%%%%%%%%%%%%%%%%%%%%%%%%%%%%%%%%%%%%%%%
\section{Proof of the main result}
\label{sec:wc}
%%%%%%%%%%%%%%%%%%%%%%%%%%%%%%%%%%%%%%%%%%%%%%%%%%%%%%%%%%%%%%%%%%%%%%%%%%%

This section is devoted to the proof of Theorem~\ref{thm:path:main},
which states that if a finitely accessible wooded adjunction
\[
\begin{tikzcd}
  \C
  \arrow[bend left=25]{r}{\Ladj}
  \arrow[phantom]{r}[description]{\textnormal{\footnotesize{$\bot$}}}
& \E
  \arrow[bend left=25]{l}{R}
\end{tikzcd}
\]

\noindent
detects path embeddings in a signature $\Sig$, then
\begin{equation}
\label{eq:wc:main}
\begin{array}{l l l}
  \text{$M,N \in \E$ equivalent in $\Lang_{\infty}(\Sig)$}
& \longimp
& M \bisim_R N.
\end{array}
\end{equation}

We know from Corollary~\ref{cor:hintikka:wooded:R-bfe}
that the latter implication
holds under the assumption that
path embeddings in $\C$ are definable.
In order to prove Theorem~\ref{thm:path:main},
we will show that path embeddings in $\C$ are definable
whenever they are detected by $\Ladj \colon \C \inadj \E \cocolon R$.

Recall that in Corollary~\ref{cor:emb:emb} (\S\ref{sec:emb:emb})
we obtained formulae which define embeddings of structures in $\E$.
To translate the latter into formulae which define path embeddings in $\C$, 
we use Hodges' \emph{word-constructions}~\cite{hodges74,hodges75la}.
We review the latter in~\S\ref{sec:wc:wc}, and complete the proof of
Theorem~\ref{thm:path:main} in~\S\ref{sec:wc:main}.
Finally, in~\S\ref{sec:wc:fin}
we discuss the special case in which~$\E$ is the category of structures
for a \emph{finite} relational signature $\sig$,
and give sufficient conditions for eq.~\eqref{eq:wc:main} to hold
when $\Lang_\infty(\sig)$ is replaced with first-order logic $\Lang_\omega(\sig)$.

%%%%%%%%%%%%%%%%%%%%%%%%%%%%%%%%%%%%%%%%%%%%%%%%%%%%%%%%%%%%%%%%%%%%%%%%%%%
\subsection{Hodges' word-constructions}
\label{sec:wc:wc}
%%%%%%%%%%%%%%%%%%%%%%%%%%%%%%%%%%%%%%%%%%%%%%%%%%%%%%%%%%%%%%%%%%%%%%%%%%%

Let $\theory$ and $\theorybis$ be cartesian theories
in signatures~$\Sig$ and~$\Sigbis$, respectively.
Recall from~\S\ref{sec:coste:interp} that
morphisms of lfp categories 
\[
\Mod(\theorybis) \to \Mod(\theory),
\]
i.e.\ finitary right adjoints (cf.\ Proposition~\ref{p:lfp-morphisms-characterisation}),
induce interpretations of formulae in~$\Sig$
as formulae in~$\Sigbis$.
Hodges' \emph{word-constructions}~\cite{hodges74,hodges75la},
see also \cite[\S 9.3]{hodges93book}, are a relaxed form
of interpretations that apply to all finitary functors 
$\Mod(\theorybis) \to \Mod(\theory)$,
and in particular to \emph{left} adjoints.

%%%%%%%%%%%%%%%%%%%%%%%%%%%%%%%%%%%%%%%%%%%%%%%%%%%%%%%%%%%%%%%%%%%%%%%%%%%
\subsubsection{Generated structures and presentations}
\label{sec:wc:genpres}
%%%%%%%%%%%%%%%%%%%%%%%%%%%%%%%%%%%%%%%%%%%%%%%%%%%%%%%%%%%%%%%%%%%%%%%%%%%
We recall the following standard notions from~\cite[1.2.2 and~\S 9.2]{hodges93book}.

Fix a signature $\Sig$.
Given a tuple of elements $\vec a$ from a $\Sig$-structure $M$,
if the smallest substructure of $M$ that contains $\vec a$ is $M$ itself,
then we say that $\vec a$ \emph{generates}~$M$.

A \emph{presentation} is a pair $(\vec x, (\atom_i)_{i \in I})$
where the $\atom_i$'s are atomic formulae in $\Sig$
and $\vec x$ 
is a possibly infinite
sequence of sorted variables containing all variables
of $(\atom_i \mid i \in I)$.

Let $\cat K$ be a class of $\Sig$-structures.
Consider a presentation $(\vec x, (\atom_i)_{i \in I})$ with $\vec x = (x_j)_{j \in J}$
and a structure $M \in \cat K$.
Let $\vec a = (a_j)_{j \in J}$
be a sequence of elements of the carriers of $M$ matching the sorts of $\vec x$.
We say that $(\vec x, (\atom_i)_{i \in I})$
\emph{presents} the expanded structure $(M,\vec a)$ in $\cat K$
when the following conditions hold:
\begin{enumerate}[(i)]
\item
$M$ is generated (as a $\Sig$-structure) by $\vec a$,

\item
$(M,\vec a)\models \atom_i$ for each $i \in I$,

\item
If $N \in \cat K$ and $(N,\vec b)\models \atom_i$ for each $i \in I$,
there is a (necessarily unique) homomorphism $(M, \vec a) \to (N,\vec b)$.
\end{enumerate}

\noindent Note that if $(M,\vec a)$ and $(N,\vec b)$
are both presented by $(\vec x,\vec\atom)$ in $\cat K$,
then the unique homomorphism $(M, \vec a) \to (N,\vec b)$ 
is an isomorphism.

The following useful fact was proved in~\cite[Lemma~9.2.1]{hodges93book}.

%%%%%%%%%%%%%%%%%%%%%%%%%%%%%%%%%%%%%%%%%%%%%%%%%%%%%%%%%%%%%%%%%%%%%%%%%%%
\begin{lemma}
\label{lem:wc:pres}
%%%%%%%%%%%%%%%%%%%%%%%%%%%%%%%%%%%%%%%%%%%%%%%%%%%%%%%%%%%%%%%%%%%%%%%%%%%
Let $\cat K$ be a class of $\Sig$-structures
and consider a presentation $(\vec x, (\atom_i)_{i \in I})$.
Then the following are equivalent for each $(M,\vec a)$
with $M \in \cat K$:
\begin{enumerate}[(i)]
\item
\label{item:wc:pres:pres}
$(\vec x,(\atom_i)_{i \in I})$ presents $(M,\vec a)$ in $\cat K$.

\item
\label{item:wc:pres:char}
$\vec a$ generates $M$ (as a $\Sig$-structure)
and for every atomic formula $\atombis(\vec x)$ of $\Sig$,
we have $(M,\vec a) \models \atombis$
if and only if,
for all $(N,\vec b)$ with $N\in\cat K$, if $(N,\vec b)\models \alpha_{i}$ for each $i\in I$ then $(N,\vec b)\models \beta$.
\end{enumerate}
%%%%%%%%%%%%%%%%%%%%%%%%%%%%%%%%%%%%%%%%%%%%%%%%%%%%%%%%%%%%%%%%%%%%%%%%%%%
\end{lemma}
%%%%%%%%%%%%%%%%%%%%%%%%%%%%%%%%%%%%%%%%%%%%%%%%%%%%%%%%%%%%%%%%%%%%%%%%%%%

%%%%%%%%%%%%%%%%%%%%%%%%%%%%%%%%%%%%%%%%%%%%%%%%%%%%%%%%%%%%%%%%%%%%%%%%%%%
\begin{fullproof}[Proof of Lemma~\ref{lem:wc:pres}]
%%%%%%%%%%%%%%%%%%%%%%%%%%%%%%%%%%%%%%%%%%%%%%%%%%%%%%%%%%%%%%%%%%%%%%%%%%%
It is convenient to use here the $\Lang_{\infty,\infty}$-formula
$(\forall \vec x)\left(\bigwedge_i \atom_i \limp \atombis\right)$.
Assume first that $(\vec x, (\atom_i)_i)$ presents $(M,\vec a)$.
Then $\vec a$ generates $M$ by definition.
Consider now an atomic formula $\atombis(\vec x)$ of $\Sig$.
If $\cat K \models (\forall \vec x)\left(\bigwedge_i \atom_i \limp \atombis\right)$,
then we in particular have $(M,\vec a) \models \atombis$.
If $\cat K \not\models (\forall \vec x)\left(\bigwedge_i \atom_i \limp \atombis\right)$,
then there is a $N \in \cat K$ and some $\vec b \in N$
such that $(N,\vec b) \models \bigwedge_i \atom_i$
but $(N,\vec b) \not\models \atombis$.
By assumption on $M$, there is a homomorphism
$h \colon M  \to N$ taking $\vec a$ to $\vec b$,
so that we must have $(M,\vec a) \not\models \atombis$ as well.

For the converse, we have to show that 
$(\vec x, (\atom_i)_i)$ presents $(M,\vec a)$.
Given $j \in I$,
we trivially have 
$\cat K \models (\forall \vec x)\left(\bigwedge_i \atom_i \limp \atom_j\right)$,
so that
$(M,\vec a) \models \atom_j$ by assumption.
It follows that
$(M,\vec a) \models \bigwedge_i \atom_i$.
Consider now some $N \in \cat K$
and some $\vec b \in N$ such that
$(N,\vec b) \models \bigwedge_i \atom_i$.
Hence for very atomic formula $\atombis$, 
if $(M ,\vec a) \models \atombis$
then $(N,\vec b) \models \atombis$
since
$\cat K \models (\forall \vec x)\left(\bigwedge_i \atom_i \limp \atombis\right)$.
Recall from \S\ref{sec:prelim}
that equalities are atomic formulae.
Since $\vec a$ generates $M$,
it follows that
the assignment $\vec a \mapsto \vec b$ extends to a function
$h \colon M \to N$ which preserves the function symbols of $\Sig$,
namely
such that for every term $t(\vec x)$ in $\Fun(\Sig)$
we have
\[
\begin{array}{l l l}
  h(t_M(\vec a))
& =
& t_N(\vec b)
\end{array}
\]

\noindent
It is equally easy to see that $h$ preserves the predicate symbols of $\Sig$.
%%%%%%%%%%%%%%%%%%%%%%%%%%%%%%%%%%%%%%%%%%%%%%%%%%%%%%%%%%%%%%%%%%%%%%%%%%%
\end{fullproof}
%%%%%%%%%%%%%%%%%%%%%%%%%%%%%%%%%%%%%%%%%%%%%%%%%%%%%%%%%%%%%%%%%%%%%%%%%%%

%%%%%%%%%%%%%%%%%%%%%%%%%%%%%%%%%%%%%%%%%%%%%%%%%%%%%%%%%%%%%%%%%%%%%%%%%%%
\begin{remark}\label{rem:variables-vs-constants-etc}
%%%%%%%%%%%%%%%%%%%%%%%%%%%%%%%%%%%%%%%%%%%%%%%%%%%%%%%%%%%%%%%%%%%%%%%%%%%
Above we considered presentations of pairs $(\vec x, \vec\atom)$
where $\vec x$ is a sequence of sorted variables.
But we could have equivalently considered presentations of pairs
$(\vec c,\vec\atom)$ where
$\vec c = (c_j)_{j \in J}$
is a sequence of elements
from a family of sets
$S = (S_\sort \mid \sort \in \Sort(\Sig))$,
and the $\vec\atom$ are atomic formulae
in $\Sig(\vec c)$, the extension of $\Sig$
with a new constant symbol of the appropriate sort
for each $c_j$.

Moreover, it is a technical convenience
to insist that $\vec x$, or $\vec c$,
form sequences rather than mere sets.
Following~\cite[Remark~9.2]{hodges93book},
when leaving implicit the order of elements in sequences,
we may
simply consider pairs $(X,\Phi)$ where $X$ is a set
and $\Phi$ is a set of atomic formulae in $\Sig(X)$.
%%%%%%%%%%%%%%%%%%%%%%%%%%%%%%%%%%%%%%%%%%%%%%%%%%%%%%%%%%%%%%%%%%%%%%%%%%%
\end{remark}
%%%%%%%%%%%%%%%%%%%%%%%%%%%%%%%%%%%%%%%%%%%%%%%%%%%%%%%%%%%%%%%%%%%%%%%%%%%

%%%%%%%%%%%%%%%%%%%%%%%%%%%%%%%%%%%%%%%%%%%%%%%%%%%%%%%%%%%%%%%%%%%%%%%%%%%
\subsubsection{Word-constructions}
\label{sec:wc:wcdef}
%%%%%%%%%%%%%%%%%%%%%%%%%%%%%%%%%%%%%%%%%%%%%%%%%%%%%%%%%%%%%%%%%%%%%%%%%%%
Consider a finitary functor $F \colon \Mod(\theorybis) \to \Mod(\theory)$,
where~$\theorybis$ and~$\theory$ are cartesian theories
in signatures $\Sigbis$ and $\Sig$, respectively.
The corresponding word-construction provides a formula translation
from $\Lang_\infty(\Sig)$ to $\Lang_\infty(\Sigbis)$
such that, for each model ${M \in \Mod(\theorybis)}$,
the theory of $F M$ can be expressed in
the theory of $M$.

More precisely, a word-construction
\[
\begin{array}{*{5}{l}}
  \WC
& :
& \Sigbis
& \longto
& 
  \Sig
\end{array}
\]

\noindent
for a functor $F$ as above is a syntactic device 
that produces for each $M \in \Mod(\theorybis)$
a presentation $\WC(M) = (X,\Phi)$ of $F M \in \Mod(\theory)$.
The key is that the ``generators''~$X$ and ``relations''~$\Phi$
are controlled by formulae in~$\Sigbis$
whose definitions are independent of~$M$.

We now recall the formal definition of word-constructions; the following is a 
finitary version of the many-sorted notion of~\cite{hodges75la}.

%%%%%%%%%%%%%%%%%%%%%%%%%%%%%%%%%%%%%%%%%%%%%%%%%%%%%%%%%%%%%%%%%%%%%%%%%%%
\begin{definition}[Word-construction]
\label{def:wc:wc}
%%%%%%%%%%%%%%%%%%%%%%%%%%%%%%%%%%%%%%%%%%%%%%%%%%%%%%%%%%%%%%%%%%%%%%%%%%%
Let $\Sigbis$ and $\Sig$ be signatures.
A \emph{word-construction}
\[
\begin{array}{*{5}{l}}
  \WC
& :
& \Sigbis
& \longto
& 
  \Sig
\end{array}
\]

\noindent
is a tuple
$\WC = \left( \Sig(\WC) , \WCTe(\WC) , \WCAt(\WC) \right)$
where:
\begin{enumerate}[(1)]

\item
$\Sig(\WC)$ is a signature 
with $\Sort(\Sig(\WC)) \coloneqq \Sort(\Sig) + \Sort(\Sigbis)$
that extends~$\Sig$ with function symbols
\[
\begin{array}{*{5}{l}}
  f
& :
& \sortbis_1 \dots \sortbis_n
& \longto
& \sort
\end{array}
\]

\noindent
for each $\sortbis_1,\dots,\sortbis_n \in \Sort(\Sigbis)$
and $\sort \in \Sort(\Sig)$.

Note that atomic formulae in $\Sig(\WC)$ are atomic formulae in $\Sig$,
but with terms possibly involving the additional function symbols of $\Sig(\WC)$.

\item
$\WCTe(\WC)$
is a function taking terms
$(\vec x : \vec\sortbis \sorting t : \sort)$
in $\Sig(\WC)$,
with $\vec\sortbis \in \Sort(\Sigbis)$ and $\sort \in \Sort(\Sig)$,
to formulae in $\Sigbis$
with free variables among $\vec x : \vec\sortbis$.

\item
$\WCAt(\WC)$ is a function taking
atomic formulae
$(\vec x : \vec\sortbis \sorting \atom)$
in $\Sig(\WC)$, with ${\vec\sortbis \in \Sort(\Sigbis)}$,
to formulae in $\Sigbis$
with free variables among $\vec x : \vec\sortbis$.
\end{enumerate}
\end{definition}

A word-construction
$\WC \colon \Sigbis \to \Sig$
induces, for each $M \in \Struct(\Sigbis)$,
a presentation $\WC(M) = (X,\Phi)$ defined as follows:
\begin{itemize}
\item
$X$ is the set of expressions
\footnote{Here we adopt the equivalent approach outlined in Remark~\ref{rem:variables-vs-constants-etc}.}\
of the form $t(\vec a)$ such that
$(\vec x : \vec\sortbis \sorting t : \sort)$ is in the domain of $\WCTe(\WC)$,
$\vec a \in M(\vec\sortbis)$,
and
$M \models \WCTe(\WC)(t)(\vec a)$.

\item
$\Phi$ is the set of atomic formulae of the form
$\atom(\vec a)$
where $(\vec x : \vec\sortbis \sorting \atom)$ is in the domain of $\WCAt(\WC)$,
$\vec a \in M(\vec\sortbis)$,
and
$M \models \WCAt(\WC)(\atom)(\vec a)$.
\end{itemize}

%%%%%%%%%%%%%%%%%%%%%%%%%%%%%%%%%%%%%%%%%%%%%%%%%%%%%%%%%%%%%%%%%%%%%%%%%%%
\begin{remark}
\label{rem:wc:pres:horn}
%%%%%%%%%%%%%%%%%%%%%%%%%%%%%%%%%%%%%%%%%%%%%%%%%%%%%%%%%%%%%%%%%%%%%%%%%%%
In general, given a class of structures $\cat K$ and a  
presentation $(\vec x, (\atom_i)_{i \in I})$, there could be no (expansion of a) member of $\cat K$ that is presented by $(\vec x, (\atom_i)_{i \in I})$. 

A characteristic property of universal Horn theories is that, 
in their categories of models,
every presentation 
presents some expanded structure;
see~\cite[Theorem 9.2.2 and Exercise 9.2.19]{hodges93book}.
In particular, if $\cat K$ is the class of models of a universal Horn theory,
then there is always an expansion of a structure in $\cat K$ 
that is presented by $\WC(M)$. 

In fact, Hodges considered (infinitary) Horn theories
in~\cite{hodges74,hodges75la,hodges93book}.
Since we allow for arbitrary cartesian theories,
we will need a slight adaptation of Hodges' results.\footnote{Recall from Remark~\ref{rem:prelim:coste:horn} that
Horn theories are a proper subclass of cartesian theories.}
%%%%%%%%%%%%%%%%%%%%%%%%%%%%%%%%%%%%%%%%%%%%%%%%%%%%%%%%%%%%%%%%%%%%%%%%%%%
\end{remark}
%%%%%%%%%%%%%%%%%%%%%%%%%%%%%%%%%%%%%%%%%%%%%%%%%%%%%%%%%%%%%%%%%%%%%%%%%%%

%%%%%%%%%%%%%%%%%%%%%%%%%%%%%%%%%%%%%%%%%%%%%%%%%%%%%%%%%%%%%%%%%%%%%%%%%%%
\subsubsection{The word-construction of a finitary functor of lfp categories}
\label{sec:wc:filtcolim}
%%%%%%%%%%%%%%%%%%%%%%%%%%%%%%%%%%%%%%%%%%%%%%%%%%%%%%%%%%%%%%%%%%%%%%%%%%%
Fix a finitary functor $F \colon \Mod(\theorybis) \to \Mod(\theory)$,
where $\theorybis, \theory$ are cartesian theories.
Write $\Sigbis$ and $\Sig$ for the respective signatures of $\theorybis$ and $\theory$.
We will associate with $F$ a word-construction
\[
\begin{array}{*{5}{l}}
  \WC
& :
& \Sigbis
& \longto
& 
  \Sig
\end{array}
\]

\noindent
such that for each $M \in \Mod(\theorybis)$,
$\WC(M)$ presents in $\Mod(\theory)$ an expansion of $F M$.
This will allow us to translate the $\Lang_{\infty}(\Sig)$-theory of $F M$
into the $\Lang_{\infty}(\Sigbis)$-theory of $M$
(see Theorem~\ref{thm:wc:cor:trans} below).

Let $M \in \Mod(\theorybis)$.
Recall that $M$ is the canonical filtered colimit
of the cocone consisting of morphisms with finitely presentable domain $\FG{\vec x \mid \varphi} \to M$.
Since $F$ preserves filtered colimits,
$F M$ is the filtered colimit 
of the $F \FG{\vec x \mid \varphi} \to F M$
with connecting maps induced by commuting triangles in $\Mod(\theorybis)$.
It is convenient to have an explicit
description of this filtered colimit.

%%%%%%%%%%%%%%%%%%%%%%%%%%%%%%%%%%%%%%%%%%%%%%%%%%%%%%%%%%%%%%%%%%%%%%%%%%%
\begin{remark}
\label{rem:wc:filtcolim}
%%%%%%%%%%%%%%%%%%%%%%%%%%%%%%%%%%%%%%%%%%%%%%%%%%%%%%%%%%%%%%%%%%%%%%%%%%%
By Lemma~\ref{lem:coste:mod:filtcolim},
$\Mod(\theory)$ is closed in $\Struct(\Sig)$
under filtered colimits.
Following~\cite[Remarks 5.1(3) and 3.4(4)(iii)]{ar94book},%
\footnote{Alternatively, one can use the fact that
filtered colimits in $\Struct(\Sig)$ are 
computed in $\funct{\cat T(\Sig),\Set}$,
where $\cat T(\Sig)$ is the syntactic category of $\Sig$.
Cf.\ \S\ref{ss:lex-and-lfp}
and eq.~\eqref{eq:coste:struct} in \S\ref{sec:coste:synt}.}\
this yields the following description of $F M$.
\begin{enumerate}[(1)]

\item
\label{item:wc:filtcolim:sort}
For each sort $\sort \in \Sort(\Sig)$,
$F M(\sort)$
is the filtered colimit in $\Set$ of the
$F \FG{\vec x \mid \varphi}(\sort)$. 
Recall 
that this amounts to saying that $F M(\sort)$
is isomorphic to the quotient of the disjoint sum of the
$F \FG{\vec x \mid \varphi}(\sort)$'s
by the equivalence relation that
identifies those pairs of elements having the same image
in cospans of the form
\[
\begin{tikzcd}
  F \FG{\vec x \mid \varphi}(\sort)
  \arrow{dr}[left, yshift=-5pt]{(F h)^{\sort}}
&
& F \FG{\vec y \mid \psi}(\sort)
  \arrow{dl}[right, yshift=-5pt]{(F k)^{\sort}}
\\
& F \FG{\vec z \mid \theta}(\sort)
\end{tikzcd}
\]

\noindent
(see e.g.~\cite[Proposition~3.1.3]{ks06book}).

\item
\label{item:wc:filtcolim:fun}
Consider now a function symbol
$f \in \Fun(\Sig)(\sort_1,\dots,\sort_n;\sort)$.
Given
$\vec a \in F M(\vec\sort)$,
since the canonical diagram of $M$ is filtered
there is some
$F h \colon F \FG{\vec x \mid \varphi} \to F M$
such that
$\vec a$ is in the image of $F h$,
say $\vec a = (F h)(\vec b)$.
We can then let $f^{F M}(\vec a)$
be $h f_{\FG{\vec x \mid \varphi}}(\vec b)$.
That this indeed uniquely defines $f^{F M}$
follows from the fact that the canonical diagram of $M$ is filtered.

\item
\label{item:wc:filtcolim:pred}
Finally, for each relation symbol
$R \in \Rel(\Sig)(\sort_1,\dots,\sort_n)$,
$R^{F M}$ is the union of the images 
of the interpretations of $R$ under the morphisms
$F h \colon F \FG{\vec x \mid \varphi} \to F M$.
\end{enumerate}
%%%%%%%%%%%%%%%%%%%%%%%%%%%%%%%%%%%%%%%%%%%%%%%%%%%%%%%%%%%%%%%%%%%%%%%%%%%
\end{remark}
%%%%%%%%%%%%%%%%%%%%%%%%%%%%%%%%%%%%%%%%%%%%%%%%%%%%%%%%%%%%%%%%%%%%%%%%%%%

We now define the word-construction $\WC \colon \Sigbis \to \Sig$
associated with $F$.
To this end, we have to devise the signature $\Sig(\WC)$ and the functions
$\WCTe(\WC)$ and $\WCAt(\WC)$.

We first describe the signature $\Sig(\WC)$ extending $\Sig$.
Recall that the sorts of $\Sig(\WC)$ are $\Sort(\Sigbis) + \Sort(\Sig)$.
The additional function symbols of $\Sig(\WC)$
are as follows.
Let $\FG{\vec x \mid \varphi}$ be a finitely presentable 
model of $\theorybis$.
Assume $\vec x = x_1 : \sortbis_1,\dots,x_k : \sortbis_k$
(where $\sortbis_1,\dots,\sortbis_k$ are sorts of $\Sigbis$).
For
each $\sort \in \Sort(\Sig)$
and each $c \in F\FG{\vec x \mid \varphi}(\sort)$
we introduce a new function symbol
\[
\begin{array}{*{5}{l}}
  f^c_{\FG{\vec x \mid \varphi}}
& :
& \sortbis_1 \ldots \sortbis_k
& \longto
& \sort.
\end{array}
\]

The intended interpretation of $f^c_{\FG{\vec x \mid \varphi}}$
is as follows.
Given $M \in \Mod(\theorybis)$,
the function symbol
$f^c_{\FG{\vec x \mid \varphi}}$
with $c \in \FG{\vec x \mid \varphi}(\sort)$
is to be thought of as the function
\[
\begin{array}{r c l}
  \I{\vec x \mid \varphi}_M
& \longto
& (F M)(\sort)
\\

  \vec a
& \longmapsto
& (F h_{\vec a})(c)
\end{array}
\]

\noindent
where $h_{\vec a} \colon \FG{\vec x \mid \varphi} \to M$
is the homomorphism induced by $\vec a \in \I{\vec x \mid \varphi}_M$
(Lemma~\ref{lem:coste:lp:hom}).
It follows that $F M$ is generated
(in the sense of~\S\ref{sec:wc:genpres})
by the valuation
\begin{equation}
\label{eq:wc:wc:val}
\begin{array}{*{5}{l}}
  C
& :
& (f^c_{\FG{\vec x \mid \varphi}} 
  ,\,
  \vec a \in \I{\vec x \mid \varphi}_M)
& \longmapsto
& (F h_{\vec a})(c) \in F M
\end{array}
\end{equation}

\noindent
where $(\vec x \sorting \varphi)$ ranges over 
the cartesian formulae over $\theorybis$.

We now turn to $\WCTe(\WC)$.
The above description of $\Sig(\WC)$
readily gives the formula
$\WCTe(\WC)(f^c_{\FG{\vec x \mid \varphi}})(\vec x)$,
namely
\[
\begin{array}{l l l}
  \WCTe(\WC)(f^c_{\FG{\vec x \mid \varphi}}(\vec x))
& \deq
& \varphi(\vec x).
\end{array}
\]

\noindent
Moreover, since $F M$ is generated by the valuation $C$
of eq.~\eqref{eq:wc:wc:val},
for each term $t$ in $\Sig(\WC)$
that is not of the form $f^c_{\FG{\vec x \mid \varphi}}(\vec x)$,
we can simply let
\[
\begin{array}{l l l}
  \WCTe(\WC)(t)
& \deq
& \False.
\end{array}
\]

The case of $\WCAt(\WC)$ is less straightforward.
Let $\cat U$ be the syntactic category of $\Mod(\theorybis)$.
Let $\atombis$ be an atomic formula in $\Sig(\WC)$
whose free variables have sorts in $\Sort(\Sigbis)$.
The formula $\WCAt(\WC)(\atombis)$ in $\Sigbis$ is defined as follows.
Recall that the input sorts of the new function symbols
$f^{c}_{\FG{\vec x \mid \varphi}}$
are sorts of $\Sigbis$,
and write
\[
\begin{array}{l l l}
  \atombis
& =
& \atom\left(
  f^{c_1}_{\FG{\vec x_1 \mid \varphi_1}}(\vec x_1)
  ,\dots,
  f^{c_n}_{\FG{\vec x_n \mid \varphi_n}}(\vec x_n)
  \right)
\end{array}
\]

\noindent
where $\atom$ is an atomic formula in $\Sig$
and the $\vec x_i$'s need not be distinct.
We devise $\WCAt(\WC)(\atombis)$
on the basis of the following characterisation.

%%%%%%%%%%%%%%%%%%%%%%%%%%%%%%%%%%%%%%%%%%%%%%%%%%%%%%%%%%%%%%%%%%%%%%%%%%%
\begin{lemma}
\label{lem:wc:locval}
%%%%%%%%%%%%%%%%%%%%%%%%%%%%%%%%%%%%%%%%%%%%%%%%%%%%%%%%%%%%%%%%%%%%%%%%%%%
Let $\atom(y_1,\dots,y_n)$ be an atomic formula in $\Sig$
and let $M \in \Mod(\theorybis)$.
Given morphisms
$h_i \colon \FG{\vec x_i \mid \varphi_i} \to M$
for $i = 1,\dots,n$,
the following are equivalent:
\begin{enumerate}[(i)]
\item
\label{item:wc:locval:models}
$F M \models \atom((F h_1)(c_1),\dots,(F h_n)(c_n))$.

\item
\label{item:wc:locval:char}
There exist
$\FG{\vec x \mid \varphi}$
and
$h \colon \FG{\vec x \mid \varphi} \to M$
such that each $h_i$ can be factored as
\[
\begin{tikzcd}
  \FG{\vec x_i \mid \varphi_i}
  \arrow{r}{h_i}
  \arrow{d}[left]{k_i}
& M
\\
  \FG{\vec x \mid \varphi}
  \arrow{ur}[below]{h}
\end{tikzcd}
\]
where the morphisms $k_{i}$'s satisfy
\[
\begin{array}{l l l}
  F \FG{\vec x \mid \varphi}
& \models
& \atom\big( (F k_1)(c_1),\dots, (F k_n)(c_n) \big).
\end{array}
\]
\end{enumerate}
\end{lemma}

%%%%%%%%%%%%%%%%%%%%%%%%%%%%%%%%%%%%%%%%%%%%%%%%%%%%%%%%%%%%%%%%%%%%%%%%%%%
\begin{proof}
%%%%%%%%%%%%%%%%%%%%%%%%%%%%%%%%%%%%%%%%%%%%%%%%%%%%%%%%%%%%%%%%%%%%%%%%%%%
We first show that
\(
  \text{\ref{item:wc:locval:char}}
  \imp
  \text{\ref{item:wc:locval:models}}
\).
Given $h \colon \FG{\vec x \mid \varphi} \to M$
and $k_i \colon \FG{\vec x_i \mid \varphi_i} \to \FG{\vec x \mid \varphi}$
as required,
since $h$ is a homomorphism we have
\[
\begin{array}{l l l}
  F M
& \models
& \atom\big( (F h)(F k_1)(c_1),\dots, (F h)(F k_n)(c_n) \big),
\end{array}
\]

\noindent
and thus we are done since
\[
\begin{array}{l l l}
  (F h_i)(c_i)
& =
& (F h)(F k_i)(c_i).
\end{array}
\]

We now show that
\(
  \text{\ref{item:wc:locval:models}}
  \imp
  \text{\ref{item:wc:locval:char}}
\).
Assume
$F M \models \atom((F h_1)(c_1),\dots,(F h_n)(c_n))$.
By Remark~\ref{rem:wc:filtcolim},
there are some
$h' \colon \FG{\vec x' \mid \varphi'} \to M$
and some $b_1,\dots,b_n \in F \FG{\vec x' \mid \varphi'}$
such that
$F \FG{\vec x' \mid \varphi'} \models \atom(b_1,\dots,b_n)$
and
$(F h')(b_i) = (F h_i)(c_i)$ for each $i = 1,\dots,n$.
Again by Remark~\ref{rem:wc:filtcolim},
using the fact that the canonical diagram of $M$ in $\Mod(\theorybis)$
is filtered,
$h'$ factors as 
\[
\begin{tikzcd}
  \FG{\vec x' \mid \varphi'}
  \arrow{r}{h''}
& \FG{\vec x \mid \varphi}
  \arrow{r}{h}
& M
\end{tikzcd}
\]

\noindent
such that for each $i = 1,\dots, n$
there is
a $k_i \colon \FG{\vec x_i \mid \varphi_i} \to \FG{\vec x \mid \varphi}$
with $h_i = h \comp k_i$
and $(F h'')(b_i) = (F k_i)(c_i)$.
Since $h''$ is a homomorphism, we have
\[
\begin{array}{l l l}
  F \FG{\vec x \mid \varphi}
& \models
& \atom\big( (F h'')(b_1),\dots, (F h'')(b_n) \big).
\end{array}
\]

\noindent
This concludes the proof. 
\end{proof}

Corollary~\ref{cor:coste:fp:triangle} 
provides us with a syntactic formulation of
condition~\ref{item:wc:locval:char}
in Lemma~\ref{lem:wc:locval}.
This leads us to define
$\WCAt(\WC)(\atombis)(\vec x_1,\dots,\vec x_n)$
as the formula
\[
  \bigvee_{\OI{\vec x \mid \varphi} \in \cat U}
  \bigvee_{\substack{\theta_1 \in \cat U\funct{\varphi,\varphi_1}
  ,\dots,
  \theta_n \in \cat U\funct{\varphi,\varphi_n}
  \\
  F\FG{\vec x \mid \varphi} \models
  \atom((F k_1)(c_1),\dots,(F k_n)(c_n))
  }}
  (\exists \vec x)\left( \varphi
  \land
  \bigwedge_{1 \leq i \leq n}
  \theta_i(\vec x, \vec x_i)
  \right)
\]

\noindent
where
$k_i \colon \FG{\vec x_i \mid\varphi_i} \to \FG{\vec x \mid \varphi}$
corresponds to $\yoneda\theta_i$
under Theorem~\ref{thm:coste:synt}.

This completes the definition of the word-construction $\WC$
associated with $F$.
We shall now prove that for each $M \in \Mod(\theorybis)$,
the expanded structure $(F M, C)$
is presented by $\WC(M)$ in $\Mod(\theory)$.

%%%%%%%%%%%%%%%%%%%%%%%%%%%%%%%%%%%%%%%%%%%%%%%%%%%%%%%%%%%%%%%%%%%%%%%%%%%
\begin{proposition}
\label{prop:wc:cor}
%%%%%%%%%%%%%%%%%%%%%%%%%%%%%%%%%%%%%%%%%%%%%%%%%%%%%%%%%%%%%%%%%%%%%%%%%%%
Given $M \in \Mod(\theorybis)$,
the presentation $\WC(M)$ presents
the expanded structure $(F M, C)$
in $\Mod(\theory)$,
where $C$ is as in eq.~\eqref{eq:wc:wc:val}.
%%%%%%%%%%%%%%%%%%%%%%%%%%%%%%%%%%%%%%%%%%%%%%%%%%%%%%%%%%%%%%%%%%%%%%%%%%%
\end{proposition}
%%%%%%%%%%%%%%%%%%%%%%%%%%%%%%%%%%%%%%%%%%%%%%%%%%%%%%%%%%%%%%%%%%%%%%%%%%%

%%%%%%%%%%%%%%%%%%%%%%%%%%%%%%%%%%%%%%%%%%%%%%%%%%%%%%%%%%%%%%%%%%%%%%%%%%%
\begin{proof}
%%%%%%%%%%%%%%%%%%%%%%%%%%%%%%%%%%%%%%%%%%%%%%%%%%%%%%%%%%%%%%%%%%%%%%%%%%%
Let $M \in \Mod(\theorybis)$
and write $\WC(M) = (X,\Phi)$ as in~\S\ref{sec:wc:wcdef}.
We apply Lemma \ref{lem:wc:pres}.

First, it follows from
Remark~\ref{rem:wc:filtcolim}.\ref{item:wc:filtcolim:sort}
that $F M$ is generated by $C$.
But the elements of~$X$ 
are precisely the $f^c_{\FG{\vec x \mid \varphi}}(\vec a)$
such that $(M,\vec a) \models \varphi$,
i.e.\ such that there is a (necessarily unique)
$h_{\vec a} \colon \FG{\vec x \mid \varphi} \to M$
taking $\vec x$ to $\vec a$.
We thus get an interpretation of the elements of~$X$
in the sequence $C$,
and this interpretation generates $F M$. 

Consider now an atomic formula $\atombis$ in $\Sig$
with parameters in $X$,
say
\[
\begin{array}{l l l}
  \atombis
& =
& \atom\left(
  f^{c_1}_{\FG{\vec x_1 \mid \varphi_1}}(\vec a_1)
  ,\dots,
  f^{c_n}_{\FG{\vec x_n \mid \varphi_n}}(\vec a_n)
  \right).
\end{array}
\]

\noindent
We need to prove that $(F M, C) \models \atombis$ if, and only if, each $N\in\Mod(\theory)$ equipped with an interpretation of the elements of $X$ satisfies $\atombis$ whenever it satisfies all formulae in $\Phi$. 
Note that since $F M$ is a model of $\theory$,
it is enough to show that
for all $\atombis$ as above, we have
\[
\begin{array}{l l l}
  (F M, C) \models \atombis
& \text{iff}
& \atombis \in \Phi,
\end{array}
\]

\noindent
that is
\[
\begin{array}{l l l}
  (F M, C) \models \atombis
& \text{iff}
& M \models \WCAt(\WC)(\atombis).
\end{array}
\]

\noindent
But this follows from Corollary~\ref{cor:coste:fp:triangle}
and Lemma~\ref{lem:wc:locval}.
%%%%%%%%%%%%%%%%%%%%%%%%%%%%%%%%%%%%%%%%%%%%%%%%%%%%%%%%%%%%%%%%%%%%%%%%%%%
\end{proof}
%%%%%%%%%%%%%%%%%%%%%%%%%%%%%%%%%%%%%%%%%%%%%%%%%%%%%%%%%%%%%%%%%%%%%%%%%%%

%%%%%%%%%%%%%%%%%%%%%%%%%%%%%%%%%%%%%%%%%%%%%%%%%%%%%%%%%%%%%%%%%%%%%%%%%%%
\begin{remark}
%%%%%%%%%%%%%%%%%%%%%%%%%%%%%%%%%%%%%%%%%%%%%%%%%%%%%%%%%%%%%%%%%%%%%%%%%%%
Note that since $\theory$ is not assumed to be a Horn theory,
we cannot deduce Proposition~\ref{prop:wc:cor}
from~\cite[Theorem 9.2.2]{hodges93book}.

Moreover, comparing Proposition~\ref{prop:wc:cor}
with~\cite[Theorem 9.3.7]{hodges93book},
note that the latter
deals with the case where $\theory$ is a universal Horn theory
and $\Mod(\theorybis) = \Struct(\Sigbis)$.

As an aside, 
note that all formulae $\varphi$ and $\theta_i$
involved in $\WCTe(\WC)$ and $\WCAt(\WC)$
are cartesian and thus primitive-positive (cf.\ Remark~\ref{rem:emb:pure}). 
Hence, the word-construction $\WC$ is presenting
in the sense of~\cite[\S 9.3]{hodges93book}.
%%%%%%%%%%%%%%%%%%%%%%%%%%%%%%%%%%%%%%%%%%%%%%%%%%%%%%%%%%%%%%%%%%%%%%%%%%%
\end{remark}
%%%%%%%%%%%%%%%%%%%%%%%%%%%%%%%%%%%%%%%%%%%%%%%%%%%%%%%%%%%%%%%%%%%%%%%%%%%

For
$F \colon \Mod(\theorybis) \to \Mod(\theory)$
and $\WC \colon \Sigbis \to \Sig$ as above,
and $M \in \Mod(\theorybis)$,
we shall now explain 
how the $\Lang_{\infty}(\Sig)$-theory of $F M$
can be translated into the $\Lang_{\infty}(\Sigbis)$-theory of~$M$.
This relies on the following formula translation $(-)^\WC$.

Let $\psi(y_1,\dots,y_n)$
be a formula in $\Sig$,
and let
\(
  \vec f =
  f^{c_1}_{\FG{\vec x_1 \mid \varphi_1}}(\vec x_1)
  ,\dots,
  f^{c_n}_{\FG{\vec x_n \mid \varphi_n}}(\vec x_n)
\)
be additional function symbols from $\Sig(\WC)$
whose output sorts match the sorts of $y_1,\dots,y_n$.
We define
$\psi(\vec f)^\WC(\vec x_1,\dots,\vec x_n) \in \Lang_\infty(\Sigbis)$
by induction on $\psi$ as follows.
\begin{itemize}
\item
If $\psi(y_1,\dots,y_n)$ is an atomic formula $\atom(y_1,\dots,y_n)$,
\[
\begin{array}{l l l}
  \psi(\vec f)^\WC(\vec x_1,\dots,\vec x_n)
& \deq
& \WCAt(\WC)(\atombis)(\vec x_1,\dots,\vec x_n)
\end{array}
\]

\noindent
where
\(
  \atombis
  =
  \atom\left(
  f^{c_1}_{\FG{\vec x_1 \mid \varphi_1}}(\vec x_1)
  ,\dots,
  f^{c_n}_{\FG{\vec x_n \mid \varphi_n}}(\vec x_n)
  \right)
\).

\item
The translation $(-)^\WC$ commutes with propositional connectives.

\item
In the case of 
$\psi(\vec y) = (\exists y : \sort)\psi'(\vec y,y)$,
we let
$\psi(\vec f)^\WC(\vec x_1,\dots,\vec x_n)$
be the formula
\[
  \bigvee_{\OI{\vec x \mid \varphi} \in \cat U} \ 
  \bigvee_{c \in (F\FG{\vec x \mid \varphi})(\sort)}
  (\exists \vec x)
  \left(
  \varphi(\vec x)
  ~\land~
  \psi'(\vec f, f^{c}_{\FG{\vec x \mid \varphi}})^\WC(\vec x_1,\dots,\vec x_n,\vec x)
  \right)
\]

\noindent
where $\cat U$ is the syntactic category of $\theorybis$.
\end{itemize}

The translation $(-)^\WC$ is correct in the following sense:

%%%%%%%%%%%%%%%%%%%%%%%%%%%%%%%%%%%%%%%%%%%%%%%%%%%%%%%%%%%%%%%%%%%%%%%%%%%
\begin{theorem}
\label{thm:wc:cor:trans}
%%%%%%%%%%%%%%%%%%%%%%%%%%%%%%%%%%%%%%%%%%%%%%%%%%%%%%%%%%%%%%%%%%%%%%%%%%%
Let $\psi(y_1,\dots,y_n) \in \Lang_\infty(\Sig)$,
and let
\(
  \vec f = 
  f^{c_1}_{\FG{\vec x_1 \mid \varphi_1}}(\vec x_1)
  ,\dots,
  f^{c_n}_{\FG{\vec x_n \mid \varphi_n}}(\vec x_n)
\)
be additional function symbols from $\Sig(\WC)$.
Let $M \in \Mod(\theorybis)$
and let $\vec a_1,\dots,\vec a_n \in M$
be such that
$(M, \vec a_i) \models \WCTe(\WC)(f^{c_i}_{\FG{\vec x_i \mid \varphi_i}}(\vec x_i))$
for each $i = 1,\dots,n$.
Then the following statements are equivalent:
\begin{enumerate}[(i)]
\item
$F M \models \psi\left( (F h_1)(c_1) ,\dots, (F h_n)(c_n) \right)$,
where $h_i \colon \FG{\vec x_i \mid \varphi_i} \to M$
takes $\vec x_i$ to $\vec a_i$.

\item
$M \models \psi(\vec f)^\WC(\vec a_1,\dots,\vec a_n)$.
\end{enumerate}
%%%%%%%%%%%%%%%%%%%%%%%%%%%%%%%%%%%%%%%%%%%%%%%%%%%%%%%%%%%%%%%%%%%%%%%%%%%
\end{theorem}
%%%%%%%%%%%%%%%%%%%%%%%%%%%%%%%%%%%%%%%%%%%%%%%%%%%%%%%%%%%%%%%%%%%%%%%%%%%

%%%%%%%%%%%%%%%%%%%%%%%%%%%%%%%%%%%%%%%%%%%%%%%%%%%%%%%%%%%%%%%%%%%%%%%%%%%
\begin{proof}
%%%%%%%%%%%%%%%%%%%%%%%%%%%%%%%%%%%%%%%%%%%%%%%%%%%%%%%%%%%%%%%%%%%%%%%%%%%
The statement is proved by structural induction on $\psi$.
We only spell out the most important cases:
\begin{itemize}
\item
The case of $\psi$ atomic follows
from 
Corollary~\ref{cor:coste:fp:triangle}
and Lemma~\ref{lem:wc:locval}.

\item
If
$\psi(\vec y) = (\exists y)\psi'(\vec y,y)$,
it follows from Remark~\ref{rem:wc:filtcolim}
that
\[
F M \models \psi\left( (F h_1)(c_1) ,\dots, (F h_n)(c_n) \right)
\]
if, and only if,
there are
$h \colon \FG{\vec x \mid \varphi} \to M$
and
$c \in F\FG{\vec x \mid \varphi}$
such that
\[
  F M \models \psi'\left( (F h_1)(c_1) ,\dots, (F h_n)(c_n), (F h)(c) \right).
\]

\noindent
But there is a homomorphism
$h \colon \FG{\vec x \mid \varphi} \to M$
if, and only if, there are $\vec a \in M$
such that $(M,\vec a) \models \varphi$.
Hence the result follows from the induction hypothesis.
\qedhere
\end{itemize}
%%%%%%%%%%%%%%%%%%%%%%%%%%%%%%%%%%%%%%%%%%%%%%%%%%%%%%%%%%%%%%%%%%%%%%%%%%%
\end{proof}
%%%%%%%%%%%%%%%%%%%%%%%%%%%%%%%%%%%%%%%%%%%%%%%%%%%%%%%%%%%%%%%%%%%%%%%%%%%

%%%%%%%%%%%%%%%%%%%%%%%%%%%%%%%%%%%%%%%%%%%%%%%%%%%%%%%%%%%%%%%%%%%%%%%%%%%
\begin{remark}
\label{rem:wc:cor:trans}
%%%%%%%%%%%%%%%%%%%%%%%%%%%%%%%%%%%%%%%%%%%%%%%%%%%%%%%%%%%%%%%%%%%%%%%%%%%
Theorem~\ref{thm:wc:cor:trans} is akin to~\cite[Theorem 9.3.2]{hodges93book}.
We make the following further comments:
\begin{enumerate}[(1)]
\item
\label{item:wc:cor:trans:lang}
By Theorem~\ref{thm:wc:cor:trans},
if $M$ and $N$ are equivalent
in $\Lang_{\infty}(\Sigbis)$,
then
$F M$ and $F N$ are equivalent in $\Lang_{\infty}(\Sig)$.

\item
Unlike the case of interpretations (cf.\ Remark~\ref{rem:coste:interp:card}), if $M$ and $N$ are equivalent
in $\Lang_{\omega}(\Sigbis)$
it does not necessarily follow that $F M$ and $F N$ are equivalent in $\Lang_{\omega}(\Sig)$.

\item
When $\Mod(\theorybis)$ and $\Mod(\theory)$
are categories of structures,
item~\ref{item:wc:cor:trans:lang}
can equivalently be deduced from~\cite[Theorems~2.18 and~2.25]{br18apal}.

Indeed, it follows from the proof of~\cite[Theorem~1.70]{ar94book}
that both $\Mod(\theorybis)$ and $\Mod(\theory)$
are $\omega$-mono generated in the sense of~\cite{br18apal}.
Then, by~\cite[Theorem~2.25]{br18apal},
if $\Mod(\theorybis)$ and $\Mod(\theory)$
are categories of structures
$\Struct(\Sigbis)$ and $\Struct(\Sig)$, respectively,
the above preservation of logical equivalence
is given by~\cite[Theorem~2.18]{br18apal}.
\end{enumerate}
%%%%%%%%%%%%%%%%%%%%%%%%%%%%%%%%%%%%%%%%%%%%%%%%%%%%%%%%%%%%%%%%%%%%%%%%%%%
\end{remark}
%%%%%%%%%%%%%%%%%%%%%%%%%%%%%%%%%%%%%%%%%%%%%%%%%%%%%%%%%%%%%%%%%%%%%%%%%%%

%%%%%%%%%%%%%%%%%%%%%%%%%%%%%%%%%%%%%%%%%%%%%%%%%%%%%%%%%%%%%%%%%%%%%%%%%%%
\subsection{Proof of Theorem~\ref{thm:path:main}}
\label{sec:wc:main}
%%%%%%%%%%%%%%%%%%%%%%%%%%%%%%%%%%%%%%%%%%%%%%%%%%%%%%%%%%%%%%%%%%%%%%%%%%%
We are now in a position to prove our main result,
which we recall for the reader's convenience.

\maintheorem*

Fix an adjunction 
\[
\begin{tikzcd}
  \C
  \ar[r, bend left, "\Ladj"]
  \ar[r, phantom, description, "\textnormal{\footnotesize{$\bot$}}"]
& \E
  \ar[l, bend left, "R"]
\end{tikzcd}
\]

\noindent
as in the statement above.
Recall that this means the following:
\begin{enumerate}[(1)]
\item
There is a cartesian theory $\theory$ in $\Sig$ such that $\E = \Mod(\theory)$.

\item
The wooded category $\C$ is lfp.
Hence, by Corollary~\ref{cor:lfp-iff-models-of-T}\ref{lfp-model-cat}
we can identify $\C$ with $\Mod(\theorybis)$
for a cartesian theory $\theorybis$ in a signature $\Sigbis$.

\item
The right adjoint $R \colon \E \to \C$ is finitary
(Definition~\ref{def:path:finaccadj});
equivalently, the left adjoint $\Ladj \colon \C \to \E$
preserves finitely presentable objects
(Proposition~\ref{p:lfp-morphisms-characterisation}).

\item
\label{item:wc:main:cond-detect-pe}
The paths of $\C$ are finitely presentable,
and a $\C$-morphism $f \colon P \to a$ with $P$ a path
is an embedding (in the sense of the wooded category $\C$) exactly
when $L f$ is an embedding of $\Sig$-structures in $\E$
(Definition~\ref{def:path:detection-path-emb}).
\setcounter{SplitEnum}{\value{enumi}}
\end{enumerate}
Condition~\ref{item:wc:main:cond-detect-pe} above
is the key part in the notion of detection of path embeddings.

We know from Corollary~\ref{cor:hintikka:wooded:R-bfe}
that the conclusion of Theorem~\ref{thm:path:main}
holds under the assumption that path embeddings in $\C$ are definable in $\Sigbis$
(Definition~\ref{def:hintikka:def-path-emb}), that is
\begin{enumerate}[(1)]
\setcounter{enumi}{\value{SplitEnum}}

\item
\label{item:wc:main:cond-definable-pe}
For each path $\FG{\vec y \mid \psi}$ in $\C$
there is a 
formula $\FEmb_{\C}\FG{\vec y | \psi}(\vec y) \in \Lang_\infty(\Sigbis)$
such that for each $a \in \C$
and each $h \colon \FG{\vec y \mid \psi} \to a$
taking $\vec y$ to $\vec b \in \I{\vec y \mid \psi}_a$,
\[
\begin{array}{l !{\quad\longiff\quad} l}
  \text{$h$ is an embedding in $\C$}
& a \models \FEmb_{\C}\FG{\vec y \mid \psi}(\vec b).
\end{array}
\]
\end{enumerate}

\noindent
Theorem~\ref{thm:path:main} is then an immediate consequence of the following lemma,
which shows that, in the setting of the theorem,
\ref{item:wc:main:cond-detect-pe} $\Rightarrow$ \ref{item:wc:main:cond-definable-pe}.
I.e., path embeddings in $\C$ are definable whenever they are detected.

%%%%%%%%%%%%%%%%%%%%%%%%%%%%%%%%%%%%%%%%%%%%%%%%%%%%%%%%%%%%%%%%%%%%%%%%%%%
\begin{lemma}
\label{lem:emb:main}
%%%%%%%%%%%%%%%%%%%%%%%%%%%%%%%%%%%%%%%%%%%%%%%%%%%%%%%%%%%%%%%%%%%%%%%%%%%
Consider an adjunction $\Ladj \colon \C \inadj \E \cocolon R$ as above, and
let $P = \FG{\vec y \mid \psi}$ be a path in $\C$.
There is a formula
\(
  \FEmb_{\C}[P](\vec y)
  \in \Lang_\infty(\Sigbis)
\)
such that for every $a \in \C$ and every homomorphism $h \colon \FG{\vec y \mid \psi} \to a$
taking $\vec y$ to $\vec b \in \I{\vec y \mid \psi}_a$,
the following are equivalent:
\begin{enumerate}[(i)]
\item h is an embedding,

\item
$a \models 
  \FEmb_{\C}[P](\vec b)$.
\end{enumerate}
%%%%%%%%%%%%%%%%%%%%%%%%%%%%%%%%%%%%%%%%%%%%%%%%%%%%%%%%%%%%%%%%%%%%%%%%%%%
\end{lemma}
%%%%%%%%%%%%%%%%%%%%%%%%%%%%%%%%%%%%%%%%%%%%%%%%%%%%%%%%%%%%%%%%%%%%%%%%%%%

%%%%%%%%%%%%%%%%%%%%%%%%%%%%%%%%%%%%%%%%%%%%%%%%%%%%%%%%%%%%%%%%%%%%%%%%%%%
\begin{proof}
%%%%%%%%%%%%%%%%%%%%%%%%%%%%%%%%%%%%%%%%%%%%%%%%%%%%%%%%%%%%%%%%%%%%%%%%%%%
Let $\FG{\vec x \mid \varphi}$ be the finitely presentable
model $\Ladj \FG{\vec y \mid \psi}=LP$. 
Corollary~\ref{cor:emb:emb}
provides us with
a formula
$\FEmb_{\E}\FG{\vec x \mid \varphi}(\vec x) \in \Lang_\infty(\Sig)$
such that
for each $a \in \C$
and each homomorphism $h \colon \FG{\vec y \mid \psi} \to a$,
the following are equivalent:
\begin{enumerate}[(A)]
\item\label{item:Lh-emb-struct}
$L h \colon \FG{\vec x \mid \varphi} \to L a$
is an embedding of structures in $\E$.

\item
\label{item:emb:main:c}
\(
  \Ladj a
  \models
  \FEmb_{\E}\FG{\vec x \mid \varphi}(\vec c)
\),
where $\Ladj h \colon \FG{\vec x \mid \varphi} \to \Ladj a$
is induced by
$\vec c \in \I{\vec x \mid \varphi}_{\Ladj a}$.

\setcounter{SplitEnum}{\value{enumi}}
\end{enumerate}
We will define the formula 
$\FEmb_{\C}[P](\vec y) \in \Lang_\infty(\Sigbis)$
as an appropriate translation of
the formula $\FEmb_{\E}\FG{\vec x \mid \varphi}(\vec x)$.
To this end, consider the word-construction
${\WC(\Ladj) \colon \Sigbis \to \Sig}$ induced
by the left adjoint ${\Ladj \colon \C \to \E}$ (\S\ref{sec:wc:filtcolim})
and the ensuing formula translation $(-)^{\WC(\Ladj)}$, cf.\ Theorem~\ref{thm:wc:cor:trans}.

Write $\vec x = x_1,\dots,x_n$.
For each $i = 1,\dots,n$,
consider
the element $\const x_i \in \FG{\vec x \mid \varphi}$
given by 
Lemma~\ref{lem:coste:fp:const}
and let
\(
  \vec f =
  f^{\const x_1}_{\FG{\vec y \mid \psi}}(\vec y_1)
  ,\dots,
  f^{\const x_n}_{\FG{\vec y \mid \psi}}(\vec y_n)
\)
be additional function symbols from $\Sig(\WC(\Ladj))$.
Recall from
Proposition~\ref{prop:wc:cor}
that each
$f^{\const x_i}_{\FG{\vec y \mid \psi}}(\vec y_i)$
is seen as the function
\[
\begin{array}{l l l}
  \vec b \in \I{\vec y \mid \psi}_{a}
& \longmapsto
& (\Ladj h)(\const x_i) \in \Ladj a
\end{array}
\]

\noindent
where $h \colon \FG{\vec y \mid \psi} \to a$
is induced by $\vec b$.
Let
\[
\begin{array}{l l l l l}
  \FEmb_{\C}[P](\vec y)  
  & \deq 
  & \FEmb_{\E}\FG{\vec x \mid \varphi}(\vec f)^{\WC(\Ladj)}(\vec y,\dots,\vec y)
  & \in 
  & \Lang_\infty(\Sigbis).
  \end{array}
\]
By Theorem~\ref{thm:wc:cor:trans},
the following are equivalent
for each $a \in \C$
and each $h \colon \FG{\vec y \mid \psi} \to a$:
\begin{enumerate}[(A)]
\setcounter{enumi}{\value{SplitEnum}}
\item
\label{item:emb:main:x}
\(
  \Ladj a \models 
  \FEmb_{\E}\FG{\vec x \mid \varphi}
  \left( (\Ladj h)(\const x_1) ,\dots, (\Ladj h)(\const x_n) \right)
\).

\item\label{item:tran-def-path-emb}
\(
  a \models 
  \FEmb_{\C}[P](\vec b)
\),
where $h \colon \FG{\vec y \mid \psi} \to a$ 
is induced by $\vec b \in \I{\vec y \mid \psi}_a$.
\end{enumerate}

Saying that $\Ladj h \colon \FG{\vec x \mid \varphi} \to \Ladj a$
is induced by
$\vec c \in \I{\vec x \mid \varphi}_{\Ladj a}$
means that $(\Ladj h)(\const x_i) = c_i$
for each $i = 1,\dots,n$,
hence~\ref{item:emb:main:c} is equivalent to~\ref{item:emb:main:x},
and so all four conditions~\ref{item:Lh-emb-struct}--\ref{item:tran-def-path-emb} are equivalent.
Since the adjunction detects path embeddings, condition~\ref{item:Lh-emb-struct} is equivalent to saying that $h$ is an embedding in $\C$. Thus the statement of the lemma follows.
%
%%%%%%%%%%%%%%%%%%%%%%%%%%%%%%%%%%%%%%%%%%%%%%%%%%%%%%%%%%%%%%%%%%%%%%%%%%%
\end{proof}
%%%%%%%%%%%%%%%%%%%%%%%%%%%%%%%%%%%%%%%%%%%%%%%%%%%%%%%%%%%%%%%%%%%%%%%%%%%

While we derived Theorem~\ref{thm:path:main} from
Corollary~\ref{cor:hintikka:wooded:R-bfe},
in \S\ref{sec:hintikka:finaccadj} we also devised Hintikka formulae in $\E$,
which are obtained from the formulae of
Proposition~\ref{prop:hintikka:wooded} (\S\ref{sec:hintikka:form})
using the interpretation induced by the finitary right adjoint $R \colon \E \to \C$.
This was stated in Corollary~\ref{cor:hintikka:interp}
which, combined with Lemma~\ref{lem:emb:main}, yields the following result.

%%%%%%%%%%%%%%%%%%%%%%%%%%%%%%%%%%%%%%%%%%%%%%%%%%%%%%%%%%%%%%%%%%%%%%%%%%%
\begin{corollary}
\label{cor:wc:hintikka}
%%%%%%%%%%%%%%%%%%%%%%%%%%%%%%%%%%%%%%%%%%%%%%%%%%%%%%%%%%%%%%%%%%%%%%%%%%%
Consider an adjunction $\Ladj \colon \C \inadj \E \cocolon R$ as above.
Let $m \colon P \emb R M$ be a path embedding in $\C$ with $M \in \E$.
Further, let $Q = \FG{\vec z \mid \varpi_Q}$ be a path in $\C$,
and write $\OI{\vec x \mid \varphi} \in \cat T$ for $\inclth \OI{\vec z \mid \varpi_Q}$.

For each ordinal $\ord$, there is a formula
$\Theta_{\E}[M,m,Q,\ord](\vec x) \in \Lang_\infty(\Sig)$
such that,
for each $N \in \E$ and each $n \colon Q \emb R N$
induced by 
$\vec c \in \I{\vec x \mid \varphi}_{N}$,
the following are equivalent:
\begin{enumerate}[(i)]
\item
$N \models \Theta_{\E}[M,m,Q,\ord](\vec c)$,

\item
$(m,n)$ is a position of rank $\geq \ord$
in the game $\G(R M,R N)$.
\end{enumerate}
%%%%%%%%%%%%%%%%%%%%%%%%%%%%%%%%%%%%%%%%%%%%%%%%%%%%%%%%%%%%%%%%%%%%%%%%%%%
\end{corollary}
%%%%%%%%%%%%%%%%%%%%%%%%%%%%%%%%%%%%%%%%%%%%%%%%%%%%%%%%%%%%%%%%%%%%%%%%%%%

%%%%%%%%%%%%%%%%%%%%%%%%%%%%%%%%%%%%%%%%%%%%%%%%%%%%%%%%%%%%%%%%%%%%%%%%%%%
\subsection{Finite relational signatures and finitary logic}
\label{sec:wc:fin}
%%%%%%%%%%%%%%%%%%%%%%%%%%%%%%%%%%%%%%%%%%%%%%%%%%%%%%%%%%%%%%%%%%%%%%%%%%%
Let $\Ladj \colon \C \inadj \E \cocolon R$ be a finitely accessible wooded adjunction which detects
path embeddings in a \emph{finite} relational signature $\sig$.
We discuss sufficient conditions under which we can strengthen Theorem~\ref{thm:path:main} by replacing infinitary logic with ordinary first-order logic, to the effect that
\[
\begin{array}{l l l}
  \text{$M,N \in \E$ equivalent in $\Lang_{\omega}(\sig)$}
& \longimp
& M \bisim_R N.
\end{array}
\]
These conditions essentially amount to the finiteness of the Hintikka formulae
$\Theta_{\E}$ of Corollary~\ref{cor:wc:hintikka}.
A useful observation (a proof of which will be given below) is the following.

%%%%%%%%%%%%%%%%%%%%%%%%%%%%%%%%%%%%%%%%%%%%%%%%%%%%%%%%%%%%%%%%%%%%%%%%%%%
\begin{proposition}
\label{prop:wc:fin}
%%%%%%%%%%%%%%%%%%%%%%%%%%%%%%%%%%%%%%%%%%%%%%%%%%%%%%%%%%%%%%%%%%%%%%%%%%%
Let $\E = \Struct(\sig)$ for a finite $\sig$
and assume that $\Ladj \colon \C \inadj \E \cocolon R$ detects path embeddings in $\sig$.
Assume further 
that for each path $P$ in $\C$ there are (up to isomorphism)
finitely many quotients $e \colon P \epi Q$ in $\C$.

Then, for each path $P$ in $\C$,
the formula $\FEmb_{\C}[P]$ of Lemma~\ref{lem:emb:main} 
is equivalent to a finite (i.e.\ first-order) formula.
%%%%%%%%%%%%%%%%%%%%%%%%%%%%%%%%%%%%%%%%%%%%%%%%%%%%%%%%%%%%%%%%%%%%%%%%%%%
\end{proposition}
%%%%%%%%%%%%%%%%%%%%%%%%%%%%%%%%%%%%%%%%%%%%%%%%%%%%%%%%%%%%%%%%%%%%%%%%%%%

Proposition~\ref{prop:wc:fin},
combined with Corollary~\ref{cor:hintikka:interp:fin},
yields finite Hintikka formulae for $R$-equivalence.
We state this observation in the next corollary, where we use the fact that given any wooded category $\C$, its subcategory $\pth\C$ defined by the paths and the embeddings between them can be identified, up to equivalence, with a forest order (see Lemma~\ref{l:forest-of-paths-wooded}).

%%%%%%%%%%%%%%%%%%%%%%%%%%%%%%%%%%%%%%%%%%%%%%%%%%%%%%%%%%%%%%%%%%%%%%%%%%%
\begin{corollary}
\label{cor:wc:fin}
%%%%%%%%%%%%%%%%%%%%%%%%%%%%%%%%%%%%%%%%%%%%%%%%%%%%%%%%%%%%%%%%%%%%%%%%%%%
Let $\E = \Struct(\sig)$ for a finite $\sig$
and assume that $\Ladj \colon \C \inadj \E \cocolon R$ detects path embeddings in $\sig$.
Assume the following:
\begin{enumerate}[(i)]
\item
\label{item:cor:wc:fin:path}
The forest $\pth\C$ is finitely branching.

\item
\label{item:cor:wc:fin:emb}
For each path $P$ in $\C$, there are (up to isomorphism)
finitely many quotients $e \colon P \epi Q$ in $\C$.

\setcounter{SplitEnum}{\value{enumi}}
\end{enumerate}

\noindent
Further, let $M, N \in \E$ such that the following conditions hold:
\begin{enumerate}[(i)]
\setcounter{enumi}{\value{SplitEnum}}
\item
\label{item:cor:wc:fin:ord}
The plays of $\G(R M, R N)$ have (finite) bounded length.

\item
\label{item:cor:wc:fin:struct}
For each path embedding $m \colon P \emb R M$ in $\C$,
there are (up to isomorphism) finitely many
embeddings $m' \colon P' \emb R M$ such that $m \prec m'$.
\end{enumerate}

\noindent
Then
\[
\begin{array}{l l l}
  \text{$M,N$ equivalent in $\Lang_{\omega}(\sig)$}
& \longimp
& M \bisim_R N.
\end{array}
\]
%%%%%%%%%%%%%%%%%%%%%%%%%%%%%%%%%%%%%%%%%%%%%%%%%%%%%%%%%%%%%%%%%%%%%%%%%%%
\end{corollary}
%%%%%%%%%%%%%%%%%%%%%%%%%%%%%%%%%%%%%%%%%%%%%%%%%%%%%%%%%%%%%%%%%%%%%%%%%%%

%%%%%%%%%%%%%%%%%%%%%%%%%%%%%%%%%%%%%%%%%%%%%%%%%%%%%%%%%%%%%%%%%%%%%%%%%%%
\begin{example}
\label{ex:fin:wc}
%%%%%%%%%%%%%%%%%%%%%%%%%%%%%%%%%%%%%%%%%%%%%%%%%%%%%%%%%%%%%%%%%%%%%%%%%%%
Consider a finite $\sig$,
let $k < \omega$ and
recall the adjunctions of Example~\ref{ex:path:main:fo}:
\[
\begin{tikzcd}
{\cat R^E_k(\sig^I)} \arrow[r, bend left=25, ""{name=U, below}, "\Ladj_k"{above}]
\arrow[r, leftarrow, bend right=25, ""{name=D}, "R_k"{below}]
& {\Struct(\sig^I)} \arrow[r, bend left=25, ""{name=U', below}, "H"{above}]
\arrow[r, leftarrow, bend right=25, ""{name=D'}, "J"{below}] & {\Struct(\sig)}
\arrow[phantom, "\textnormal{\footnotesize{$\bot$}}", from=U, to=D] 
\arrow[phantom, "\textnormal{\footnotesize{$\bot$}}", from=U', to=D'] 
\end{tikzcd}
\]

\noindent
Given some \emph{finite} $M \in \Struct(\sig)$, we obtain, as expected,
finite Hintikka formulae in $\sig$ for $R_k^I$-equivalence with $M$:
The $\sig^I$-structure $J M$ is finite, and it is easy to see that
all conditions of Corollary~\ref{cor:wc:fin} are met for $\Ladj_k \adj R_k$.
Now, given $N \in \Struct(\sig)$, if $N$ is $\Lang_\omega(\sig)$-equivalent to $M$,
then it follows from Remark~\ref{rem:coste:interp:card}
that $J N$ is $\Lang_\omega(\sig^I)$-equivalent to $J M$.
Hence, Corollary~\ref{cor:wc:fin} yields
that $J M$ and $J N$ are $R_k$-equivalent.
In other words, 
\begin{equation}
\label{eq:wc:fin}
\begin{array}{l l l}
  \text{$k$ finite, $M$ finite and $M,N$ equivalent in $\Lang_\omega(\sig)$}
& \longimp
& \text{$M,N$ $R_k^I$-equivalent}
\end{array}
\end{equation}
%%%%%%%%%%%%%%%%%%%%%%%%%%%%%%%%%%%%%%%%%%%%%%%%%%%%%%%%%%%%%%%%%%%%%%%%%%%
\end{example}
%%%%%%%%%%%%%%%%%%%%%%%%%%%%%%%%%%%%%%%%%%%%%%%%%%%%%%%%%%%%%%%%%%%%%%%%%%%

%%%%%%%%%%%%%%%%%%%%%%%%%%%%%%%%%%%%%%%%%%%%%%%%%%%%%%%%%%%%%%%%%%%%%%%%%%%
\begin{remark}
\label{rem:fin:wc}
%%%%%%%%%%%%%%%%%%%%%%%%%%%%%%%%%%%%%%%%%%%%%%%%%%%%%%%%%%%%%%%%%%%%%%%%%%%
The assumption that $M$ is finite in eq.~\eqref{eq:wc:fin}
is due to condition~\ref{item:cor:wc:fin:struct} of Corollary~\ref{cor:wc:fin}.
The latter is precisely condition~\ref{item:hintikka:interp:fin:struct} of
Corollary~\ref{cor:hintikka:interp:fin}.
It comes from the fact that 
our Hintikka formulae $\Theta_{\C}[R M,-,-,-]$
in Figure~\ref{fig:hintikka:arb} (\S\ref{sec:hintikka:form})
hide conjunctions and disjunctions which may (essentially)
range over all elements of $R M$.

This is in line with the Hintikka formulae 
in~\cite[Theorem 3.5.1]{hodges93book} and~\cite[Definition 2.2.5]{ef99book},
which are infinite when the involved structures are infinite.
But note that this differs from the formulae used in the 
``Hintikka-Fraïssé'' Theorem (see~\cite[Theorem 3.3.2]{hodges93book}):
the latter formulae represent elementary equivalence up to a given quantifier depth,
and are finite whenever so is $\sig$.
%%%%%%%%%%%%%%%%%%%%%%%%%%%%%%%%%%%%%%%%%%%%%%%%%%%%%%%%%%%%%%%%%%%%%%%%%%%
\end{remark}
%%%%%%%%%%%%%%%%%%%%%%%%%%%%%%%%%%%%%%%%%%%%%%%%%%%%%%%%%%%%%%%%%%%%%%%%%%%

The proof of Proposition~\ref{prop:wc:fin}
relies on Lemma~\ref{lem:wc:fin:arblocval} below.
We assume the same notations as in Lemma~\ref{lem:emb:main} and its proof.
Furthermore, we write $\cat T$ for the syntactic category of $\theory(\sig)$.

Let $\FG{\vec x \mid \varphi}$ be the finitely presentable
model $\Ladj\FG{\vec y \mid \psi}$,
where $\varphi$ is quantifier-free, and write $\vec x = x_1,\dots,x_n$.
In the proof of Lemma~\ref{lem:emb:main},
we started from a formula
$\FEmb_{\E}\FG{\vec x \mid \varphi}(\vec x)$
provided by Corollary~\ref{cor:emb:emb},
to which we applied the formula translation $(-)^{\WC(\Ladj)}$
induced by the word-construction $\WC(\Ladj)$ derived from the
left adjoint $\Ladj \colon \C \to \E$.
But since $\E = \Struct(\sig)$ for a (mono-sorted) purely relational $\sig$,
instead of starting from $\FEmb_{\E}\FG{\vec x \mid \varphi}(\vec x)$,
we can start from the formula
\[
  \FEmb^{\sig}\FG{\vec x \mid \varphi}(\vec x)
  ~=~
  \mathord{\bigwedge}_{\substack{
  \text{$\vec x \sorting \atom$ atomic}
  \\
  \FG{\vec x \mid \varphi} \models \lnot \atom(\const x_1,\dots,\const x_n)
  }}
  \lnot \atom(\vec x)
\]

\noindent
devised in eq.~\eqref{eq:emb:sig} (\S\ref{sec:emb}).
It follows that in Lemma~\ref{lem:emb:main},
we can take for
 $\FEmb_{\C}[P](\vec y)$
the formula
\[
  \FEmb^{\sig}\FG{\vec x \mid \varphi}(\vec f)^{\WC(\Ladj)}(\vec y,\dots,\vec y)
\]

\noindent
where
\(
  \vec f
  =
  f^{\const x_1}_{\FG{\vec y \mid \psi}}(\vec y_1)
  ,\dots,
  f^{\const x_n}_{\FG{\vec y \mid \psi}}(\vec y_n)
\)
and where
\(
  \FEmb^{\sig}\FG{\vec x \mid \varphi}(\vec f)^{\WC(\Ladj)}(\vec y_1,\dots,\vec y_n)
\)
is given
by Theorem~\ref{thm:wc:cor:trans} from the above formula
$\FEmb^{\sig}\FG{\vec x \mid \varphi}(\vec x)$.

It remains to show that the formula
$\FEmb^{\sig}\FG{\vec x \mid \varphi}(\vec f)^{\WC(\Ladj)}$
can be assumed to be finite
under the assumptions of Proposition~\ref{prop:wc:fin}.
We rely on the following variation on Lemma~\ref{lem:wc:locval}.

%%%%%%%%%%%%%%%%%%%%%%%%%%%%%%%%%%%%%%%%%%%%%%%%%%%%%%%%%%%%%%%%%%%%%%%%%%%
\begin{lemma}
\label{lem:wc:fin:arblocval}
%%%%%%%%%%%%%%%%%%%%%%%%%%%%%%%%%%%%%%%%%%%%%%%%%%%%%%%%%%%%%%%%%%%%%%%%%%%
Let $a \in \C = \Mod(\theorybis)$
and let $P$ be a path
of the form $\FG{\vec y \mid \psi}$.
Assume $\Ladj P = \FG{\vec x \mid \varphi}$
with $\vec x = x_1,\dots,x_n$.
Fix an atomic formula $(\vec x \sorting \atom)$ in $\sig$.

Given $h \colon P \to a$,
and given
$c_i \in \Ladj\FG{\vec y \mid \psi}$ for $i=1,\dots,n$,
the following are equivalent.
\begin{enumerate}[(i)]
\item
\label{item:wc:fin:arblocval:mod}
$\Ladj a \models \atom((\Ladj h)(c_1),\dots,(\Ladj h)(c_n))$

\item
\label{item:wc:fin:arblocval:diag}
$h$ factors as
\[
\begin{tikzcd}
  P
  \arrow{r}[above]{h}
  \arrow[twoheadrightarrow]{d}[left]{e}
& a
\\
  Q
  \arrow{ur}[right]{g}
\end{tikzcd}
\]

\noindent
for some path $Q$ and some quotient $e \colon P \epi Q$ such that
\[
\begin{array}{l l l}
  \Ladj Q
& \models
& \atom\big( (\Ladj e)(c_1),\dots, (\Ladj e)(c_n) \big)
\end{array}
\]
\end{enumerate}
%%%%%%%%%%%%%%%%%%%%%%%%%%%%%%%%%%%%%%%%%%%%%%%%%%%%%%%%%%%%%%%%%%%%%%%%%%%
\end{lemma}
%%%%%%%%%%%%%%%%%%%%%%%%%%%%%%%%%%%%%%%%%%%%%%%%%%%%%%%%%%%%%%%%%%%%%%%%%%%

%%%%%%%%%%%%%%%%%%%%%%%%%%%%%%%%%%%%%%%%%%%%%%%%%%%%%%%%%%%%%%%%%%%%%%%%%%%
\begin{proof}
%%%%%%%%%%%%%%%%%%%%%%%%%%%%%%%%%%%%%%%%%%%%%%%%%%%%%%%%%%%%%%%%%%%%%%%%%%%
Assume $\Ladj h \colon \FG{\vec x \mid \varphi} \to \Ladj a$
takes $\vec x$ to $\vec b \in \I{\vec x \mid \varphi}_{\Ladj a}$.

We first show that
\(
   \text{\ref{item:wc:fin:arblocval:diag}}
   \imp
   \text{\ref{item:wc:fin:arblocval:mod}}
\).
For each $i = 1,\dots,n$ we have
$(\Ladj h)(c_i) = (\Ladj g)((\Ladj e)(c_i))$.
Hence 
$\Ladj a \models \atom((\Ladj h)(c_1),\dots,(\Ladj h)(c_n))$
since $\Ladj g \colon \Ladj Q \to \Ladj a$ is a homomorphism.

For the converse, factor $h$ as $m \comp e$,
where $e \in \Q$ and $m \in \M$.
Then by Lemma~\ref{lem:path:base}
the codomain of $e$ is a path, say $Q$,
and we have
\(
  \Ladj Q
  \models
  \atom\big( (\Ladj e)(c_1),\dots, (\Ladj e)(c_n) \big)
\)
since $\Ladj m$ is an embedding of structures with
$(\Ladj h)(c_i) = (\Ladj m)((\Ladj e)(c_i))$
for all $i = 1,\dots,n$.
%%%%%%%%%%%%%%%%%%%%%%%%%%%%%%%%%%%%%%%%%%%%%%%%%%%%%%%%%%%%%%%%%%%%%%%%%%%
\end{proof}
%%%%%%%%%%%%%%%%%%%%%%%%%%%%%%%%%%%%%%%%%%%%%%%%%%%%%%%%%%%%%%%%%%%%%%%%%%%

%%%%%%%%%%%%%%%%%%%%%%%%%%%%%%%%%%%%%%%%%%%%%%%%%%%%%%%%%%%%%%%%%%%%%%%%%%%
\begin{remark}
%%%%%%%%%%%%%%%%%%%%%%%%%%%%%%%%%%%%%%%%%%%%%%%%%%%%%%%%%%%%%%%%%%%%%%%%%%%
Note that the proof of 
\(
   \text{\ref{item:wc:fin:arblocval:mod}}
   \imp
   \text{\ref{item:wc:fin:arblocval:diag}}
\)
implies that we can always take an arboreal embedding for $g$,
while this is not assumed in
condition \ref{item:wc:fin:arblocval:diag}.
Also, the assumption that $e$ is a quotient is not used in
the implication
\(
   \text{\ref{item:wc:fin:arblocval:diag}}
   \imp
   \text{\ref{item:wc:fin:arblocval:mod}}
\).

The actual formulation of
condition \ref{item:wc:fin:arblocval:diag}
is tailored to obtain finite formulae under the
word-construction $\WC(\Ladj)$.
%%%%%%%%%%%%%%%%%%%%%%%%%%%%%%%%%%%%%%%%%%%%%%%%%%%%%%%%%%%%%%%%%%%%%%%%%%%
\end{remark}
%%%%%%%%%%%%%%%%%%%%%%%%%%%%%%%%%%%%%%%%%%%%%%%%%%%%%%%%%%%%%%%%%%%%%%%%%%%

We can now prove Proposition~\ref{prop:wc:fin}.

%%%%%%%%%%%%%%%%%%%%%%%%%%%%%%%%%%%%%%%%%%%%%%%%%%%%%%%%%%%%%%%%%%%%%%%%%%%
\begin{proof}[Proof of Proposition~\ref{prop:wc:fin}]
%%%%%%%%%%%%%%%%%%%%%%%%%%%%%%%%%%%%%%%%%%%%%%%%%%%%%%%%%%%%%%%%%%%%%%%%%%%
Lemma~\ref{lem:wc:fin:arblocval}
(together with Corollary~\ref{cor:coste:fp:triangle})
allows for some simplication
of the formula $\WCAt(\WC(\Ladj))(\atombis)(\vec y)$
with
\[
\begin{array}{l l l}
  \atombis
& =
& \atom \left(
  f^{\const x_1}_{\FG{\vec y \mid \psi}}(\vec y)
  ,\dots,
  f^{\const x_n}_{\FG{\vec y \mid \psi}}(\vec y)
  \right)
\end{array}
\]

\noindent
where $\vec x \sorting \atom$ is atomic (in $\sig$),
and where the path $P$ is $\FG{\vec y \mid \psi}$.
Namely, assuming $\Ladj P = \FG{\vec x \mid \varphi}$
we can take
$\WCAt(\WC(\Ladj))(\atombis)(\vec y)$
to be the formula
\[
  \bigvee_{\substack{
    e \colon P \epi Q
    \\
    \Ladj Q \models \atom( (\Ladj e)(\const x_1),\dots, (\Ladj e)(\const x_n) )
  }}
  (\exists \vec z) \theta_e(\vec z,\vec y)
\]

\noindent
where $Q \cong \FG{\vec z \mid \vartheta}$
and where
\[
\begin{array}{*{5}{l}}
  \MI{\vec z, \vec y \mid \theta_e}
& :
& \OI{\vec z \mid \vartheta}
& \longto
& \OI{\vec y \mid \psi}
\end{array}
\]

\noindent
is such that $\yoneda \theta_e = e$.

The formula 
$\WCAt(\WC(\Ladj))(\atombis)(\vec y)$
can be assumed to be finite if there are (up to isomorphism)
finitely many
arboreal quotients $e \colon P \epi Q$ in $\C$.
It follows that the 
above formulae
$\FEmb^{\sig}\FG{\vec x \mid \varphi}(\vec f)^{\WC(\Ladj)}$
can be assumed to be finite in this case.
%%%%%%%%%%%%%%%%%%%%%%%%%%%%%%%%%%%%%%%%%%%%%%%%%%%%%%%%%%%%%%%%%%%%%%%%%%%
\end{proof}
%%%%%%%%%%%%%%%%%%%%%%%%%%%%%%%%%%%%%%%%%%%%%%%%%%%%%%%%%%%%%%%%%%%%%%%%%%%

%%%%%%%%%%%%%%%%%%%%%%%%%%%%%%%%%%%%%%%%%%%%%%%%%%%%%%%%%%%%%%%%%%%%%%%%%%%
\section{The factorisation property}
\label{sec:fact}
%%%%%%%%%%%%%%%%%%%%%%%%%%%%%%%%%%%%%%%%%%%%%%%%%%%%%%%%%%%%%%%%%%%%%%%%%%%
In the main examples of finitely accessible wooded adjunctions
\[
\begin{tikzcd}
  \C
  \arrow[r, bend left=25, ""{name=U, below}, "\Ladj"{above}]
  \arrow[r, leftarrow, bend right=25, ""{name=D}, "R"{below}]
& \E
  \arrow[phantom, "\textnormal{\footnotesize{$\bot$}}", from=U, to=D] 
\end{tikzcd}
\]

\noindent
presented so far, the extensional category $\E$
is a category of structures, while Theorem~\ref{thm:path:main}
covers the more general case where $\E$ is the category of models for a cartesian theory. This was done mainly with an eye to future applications and also because it is the natural setting for Gabriel--Ulmer duality.

On the other hand, note that given a finitely accessible wooded adjunction
\[
\C \inadj \Struct(\Sig)
\] 

\noindent
and a cartesian theory~$\theory$ in the signature $\Sig$,
we can compose the latter adjunction with the reflection of
$\Struct(\Sig)$ into $\Mod(\theory)$,
which is in particular a finitely accessible adjunction
$\Struct(\Sig) \inadj \Mod(\theory)$
(cf.\ Remark~\ref{rem:lfp:struct} and Lemma~\ref{lem:coste:mod:filtcolim}).
The result is a finitely accessible wooded adjunction between $\C$ and $\Mod(\theory)$:
\begin{equation}\label{eq:decomp-reflection}
\C \inadj \Struct(\Sig) \inadj \Mod(\theory).
\end{equation}

The way in which the equality symbol is encoded in the context of game comonads (cf.\ Example~\ref{ex:path:bisim-FOk-equivalence} and eq.~\eqref{eq:EF-wooded-adjunction}) can be seen as an example of this phenomenon, by identifying the category $\Struct(\sig)$ with the category of models for the cartesian theory 
\[
I(x,y) \thesis_{x,y} x \Eq y, \ \ \ \ \ x \Eq y \thesis_{x,y} I(x,y)
\] 
in the signature $\sig^{I}$. We now provide another example:

\begin{example}
Recall that the modal logic $\textbf{S4}$ is the logic of reflexive and transitive Kripke frames (see e.g.\ \cite[Theorem 4.29]{brv02modal}). Assume for simplicity that $\sig$ is a modal vocabulary consisting of a single binary relation and no unary relations. Then the corresponding (pointed) Kripke models and reflexive transitive Kripke models are, respectively, (pointed) graphs and preorders. 
Lemma~\ref{lem:coste:mod:filtcolim} gives a finitely accessible adjunction
\[
\begin{tikzcd}
  \Gph_{\bullet}
  \arrow[r, bend left=25, ""{name=U, below}, "F"{above}]
  \arrow[r, leftarrow, bend right=25, ""{name=D}, "U"{below}]
& \Pre_{\bullet}
  \arrow[phantom, "\textnormal{\footnotesize{$\bot$}}", from=U, to=D] 
\end{tikzcd}
\]

\noindent
where $U$ is the inclusion of the full subcategory $\Pre_{\bullet}$ of pointed preorders
into the category $\Gph_{\bullet}$ of pointed graphs and their homomorphisms,\footnote{By \emph{graphs} we mean directed graphs, possibly with self-loops. Note that $\Gph_{\bullet}$ can be identified with a category of structures by adding a constant symbol to the signature, and $\Pre_{\bullet}$ is then the category of models for a cartesian theory in this extended signature.}
and $F$ produces a preorder from a graph by 
taking the reflexive-transitive closure of the edge relation. Composing this adjunction with the finitely accessible arboreal adjunction induced by the modal comonad (see Example~\ref{ex:path:misc:games}\ref{item:path:misc:games:modal}), for each $k\in\mathbb{N}$ we get a finitely accessible arboreal adjunction 
\begin{equation*}
\begin{tikzcd}
{\cat R^M_k(\sig)} \arrow[r, bend left=25, ""{name=U, below}, "L^M_k"{above}]
\arrow[r, leftarrow, bend right=25, ""{name=D}, "R^M_k"{below}]
& {\Gph_{\bullet}} \arrow[r, bend left=25, ""{name=U', below}, "F"{above}]
\arrow[r, leftarrow, bend right=25, ""{name=D'}, "U"{below}] & {\Pre_{\bullet}}
\arrow[phantom, "\textnormal{\footnotesize{$\bot$}}", from=U, to=D] 
\arrow[phantom, "\textnormal{\footnotesize{$\bot$}}", from=U', to=D'] 
\end{tikzcd}
\end{equation*}
which encodes the $\textbf{S4}$\emph{-unravelling up to level $k$} of a pointed Kripke frame, cf.\ e.g.\ \cite[pp.~220--221]{brv02modal}. This can be extended in a straightforward way to pointed Kripke models in the case where~$\sig$ contains an arbitrary number of unary relations, by considering the reflection of Kripke models into reflexive transitive Kripke models.
\end{example}

The finitely accessible wooded adjunctions of the form $\C \inadj \Struct(\Sig)$ are technically easier to deal with, especially when $\Sig$ is a mono-sorted relational signature. In this case, the formulae $ \FEmb_{\C}[P](\vec y)$
of Lemma~\ref{lem:emb:main} can be considerably simplified (cf.~Proposition~\ref{prop:wc:fin} and its proof).
Ultimately, this avoids much of the technical development of~\S\ref{sec:emb}, and requires applying the formula translation induced by word-constructions only to quantifier-free formulae.

One might ask whether it is always possible to reduce ourselves to the case where $\E$ is a category of structures; this leads us to the following definition.

%%%%%%%%%%%%%%%%%%%%%%%%%%%%%%%%%%%%%%%%%%%%%%%%%%%%%%%%%%%%%%%%%%%%%%%%%%%
\begin{definition}
\label{def:fact}
%%%%%%%%%%%%%%%%%%%%%%%%%%%%%%%%%%%%%%%%%%%%%%%%%%%%%%%%%%%%%%%%%%%%%%%%%%%
Let $\E = \Mod(\theory)$
with $\theory$ a cartesian theory in a signature~$\Sig$.
A finitely accessible wooded adjunction $\Ladj \colon \C \inadj \E \cocolon R$
has the \emph{factorisation property} provided that
it can be decomposed (up to natural isomorphism) as a finitely accessible adjunction $\C \inadj \Struct(\Sig)$
followed by the reflection $\Struct(\Sig) \inadj \E$,
as in eq.~\eqref{eq:decomp-reflection}.
%%%%%%%%%%%%%%%%%%%%%%%%%%%%%%%%%%%%%%%%%%%%%%%%%%%%%%%%%%%%%%%%%%%%%%%%%%%
\end{definition}
%%%%%%%%%%%%%%%%%%%%%%%%%%%%%%%%%%%%%%%%%%%%%%%%%%%%%%%%%%%%%%%%%%%%%%%%%%%

%%%%%%%%%%%%%%%%%%%%%%%%%%%%%%%%%%%%%%%%%%%%%%%%%%%%%%%%%%%%%%%%%%%%%%%%%%%
\begin{example}
\label{ex:fact:Diaconescu-cover}
%%%%%%%%%%%%%%%%%%%%%%%%%%%%%%%%%%%%%%%%%%%%%%%%%%%%%%%%%%%%%%%%%%%%%%%%%%%
Let $\cat D$ be a small category and recall from
Example~\ref{ex:Diaconescu-cover} the finitely accessible arboreal adjunction
\[
\begin{tikzcd}
  \presh{\forest}
  \arrow[bend left=25]{r}{\eadj\pi}
  \arrow[phantom]{r}[description]{\textnormal{\footnotesize{$\bot$}}}
& \presh{\cat D}
  \arrow[bend left=25]{l}{\ladj\pi}
\end{tikzcd}
\]

\noindent
where $\forest = \forest(\cat D)$ is the forest of finite sequences of composable
arrows in $\cat D$
(see also Example~\ref{ex:hintikka:mono:forest} and \S\ref{sec:emb:presh}).
We saw in Example~\ref{ex:prelim:coste:funct} that there exists an (equational)
theory $\theory = \theory(\cat D^\op)$ in a signature $\Sig = \Sig(\cat D^\op)$
such that $\presh{\cat D} \cong \Mod(\theory)$.
It is possible to show that the adjunction above has the factorisation property,
i.e.\ it can be decomposed as a finitely accessible adjunction
$\presh{\forest} \inadj \Struct(\Sig)$ followed by the reflection
$\Struct(\Sig) \inadj \presh{\cat D}$.
%%%%%%%%%%%%%%%%%%%%%%%%%%%%%%%%%%%%%%%%%%%%%%%%%%%%%%%%%%%%%%%%%%%%%%%%%%%
\end{example}
%%%%%%%%%%%%%%%%%%%%%%%%%%%%%%%%%%%%%%%%%%%%%%%%%%%%%%%%%%%%%%%%%%%%%%%%%%%

%%%%%%%%%%%%%%%%%%%%%%%%%%%%%%%%%%%%%%%%%%%%%%%%%%%%%%%%%%%%%%%%%%%%%%%%%%%
\begin{fullproof}
%%%%%%%%%%%%%%%%%%%%%%%%%%%%%%%%%%%%%%%%%%%%%%%%%%%%%%%%%%%%%%%%%%%%%%%%%%%

Recall from Example~\ref{ex:Diaconescu-cover}
the functor $\pi \colon \forest \to \cat D$
which induces the adjunction $\eadj\pi \adj \ladj\pi$.
Also, recall from Example~\ref{ex:prelim:coste:funct}
the signature $\Sig = \Sig(\cat D^\op)$ and
the (equational) theory $\theory = \theory(\cat D^\op)$ in $\Sig$
such that $\presh{\cat D} \cong \Mod(\theory)$.
Write $F \colon \Struct(\Sig) \inadj \presh{\cat D} \cocolon U$
for the adjunction induced from Lemma~\ref{lem:coste:mod:filtcolim}.

The adjunction $\eadj\pi \adj \ladj\pi$ factors (up to iso) as
\[
\begin{tikzcd}
  \presh{\forest}
  \arrow[phantom]{r}[description]{\bot}
  \arrow[bend left]{r}{L}
& \Struct(\Sig)
  \arrow[phantom]{r}[description]{\bot}
  \arrow[bend left]{l}{R}
  \arrow[bend left]{r}{F}
& \presh{\cat D}.
  \arrow[bend left]{l}{U}
\end{tikzcd}
\]

The finitely accessible adjunction $L \adj R$ is obtained as follows.
Write $\cat S$ for the syntactic category of the empty cartesian theory in $\Sig$.
There is a functor $\forest \to \cat S^\op$
which takes a node $s \in \forest$, say
\[
\begin{tikzcd}
  a_0
  \arrow{r}{k_1}
& a_1
  \arrow[dashed]{r}
& a_{n-1}
  \arrow{r}{k_{n}}
& a_n
\end{tikzcd}
\]

\noindent
to the $\cat S$-object $\OI{x:a_n | \True}$
and such that given a node $t \in \forest$ of the form
\[
\begin{tikzcd}
  a_0
  \arrow[dashed]{r}{s}
& a_n
  \arrow{r}{k_{n+1}}
& a_{n+1}
  \arrow[dashed]{r}
& a_{n+m-1}
  \arrow{r}{k_{n+m}}
& a_{n+m}
\end{tikzcd}
\]

\noindent
the $\forest$-morphism $s \leq t$ 
is taken to the $\cat S$-morphism
\[
\begin{array}{*{5}{l}}
  \MI{x : a_{n+m},y:a_n \mid y \Eq f_{k_{n+1}}(\dots f_{k_{n+m}}(x))}
& :
& \OI{x : a_{n+m} \mid \True}
& \longto
& \OI{y:a_n \mid \True}
\end{array}
\]

\noindent
Note that the identity on $s \in \forest$
is taken to the $\cat S$-identity $\MI{x:a_n,y:a_n \mid y \Eq x}$,
while $\forest$-compositions are taken to term substitutions.

By Theorem~\ref{thm:coste:synt},
this induces a functor
$\varpi \colon \forest \to \Struct(\Sig)$.
The functor $R \colon \Struct(\Sig) \to \presh\forest$
taking $M$ to $\Struct(\Sig)\funct{\varpi(-),M}$
is a right adjoint (see e.g.~\cite[Theorem I.5.2]{mm92book}).
Note that $R$ is finitary since
the image of $\varpi$ consists of finitely presentable $\Sig$-structures.

Moreover, $F \comp \varpi$ takes $s \in \forest$ as above to the presheaf
corresponding to the $\theory$-model $\FG{x:a_n \mid \True}$,
that is, to $\yoneda a_n$.
Hence,
$F \comp \varpi \cong \yoneda \comp \pi$
and
$R \comp U$ takes a presheaf $P \in \presh{\cat D}$
to
\[
\begin{array}{*{7}{l}}
  \Struct(\Sig)\funct{\varpi(-),U P}
& \cong
& \presh{\cat D}\funct{F \varpi(-), P}
& \cong
& \presh{\cat D}\funct{(\yoneda \comp \pi)(-), P}
& \cong
& P \pi
\end{array}
\]

\noindent
It follows that $R \comp U \cong \ladj\pi$.
%%%%%%%%%%%%%%%%%%%%%%%%%%%%%%%%%%%%%%%%%%%%%%%%%%%%%%%%%%%%%%%%%%%%%%%%%%%
\end{fullproof}
%%%%%%%%%%%%%%%%%%%%%%%%%%%%%%%%%%%%%%%%%%%%%%%%%%%%%%%%%%%%%%%%%%%%%%%%%%%

We now give two examples of finitely wooded adjunctions that do \emph{not}
satisfy the factorisation property (the second is even an arboreal adjunction).
The first example is less relevant from a logical point of view,
but it is a useful warm-up for the second.

%%%%%%%%%%%%%%%%%%%%%%%%%%%%%%%%%%%%%%%%%%%%%%%%%%%%%%%%%%%%%%%%%%%%%%%%%%%
\begin{example}
\label{ex:fact:pos}
%%%%%%%%%%%%%%%%%%%%%%%%%%%%%%%%%%%%%%%%%%%%%%%%%%%%%%%%%%%%%%%%%%%%%%%%%%%
Let $\sig$ be the signature $\{\Leq\}$ of posets (see Example~\ref{ex:prelim:coste:pos}),
and recall from Example~\ref{ex:path:pos} the finitely accessible wooded (self-)adjunction
\[
\begin{tikzcd}
  \Pos
  \arrow[r, bend left=25, ""{name=U, below}, "\Id"{above}]
  \arrow[r, leftarrow, bend right=25, ""{name=D}, "\Id"{below}]
& \Pos
  \arrow[phantom, "\textnormal{\footnotesize{$\bot$}}", from=U, to=D] 
\end{tikzcd}
\]

\noindent
which also detects path embeddings.
Lemma~\ref{lem:coste:mod:filtcolim} yields a finitely accessible adjunction
\[
\begin{tikzcd}
  \Gph
  \arrow[r, bend left=25, ""{name=U, below}, "F"{above}]
  \arrow[r, leftarrow, bend right=25, ""{name=D}, "U"{below}]
& \Pos
  \arrow[phantom, "\textnormal{\footnotesize{$\bot$}}", from=U, to=D] 
\end{tikzcd}
\]
where $U$ is the inclusion functor and
its left adjoint $F$
takes a graph to the poset of its strongly connected components (SCCs).

We claim that there is no adjunction
$L \colon \Pos \inadj \Gph \cocolon R$
such that $\Id \adj \Id$ factors (up to natural isomorphism) as
\[
\begin{tikzcd}
  \Pos
  \arrow[phantom]{r}[description]{\textnormal{\footnotesize{$\bot$}}}
  \arrow[bend left]{r}{L}
& \Gph
  \arrow[phantom]{r}[description]{\textnormal{\footnotesize{$\bot$}}}
  \arrow[bend left]{l}{R}
  \arrow[bend left]{r}{F}
& \Pos.
  \arrow[bend left]{l}{U}
\end{tikzcd}
\]
To see this, assume towards a contradiction that
such an adjunction $L \adj R$ exists, i.e.\  
$F L \cong \Id$ and $R U \cong \Id$.

We shall use the well-known fact that the inclusion functor
$U \colon \Pos \to \Gph$ preserves coproducts but not coequalisers.
Let $|3|$ be the discrete poset $\{0,1,2\}$,
and let $2$ be the poset $\{0 \leq 1\}$.
Consider the monotone maps $f,g \colon |3| \rightrightarrows 2+2$
\[
\begin{tikzcd}[column sep=tiny]
& 0
  \arrow[mapsto]{dl}[left]{f,g}
&
& 1
  \arrow[mapsto]{dl}[left]{f}
  \arrow[mapsto]{dr}{g}
& 
& 2
  \arrow[mapsto]{dr}{f,g}
& 
\\
  a
  \arrow[phantom]{rr}[description]{\leq}
&
& b
& 
& a'
  \arrow[phantom]{rr}[description]{\leq}
&
& b'
\end{tikzcd}
\]

\noindent
where we represent $2+2$ as $\{a\leq b\} + \{a' \leq b'\}$.
The coequaliser of $f$ and $g$
in $\Pos$ is $3 = \{0 \leq 1 \leq 2\}$,
while their coequaliser
in $\Gph$ is the following non-transitive graph:
\[
\begin{tikzcd}
  0
  \arrow[loop left]
  \arrow{r}
& 1
  \arrow[loop above]
  \arrow{r}
& 2
  \arrow[loop right]
\end{tikzcd}
\]

Since $L \colon \Pos \to \Gph$ is left adjoint,
in $\Gph$ we get a coequaliser diagram
\[
\begin{tikzcd}
  L \vert 3 \vert
  \arrow[shift left]{r}[above]{L f}
  \arrow[shift right]{r}[below]{L g}
& L 2 + L 2
  \arrow[twoheadrightarrow]{r}
& L 3.
\end{tikzcd}
\]

\noindent
On the other hand, since $F L \cong \Id$,
the graph $L|3|$ has three SCCs (and no edges between them),
$L 2$ is a graph with two SCCs, say $S_0, S_1$, and with at least one edge
from~$S_0$ to~$S_1$ but no edges from~$S_1$ to~$S_0$.
Similarly, $L 3$ has three SCCs, say $T_0,T_1,T_2$,
with edges from $T_0$ to $T_1$ and from~$T_1$ to~$T_2$
(and not the other way around).
Moreover, since $L 3$ is the coequaliser of $L f$ and $L g$ in $\Gph$,
there are no edges from $T_0$ to $T_2$ in $L 3$.

Consider now the monotone map $h \colon 2 \to 3$
that takes $0$ to $0$, and $1$ to $2$.
In $L 2$, there are vertices $u \in S_0$ and $v \in S_1$
with an edge from $u$ to $v$.
Using again the fact that~$F L \cong \Id$, it must be the case that
$L h(u) \in T_0$ and $L h(v) \in T_2$.
But then $L h$ cannot be a graph homomorphism because 
there is no edge from $T_0$ to $T_2$ in $L 3$.
A contradiction.
\qed
%%%%%%%%%%%%%%%%%%%%%%%%%%%%%%%%%%%%%%%%%%%%%%%%%%%%%%%%%%%%%%%%%%%%%%%%%%%
\end{example}
%%%%%%%%%%%%%%%%%%%%%%%%%%%%%%%%%%%%%%%%%%%%%%%%%%%%%%%%%%%%%%%%%%%%%%%%%%%

We saw in~\S\ref{sec:arboreal:def} that the category $\Pos$ is not arboreal,
so the wooded adjunction in the previous example is not arboreal.
We will now present an example of a finitely accessible arboreal adjunction,
involving ``deterministic'' Kripke structures,
which does not have the factorisation property.

%%%%%%%%%%%%%%%%%%%%%%%%%%%%%%%%%%%%%%%%%%%%%%%%%%%%%%%%%%%%%%%%%%%%%%%%%%%
\begin{example}
\label{ex:fact:det}
%%%%%%%%%%%%%%%%%%%%%%%%%%%%%%%%%%%%%%%%%%%%%%%%%%%%%%%%%%%%%%%%%%%%%%%%%%%
Let $\cat D$ be the full subcategory of $\Gph$ defined by \emph{deterministic graphs}, 
i.e.\ graphs $G=(V,E)$ such that for all $u \in V$
there is at most one $v \in V$ satisfying~$E(u,v)$.
Note that $\cat D$ is the category of models of the Horn (in particular, cartesian)
theory 
\[
\{E(x,y)\land E(x,z) \thesis_{x,y,z} y \Eq z\}.
\]

Consider the finitely accessible adjunction
$\Ladj^M \colon \Forest \inadj \Gph \cocolon R^M$
where $\Ladj^M$ sends a forest to the graph of its covering relation,
and $R^M$ sends a graph to the forest of its finite walks.
This restricts to a wooded adjunction
\[
\begin{tikzcd}
  \A
  \arrow[bend left]{r}{L^D}
  \arrow[phantom]{r}[description]{\textnormal{\footnotesize{$\bot$}}}
& \cat D
  \arrow[bend left]{l}{R^D}
\end{tikzcd}
\]
where~$\A$ is 
the full subcategory of $\Forest$
consisting of those forests in which the covering relation is \emph{deterministic},
i.e.\ each node is covered by at most one element.
In other words, $\A$ consists of those forests which are disjoint unions of branches.
The category~$\A$ is lfp, and is arboreal
with respect to the (surjective, injective) factorisation system.
Therefore, $L^D \colon \A \inadj \cat D \cocolon R^D$ is a finitely accessible arboreal adjunction detecting path embeddings.

In view of Lemma~\ref{lem:coste:mod:filtcolim}, there is a finitely accessible adjunction
\[
\begin{tikzcd}
  \Gph
  \arrow[bend left]{r}{F}
  \arrow[phantom]{r}[description]{\textnormal{\footnotesize{$\bot$}}}
& \cat D
  \arrow[bend left]{l}{U}
\end{tikzcd}
\]
where $U$ is the full embedding
and $F$ takes a graph $G = (V,E)$ to its quotient under the
equivalence relation generated by the relation 
\[
\{(u,v)\in V\times V \mid \exists w\in V \ \text{such that $E(w,u)$ and $E(w,v)$}\}.
\]
We claim that there is no finitely accessible adjunction
$L \colon \A \inadj \Gph \cocolon R$
such that the adjunction $L^D \adj R^D$
factors (up to natural isomorphism) as
\[
\begin{tikzcd}
  \A
  \arrow[phantom]{r}[description]{\textnormal{\footnotesize{$\bot$}}}
  \arrow[bend left]{r}{L}
& \Gph
  \arrow[phantom]{r}[description]{\textnormal{\footnotesize{$\bot$}}}
  \arrow[bend left]{l}{R}
  \arrow[bend left]{r}{F}
& \cat D.
  \arrow[bend left]{l}{U}
\end{tikzcd}
\]

Assume toward a contradiction that such an adjunction exists 
and so $F L \cong L^D$ and $R U \cong R^D$.
For each $n \in \NN$, write $\mathbf{n}$ for the chain $0<1<\cdots< n-1$ of length $n$ in $\A$ and let $\gr{n}$ be its image under the functor $L^{D}$; that is, $\gr{n}$ is the deterministic graph
\[
\begin{tikzcd}
  0
  \arrow{r}
& 1
  \arrow{r}
& \cdots
  \arrow{r}
& n-1.
\end{tikzcd}
\]

\noindent
In particular, $\gr{0}$ is the empty graph,
$\gr{1}$ is the graph with a vertex $0$ and no edges,
and $\gr{2}$ is the graph $0 \longto 1$.
Consider the functions $f,g \colon \gr{1} \rightrightarrows \gr{2}+\gr{2}$ below:
\[
\begin{tikzcd}
  \gr{1}
  \arrow[shift right]{d}[left]{f}
  \arrow[shift left]{d}[right]{g}
& 0
  \arrow[mapsto]{d}[left]{f}
  \arrow[mapsto]{drr}{g}
&
\\
  \gr{2} + \gr{2}
& 0
  \arrow{r}
& 1
& 0
  \arrow{r}
& 1
\end{tikzcd}
\]
In the category $\A$, denote by $!\colon \mathbf{1}\to\mathbf{2}$ the unique forest morphism and let $\varphi,\psi\colon \mathbf{1}\to \mathbf{2} + \mathbf{2}$ be obtained by composing $!$ with the first and second coproduct injection, respectively. Then $L^{D}(\mathbf{2} + \mathbf{2})$ can be identified with $\gr{2}+\gr{2}$, and $L^{D}\varphi, L^{D}\psi$ with $f$ and $g$, respectively.

The coequaliser of $\varphi$ and $\psi$ in $\A$ is $\mathbf{2}$,
while the coequaliser of $U f$ and $U g$
in $\Gph$ is the graph $H_2$, where $H_n$ denotes the graph displayed below:
\[
\begin{tikzcd}
  1
& \cdots
& n
\\
& 0
  \arrow{ul}
  \arrow{ur}
\end{tikzcd}
\]

\noindent
Now, since $L \colon \A \to \Gph$ is left adjoint,
in $\Gph$ we obtain a coequaliser diagram
\[
\begin{tikzcd}
  L\mathbf{1}
  \arrow[shift left]{r}[above]{L \varphi}
  \arrow[shift right]{r}[below]{L \psi}
& L\mathbf{2} + L\mathbf{2}
  \arrow[twoheadrightarrow]{r}
& L\mathbf{2}.
\end{tikzcd}
\]

\noindent
It is not difficult to see that $L\mathbf{1} \cong U\gr{1}$
and, using the fact that the adjunction $L\dashv R$ is finitely accessible, that $L\mathbf{2}\cong H_n$ for some $n > 0$.
Moreover, it follows from the description of $F \colon \Gph \to \cat D$
given above that $L\varphi \colon L\mathbf{1} \to L\mathbf{2} + L\mathbf{2}$
takes $0$ to one root of $L\mathbf{2} + L\mathbf{2}$,
while $L \psi$ takes $0$ to the other root of $L\mathbf{2} + L\mathbf{2}$.
Therefore, $H_{2n}$ is a coequaliser of $L\varphi$ and $L\psi$ in $\Gph$.
But $H_n \not\cong H_{2n}$ for $n > 0$, a contradiction.
\qed
%%%%%%%%%%%%%%%%%%%%%%%%%%%%%%%%%%%%%%%%%%%%%%%%%%%%%%%%%%%%%%%%%%%%%%%%%%%
\end{example}
%%%%%%%%%%%%%%%%%%%%%%%%%%%%%%%%%%%%%%%%%%%%%%%%%%%%%%%%%%%%%%%%%%%%%%%%%%%

%%%%%%%%%%%%%%%%%%%%%%%%%%%%%%%%%%%%%%%%%%%%%%%%%%%%%%%%%%%%%%%%%%%%%%%%%%%
\begin{fullproof}
%%%%%%%%%%%%%%%%%%%%%%%%%%%%%%%%%%%%%%%%%%%%%%%%%%%%%%%%%%%%%%%%%%%%%%%%%%%
We prove various statements of Example~\ref{ex:fact:det}.

\begin{claim*}
The category $\A$ is lfp.
\end{claim*}

\begin{proof}[Proof of the Claim] 
We rely on Example~\ref{ex:prelim:forests},
and more precisely on the equivalence $\Forest \cong \presh{\NN}$.
Note that this equivalence takes a forest $\forest$ to the
presheaf $X_{\forest} \in \funct{\NN^\op, \Set}$
where $X_{\forest}(n)$ is the set of all $u \in \forest$ such that $\down u$
has cardinality $n$, and where the restriction map
$X_{\forest}(n+1) \to X_{\forest}(n)$ takes $u$ with
$\down u =\{v_1 \prec \dots \prec v_n \prec v_{n+1} \}$ to $v_n$.
Hence we have $\forest \in \A$ if, and only if,
all the restriction maps of $X_{\forest}$ are injective.
But this latter condition can be expressed by a cartesian (in fact Horn)
theory in the signature $\Sig(\NN^\op)$ associated to $\NN^\op$
in Example~\ref{ex:prelim:coste:funct}.
Hence $\A$ is lfp.
\end{proof}

\begin{claim*}
$\A$ is an arboreal category.
\end{claim*}

\begin{proof}[Proof of the Claim] 
We are going to prove that $\A$ is arboreal when equipped with
the factorisation system $(\Q,\M)$,
where $\Q$ (resp.\ $\M$) consists of the surjective (resp.\ injective) morphisms.

We have seen above that under the equivalence $\Forest \cong \presh{\NN}$,
the objects of $\A$ correspond exactly to the models
of a Horn extension of the Horn theory associated to
$\presh{\NN}$ in Example~\ref{ex:prelim:coste:funct}.
Hence we get from (a direct adaptation of) Lemma~\ref{lem:coste:mod:filtcolim}
that the full inclusion $\A \into \Forest$ is a finitary right adjoint.

It follows from Example~\ref{ex:forests-trees-presheaves}
and Example~\ref{ex:prelim:fact:forest}
that $\Forest$ is arboreal w.r.t.\ 
(surjective morphisms, injective morphisms) factorisations.
Moreover, it follows from Example~\ref{ex:path:tree} that in $\Forest$,
the paths are exactly the finite chains.

We now show that $(\Q,\M)$ is a (proper) factorisation system on $\A$.
Let $f \in \A \funct{a,b}$.
We can factor $f$ as $f = m \comp e$ where $e$ (resp.\ $m$)
is a surjective (resp.\ injective) morphism of $\Forest$.
But it follows from Example~\ref{ex:prelim:fact:forest}
(actually \cite[4.3.10.g]{borceux94vol1})
that $m$ is a strong monomorphism in $\Forest$.
Then by Remark~\ref{rem:emb:strong},
$m$ corresponds under $\Forest \cong \presh{\NN}$
to an embedding of $\Sig(\NN^\op)$ structures,
where $\Sig(\NN^\op)$ is the signature 
associated to $\NN^\op$ in Example~\ref{ex:prelim:coste:funct}.
Now, $\A$ corresponds under $\Forest \cong \presh{\NN}$, 
to the category of models of a Horn theory in $\Sig(\NN^\op)$.
This category is thus stable under substructures
(see e.g.~\cite[Theorem 5.12]{ar94book}).
Hence the domain of $m$ is actually an object of $\A$,
and it follows that $m \comp e$ is a $(\Q,\M)$ factorisation of $f$.

The proper factorization system $(\Q,\M)$ is moreover stable
since the full inclusion $\A \into \Forest$
preserves pullbacks, and preserves as well as reflects $\Q$ and $\M$.

We thus obtain that $\A$ is wooded, when equipped with $(\Q,\M)$.

It remains to check conditions \ref{ax:2-out-of-3}--\ref{ax:path-generated}
of Definition~\ref{def:arboreal-cat}.
Condition~\ref{ax:colimits} (coproducts of paths)
follows from the fact that $\A$ is lfp, and thus cocomplete.

For the other conditions, first note that the paths of $\A$
are exactly the paths of $\Forest$.
Indeed, from Example~\ref{ex:path:tree} we get
that all the paths of $\Forest$ are objects of $\A$,
from which we obtain that all the paths of $\Forest$ are actually
paths in $\A$.
Conversely, since $\A$ is closed in $\Forest$ under domains
of injective morphisms, we get that all paths of $\A$
are paths in $\Forest$.

This yields condition~\ref{ax:2-out-of-3} (``2-out-of-3'')
of Definition~\ref{def:arboreal-cat}.

As for condition~\ref{ax:connected} (on connectedness),
notice that the inclusion $\A \into \Forest$ preserves coproducts.

Finally, we obtain condition~\ref{ax:path-generated} (path-generation)
from the above and the fact that the inclusion $\A \into \Forest$
is full.
\end{proof}

\begin{claim*}
$L^D \colon \A \inadj \cat D \cocolon R^D$
is a finitely accessible adjunction.
\end{claim*}

\begin{proof}[Proof of the Claim] 
It follows from the above that the finitely accessible adjunction
$L^M : \Forest \inadj \Gph \cocolon R^M$
restricts to an adjunction
$L^D \colon \A \inadj \cat D \cocolon R^D$
where $\A$ and $\cat D$ are lfp.

In view of Proposition~\ref{p:lfp-morphisms-characterisation},
it remains to show that
$L^D \colon \A \to \cat D$
preserves finitely presentable objects.
But this follows from the fact that
in both $\A$ and $\cat D$,
the finitely presentable objects are the finite structures.
\end{proof}

\begin{claim*}
$L^D \colon \A \inadj \cat D \cocolon R^D$
detects path embeddings.
\end{claim*}

\begin{proof}[Proof of the Claim]
The embeddings in $\A$ are by definition the injective morphisms.
Since $L^D$ is a (non full) inclusion,
we are done if we show that a morphism $f \colon a \to b$ in $\A$
is injective exactly when it is an embedding of $\prec$-structures.

It is clear that embeddings of $\prec$-structures are injective.

For the converse, consider an injective morphism $f \in \A\funct{a,b}$.
Let $u$ and $v$ be vertices of $a$ such that $f(u) \prec f(v)$ in $b$.
Hence $f(v)$ is not a root of $b$, so that $v$ cannot be a root of $a$.
It follows that there is some vertex $w$ of $a$ such that $w \prec v$ in $a$.
But then $f(w) \prec f(v)$, so that $f(w) = f(u)$ since $b$ is a forest.
Hence $w = u$ since $f$ is injective.
\end{proof}

\begin{claim*}
$F \colon \Gph \inadj \cat D \cocolon U$
is a finitely accessible adjunction.
\end{claim*}

\begin{proof}[Proof of the Claim] 
Since $\cat D = \Mod(\theory)$ with $\theory$ cartesian in $\sig$,
it follows from Lemma~\ref{lem:coste:mod:filtcolim} (\S\ref{sec:coste:synt})
that the inclusion functor $U \colon \cat D \to \Gph$
is a finitary right adjoint.
Hence, we have to show that the indicated $F$ is indeed left adjoint to $U$.
To this end, we just have to prove the usual universal
property (see e.g.~\cite[Theorem IV.1.2]{maclane98book}).

Recall that $F$ takes a graph $G = (V,E)$ to its quotient under the
least equivalence relation $\simeq$ on $V$ such that $u \simeq v$
if there is some $w$ with $E(w,u)$ and $E(w,v)$.
The edge relation of $F G$ is the set of those pairs
$([u]_{\simeq}, [v]_{\simeq})$ such that
for some $u' \simeq u$ and some $v' \simeq v$, we have $E(u',v')$ in $G$.

Given a graph $G \in \Gph$, we let
$\eta_G \colon G \to U F G$ be the function 
which takes each vertex $u$ of $G$ to its $\simeq$-class $[u]_{\simeq}$ in $U F G$.
It is clear that $\eta_G$ is a graph morphism, since
if $E(u,v)$ in $G$, then there is an edge from
$[u]_{\simeq}$ to $[v]_{\simeq}$ in $U F G$.

Given some $K \in \cat D$ and some graph morphism $h \colon G \to U K$,
we have to show that there is a unique $\cat D$-morphism
$\ladj h \colon F G \to K$ such that the following commutes.
\[
\begin{tikzcd}
  G
  \arrow{d}[left]{\eta_G}
  \arrow{r}{h}
& U K
\\
  U F G
  \arrow{ur}[below, xshift=5pt]{U \ladj h}
\end{tikzcd}
\]

We first define $\ladj h \colon F G \to K$.
To this end, we show by induction on $\simeq$ that $u \simeq v$ implies $h(u) = h(v)$.
\begin{description}
\item[Base case]
In the base case, there is some $w \in V$ such that $E(w,u)$ and $E(w,v)$.
Hence, in $K$ there are edges from $h w$ to $h u$ and to $h v$.
But this implies $h u = h v$, as required.

\item[Inductive step (symmetry)]
In this case we obtained $u \simeq v$ from $v \simeq u$,
and the result directly follows from the induction hypothesis.

\item[Inductive step (transitivity)]
In this case, we obtained $u \simeq v$ from $u \simeq w$ and $w \simeq v$,
and the result again follows from the induction hypothesis.
\end{description}

Hence, if $\eta_G(u) = \eta_G(v)$, then $h(u) = h(v)$,
and this defines the value of $\ladj h$ on $[u]_{\simeq} = [v]_{\simeq}$.
Moreover, $\ladj h$ is indeed a $\cat D$-morphism, since
if $E(u',v')$ for some $u' \simeq u$ and some $v' \simeq v$,
then there is an edge from $h(u')$ to $h(v')$ in $U K$,
while $\ladj h[u]_\simeq = h(u')$ and $\ladj h[v]_\simeq = h(v')$.
Finally, we indeed have $U \ladj h \comp \eta_G = h$
from the definition of $\ladj h$.

It remains to show the uniqueness of $\ladj h \colon F G \to K$.
Let $g \colon F G \to K$ in $\cat D$ be such that $U g \comp \eta_G = h$.
This means $g[u]_\simeq = h(u) = \ladj h [u]_\simeq$ for each $[u]_\simeq \in F$,
so that $g = \ladj h$.
\end{proof}

\begin{claim*}
The full inclusion $U \colon \cat D \to \Gph$ preserves coproducts.
\end{claim*}

\begin{proof}[Proof of the Claim]
Consider a family $(K_i \mid i \in I)$ of objects of $\cat D$,
with $K_i = (W_i,E_i)$,
and let $G \deq \coprod_{i \in I} K_i$ be their coproduct in $\Gph$.
By~\cite[Remark 1.5(3)]{ar94book}, the edge relation $E$ of $G$
contains exactly those pairs $((i,u_i),(j,v_j))$ 
such that $i = j$ and $E_i(u_i,v_i)$.
It follows that $G \in \cat D$.
\end{proof}

\begin{claim*}
The coproduct $\gr{2}+\gr{2}$ is the same in $\Gph$ and $\cat D$. 
The coproduct $\mathbf{2} + \mathbf{2}$ in $\A$ is the forest
with two copies of $\mathbf{2}$.
\end{claim*}
\begin{proof}
Indeed, the above displayed $\gr{2}+\gr{2}$ is the coproduct in $\Gph$
(\cite[Remark 1.5(3)]{ar94book}).
Hence it is also the coproduct in the cocomplete category $\cat D$,
since the inclusion $U \colon \cat D \to \Gph$ preserves coproducts.
As the (non-full) inclusion $L^D \colon \A \to \cat D$ is a left adjoint,
it follows that coproduct $\mathbf{2} + \mathbf{2}$ in 
the cocomplete category $\A$
is the forest with two copies of $\mathbf{2}$.
\end{proof}

\begin{claim*}
The coequaliser of $\varphi$ and $\psi$ in $\A$ is $\mathbf{2}$.
\end{claim*}

\begin{proof}[Proof of the Claim] 
We show that the coequaliser of $f$ and $g$ in $\cat D$ is $\gr 2$.
The result then follows from the fact that the (non full) inclusion
$L^D \colon \A \to \cat D$ is a left adjoint.

Let $e \colon \gr 2 + \gr 2 \epi \gr 2$ take
both $0$'s to $0$, and both $1$'s to $1$.
Then $e$ is a morphism of $\cat D$, and moreover we have $e \comp f = e \comp g$.

Consider now some $\cat D$-morphism $h \colon \gr 2 + \gr 2 \to K$
such that $h \comp f = h \comp g$.
In particular, $h(f(0)) = h(g(0))$,
hence $h$ identifies the two $0$'s of $\gr 2 + \gr 2$.
But since $K \in \cat D$, it follows that $h$ must also identify the two $1$'s
of $\gr 2 + \gr 2$
(since otherwise $h(f(0)) = h(g(0))$ would have two distinct outgoing edges in $K$).
It follows that $h\colon \gr 2 + \gr 2 \to K$ factors uniquely through
$e \colon \gr 2 + \gr 2 \epi \gr 2$.
\end{proof}

\begin{claim*}
$L\mathbf{1} \cong \gr{1}$.
\end{claim*}

\begin{proof}[Proof of the Claim] 
First, note that $R\gr{1} \cong \mathbf{1}$
since
$R\gr{1} \cong R U \gr{1} \cong R^D \gr{1} \cong \mathbf{1}$.
It follows that
$\Gph\funct{L\mathbf{1},\gr{1}} \cong \A\funct{\mathbf{1},\mathbf{1}}$
is a singleton.
In particular, the graph $L\mathbf{1}$ has no edges,
and it follows from the above description of $F \colon \Gph \to \cat D$
that $F L\mathbf{1} \cong L\mathbf{1}$.
But $F L \cong L^D$, so $L\mathbf{1} \cong L^D\mathbf{1} \cong \gr{1}$
since $L^D$ is the identity on objects.
\end{proof}

\begin{claim*}
For some $n > 0$, we have $L\mathbf{2} \cong H_n$.
\end{claim*}

\begin{proof}[Proof of the Claim] 
First, note that $\Gph\funct{L\mathbf{2},\gr{2}}$ is a singleton.
This follows from the fact that $R\gr{2} \cong R^D\gr{2} \cong \mathbf{2}+\mathbf{1}$,
so that
\[
\begin{array}{*{5}{c}}
  \Gph\funct{L\mathbf{2},\gr{2}}
& \cong
& \A\funct{\mathbf{2},\mathbf{2}+\mathbf{1}}
& \cong
& \A\funct{\mathbf{2},\mathbf{2}}
\end{array}
\]

\noindent
while the only forest morphism $\mathbf{2} \to \mathbf{2}$ is the identity.

Moreover, $L\mathbf{2}$ is a finite graph since $L \colon \A \to \Gph$
preserves finitely presentable objects
(see Proposition~\ref{p:lfp-morphisms-characterisation}).

Write $L\mathbf{2} = (V,E)$ and
let $h \colon L\mathbf{2} \to \gr{2}$ be the unique graph morphism.
We have $V = V_0 + V_1$ where $V_i = h^{-1}(\{i\})$.
Note that there are no edges between any two vertices of $V_i$.
Further, each $u \in V_{0}$ must have and edge to a $v \in V_1$,
since otherwise we would get a distinct morphism $h' \colon L\mathbf{2} \to \gr{2}$
by taking $u$ to $1$.
Similarly, each $v \in V_1$ must have an edge from some $u \in V_0$.
In particular, $V_0$ is non-empty if, and only if, $V_1$ is non-empty.

It remains to show that $V_0$ is a singleton.
Note that $FL\mathbf{2} \cong \gr{2}$.
Hence, it follows from the above description of $F$ that
all vertices in $V_0$ must be identified by $F$.
But since no $u \in V_0$ is the target of an edge,
this implies that $V_0$ has at most one element.
Assume toward a contradiction that $V_0= \emptyset$.
Then (by the above), we also get $V_1 = \emptyset$.
But then $L\mathbf{2}$ is $\gr{0}$, the initial object of $\A$, $\Gph$ and $\cat D$.
Since $F$ is a left adjoint, this yields 
$FL\mathbf{2} \cong \gr{0}$,
which contradicts $F L \mathbf{2} \cong L^D \mathbf{2} \cong \gr{2}$.
\end{proof}

\begin{claim*}
Let $k,\ell \colon \gr 1 \to H_n + H_n$ be the $\Gph$-morphisms such that
$k$ takes $0$ to the left root of $H_n + H_n$,
while $\ell$ takes $0$ to the right root of $H_n + H_n$.
Then the coequaliser of $k$ and $\ell$ in $\Gph$ is $H_{2n}$.
\end{claim*}

\begin{proof}[Proof of the Claim] 
Let $e \colon H_n + H_n \epi H_{2n}$
be the graph morphism
which identifies the two roots of $H_n + H_n$ with the root of $H_{2n}$,
and which is injective on the leaves of $H_n + H_n$.
We are going to show that $e$ is the coequaliser of $k$ and $\ell$ in $\Gph$.

First, $e$ obviously coequalises $k$ and $\ell$.
Moreover, if $h \colon H_n + H_n \to G$
coequalises $k$ and $\ell$, we have $h(k(0)) = h(\ell(0))$,
and this directly yields the existence of a unique
$\ladj h \colon H_{2n} \to G$ such that $\ladj h \comp e = h$.
\end{proof}

This concludes the proof of Example~\ref{ex:fact:det}.
%%%%%%%%%%%%%%%%%%%%%%%%%%%%%%%%%%%%%%%%%%%%%%%%%%%%%%%%%%%%%%%%%%%%%%%%%%%
\end{fullproof}
%%%%%%%%%%%%%%%%%%%%%%%%%%%%%%%%%%%%%%%%%%%%%%%%%%%%%%%%%%%%%%%%%%%%%%%%%%%

%%%%%%%%%%%%%%%%%%%%%%%%%%%%%%%%%%%%%%%%%%%%%%%%%%%%%%%%%%%%%%%%%%%%%%%%%%%
\section{Conclusion}
\label{sec:conc}
%%%%%%%%%%%%%%%%%%%%%%%%%%%%%%%%%%%%%%%%%%%%%%%%%%%%%%%%%%%%%%%%%%%%%%%%%%%

In this work, we have identified a tight upper bound (namely, $\Lang_{\infty}$) for the expressive power of back-and-forth equivalence relations induced by a finitely accessible wooded adjunction satisfying mild assumptions concerning the definability (more generally, detection) of path embeddings.
These assumptions are met in a broad range of examples.
Besides Ehrenfeucht-Fraïssé games, our results apply to pebble games, and to modal and hybrid bisimulation games
\cite{as21jlc,ar21arboreal,am22mfcs}.
We expect that the same methodology will extend to
e.g.\ guarded fragments~\cite{am21lics}.
On the other hand, our results show that game
comonads for $\MSO$ cannot lead to finitely accessible wooded adjunctions,
at least with the obvious choice of path embeddings (see Example~\ref{ex:mso}).

This suggests that in order to handle stronger logics,
one may need to consider
$\kappa$-accessible wooded adjunctions
for uncountable regular cardinals $\kappa$.
We believe that this is a natural avenue for future investigation.
In particular, games for generalised Lindstr\"om quantifiers
\cite{caicedo80mlla, ocd2021} may 
be instrumental in developing a systematic study.

%%%%%%%%%%%%%%%%%%%%%%%%%%%%%%%%%%%%%%%%%%%%%%%%%%%%%%%%%%%%%%%%%%%%%%%%%%%
\subsection*{Acknowledgment}
%%%%%%%%%%%%%%%%%%%%%%%%%%%%%%%%%%%%%%%%%%%%%%%%%%%%%%%%%%%%%%%%%%%%%%%%%%%

Research supported by the EPSRC grant EP/V040944/1
and the ANR project ReCiProg (ANR-21-CE48-0019).

%\clearpage
\bibliographystyle{amsplain}
\bibliography{IEEEabrv,bibliographie}

\opt{fullproof}{\clearpage}
\appendix
\opt{fullproof}{%%%%%%%%%%%%%%%%%%%%%%%%%%%%%%%%%%%%%%%%%%%%%%%%%%%%%%%%%%%%%%%%%%%%%%%%%%%
\section{Lfp categories as categories of structures}
\label{s:lfp-cats-of-substructures}
%%%%%%%%%%%%%%%%%%%%%%%%%%%%%%%%%%%%%%%%%%%%%%%%%%%%%%%%%%%%%%%%%%%%%%%%%%%

In Remark~\ref{rem:path:structemb}\ref{item:path:structemb:rabin}
we recalled that, given a full subcategory~$\E$ of $\Struct(\Sig)$
that is reflective and closed under filtered colimits in $\Struct(\Sig)$,
there need not be a cartesian theory~$\theory$ in~$\Sig$
such that $\E \cong \Mod(\theory)$.
This can be seen with a slight adaptation of an example given in~\cite{rabin62ascfm},
which we detail here; see also~\cite{volger79mz,har01ctgdc}.

Let $\cat K$ be a class of $\Sig$-structures.
Following~\cite[Definition 1, \S 2]{rabin62ascfm},
we say that $\cat K$ has the \emph{intersection property}
if given $M \in \cat K$ and given
a family $(N_i \mid i \in I)$ of structures $N_i \in \cat K$
such that each $N_i$ is a substructure of $M$,
we have $\bigcap_{i \in I} N_i \in \cat K$.
The above intersection of structures $\bigcap_{i \in I} N_i$
is the wide pullback in $\Struct(\Sig)$ of the substructure embeddings
$N_i \emb M$.
In particular, given a cartesian theory $\theory$ in~$\Sig$,
since $\Mod(\theory)$ is closed under limits in $\Struct(\Sig)$
(Lemma~\ref{lem:coste:mod:filtcolim}),
we get that $\Mod(\theory)$, seen as a class of $\Sig$-structures,
has the intersection property.%
\footnote{Recall that 
substructure embeddings are regular monomorphisms in $\Struct(\Sig)$,
cf.\ Example~\ref{ex:prelim:fact:struct}.}

%%%%%%%%%%%%%%%%%%%%%%%%%%%%%%%%%%%%%%%%%%%%%%%%%%%%%%%%%%%%%%%%%%%%%%%%%%%
\begin{example}
%%%%%%%%%%%%%%%%%%%%%%%%%%%%%%%%%%%%%%%%%%%%%%%%%%%%%%%%%%%%%%%%%%%%%%%%%%%
The following is adapted from~\cite{rabin62ascfm,volger79mz}.
Let $\sig$ be the signature with one
$(n+1)$-ary relation symbol $P_n$ for each $n \in \NN$.
The following sentences were devised in~\cite[\S 8]{rabin62ascfm}:
define, for each $n \in \NN$,
\[
\begin{array}{r c l}
  \varpi_n
& \deq
& (\forall x)(\forall y_1,\dots,y_n,y_{n+1})
  \big(
  P_{n+1}(x,y_1,\dots,y_n,y_{n+1})
  ~\limp~
  P_n(x,y_1,\dots,y_n)
  \big)
\\

  \varphi_n
& \deq
& (\forall x)
  (\forall \vec y,y_{n+1})
  (\forall \vec z,z_{n+1})
  \Big(
  \big(
  P_{n+1}(x,\vec y,y_{n+1})
  \land
  P_{n+1}(x,\vec z,z_{n+1})
  \big)
  ~\limp~
  \vec y \Eq \vec z
  \Big)
\\

  \gamma_n
& \deq
& (\forall x)(\exists y_1,\dots,y_n) P_n(x,y_1,\dots,y_n)
\end{array}
\]

\noindent Let $\E$ be the full subcategory of $\Struct(\sig)$ on the models
of all the above sentences, where~$n$ ranges over $\NN$.
It has been observed in~\cite{volger79mz}
that $\E$ is closed under finite limits in $\Struct(\sig)$.
By a result of~\cite{volger79mz},
this implies that $\E$ is closed in $\Struct(\sig)$
under limits and filtered colimits,
from which we get the following
(see~\cite[Theorem 5.20 and Corollary 5.21]{ar94book}).

%%%%%%%%%%%%%%%%%%%%%%%%%%%%%%%%%%%%%%%%%%%%%%%%%%%%%%%%%%%%%%%%%%%%%%%%%%%
\begin{claim*}
%%%%%%%%%%%%%%%%%%%%%%%%%%%%%%%%%%%%%%%%%%%%%%%%%%%%%%%%%%%%%%%%%%%%%%%%%%%
$\E$ is reflective and closed under filtered colimits in $\Struct(\sig)$.
%%%%%%%%%%%%%%%%%%%%%%%%%%%%%%%%%%%%%%%%%%%%%%%%%%%%%%%%%%%%%%%%%%%%%%%%%%%
\end{claim*}
%%%%%%%%%%%%%%%%%%%%%%%%%%%%%%%%%%%%%%%%%%%%%%%%%%%%%%%%%%%%%%%%%%%%%%%%%%%

%%%%%%%%%%%%%%%%%%%%%%%%%%%%%%%%%%%%%%%%%%%%%%%%%%%%%%%%%%%%%%%%%%%%%%%%%%%
\begin{fullproof}[Proof of the Claim] 
%%%%%%%%%%%%%%%%%%%%%%%%%%%%%%%%%%%%%%%%%%%%%%%%%%%%%%%%%%%%%%%%%%%%%%%%%%%
This follows from~\cite[Corollary 5.21]{ar94book}, which
only requires us to show that $\E$ is closed under finite limits
in $\Struct(\sig)$.
The sentences $\varpi_n$, $\varphi_n$ and $\gamma_n$
are regular, hence it follows from~\cite[Lemma 9.1.4]{hodges93book}
that $\E$ is closed in $\Struct(\sig)$
under products (see also~\cite[D2.4.3]{johnstone02book}).
It remains to show that $\E$ is closed
under equalisers in $\Struct(\sig)$.

We rely on the description of equalisers in $\Struct(\sig)$
given in~\cite[Remark 5.1(2)]{ar94book}.
Namely, 
the equaliser in $\Struct(\sig)$ of parallel arrows $f,g \colon M \to N$ in $\E$
is the substructure $E$ of $M$ on those $a \in M$ such that $f(a) = g(a)$.
We have to show that $E \models \varpi_n \land \varphi_n \land \gamma_n$
for all $n \in \NN$.
Note that $\varpi_n$ and $\varphi_n$ are Horn sentences, so 
their classes of models are closed under substructures
(this follows from \cite[Theorem 5.12]{ar94book}).
Hence
$E \models \varpi_n \land \varphi_n$ for all $n \in \NN$.
It remains to show that $E \models \gamma_n$ for all $n \in \NN$.
Fix $n \in \NN$ and let $a \in E$.
Since~$M$ is a model of $\gamma_{n+1}$ and $\varpi_n$,
there are $\vec b, b \in M$ such that $P_{n+1}(a,\vec b,b)$
and $P_{n}(a,\vec b)$ hold in~$M$.
Hence $P_{n+1}(f(a),f(\vec b),f(b))$
and $P_{n+1}(g(a),g(\vec b),g(b))$
both hold in $N$.
But we have $f(a) = g(a)$ because $a \in E$,
and so $f(\vec b)=g(\vec b)$ since $N\models \varphi_n$.
It follows that $\vec b \in E$, and we are done
because $E$ is a substructure of $M$.
%%%%%%%%%%%%%%%%%%%%%%%%%%%%%%%%%%%%%%%%%%%%%%%%%%%%%%%%%%%%%%%%%%%%%%%%%%%
\end{fullproof}
%%%%%%%%%%%%%%%%%%%%%%%%%%%%%%%%%%%%%%%%%%%%%%%%%%%%%%%%%%%%%%%%%%%%%%%%%%%

In particular $\E$, seen as a class of $\sig$-structures, has the intersection property.
However, by \cite[Theorem 11]{rabin62ascfm}
there is no set $\Psi$ of sentences in $\sig$
such that $\E$ is the class of models of $\Psi$,
and such that for each $\psi \in \Psi$, the class of models of $\psi$
has the intersection property.
It is then not difficult to derive the following.

%%%%%%%%%%%%%%%%%%%%%%%%%%%%%%%%%%%%%%%%%%%%%%%%%%%%%%%%%%%%%%%%%%%%%%%%%%%
\begin{claim*}
%%%%%%%%%%%%%%%%%%%%%%%%%%%%%%%%%%%%%%%%%%%%%%%%%%%%%%%%%%%%%%%%%%%%%%%%%%%
There is no cartesian theory $\theory$ in $\sig$ such that $\E = \Mod(\theory)$.
%%%%%%%%%%%%%%%%%%%%%%%%%%%%%%%%%%%%%%%%%%%%%%%%%%%%%%%%%%%%%%%%%%%%%%%%%%%
\end{claim*}
%%%%%%%%%%%%%%%%%%%%%%%%%%%%%%%%%%%%%%%%%%%%%%%%%%%%%%%%%%%%%%%%%%%%%%%%%%%

%%%%%%%%%%%%%%%%%%%%%%%%%%%%%%%%%%%%%%%%%%%%%%%%%%%%%%%%%%%%%%%%%%%%%%%%%%%
\begin{fullproof}[Proof of the Claim] 
%%%%%%%%%%%%%%%%%%%%%%%%%%%%%%%%%%%%%%%%%%%%%%%%%%%%%%%%%%%%%%%%%%%%%%%%%%%
Assume toward a contradiction that $\E = \Mod(\theory)$
for some cartesian theory $\theory$ in $\sig$.
By Definition~\ref{def:prelim:coste:th},
for each sequent $(\psi_1 \thesis_{\vec x} \psi_2)$ of $\theory$, 
there is a \emph{finite} cartesian
$\theory_{(\psi_1 \thesis_{\vec x} \psi_2)} \sle \theory$
with
$(\psi_1 \thesis_{\vec x} \psi_2) \in \theory_{(\psi_1 \thesis_{\vec x} \psi_2)}$.
Then, for each
$(\psi_1 \thesis_{\vec x} \psi_2) \in \theory$, 
there is a sentence $\psi$ in $\sig$ whose class of models is exactly
$\Mod(\theory_{(\psi_1 \thesis_{\vec x} \psi_2)})$,
and thus has the intersection property.
Hence, $\E$ is the class of models of a set of sentences $\Psi$,
such that for each $\psi \in \Psi$, the class of models of $\psi$
has the intersection property.
This contradicts \cite[Theorem 11]{rabin62ascfm}.
%%%%%%%%%%%%%%%%%%%%%%%%%%%%%%%%%%%%%%%%%%%%%%%%%%%%%%%%%%%%%%%%%%%%%%%%%%%
\end{fullproof}
%%%%%%%%%%%%%%%%%%%%%%%%%%%%%%%%%%%%%%%%%%%%%%%%%%%%%%%%%%%%%%%%%%%%%%%%%%%
%%%%%%%%%%%%%%%%%%%%%%%%%%%%%%%%%%%%%%%%%%%%%%%%%%%%%%%%%%%%%%%%%%%%%%%%%%%
\end{example}
%%%%%%%%%%%%%%%%%%%%%%%%%%%%%%%%%%%%%%%%%%%%%%%%%%%%%%%%%%%%%%%%%%%%%%%%%%%

}
\opt{fullproof,draft}{%%%%%%%%%%%%%%%%%%%%%%%%%%%%%%%%%%%%%%%%%%%%%%%%%%%%%%%%%%%%%%%%%%%%%%%%%%%
\section{Additional material for~\S\ref{sec:prelim} (\nameref{sec:prelim})}
\label{sec:app:prelim}
%%%%%%%%%%%%%%%%%%%%%%%%%%%%%%%%%%%%%%%%%%%%%%%%%%%%%%%%%%%%%%%%%%%%%%%%%%%

%%%%%%%%%%%%%%%%%%%%%%%%%%%%%%%%%%%%%%%%%%%%%%%%%%%%%%%%%%%%%%%%%%%%%%%%%%%
\subsection{Structures and homomorphisms}
%%%%%%%%%%%%%%%%%%%%%%%%%%%%%%%%%%%%%%%%%%%%%%%%%%%%%%%%%%%%%%%%%%%%%%%%%%%

A \emph{many-sorted} signature $\Sig$
consists of 
a set $\Sort(\Sig)$ of \emph{sorts}
together with collections
\begin{align*}
  \Fun(\Sig)
& =
  \mathord{\bigcup}_{(\sort_1,\dots,\sort_n;\sort) \in \Sort(\Sig)^{n+1}}
  \Fun(\Sig)(\sort_1,\dots,\sort_n;\sort)
\\
  \Rel(\Sig)
& =
  \mathord{\bigcup}_{(\sort_1,\dots,\sort_n) \in \Sort(\Sig)^{n}}
  \Rel(\Sig)(\sort_1,\dots,\sort_n)
\end{align*}

\noindent
of, respectively,
\emph{function}
and \emph{relation}
symbols.
We write
\[f \colon \sort_1 \dots \sort_n \to \sort\]
if
$f \in \Fun(\Sig)(\sort_1,\dots,\sort_n;\sort)$,
and
\[R \into \sort_1 \dots \sort_n\]
if
$R \in \Rel(\Sig)(\sort_1,\dots,\sort_n)$.

%%%%%%%%%%%%%%%%%%%%%%%%%%%%%%%%%%%%%%%%%%%%%%%%%%%%%%%%%%%%%%%%%%%%%%%%%%%
\begin{definition}[Structures]
%%%%%%%%%%%%%%%%%%%%%%%%%%%%%%%%%%%%%%%%%%%%%%%%%%%%%%%%%%%%%%%%%%%%%%%%%%%
Let $\Sig$ be a signature.
\begin{enumerate}[(1)]
\item
A \emph{$\Sig$-structure} $M$ 
is given by the following data:
\begin{enumerate}[(i)]
\item
for each sort 
$\sort \in \Sort(\Sig)$,
a set $M(\sort)$;

\item
for each function symbol
$f \colon \sort_1 \dots \sort_n \to \sort$,
a function
$f_M \colon M(\sort_1) \times \dots \times M(\sort_n) \to M(\sort)$;

\item
for each relation symbol
$R \into \sort_1 \dots \sort_n$,
a relation
$R_M \sle M(\sort_1) \times \dots \times M(\sort_n)$.
\end{enumerate}

\item
A \emph{homomorphism of $\Sig$-structures}
$h \colon M \to N$
is
a family
of functions
\[
(
h^\sort \colon M(\sort) \to N(\sort)
\mid
\sort \in \Sort(\Sig)
)
\]
such that
for every
$f \colon \sort_1 \dots \sort_n \to \sort$
and
$R\into \sort_1 \dots \sort_n$,
and for all
$(a_1,\dots,a_n)\in M(\sort_1)\times \cdots \times M(\sort_n)$,
we have
\[
\begin{array}{r c r}
  h^{\sort}( f_M(a_1,\dots,a_n) )
& =
& f_N( h^{\sort_1}(a_1), \dots, h^{\sort_n}(a_n) )
\\

  R_M(a_1,\dots,a_n)
& \imp
& R_N\left( h^{\sort_1}(a_1), \dots, h^{\sort_n}(a_n) \right)
\end{array}
\]

\item
We write $\Struct(\Sig)$ for the category
whose objects are $\Sig$-structures
and whose morphisms are homomorphisms.
\end{enumerate}
\end{definition}

%%%%%%%%%%%%%%%%%%%%%%%%%%%%%%%%%%%%%%%%%%%%%%%%%%%%%%%%%%%%%%%%%%%%%%%%%%%
\subsection{Infinitary first-order logic}
%%%%%%%%%%%%%%%%%%%%%%%%%%%%%%%%%%%%%%%%%%%%%%%%%%%%%%%%%%%%%%%%%%%%%%%%%%%

We shall consider different fragments of infinitary first-order logic.
In contrast with~\cite[\S 5]{ar94book} but in accordance
with~\cite[\S D1]{johnstone02book}
we work with formulae having at most finitely many free variables.
These variables are declared \emph{contexts}
$\Env$
of the form $x_1 : \sort_1,\dots,x_n : \sort_n$
where the variables $x_1,\dots,x_n$ are pairwise distinct.

The \emph{terms-in-context}
$\Env \sorting t : \sort$
in a signature $\Sig$
are inductively defined as follows:
\begin{itemize}
\item
$x_1:\sort_1,\dots,x_n:\sort_n \sorting x_i : \sort_i$
is a term-in-context for every $i =1,\dots,n$,

\item
$\Env \sorting f(t_1,\dots,t_n) : \sort$
is a term-in-context
if
$f \colon \sort_1 \dots \sort_n \to \sort$
and
if
$\Env \sorting t_i : \sort_i$
is a
term-in-context
for each $i = 1,\dots,n$.
\end{itemize}

We now turn to formulae.

%%%%%%%%%%%%%%%%%%%%%%%%%%%%%%%%%%%%%%%%%%%%%%%%%%%%%%%%%%%%%%%%%%%%%%%%%%%
\begin{definition}[Formulae]
\label{def:app:prelim:formulae}
%%%%%%%%%%%%%%%%%%%%%%%%%%%%%%%%%%%%%%%%%%%%%%%%%%%%%%%%%%%%%%%%%%%%%%%%%%%
Let $\Sig$ be a signature.
The many-sorted language
$\Lang_\infty(\Sig) = \Lang_{\infty} (= \Lang_{\infty,\omega})$
consists of the \emph{formulae-in-context}
$\Env \sorting \varphi$
inductively defined as follows.
\begin{itemize}
\item
$\Env \sorting (t \Eq_\sort u)$ is a formula-in-context
if
$\Env \sorting t : \sort$ and $\Env \sorting u : \sort$
are terms-in-context.

\item
$\Env \sorting R(t_1,\dots,t_n)$ is a formula-in-context
if
$R \into \sort_1 \dots \sort_n$
and
if
$\Env \sorting t_i : \sort_i$ 
is a term-in-context for each $i = 1,\dots,n$.

\item
$\Env \sorting \lnot \varphi$ is a formula-in-context
if
$\Env \sorting \varphi$
is a formula-in-context.

\item
$\Env \sorting (\forall x : \sort)\varphi$
and
$\Env \sorting (\exists x : \sort)\varphi$
are formulae-in-context
if
$\Env,x : \sort \sorting \varphi$
is a formula-in-context.

\item
$\Env \sorting \bigwedge_{i \in I}\varphi_i$
and
$\Env \sorting \bigvee_{i \in I}\varphi_i$
are formulae-in-context
if $I$ is a small set and if
$\Env \sorting \varphi_i$
is a formula-in-context for each $i \in I$.
\end{itemize}
\end{definition}

\noindent
Note that $\Lang_\infty$ is a large set.
Note also that by construction, formulae-in-context
have at most finitely many free variables.
Finally, note that in case $\Sig$ is purely relational
(i.e.\ $\Fun(\Sig) = \emptyset$),
the notion of term-in-context still properly
handles the contexts of atomic formulae:
e.g.\
$\Env \sorting P(x_1,\dots,x_n)$
requires all the $x_i$ to be declared in $\Env$ with
appropriate sort,
but $\Env$ may declare additional variables.

Given a cardinal $\kappa$, we write $\Lang_\kappa (=\Lang_{\kappa,\omega})$
for the set of formulae $\varphi \in \Lang_\infty$ whose conjunctions
and disjunctions have cardinality $<\kappa$.
The following is a trivial consequence of the definitions.

%%%%%%%%%%%%%%%%%%%%%%%%%%%%%%%%%%%%%%%%%%%%%%%%%%%%%%%%%%%%%%%%%%%%%%%%%%%
\begin{fact}[$\AC$]
\label{fact:app:prelim:infty-kappa}
%%%%%%%%%%%%%%%%%%%%%%%%%%%%%%%%%%%%%%%%%%%%%%%%%%%%%%%%%%%%%%%%%%%%%%%%%%%
If $\varphi \in \Lang_\infty$,
then $\varphi \in \Lang_\kappa$ for some $\kappa$.
\end{fact}

%%%%%%%%%%%%%%%%%%%%%%%%%%%%%%%%%%%%%%%%%%%%%%%%%%%%%%%%%%%%%%%%%%%%%%%%%%%
\begin{proof}
%%%%%%%%%%%%%%%%%%%%%%%%%%%%%%%%%%%%%%%%%%%%%%%%%%%%%%%%%%%%%%%%%%%%%%%%%%%
By induction on $\varphi$.
We only discuss the case of $\varphi = \bigwedge_{i \in I} \varphi_i$.
By induction hypothesis, for each $i \in I$
there is a cardinal $\kappa_i > 0$ such that $\varphi_i \in \Lang_{\kappa_i}$.
Using the Axiom of Choice (see e.g.~\cite[\S 5]{jech06set}),
let $\kappa$ be the cardinal of $\sum_{i \in I}\kappa_i$.
Then $\varphi_i \in \Lang_{\kappa}$ for all $i \in I$.
Since $I$ has cardinality $\leq \kappa$,
we get $\varphi \in \Lang_\kappa$.
\end{proof}}

% that's all folks
\end{document}